\documentclass[oneside,french,english]{amsart}
\usepackage[T1]{fontenc}
\usepackage[latin9]{inputenc}
\usepackage{color}
\usepackage{babel}
\addto\extrasfrench{\providecommand{\fg}{\ifdim\lastskip>\z@\unskip\fi~\frqq}}

\usepackage{amsthm}
\usepackage{amsbsy}
\usepackage{amstext}
\makeindex
\usepackage{amssymb}
\usepackage{esint}
\usepackage{nomencl}
\providecommand{\makenomenclature}{\makeglossary}
\makenomenclature
\usepackage[unicode=true, 
 bookmarks=true,bookmarksnumbered=false,bookmarksopen=false,
 breaklinks=false,pdfborder={0 0 0},backref=false,colorlinks=true]
 {hyperref}
\hypersetup{pdftitle={A derivation of the linear Boltzmann equation},
 pdfauthor={Sébastien Breteaux}}

\makeatletter

\numberwithin{equation}{section}
\numberwithin{figure}{section}
\theoremstyle{plain}
\newtheorem{thm}{Theorem}[section]
  \theoremstyle{plain}
  \newtheorem*{assumption*}{Assumption}
  \theoremstyle{remark}
  \newtheorem{rem}[thm]{Remark}
  \theoremstyle{definition}
  \newtheorem{defn}[thm]{Definition}
  \theoremstyle{plain}
  \newtheorem{prop}[thm]{Proposition}
  \theoremstyle{remark}
  \newtheorem{notation}[thm]{Notation}
  \theoremstyle{plain}
  \newtheorem{lem}[thm]{Lemma}
  \theoremstyle{definition}
  
  \theoremstyle{plain}
  
  \theoremstyle{remark}
  \newtheorem*{acknowledgement*}{Acknowledgement}

\newcommand{\Id}{\mathrm{Id}}

\DeclareMathOperator*{\Tr}{Tr}
\DeclareMathOperator*{\Ad}{Ad}

\DeclareMathOperator*{\Symb}{Symb}
\DeclareMathOperator*{\proj}{proj}
\DeclareMathOperator*{\Supp}{Supp}

\DeclareMathOperator*{\sign}{sign}
\DeclareMathOperator*{\diff}{d\!}

\def\db{{\mathchar'26\mkern-12mu d}} 
\DeclareMathOperator*{\dbar}{\db\!}
\usepackage{mathrsfs}

\makeatother

\begin{document}
% fig4tex.tex  Version 1.8.4, May 5, 2007
% Copyright 2001-2007 Yvon Lafranche
%
% This material is subject to the LaTeX Project Public License.
% See http://www.ctan.org/tex-archive/help/Catalogue/licenses.lppl.html
% for the details of that license.
%
% Authors:  Yvon Lafranche, Daniel Martin
%           IRMAR, Universite de Rennes 1 - France
%           Yvon.Lafranche@univ-rennes1.fr
%           Daniel.Martin@univ-rennes1.fr
%
% Web Site: http://perso.univ-rennes1.fr/yvon.lafranche/fig4tex
%
%%%%%%%%%%%%%%%%%%%%%%%%%%%%%%%%%%%%%%%%%%%%%%%%%%%%%%%%%%%%%%%%%%%%%%%%%%%%%%%
\ifx\figforTeXisloaded\relax \else\global\let\figforTeXisloaded=\relax\fi
\message{version 1.8.4}
\catcode`\@=11
\ifx\ctr@ln@m\undefined\else%
    \immediate\write16{*** Fig4TeX WARNING : \string\ctr@ln@m\space already defined.}\fi
% Warns if \ControlSequence is already defined. The test rely on the control sequence
% \undefined which is supposed not to be defined.
% Nota: The test fills TeX's memory since the \ControlSequence passed as argument is
%       put into the tables strings, string characters and multiletter control sequences.
%       This not a problem since we need those control sequences.
%       Outside a macro, we could use the fact that "TeX does not put undefined control
%       sequences into its internal tables if they follow \ifx" (cf The TeXBook, p 384),
%       which does not seem possible here, in the context of a macro call.
% Appel (interne) : \ctr@ln@m{\ControlSequence}
\def\ctr@ln@m#1{\ifx#1\undefined\else%
    \immediate\write16{*** Fig4TeX WARNING : \string#1 already defined.}\fi}
% Apply \Command to \ControlSequence after having tested if \ControlSequence
% is already defined, in which case a message is issued.
% \Command is typically one of the commands used to define a control sequence
% namely \def, \edef, \gdef, \xdef, \chardef, \mathchardef, \let.
% Appel (interne) : \ctr@ld@f{\Command}{\ControlSequence}
\ctr@ln@m\ctr@ld@f
\def\ctr@ld@f#1#2{\ctr@ln@m#2#1#2}
% Apply \csname Command\endcsname to \ControlSequence after having tested if
% \ControlSequence is already defined, in which case a message is issued.
% Command is typically the name of the commands (without the leading \) used
% to reserve a register, namely newbox, newcount, newdimen, newif, newtoks,
% newread, newwrite.
% It is analogous to \ctr@ld@f, but \ctr@ld@f cannot be used for this
% purpose since the commands \newbox, \newcount... cannot be passed as argument.
% Appel (interne) : \ctr@ln@w{Command}{\ControlSequence}
\ctr@ld@f\def\ctr@ln@w#1#2{\ctr@ln@m#2\csname#1\endcsname#2}
{\catcode`\/=0 \catcode`/\=12 /ctr@ld@f/gdef/BS@{\}}
% Warns if \csname Text\endcsname is already defined. Used to test internal macros
% like points and associated text.
% Nota: The test is based on the fact that the expansion of \csname ...\endcsname
%       is "defined to be like \relax if its meaning is currently undefined" (cf.
%       The TeXBook, p 40 and 213). However, by itself, the expansion of \csname
%       ...\endcsname defines the macro and fills TeX's memory, so that after
%       \ctr@lcsn@m{Text}, the command \ctr@ln@m{\Text} will warn.
% Appel (interne) : \ctr@lcsn@m{Text}
%                   \ctr@lcsn@m{\ControlSequence} with \ControlSequence expands to Text
\ctr@ld@f\def\ctr@lcsn@m#1{\expandafter\ifx\csname#1\endcsname\relax\else%
    \immediate\write16{*** Fig4TeX WARNING : \BS@\expandafter\string#1\space already defined.}\fi}
\ctr@ld@f\edef\colonc@tcode{\the\catcode`\:}
\ctr@ld@f\edef\semicolonc@tcode{\the\catcode`\;}
\ctr@ld@f\def\t@stc@tcodech@nge{{\let\c@tcodech@nged=\z@%
    \ifnum\colonc@tcode=\the\catcode`\:\else\let\c@tcodech@nged=\@ne\fi%
    \ifnum\semicolonc@tcode=\the\catcode`\;\else\let\c@tcodech@nged=\@ne\fi%
    \ifx\c@tcodech@nged\@ne%
    \immediate\write16{}
    \immediate\write16{!!!=============================================================!!!}
    \immediate\write16{ Fig4TeX WARNING :}
    \immediate\write16{ The category code of some characters has been changed, which will}
    \immediate\write16{ result in an error (message "Runaway argument?").}
    \immediate\write16{ This probably comes from another package that changed the category}
    \immediate\write16{ code after Fig4TeX was loaded. If that proves to be exact, the}
    \immediate\write16{ solution is to exchange the loading commands on top of your file}
    \immediate\write16{ so that Fig4TeX is loaded last. For example, in LaTeX, we should}
    \immediate\write16{ say :}
    \immediate\write16{\BS@ usepackage[french]{babel}}
    \immediate\write16{\BS@ usepackage{fig4tex}}
    \immediate\write16{!!!=============================================================!!!}
    \immediate\write16{}
    \fi}}
% Fig4TeX logo
\ctr@ld@f\def\FigforTeX{F\kern-.05em i\kern-.05em g\kern-.1em\raise-.14em\hbox{4}\kern-.19em\TeX}
%%%%%%%%%%%%%%%%%%%%%%%%%%%%%%%%%%%%%%%%%%%%%%%%%%%%%%%%%%%%%%%%%%%%%%%%%%%%%%%
% Points with numbers >= 0 are devoted to the user.
% Points with numbers <  0 are reserved to internal use.
%%%%%%%%%%%%%%%%%%%%%%%%%%%%%%%%%%%%%%%%%%%%%%%%%%%%%%%%%%%%%%%%%%%%%%%%%%%%%%%
\ctr@ln@w{newdimen}\epsil@n\epsil@n=0.00005pt
\ctr@ln@w{newdimen}\Cepsil@n\Cepsil@n=0.005pt
\ctr@ln@w{newdimen}\dcq@\dcq@=254pt
\ctr@ln@w{newdimen}\PI@\PI@=3.141592pt
\ctr@ln@w{newdimen}\DemiPI@deg\DemiPI@deg=90pt
\ctr@ln@w{newdimen}\PI@deg\PI@deg=180pt
\ctr@ln@w{newdimen}\DePI@deg\DePI@deg=360pt
\ctr@ld@f\chardef\t@n=10
\ctr@ld@f\chardef\c@nt=100
\ctr@ld@f\chardef\@lxxiv=74
\ctr@ld@f\chardef\@xci=91
\ctr@ld@f\mathchardef\@nMnCQn=9949
\ctr@ld@f\chardef\@vi=6
\ctr@ld@f\chardef\@xxx=30
\ctr@ld@f\chardef\@lvi=56
\ctr@ld@f\chardef\@@lxxi=71
\ctr@ld@f\chardef\@lxxxv=85
\ctr@ld@f\mathchardef\@@mmmmlxviii=4068
\ctr@ld@f\mathchardef\@ccclx=360
\ctr@ld@f\mathchardef\@dccxx=720
\ctr@ln@w{newcount}\p@rtent \ctr@ln@w{newcount}\f@ctech \ctr@ln@w{newcount}\result@tent
\ctr@ln@w{newdimen}\v@lmin \ctr@ln@w{newdimen}\v@lmax \ctr@ln@w{newdimen}\v@leur
\ctr@ln@w{newdimen}\result@t\ctr@ln@w{newdimen}\result@@t
\ctr@ln@w{newdimen}\mili@u \ctr@ln@w{newdimen}\c@rre \ctr@ln@w{newdimen}\delt@
\ctr@ld@f\def\degT@rd{0.017453 }  % pi/180
\ctr@ld@f\def\rdT@deg{57.295779 } % 180/pi
\ctr@ln@m\v@leurseule
{\catcode`p=12 \catcode`t=12 \gdef\v@leurseule#1pt{#1}}
\ctr@ld@f\def\repdecn@mb#1{\expandafter\v@leurseule\the#1\space}
\ctr@ld@f\def\arct@n#1(#2,#3){{\v@lmin=#2\v@lmax=#3%
    \maxim@m{\mili@u}{-\v@lmin}{\v@lmin}\maxim@m{\c@rre}{-\v@lmax}{\v@lmax}%
    \delt@=\mili@u\m@ech\mili@u%
    \ifdim\c@rre>\@nMnCQn\mili@u\divide\v@lmax\tw@\c@lATAN\v@leur(\z@,\v@lmax)% DY > 9949 DX
    \else%
    \maxim@m{\mili@u}{-\v@lmin}{\v@lmin}\maxim@m{\c@rre}{-\v@lmax}{\v@lmax}%
    \m@ech\c@rre%
    \ifdim\mili@u>\@nMnCQn\c@rre\divide\v@lmin\tw@% DX > 9949 DY
    \maxim@m{\mili@u}{-\v@lmin}{\v@lmin}\c@lATAN\v@leur(\mili@u,\z@)%
    \else\c@lATAN\v@leur(\delt@,\v@lmax)\fi\fi%
    \ifdim\v@lmin<\z@\v@leur=-\v@leur\ifdim\v@lmax<\z@\advance\v@leur-\PI@%
    \else\advance\v@leur\PI@\fi\fi%
    \global\result@t=\v@leur}#1=\result@t}
\ctr@ld@f\def\m@ech#1{\ifdim#1>1.646pt\divide\mili@u\t@n\divide\c@rre\t@n\m@ech#1\fi}
\ctr@ld@f\def\c@lATAN#1(#2,#3){{\v@lmin=#2\v@lmax=#3\v@leur=\z@\delt@=\tw@ pt%
    \un@iter{0.785398}{\v@lmax<}%
    \un@iter{0.463648}{\v@lmax<}%
    \un@iter{0.244979}{\v@lmax<}%
    \un@iter{0.124355}{\v@lmax<}%
    \un@iter{0.062419}{\v@lmax<}%
    \un@iter{0.031240}{\v@lmax<}%
    \un@iter{0.015624}{\v@lmax<}%
    \un@iter{0.007812}{\v@lmax<}%
    \un@iter{0.003906}{\v@lmax<}%
    \un@iter{0.001953}{\v@lmax<}%
    \un@iter{0.000976}{\v@lmax<}%
    \un@iter{0.000488}{\v@lmax<}%
    \un@iter{0.000244}{\v@lmax<}%
    \un@iter{0.000122}{\v@lmax<}%
    \un@iter{0.000061}{\v@lmax<}%
    \un@iter{0.000030}{\v@lmax<}%
    \un@iter{0.000015}{\v@lmax<}%
    \global\result@t=\v@leur}#1=\result@t}
\ctr@ld@f\def\un@iter#1#2{%
    \divide\delt@\tw@\edef\dpmn@{\repdecn@mb{\delt@}}%
    \mili@u=\v@lmin%
    \ifdim#2\z@%
      \advance\v@lmin-\dpmn@\v@lmax\advance\v@lmax\dpmn@\mili@u%
      \advance\v@leur-#1pt%
    \else%
      \advance\v@lmin\dpmn@\v@lmax\advance\v@lmax-\dpmn@\mili@u%
      \advance\v@leur#1pt%
    \fi}
\ctr@ld@f\def\c@ssin#1#2#3{\expandafter\ifx\csname COS@\number#3\endcsname\relax\c@lCS{#3pt}%
    \expandafter\xdef\csname COS@\number#3\endcsname{\repdecn@mb\result@t}%
    \expandafter\xdef\csname SIN@\number#3\endcsname{\repdecn@mb\result@@t}\fi%
    \edef#1{\csname COS@\number#3\endcsname}\edef#2{\csname SIN@\number#3\endcsname}}
\ctr@ld@f\def\c@lCS#1{{\mili@u=#1\p@rtent=\@ne%
    \relax\ifdim\mili@u<\z@\red@ng<-\else\red@ng>+\fi\f@ctech=\p@rtent%
    \relax\ifdim\mili@u<\z@\mili@u=-\mili@u\f@ctech=-\f@ctech\fi\c@@lCS}}
\ctr@ld@f\def\c@@lCS{\v@lmin=\mili@u\c@rre=-\mili@u\advance\c@rre\DemiPI@deg\v@lmax=\c@rre%
    \mili@u\@@lxxi\mili@u\divide\mili@u\@@mmmmlxviii%
    \edef\v@larg{\repdecn@mb{\mili@u}}\mili@u=-\v@larg\mili@u%
    \edef\v@lmxde{\repdecn@mb{\mili@u}}%
    \c@rre\@@lxxi\c@rre\divide\c@rre\@@mmmmlxviii%
    \edef\v@largC{\repdecn@mb{\c@rre}}\c@rre=-\v@largC\c@rre%
    \edef\v@lmxdeC{\repdecn@mb{\c@rre}}%
    \fctc@s\mili@u\v@lmin\global\result@t\p@rtent\v@leur%
    \let\t@mp=\v@larg\let\v@larg=\v@largC\let\v@largC=\t@mp%
    \let\t@mp=\v@lmxde\let\v@lmxde=\v@lmxdeC\let\v@lmxdeC=\t@mp%
    \fctc@s\c@rre\v@lmax\global\result@@t\f@ctech\v@leur}
\ctr@ld@f\def\fctc@s#1#2{\v@leur=#1\relax\ifdim#2<\@lxxxv\p@\cosser@h\else\sinser@t\fi}
\ctr@ld@f\def\cosser@h{\advance\v@leur\@lvi\p@\divide\v@leur\@lvi%
    \v@leur=\v@lmxde\v@leur\advance\v@leur\@xxx\p@%
    \v@leur=\v@lmxde\v@leur\advance\v@leur\@ccclx\p@%
    \v@leur=\v@lmxde\v@leur\advance\v@leur\@dccxx\p@\divide\v@leur\@dccxx}
\ctr@ld@f\def\sinser@t{\v@leur=\v@lmxdeC\p@\advance\v@leur\@vi\p@%
    \v@leur=\v@largC\v@leur\divide\v@leur\@vi}
\ctr@ld@f\def\red@ng#1#2{\relax\ifdim\mili@u#1#2\DemiPI@deg\advance\mili@u#2-\PI@deg%
    \p@rtent=-\p@rtent\red@ng#1#2\fi}
\ctr@ld@f\def\pr@c@lCS#1#2#3{\ctr@lcsn@m{COS@\number#3 }%
    \expandafter\xdef\csname COS@\number#3\endcsname{#1}%
    \expandafter\xdef\csname SIN@\number#3\endcsname{#2}}
\pr@c@lCS{1}{0}{0}
\pr@c@lCS{0.7071}{0.7071}{45}\pr@c@lCS{0.7071}{-0.7071}{-45}
\pr@c@lCS{0}{1}{90}          \pr@c@lCS{0}{-1}{-90}
\pr@c@lCS{-1}{0}{180}        \pr@c@lCS{-1}{0}{-180}
\pr@c@lCS{0}{-1}{270}        \pr@c@lCS{0}{1}{-270}
\ctr@ld@f\def\invers@#1#2{{\v@leur=#2\maxim@m{\v@lmax}{-\v@leur}{\v@leur}%
    \f@ctech=\@ne\m@inv@rs%
    \multiply\v@leur\f@ctech\edef\v@lv@leur{\repdecn@mb{\v@leur}}%
    \p@rtentiere{\p@rtent}{\v@leur}\v@lmin=\p@\divide\v@lmin\p@rtent%
    \inv@rs@\multiply\v@lmax\f@ctech\global\result@t=\v@lmax}#1=\result@t}
\ctr@ld@f\def\m@inv@rs{\ifdim\v@lmax<\p@\multiply\v@lmax\t@n\multiply\f@ctech\t@n\m@inv@rs\fi}
\ctr@ld@f\def\inv@rs@{\v@lmax=-\v@lmin\v@lmax=\v@lv@leur\v@lmax%
    \advance\v@lmax\tw@ pt\v@lmax=\repdecn@mb{\v@lmin}\v@lmax%
    \delt@=\v@lmax\advance\delt@-\v@lmin\ifdim\delt@<\z@\delt@=-\delt@\fi%
    \ifdim\delt@>\epsil@n\v@lmin=\v@lmax\inv@rs@\fi}
\ctr@ld@f\def\minim@m#1#2#3{\relax\ifdim#2<#3#1=#2\else#1=#3\fi}
\ctr@ld@f\def\maxim@m#1#2#3{\relax\ifdim#2>#3#1=#2\else#1=#3\fi}
\ctr@ld@f\def\p@rtentiere#1#2{#1=#2\divide#1by65536 }
\ctr@ld@f\def\r@undint#1#2{{\v@leur=#2\divide\v@leur\t@n\p@rtentiere{\p@rtent}{\v@leur}%
    \v@leur=\p@rtent pt\global\result@t=\t@n\v@leur}#1=\result@t}
\ctr@ld@f\def\sqrt@#1#2{{\v@leur=#2%
    \minim@m{\v@lmin}{\p@}{\v@leur}\maxim@m{\v@lmax}{\p@}{\v@leur}%
    \f@ctech=\@ne\m@sqrt@\sqrt@@%
    \mili@u=\v@lmin\advance\mili@u\v@lmax\divide\mili@u\tw@\multiply\mili@u\f@ctech%
    \global\result@t=\mili@u}#1=\result@t}
\ctr@ld@f\def\m@sqrt@{\ifdim\v@leur>\dcq@\divide\v@leur\c@nt\v@lmax=\v@leur%
    \multiply\f@ctech\t@n\m@sqrt@\fi}
\ctr@ld@f\def\sqrt@@{\mili@u=\v@lmin\advance\mili@u\v@lmax\divide\mili@u\tw@%
    \c@rre=\repdecn@mb{\mili@u}\mili@u%
    \ifdim\c@rre<\v@leur\v@lmin=\mili@u\else\v@lmax=\mili@u\fi%
    \delt@=\v@lmax\advance\delt@-\v@lmin\ifdim\delt@>\epsil@n\sqrt@@\fi}
\ctr@ld@f\def\extrairelepremi@r#1\de#2{\expandafter\lepremi@r#2@#1#2}
\ctr@ld@f\def\lepremi@r#1,#2@#3#4{\def#3{#1}\def#4{#2}\ignorespaces}
\ctr@ld@f\def\@cfor#1:=#2\do#3{%
  \edef\@fortemp{#2}%
  \ifx\@fortemp\empty\else\@cforloop#2,\@nil,\@nil\@@#1{#3}\fi}
\ctr@ln@m\@nextwhile
\ctr@ld@f\def\@cforloop#1,#2\@@#3#4{%
  \def#3{#1}%
  \ifx#3\Fig@nnil\let\@nextwhile=\Fig@fornoop\else#4\relax\let\@nextwhile=\@cforloop\fi%
  \@nextwhile#2\@@#3{#4}}

\ctr@ld@f\def\@ecfor#1:=#2\do#3{%
  \def\@@cfor{\@cfor#1:=}%
  \edef\@@@cfor{#2}%
  \expandafter\@@cfor\@@@cfor\do{#3}}
\ctr@ld@f\def\Fig@nnil{\@nil}
\ctr@ld@f\def\Fig@fornoop#1\@@#2#3{}
\ctr@ln@m\list@@rg
\ctr@ld@f\def\trtlis@rg#1#2{\def\list@@rg{#1}%
    \@ecfor\p@rv@l:=\list@@rg\do{\expandafter#2\p@rv@l|}}
\ctr@ln@w{newbox}\b@xvisu
\ctr@ln@w{newtoks}\let@xte
\ctr@ln@w{newif}\ifitis@K
\ctr@ln@w{newcount}\s@mme
\ctr@ln@w{newcount}\l@mbd@un \ctr@ln@w{newcount}\l@mbd@de
\ctr@ln@w{newcount}\superc@ntr@l\superc@ntr@l=\@ne        % Controle impose
\ctr@ln@w{newcount}\typec@ntr@l\typec@ntr@l=\superc@ntr@l % Controle souhaite
\ctr@ln@w{newdimen}\v@lX  \ctr@ln@w{newdimen}\v@lY  \ctr@ln@w{newdimen}\v@lZ
\ctr@ln@w{newdimen}\v@lXa \ctr@ln@w{newdimen}\v@lYa \ctr@ln@w{newdimen}\v@lZa
\ctr@ln@w{newdimen}\unit@\unit@=\p@ % Initialisation a la valeur par defaut.
\ctr@ld@f\def\unit@util{pt}
\ctr@ld@f\def\ptT@ptps{0.996264 }
\ctr@ld@f\def\ptpsT@pt{1.00375 }
\ctr@ld@f\def\ptT@unit@{1} % Initialisation correspondant a la valeur par defaut de \unit@
\ctr@ld@f\def\setunit@#1{\def\unit@util{#1}\setunit@@#1:\invers@{\result@t}{\unit@}%
    \edef\ptT@unit@{\repdecn@mb\result@t}}
\ctr@ld@f\def\setunit@@#1#2:{\ifcat#1a\unit@=\@ne#1#2\else\unit@=#1#2\fi}
\ctr@ld@f\def\d@fm@cdim#1#2{{\v@leur=#2\v@leur=\ptT@unit@\v@leur\xdef#1{\repdecn@mb\v@leur}}}
\ctr@ln@w{newif}\ifBdingB@x\BdingB@xtrue
\ctr@ln@w{newdimen}\c@@rdXmin \ctr@ln@w{newdimen}\c@@rdYmin  % Dimensions de la BoundingBox
\ctr@ln@w{newdimen}\c@@rdXmax \ctr@ln@w{newdimen}\c@@rdYmax
\ctr@ld@f\def\b@undb@x#1#2{\ifBdingB@x%
    \relax\ifdim#1<\c@@rdXmin\global\c@@rdXmin=#1\fi%
    \relax\ifdim#2<\c@@rdYmin\global\c@@rdYmin=#2\fi%
    \relax\ifdim#1>\c@@rdXmax\global\c@@rdXmax=#1\fi%
    \relax\ifdim#2>\c@@rdYmax\global\c@@rdYmax=#2\fi\fi}
\ctr@ld@f\def\b@undb@xP#1{{\Figg@tXY{#1}\b@undb@x{\v@lX}{\v@lY}}}
\ctr@ld@f\def\ellBB@x#1;#2,#3(#4,#5,#6){{\s@uvc@ntr@l\et@tellBB@x%
    \setc@ntr@l{2}\figptell-2::#1;#2,#3(#4,#6)\b@undb@xP{-2}%
    \figptell-2::#1;#2,#3(#5,#6)\b@undb@xP{-2}%
    \c@ssin{\C@}{\S@}{#6}\v@lmin=\C@ pt\v@lmax=\S@ pt%
    \mili@u=#3\v@lmin\delt@=#2\v@lmax\arct@n\v@leur(\delt@,\mili@u)%
    \mili@u=-#3\v@lmax\delt@=#2\v@lmin\arct@n\c@rre(\delt@,\mili@u)%
    \v@leur=\rdT@deg\v@leur\advance\v@leur-\DePI@deg%
    \c@rre=\rdT@deg\c@rre\advance\c@rre-\DePI@deg%
    \v@lmin=#4pt\v@lmax=#5pt%
    \loop\ifdim\v@leur<\v@lmax\ifdim\v@leur>\v@lmin%
    \edef\@ngle{\repdecn@mb\v@leur}\figptell-2::#1;#2,#3(\@ngle,#6)%
    \b@undb@xP{-2}\fi\advance\v@leur\PI@deg\repeat%
    \loop\ifdim\c@rre<\v@lmax\ifdim\c@rre>\v@lmin%
    \edef\@ngle{\repdecn@mb\c@rre}\figptell-2::#1;#2,#3(\@ngle,#6)%
    \b@undb@xP{-2}\fi\advance\c@rre\PI@deg\repeat%
    \resetc@ntr@l\et@tellBB@x}\ignorespaces}
\ctr@ld@f\def\initb@undb@x{\c@@rdXmin=\maxdimen\c@@rdYmin=\maxdimen%
    \c@@rdXmax=-\maxdimen\c@@rdYmax=-\maxdimen}
\ctr@ld@f\def\c@ntr@lnum#1{%
    \relax\ifnum\typec@ntr@l=\@ne%
    \ifnum#1<\z@%
    \immediate\write16{*** Forbidden point number (#1). Abort.}\end\fi\fi%
    \set@bjc@de{#1}}
\ctr@ln@m\objc@de
\ctr@ld@f\def\set@bjc@de#1{\edef\objc@de{@BJ\ifnum#1<\z@ M\romannumeral-#1\else\romannumeral#1\fi}}
\s@mme=\m@ne\loop\ifnum\s@mme>-19
  \set@bjc@de{\s@mme}\ctr@lcsn@m\objc@de\ctr@lcsn@m{\objc@de T}
\advance\s@mme\m@ne\repeat
\s@mme=\@ne\loop\ifnum\s@mme<6
  \set@bjc@de{\s@mme}\ctr@lcsn@m\objc@de\ctr@lcsn@m{\objc@de T}
\advance\s@mme\@ne\repeat
\ctr@ld@f\def\setc@ntr@l#1{\ifnum\superc@ntr@l>#1\typec@ntr@l=\superc@ntr@l%
    \else\typec@ntr@l=#1\fi}
\ctr@ld@f\def\resetc@ntr@l#1{\global\superc@ntr@l=#1\setc@ntr@l{#1}}
\ctr@ld@f\def\s@uvc@ntr@l#1{\edef#1{\the\superc@ntr@l}}
\ctr@ln@m\c@lproscal
\ctr@ld@f\def\c@lproscalDD#1[#2,#3]{{\Figg@tXY{#2}%
    \edef\Xu@{\repdecn@mb{\v@lX}}\edef\Yu@{\repdecn@mb{\v@lY}}\Figg@tXY{#3}%
    \global\result@t=\Xu@\v@lX\global\advance\result@t\Yu@\v@lY}#1=\result@t}
\ctr@ld@f\def\c@lproscalTD#1[#2,#3]{{\Figg@tXY{#2}\edef\Xu@{\repdecn@mb{\v@lX}}%
    \edef\Yu@{\repdecn@mb{\v@lY}}\edef\Zu@{\repdecn@mb{\v@lZ}}%
    \Figg@tXY{#3}\global\result@t=\Xu@\v@lX\global\advance\result@t\Yu@\v@lY%
    \global\advance\result@t\Zu@\v@lZ}#1=\result@t}
\ctr@ld@f\def\c@lprovec#1{%
    \det@rmC\v@lZa(\v@lX,\v@lY,\v@lmin,\v@lmax)%
    \det@rmC\v@lXa(\v@lY,\v@lZ,\v@lmax,\v@leur)%
    \det@rmC\v@lYa(\v@lZ,\v@lX,\v@leur,\v@lmin)%
    \Figv@ctCreg#1(\v@lXa,\v@lYa,\v@lZa)}
\ctr@ld@f\def\det@rm#1[#2,#3]{{\Figg@tXY{#2}\Figg@tXYa{#3}%
    \delt@=\repdecn@mb{\v@lX}\v@lYa\advance\delt@-\repdecn@mb{\v@lY}\v@lXa%
    \global\result@t=\delt@}#1=\result@t}
\ctr@ld@f\def\det@rmC#1(#2,#3,#4,#5){{\global\result@t=\repdecn@mb{#2}#5%
    \global\advance\result@t-\repdecn@mb{#3}#4}#1=\result@t}
\ctr@ld@f\def\getredf@ctDD#1(#2,#3){{\maxim@m{\v@lXa}{-#2}{#2}\maxim@m{\v@lYa}{-#3}{#3}%
    \maxim@m{\v@lXa}{\v@lXa}{\v@lYa}% \v@lXa = ||X||inf
    \ifdim\v@lXa>\@xci pt\divide\v@lXa\@xci%
    \p@rtentiere{\p@rtent}{\v@lXa}\advance\p@rtent\@ne\else\p@rtent=\@ne\fi%
    \global\result@tent=\p@rtent}#1=\result@tent\ignorespaces}
\ctr@ld@f\def\getredf@ctTD#1(#2,#3,#4){{\maxim@m{\v@lXa}{-#2}{#2}\maxim@m{\v@lYa}{-#3}{#3}%
    \maxim@m{\v@lZa}{-#4}{#4}\maxim@m{\v@lXa}{\v@lXa}{\v@lYa}%
    \maxim@m{\v@lXa}{\v@lXa}{\v@lZa}% \v@lXa = ||X||inf
    \ifdim\v@lXa>\@lxxiv pt\divide\v@lXa\@lxxiv%
    \p@rtentiere{\p@rtent}{\v@lXa}\advance\p@rtent\@ne\else\p@rtent=\@ne\fi%
    \global\result@tent=\p@rtent}#1=\result@tent\ignorespaces}
\ctr@ld@f\def\FigptintercircB@zDD#1:#2:#3,#4[#5,#6,#7,#8]{{\s@uvc@ntr@l\et@tfigptintercircB@zDD%
    \setc@ntr@l{2}\figvectPDD-1[#5,#8]\Figg@tXY{-1}\getredf@ctDD\f@ctech(\v@lX,\v@lY)%
    \mili@u=#4\unit@\divide\mili@u\f@ctech\c@rre=\repdecn@mb{\mili@u}\mili@u%
    \figptBezierDD-5::#3[#5,#6,#7,#8]%
    \v@lmin=#3\p@\v@lmax=\v@lmin\advance\v@lmax0.1\p@%
    \loop\edef\T@{\repdecn@mb{\v@lmax}}\figptBezierDD-2::\T@[#5,#6,#7,#8]%
    \figvectPDD-1[-5,-2]\n@rmeucCDD{\delt@}{-1}\ifdim\delt@<\c@rre\v@lmin=\v@lmax%
    \advance\v@lmax0.1\p@\repeat%
    \loop\mili@u=\v@lmin\advance\mili@u\v@lmax%
    \divide\mili@u\tw@\edef\T@{\repdecn@mb{\mili@u}}\figptBezierDD-2::\T@[#5,#6,#7,#8]%
    \figvectPDD-1[-5,-2]\n@rmeucCDD{\delt@}{-1}\ifdim\delt@>\c@rre\v@lmax=\mili@u%
    \else\v@lmin=\mili@u\fi\v@leur=\v@lmax\advance\v@leur-\v@lmin%
    \ifdim\v@leur>\epsil@n\repeat\figptcopyDD#1:#2/-2/%
    \resetc@ntr@l\et@tfigptintercircB@zDD}\ignorespaces}
\ctr@ln@m\figptinterlines
\ctr@ld@f\def\inters@cDD#1:#2[#3,#4;#5,#6]{{\s@uvc@ntr@l\et@tinters@cDD%
    \setc@ntr@l{2}\vecunit@{-1}{#4}\vecunit@{-2}{#6}%
    \Figg@tXY{-1}\setc@ntr@l{1}\Figg@tXYa{#3}%
    \edef\A@{\repdecn@mb{\v@lX}}\edef\B@{\repdecn@mb{\v@lY}}%
    \v@lmin=\B@\v@lXa\advance\v@lmin-\A@\v@lYa%
    \Figg@tXYa{#5}\setc@ntr@l{2}\Figg@tXY{-2}%
    \edef\C@{\repdecn@mb{\v@lX}}\edef\D@{\repdecn@mb{\v@lY}}%
    \v@lmax=\D@\v@lXa\advance\v@lmax-\C@\v@lYa%
    \delt@=\A@\v@lY\advance\delt@-\B@\v@lX%
    \invers@{\v@leur}{\delt@}\edef\v@ldelta{\repdecn@mb{\v@leur}}%
    \v@lXa=\A@\v@lmax\advance\v@lXa-\C@\v@lmin%
    \v@lYa=\B@\v@lmax\advance\v@lYa-\D@\v@lmin%
    \v@lXa=\v@ldelta\v@lXa\v@lYa=\v@ldelta\v@lYa%
    \setc@ntr@l{1}\Figp@intregDD#1:{#2}(\v@lXa,\v@lYa)%
    \resetc@ntr@l\et@tinters@cDD}\ignorespaces}
\ctr@ld@f\def\inters@cTD#1:#2[#3,#4;#5,#6]{{\s@uvc@ntr@l\et@tinters@cTD%
    \setc@ntr@l{2}\figvectNVTD-1[#4,#6]\figvectNVTD-2[#6,-1]\figvectPTD-1[#3,#5]%
    \r@pPSTD\v@leur[-2,-1,#4]\edef\v@lcoef{\repdecn@mb{\v@leur}}%
    \figpttraTD#1:{#2}=#3/\v@lcoef,#4/\resetc@ntr@l\et@tinters@cTD}\ignorespaces}
\ctr@ld@f\def\r@pPSTD#1[#2,#3,#4]{{\Figg@tXY{#2}\edef\Xu@{\repdecn@mb{\v@lX}}%
    \edef\Yu@{\repdecn@mb{\v@lY}}\edef\Zu@{\repdecn@mb{\v@lZ}}%
    \Figg@tXY{#3}\v@lmin=\Xu@\v@lX\advance\v@lmin\Yu@\v@lY\advance\v@lmin\Zu@\v@lZ%
    \Figg@tXY{#4}\v@lmax=\Xu@\v@lX\advance\v@lmax\Yu@\v@lY\advance\v@lmax\Zu@\v@lZ%
    \invers@{\v@leur}{\v@lmax}\global\result@t=\repdecn@mb{\v@leur}\v@lmin}%
    #1=\result@t}
\ctr@ln@m\n@rminf
\ctr@ld@f\def\n@rminfDD#1#2{{\Figg@tXY{#2}\maxim@m{\v@lX}{\v@lX}{-\v@lX}%
    \maxim@m{\v@lY}{\v@lY}{-\v@lY}\maxim@m{\global\result@t}{\v@lX}{\v@lY}}%
    #1=\result@t}
\ctr@ld@f\def\n@rminfTD#1#2{{\Figg@tXY{#2}\maxim@m{\v@lX}{\v@lX}{-\v@lX}%
    \maxim@m{\v@lY}{\v@lY}{-\v@lY}\maxim@m{\v@lZ}{\v@lZ}{-\v@lZ}%
    \maxim@m{\v@lX}{\v@lX}{\v@lY}\maxim@m{\global\result@t}{\v@lX}{\v@lZ}}%
    #1=\result@t}
\ctr@ld@f\def\n@rmeucCDD#1#2{\Figg@tXY{#2}\divide\v@lX\f@ctech\divide\v@lY\f@ctech%
    #1=\repdecn@mb{\v@lX}\v@lX\v@lX=\repdecn@mb{\v@lY}\v@lY\advance#1\v@lX}
\ctr@ld@f\def\n@rmeucCTD#1#2{\Figg@tXY{#2}%
    \divide\v@lX\f@ctech\divide\v@lY\f@ctech\divide\v@lZ\f@ctech%
    #1=\repdecn@mb{\v@lX}\v@lX\v@lX=\repdecn@mb{\v@lY}\v@lY\advance#1\v@lX%
    \v@lX=\repdecn@mb{\v@lZ}\v@lZ\advance#1\v@lX}
\ctr@ln@m\n@rmeucSV
\ctr@ld@f\def\n@rmeucSVDD#1#2{{\Figg@tXY{#2}%
    \v@lXa=\repdecn@mb{\v@lX}\v@lX\v@lYa=\repdecn@mb{\v@lY}\v@lY%
    \advance\v@lXa\v@lYa\sqrt@{\global\result@t}{\v@lXa}}#1=\result@t}
\ctr@ld@f\def\n@rmeucSVTD#1#2{{\Figg@tXY{#2}\v@lXa=\repdecn@mb{\v@lX}\v@lX%
    \v@lYa=\repdecn@mb{\v@lY}\v@lY\v@lZa=\repdecn@mb{\v@lZ}\v@lZ%
    \advance\v@lXa\v@lYa\advance\v@lXa\v@lZa\sqrt@{\global\result@t}{\v@lXa}}#1=\result@t}
\ctr@ln@m\n@rmeuc
\ctr@ld@f\def\n@rmeucDD#1#2{{\Figg@tXY{#2}\getredf@ctDD\f@ctech(\v@lX,\v@lY)%
    \divide\v@lX\f@ctech\divide\v@lY\f@ctech%
    \v@lXa=\repdecn@mb{\v@lX}\v@lX\v@lYa=\repdecn@mb{\v@lY}\v@lY%
    \advance\v@lXa\v@lYa\sqrt@{\global\result@t}{\v@lXa}%
    \global\multiply\result@t\f@ctech}#1=\result@t}
\ctr@ld@f\def\n@rmeucTD#1#2{{\Figg@tXY{#2}\getredf@ctTD\f@ctech(\v@lX,\v@lY,\v@lZ)%
    \divide\v@lX\f@ctech\divide\v@lY\f@ctech\divide\v@lZ\f@ctech%
    \v@lXa=\repdecn@mb{\v@lX}\v@lX%
    \v@lYa=\repdecn@mb{\v@lY}\v@lY\v@lZa=\repdecn@mb{\v@lZ}\v@lZ%
    \advance\v@lXa\v@lYa\advance\v@lXa\v@lZa\sqrt@{\global\result@t}{\v@lXa}%
    \global\multiply\result@t\f@ctech}#1=\result@t}
\ctr@ln@m\vecunit@
\ctr@ld@f\def\vecunit@DD#1#2{{\Figg@tXY{#2}\getredf@ctDD\f@ctech(\v@lX,\v@lY)%
    \divide\v@lX\f@ctech\divide\v@lY\f@ctech%
    \Figv@ctCreg#1(\v@lX,\v@lY)\n@rmeucSV{\v@lYa}{#1}%
    \invers@{\v@lXa}{\v@lYa}\edef\v@lv@lXa{\repdecn@mb{\v@lXa}}%
    \v@lX=\v@lv@lXa\v@lX\v@lY=\v@lv@lXa\v@lY%
    \Figv@ctCreg#1(\v@lX,\v@lY)\multiply\v@lYa\f@ctech\global\result@t=\v@lYa}}
\ctr@ld@f\def\vecunit@TD#1#2{{\Figg@tXY{#2}\getredf@ctTD\f@ctech(\v@lX,\v@lY,\v@lZ)%
    \divide\v@lX\f@ctech\divide\v@lY\f@ctech\divide\v@lZ\f@ctech%
    \Figv@ctCreg#1(\v@lX,\v@lY,\v@lZ)\n@rmeucSV{\v@lYa}{#1}%
    \invers@{\v@lXa}{\v@lYa}\edef\v@lv@lXa{\repdecn@mb{\v@lXa}}%
    \v@lX=\v@lv@lXa\v@lX\v@lY=\v@lv@lXa\v@lY\v@lZ=\v@lv@lXa\v@lZ%
    \Figv@ctCreg#1(\v@lX,\v@lY,\v@lZ)\multiply\v@lYa\f@ctech\global\result@t=\v@lYa}}
\ctr@ld@f\def\vecunitC@TD[#1,#2]{\Figg@tXYa{#1}\Figg@tXY{#2}%
    \advance\v@lX-\v@lXa\advance\v@lY-\v@lYa\advance\v@lZ-\v@lZa\c@lvecunitTD}
\ctr@ld@f\def\vecunitCV@TD#1{\Figg@tXY{#1}\c@lvecunitTD}
\ctr@ld@f\def\c@lvecunitTD{\getredf@ctTD\f@ctech(\v@lX,\v@lY,\v@lZ)%
    \divide\v@lX\f@ctech\divide\v@lY\f@ctech\divide\v@lZ\f@ctech%
    \v@lXa=\repdecn@mb{\v@lX}\v@lX%
    \v@lYa=\repdecn@mb{\v@lY}\v@lY\v@lZa=\repdecn@mb{\v@lZ}\v@lZ%
    \advance\v@lXa\v@lYa\advance\v@lXa\v@lZa\sqrt@{\v@lYa}{\v@lXa}%
    \invers@{\v@lXa}{\v@lYa}\edef\v@lv@lXa{\repdecn@mb{\v@lXa}}%
    \v@lX=\v@lv@lXa\v@lX\v@lY=\v@lv@lXa\v@lY\v@lZ=\v@lv@lXa\v@lZ}
\ctr@ln@m\figgetangle
\ctr@ld@f\def\figgetangleDD#1[#2,#3,#4]{\ifps@cri{\s@uvc@ntr@l\et@tfiggetangleDD\setc@ntr@l{2}%
    \figvectPDD-1[#2,#3]\figvectPDD-2[#2,#4]\vecunit@{-1}{-1}%
    \c@lproscalDD\delt@[-2,-1]\figvectNVDD-1[-1]\c@lproscalDD\v@leur[-2,-1]%
    \arct@n\v@lmax(\delt@,\v@leur)\v@lmax=\rdT@deg\v@lmax%
    \ifdim\v@lmax<\z@\advance\v@lmax\DePI@deg\fi\xdef#1{\repdecn@mb{\v@lmax}}%
    \resetc@ntr@l\et@tfiggetangleDD}\ignorespaces\fi}
\ctr@ld@f\def\figgetangleTD#1[#2,#3,#4,#5]{\ifps@cri{\s@uvc@ntr@l\et@tfiggetangleTD\setc@ntr@l{2}%
    \figvectPTD-1[#2,#3]\figvectPTD-2[#2,#5]\figvectNVTD-3[-1,-2]%
    \figvectPTD-2[#2,#4]\figvectNVTD-4[-3,-1]%
    \vecunit@{-1}{-1}\c@lproscalTD\delt@[-2,-1]\c@lproscalTD\v@leur[-2,-4]%
    \arct@n\v@lmax(\delt@,\v@leur)\v@lmax=\rdT@deg\v@lmax%
    \ifdim\v@lmax<\z@\advance\v@lmax\DePI@deg\fi\xdef#1{\repdecn@mb{\v@lmax}}%
    \resetc@ntr@l\et@tfiggetangleTD}\ignorespaces\fi}    
\ctr@ld@f\def\figgetdist#1[#2,#3]{\ifps@cri{\s@uvc@ntr@l\et@tfiggetdist\setc@ntr@l{2}%
    \figvectP-1[#2,#3]\n@rmeuc{\v@lX}{-1}\v@lX=\ptT@unit@\v@lX\xdef#1{\repdecn@mb{\v@lX}}%
    \resetc@ntr@l\et@tfiggetdist}\ignorespaces\fi}
\ctr@ld@f\def\Figg@tT#1{\c@ntr@lnum{#1}%
    {\expandafter\expandafter\expandafter\extr@ctT\csname\objc@de\endcsname:%
     \ifnum\B@@ltxt=\z@\ptn@me{#1}\else\csname\objc@de T\endcsname\fi}}
\ctr@ld@f\def\extr@ctT#1,#2,#3/#4:{\def\B@@ltxt{#3}}
\ctr@ld@f\def\Figg@tXY#1{\c@ntr@lnum{#1}%
    \expandafter\expandafter\expandafter\extr@ctC\csname\objc@de\endcsname:}
\ctr@ln@m\extr@ctC
\ctr@ld@f\def\extr@ctCDD#1/#2,#3,#4:{\v@lX=#2\v@lY=#3}
\ctr@ld@f\def\extr@ctCTD#1/#2,#3,#4:{\v@lX=#2\v@lY=#3\v@lZ=#4}
\ctr@ld@f\def\Figg@tXYa#1{\c@ntr@lnum{#1}%
    \expandafter\expandafter\expandafter\extr@ctCa\csname\objc@de\endcsname:}
\ctr@ln@m\extr@ctCa
\ctr@ld@f\def\extr@ctCaDD#1/#2,#3,#4:{\v@lXa=#2\v@lYa=#3}
\ctr@ld@f\def\extr@ctCaTD#1/#2,#3,#4:{\v@lXa=#2\v@lYa=#3\v@lZa=#4}
\ctr@ln@m\t@xt@
\ctr@ld@f\def\figinit#1{\t@stc@tcodech@nge\initpr@lim\Figinit@#1,:\initpss@ttings\ignorespaces}
\ctr@ld@f\def\Figinit@#1,#2:{\setunit@{#1}\def\t@xt@{#2}\ifx\t@xt@\empty\else\Figinit@@#2:\fi}
\ctr@ld@f\def\Figinit@@#1#2:{\if#12 \else\Figs@tproj{#1}\initTD@\fi}
\ctr@ln@w{newif}\ifTr@isDim
\ctr@ld@f\def\UnD@fined{UNDEFINED}
\ctr@ld@f\def\ifundefined#1{\expandafter\ifx\csname#1\endcsname\relax}
\ctr@ln@m\@utoFN
\ctr@ln@m\@utoFInDone
\ctr@ln@m\disob@unit
\ctr@ld@f\def\initpr@lim{\initb@undb@x\figsetmark{}\figsetptname{$A_{##1}$}\def\Sc@leFact{1}%
    \initDD@\figsetroundcoord{yes}\ps@critrue\expandafter\setupd@te\defaultupdate:%
    \edef\disob@unit{\UnD@fined}\edef\t@rgetpt{\UnD@fined}\gdef\@utoFInDone{1}\gdef\@utoFN{0}}
\ctr@ld@f\def\initDD@{\Tr@isDimfalse%
    \ifPDFm@ke%
     \let\Ps@rcerc=\Ps@rcercBz%
     \let\Ps@rell=\Ps@rellBz%
    \fi
    \let\c@lDCUn=\c@lDCUnDD%
    \let\c@lDCDeux=\c@lDCDeuxDD%
    \let\c@ldefproj=\relax%
    \let\c@lproscal=\c@lproscalDD%
    \let\c@lprojSP=\relax%
    \let\extr@ctC=\extr@ctCDD%
    \let\extr@ctCa=\extr@ctCaDD%
    \let\extr@ctCF=\extr@ctCFDD%
    \let\Figp@intreg=\Figp@intregDD%
    \let\Figpts@xes=\Figpts@xesDD%
    \let\n@rmeucSV=\n@rmeucSVDD\let\n@rmeuc=\n@rmeucDD\let\n@rminf=\n@rminfDD%
    \let\pr@dMatV=\pr@dMatVDD%
    \let\ps@xes=\ps@xesDD%
    \let\vecunit@=\vecunit@DD%
    \let\figcoord=\figcoordDD%
    \let\figgetangle=\figgetangleDD%
    \let\figpt=\figptDD%
    \let\figptBezier=\figptBezierDD%
    \let\figptbary=\figptbaryDD%
    \let\figptcirc=\figptcircDD%
    \let\figptcircumcenter=\figptcircumcenterDD%
    \let\figptcopy=\figptcopyDD%
    \let\figptcurvcenter=\figptcurvcenterDD%
    \let\figptell=\figptellDD%
    \let\figptendnormal=\figptendnormalDD%
    \let\figptinterlineplane=\figptinterlineplaneDD%
    \let\figptinterlines=\inters@cDD%
    \let\figptorthocenter=\figptorthocenterDD%
    \let\figptorthoprojline=\figptorthoprojlineDD%
    \let\figptorthoprojplane=\figptorthoprojplaneDD%
    \let\figptrot=\figptrotDD%
    \let\figptscontrol=\figptscontrolDD%
    \let\figptsintercirc=\figptsintercircDD%
    \let\figptsinterlinell=\figptsinterlinellDD%
    \let\figptsorthoprojline=\figptsorthoprojlineDD%
    \let\figptorthoprojplane=\figptorthoprojplaneDD%
    \let\figptsrot=\figptsrotDD%
    \let\figptssym=\figptssymDD%
    \let\figptstra=\figptstraDD%
    \let\figptsym=\figptsymDD%
    \let\figpttraC=\figpttraCDD%
    \let\figpttra=\figpttraDD%
    \let\figptvisilimSL=\figptvisilimSLDD%
    \let\figsetobdist=\figsetobdistDD%
    \let\figsettarget=\figsettargetDD%
    \let\figsetview=\figsetviewDD%
    \let\figvectDBezier=\figvectDBezierDD%
    \let\figvectN=\figvectNDD%
    \let\figvectNV=\figvectNVDD%
    \let\figvectP=\figvectPDD%
    \let\figvectU=\figvectUDD%
    \let\psarccircP=\psarccircPDD%
    \let\psarccirc=\psarccircDD%
    \let\psarcell=\psarcellDD%
    \let\psarcellPA=\psarcellPADD%
    \let\psarrowBezier=\psarrowBezierDD%
    \let\psarrowcircP=\psarrowcircPDD%
    \let\psarrowcirc=\psarrowcircDD%
    \let\psarrowhead=\psarrowheadDD%
    \let\psarrow=\psarrowDD%
    \let\psBezier=\psBezierDD%
    \let\pscirc=\pscircDD%
    \let\pscurve=\pscurveDD%
    \let\psnormal=\psnormalDD%
    }
\ctr@ld@f\def\initTD@{\Tr@isDimtrue\initb@undb@xTD\newt@rgetptfalse\newdis@bfalse%
    \let\c@lDCUn=\c@lDCUnTD%
    \let\c@lDCDeux=\c@lDCDeuxTD%
    \let\c@ldefproj=\c@ldefprojTD%
    \let\c@lproscal=\c@lproscalTD%
    \let\extr@ctC=\extr@ctCTD%
    \let\extr@ctCa=\extr@ctCaTD%
    \let\extr@ctCF=\extr@ctCFTD%
    \let\Figp@intreg=\Figp@intregTD%
    \let\Figpts@xes=\Figpts@xesTD%
    \let\n@rmeucSV=\n@rmeucSVTD\let\n@rmeuc=\n@rmeucTD\let\n@rminf=\n@rminfTD%
    \let\pr@dMatV=\pr@dMatVTD%
    \let\ps@xes=\ps@xesTD%
    \let\vecunit@=\vecunit@TD%
    \let\figcoord=\figcoordTD%
    \let\figgetangle=\figgetangleTD%
    \let\figpt=\figptTD%
    \let\figptBezier=\figptBezierTD%
    \let\figptbary=\figptbaryTD%
    \let\figptcirc=\figptcircTD%
    \let\figptcircumcenter=\figptcircumcenterTD%
    \let\figptcopy=\figptcopyTD%
    \let\figptcurvcenter=\figptcurvcenterTD%
    \let\figptinterlineplane=\figptinterlineplaneTD%
    \let\figptinterlines=\inters@cTD%
    \let\figptorthocenter=\figptorthocenterTD%
    \let\figptorthoprojline=\figptorthoprojlineTD%
    \let\figptorthoprojplane=\figptorthoprojplaneTD%
    \let\figptrot=\figptrotTD%
    \let\figptscontrol=\figptscontrolTD%
    \let\figptsintercirc=\figptsintercircTD%
    \let\figptsorthoprojline=\figptsorthoprojlineTD%
    \let\figptsorthoprojplane=\figptsorthoprojplaneTD%
    \let\figptsrot=\figptsrotTD%
    \let\figptssym=\figptssymTD%
    \let\figptstra=\figptstraTD%
    \let\figptsym=\figptsymTD%
    \let\figpttraC=\figpttraCTD%
    \let\figpttra=\figpttraTD%
    \let\figptvisilimSL=\figptvisilimSLTD%
    \let\figsetobdist=\figsetobdistTD%
    \let\figsettarget=\figsettargetTD%
    \let\figsetview=\figsetviewTD%
    \let\figvectDBezier=\figvectDBezierTD%
    \let\figvectN=\figvectNTD%
    \let\figvectNV=\figvectNVTD%
    \let\figvectP=\figvectPTD%
    \let\figvectU=\figvectUTD%
    \let\psarccircP=\psarccircPTD%
    \let\psarccirc=\psarccircTD%
    \let\psarcell=\psarcellTD%
    \let\psarcellPA=\psarcellPATD%
    \let\psarrowBezier=\psarrowBezierTD%
    \let\psarrowcircP=\psarrowcircPTD%
    \let\psarrowcirc=\psarrowcircTD%
    \let\psarrowhead=\psarrowheadTD%
    \let\psarrow=\psarrowTD%
    \let\psBezier=\psBezierTD%
    \let\pscirc=\pscircTD%
    \let\pscurve=\pscurveTD%
    }
\ctr@ld@f\def\un@v@ilable#1{\immediate\write16{*** The macro #1 is not available in the current context.}}
\ctr@ld@f\def\figinsert#1{{\def\t@xt@{#1}\relax%
    \ifx\t@xt@\empty\ifnum\@utoFInDone>\z@\Figinsert@\DefGIfilen@me,:\fi%
    \else\expandafter\FiginsertNu@#1 :\fi}\ignorespaces}
\ctr@ld@f\def\FiginsertNu@#1 #2:{\def\t@xt@{#1}\relax\ifx\t@xt@\empty\def\t@xt@{#2}%
    \ifx\t@xt@\empty\ifnum\@utoFInDone>\z@\Figinsert@\DefGIfilen@me,:\fi%
    \else\FiginsertNu@#2:\fi\else\expandafter\FiginsertNd@#1 #2:\fi}
\ctr@ld@f\def\FiginsertNd@#1#2:{\ifcat#1a\Figinsert@#1#2,:\else%
    \ifnum\@utoFInDone>\z@\Figinsert@\DefGIfilen@me,#1#2,:\fi\fi}
\ctr@ln@m\Sc@leFact
\ctr@ld@f\def\Figinsert@#1,#2:{\def\t@xt@{#2}\ifx\t@xt@\empty\xdef\Sc@leFact{1}\else%
    \X@rgdeux@#2\xdef\Sc@leFact{\@rgdeux}\fi%
    \Figdisc@rdLTS{#1}{\t@xt@}\@psfgetbb{\t@xt@}%
    \v@lX=\@psfllx\p@\v@lX=\ptpsT@pt\v@lX\v@lX=\Sc@leFact\v@lX%
    \v@lY=\@psflly\p@\v@lY=\ptpsT@pt\v@lY\v@lY=\Sc@leFact\v@lY%
    \b@undb@x{\v@lX}{\v@lY}%
    \v@lX=\@psfurx\p@\v@lX=\ptpsT@pt\v@lX\v@lX=\Sc@leFact\v@lX%
    \v@lY=\@psfury\p@\v@lY=\ptpsT@pt\v@lY\v@lY=\Sc@leFact\v@lY%
    \b@undb@x{\v@lX}{\v@lY}%
    \ifPDFm@ke\Figinclud@PDF{\t@xt@}{\Sc@leFact}\else%
    \v@lX=\c@nt pt\v@lX=\Sc@leFact\v@lX\edef\F@ct{\repdecn@mb{\v@lX}}%
    \ifx\TeXturesonMacOSltX\special{postscriptfile #1 vscale=\F@ct\space hscale=\F@ct}%
    \else\includegraphics{#1}\fi\fi%
    \message{[\t@xt@]}\ignorespaces}
\ctr@ld@f\def\Figdisc@rdLTS#1#2{\expandafter\Figdisc@rdLTS@#1 :#2}
\ctr@ld@f\def\Figdisc@rdLTS@#1 #2:#3{\def#3{#1}\relax\ifx#3\empty\expandafter\Figdisc@rdLTS@#2:#3\fi}
\ctr@ld@f\def\figinsertE#1{\FiginsertE@#1,:\ignorespaces}
\ctr@ld@f\def\FiginsertE@#1,#2:{{\def\t@xt@{#2}\ifx\t@xt@\empty\xdef\Sc@leFact{1}\else%
    \X@rgdeux@#2\xdef\Sc@leFact{\@rgdeux}\fi%
    \Figdisc@rdLTS{#1}{\t@xt@}\pdfximage{\t@xt@}%
    \setbox\Gb@x=\hbox{\pdfrefximage\pdflastximage}%
    \v@lX=\z@\v@lY=-\Sc@leFact\dp\Gb@x\b@undb@x{\v@lX}{\v@lY}%
    \advance\v@lX\Sc@leFact\wd\Gb@x\advance\v@lY\Sc@leFact\dp\Gb@x%
    \advance\v@lY\Sc@leFact\ht\Gb@x\b@undb@x{\v@lX}{\v@lY}%
    \v@lX=\Sc@leFact\wd\Gb@x\pdfximage width \v@lX {\t@xt@}%
    \rlap{\pdfrefximage\pdflastximage}\message{[\t@xt@]}}\ignorespaces}
\ctr@ld@f\def\X@rgdeux@#1,{\edef\@rgdeux{#1}}
\ctr@ln@m\figpt
\ctr@ld@f\def\figptDD#1:#2(#3,#4){\ifps@cri\c@ntr@lnum{#1}%
    {\v@lX=#3\unit@\v@lY=#4\unit@\Fig@dmpt{#2}{\z@}}\ignorespaces\fi}
\ctr@ld@f\def\Fig@dmpt#1#2{\def\t@xt@{#1}\ifx\t@xt@\empty\def\B@@ltxt{\z@}%
    \else\expandafter\gdef\csname\objc@de T\endcsname{#1}\def\B@@ltxt{\@ne}\fi%
    \expandafter\xdef\csname\objc@de\endcsname{\ifitis@vect@r\C@dCl@svect%
    \else\C@dCl@spt\fi,\z@,\B@@ltxt/\the\v@lX,\the\v@lY,#2}}
\ctr@ld@f\def\C@dCl@spt{P}
\ctr@ld@f\def\C@dCl@svect{V}
\ctr@ln@m\c@@rdYZ
\ctr@ln@m\c@@rdY
\ctr@ld@f\def\figptTD#1:#2(#3,#4){\ifps@cri\c@ntr@lnum{#1}%
    \def\c@@rdYZ{#4,0,0}\extrairelepremi@r\c@@rdY\de\c@@rdYZ%
    \extrairelepremi@r\c@@rdZ\de\c@@rdYZ%
    {\v@lX=#3\unit@\v@lY=\c@@rdY\unit@\v@lZ=\c@@rdZ\unit@\Fig@dmpt{#2}{\the\v@lZ}%
    \b@undb@xTD{\v@lX}{\v@lY}{\v@lZ}}\ignorespaces\fi}
\ctr@ln@m\Figp@intreg
\ctr@ld@f\def\Figp@intregDD#1:#2(#3,#4){\c@ntr@lnum{#1}%
    {\result@t=#4\v@lX=#3\v@lY=\result@t\Fig@dmpt{#2}{\z@}}\ignorespaces}
\ctr@ld@f\def\Figp@intregTD#1:#2(#3,#4){\c@ntr@lnum{#1}%
    \def\c@@rdYZ{#4,\z@,\z@}\extrairelepremi@r\c@@rdY\de\c@@rdYZ%
    \extrairelepremi@r\c@@rdZ\de\c@@rdYZ%
    {\v@lX=#3\v@lY=\c@@rdY\v@lZ=\c@@rdZ\Fig@dmpt{#2}{\the\v@lZ}%
    \b@undb@xTD{\v@lX}{\v@lY}{\v@lZ}}\ignorespaces}
\ctr@ln@m\figptBezier
\ctr@ld@f\def\figptBezierDD#1:#2:#3[#4,#5,#6,#7]{\ifps@cri{\s@uvc@ntr@l\et@tfigptBezierDD%
    \FigptBezier@#3[#4,#5,#6,#7]\Figp@intregDD#1:{#2}(\v@lX,\v@lY)%
    \resetc@ntr@l\et@tfigptBezierDD}\ignorespaces\fi}
\ctr@ld@f\def\figptBezierTD#1:#2:#3[#4,#5,#6,#7]{\ifps@cri{\s@uvc@ntr@l\et@tfigptBezierTD%
    \FigptBezier@#3[#4,#5,#6,#7]\Figp@intregTD#1:{#2}(\v@lX,\v@lY,\v@lZ)%
    \resetc@ntr@l\et@tfigptBezierTD}\ignorespaces\fi}
\ctr@ld@f\def\FigptBezier@#1[#2,#3,#4,#5]{\setc@ntr@l{2}%
    \edef\T@{#1}\v@leur=\p@\advance\v@leur-#1pt\edef\UNmT@{\repdecn@mb{\v@leur}}%
    \figptcopy-4:/#2/\figptcopy-3:/#3/\figptcopy-2:/#4/\figptcopy-1:/#5/%
    \l@mbd@un=-4 \l@mbd@de=-\thr@@\p@rtent=\m@ne\c@lDecast%
    \l@mbd@un=-4 \l@mbd@de=-\thr@@\p@rtent=-\tw@\c@lDecast%
    \l@mbd@un=-4 \l@mbd@de=-\thr@@\p@rtent=-\thr@@\c@lDecast\Figg@tXY{-4}}
\ctr@ln@m\c@lDCUn
\ctr@ld@f\def\c@lDCUnDD#1#2{\Figg@tXY{#1}\v@lX=\UNmT@\v@lX\v@lY=\UNmT@\v@lY%
    \Figg@tXYa{#2}\advance\v@lX\T@\v@lXa\advance\v@lY\T@\v@lYa%
    \Figp@intregDD#1:(\v@lX,\v@lY)}
\ctr@ld@f\def\c@lDCUnTD#1#2{\Figg@tXY{#1}\v@lX=\UNmT@\v@lX\v@lY=\UNmT@\v@lY\v@lZ=\UNmT@\v@lZ%
    \Figg@tXYa{#2}\advance\v@lX\T@\v@lXa\advance\v@lY\T@\v@lYa\advance\v@lZ\T@\v@lZa%
    \Figp@intregTD#1:(\v@lX,\v@lY,\v@lZ)}
\ctr@ld@f\def\c@lDecast{\relax\ifnum\l@mbd@un<\p@rtent\c@lDCUn{\l@mbd@un}{\l@mbd@de}%
    \advance\l@mbd@un\@ne\advance\l@mbd@de\@ne\c@lDecast\fi}
\ctr@ld@f\def\figptmap#1:#2=#3/#4/#5/{\ifps@cri{\s@uvc@ntr@l\et@tfigptmap%
    \setc@ntr@l{2}\figvectP-1[#4,#3]\Figg@tXY{-1}%
    \pr@dMatV/#5/\figpttra#1:{#2}=#4/1,-1/%
    \resetc@ntr@l\et@tfigptmap}\ignorespaces\fi}
\ctr@ln@m\pr@dMatV
\ctr@ld@f\def\pr@dMatVDD/#1,#2;#3,#4/{\v@lXa=#1\v@lX\advance\v@lXa#2\v@lY%
    \v@lYa=#3\v@lX\advance\v@lYa#4\v@lY\Figv@ctCreg-1(\v@lXa,\v@lYa)}
\ctr@ld@f\def\pr@dMatVTD/#1,#2,#3;#4,#5,#6;#7,#8,#9/{%
    \v@lXa=#1\v@lX\advance\v@lXa#2\v@lY\advance\v@lXa#3\v@lZ%
    \v@lYa=#4\v@lX\advance\v@lYa#5\v@lY\advance\v@lYa#6\v@lZ%
    \v@lZa=#7\v@lX\advance\v@lZa#8\v@lY\advance\v@lZa#9\v@lZ%
    \Figv@ctCreg-1(\v@lXa,\v@lYa,\v@lZa)}
\ctr@ln@m\figptbary
\ctr@ld@f\def\figptbaryDD#1:#2[#3;#4]{\ifps@cri{\edef\list@num{#3}\extrairelepremi@r\p@int\de\list@num%
    \s@mme=\z@\@ecfor\c@ef:=#4\do{\advance\s@mme\c@ef}%
    \edef\listec@ef{#4,0}\extrairelepremi@r\c@ef\de\listec@ef%
    \Figg@tXY{\p@int}\divide\v@lX\s@mme\divide\v@lY\s@mme%
    \multiply\v@lX\c@ef\multiply\v@lY\c@ef%
    \@ecfor\p@int:=\list@num\do{\extrairelepremi@r\c@ef\de\listec@ef%
           \Figg@tXYa{\p@int}\divide\v@lXa\s@mme\divide\v@lYa\s@mme%
           \multiply\v@lXa\c@ef\multiply\v@lYa\c@ef%
           \advance\v@lX\v@lXa\advance\v@lY\v@lYa}%
    \Figp@intregDD#1:{#2}(\v@lX,\v@lY)}\ignorespaces\fi}
\ctr@ld@f\def\figptbaryTD#1:#2[#3;#4]{\ifps@cri{\edef\list@num{#3}\extrairelepremi@r\p@int\de\list@num%
    \s@mme=\z@\@ecfor\c@ef:=#4\do{\advance\s@mme\c@ef}%
    \edef\listec@ef{#4,0}\extrairelepremi@r\c@ef\de\listec@ef%
    \Figg@tXY{\p@int}\divide\v@lX\s@mme\divide\v@lY\s@mme\divide\v@lZ\s@mme%
    \multiply\v@lX\c@ef\multiply\v@lY\c@ef\multiply\v@lZ\c@ef%
    \@ecfor\p@int:=\list@num\do{\extrairelepremi@r\c@ef\de\listec@ef%
           \Figg@tXYa{\p@int}\divide\v@lXa\s@mme\divide\v@lYa\s@mme\divide\v@lZa\s@mme%
           \multiply\v@lXa\c@ef\multiply\v@lYa\c@ef\multiply\v@lZa\c@ef%
           \advance\v@lX\v@lXa\advance\v@lY\v@lYa\advance\v@lZ\v@lZa}%
    \Figp@intregTD#1:{#2}(\v@lX,\v@lY,\v@lZ)}\ignorespaces\fi}
\ctr@ld@f\def\figptbaryR#1:#2[#3;#4]{\ifps@cri{%
    \v@leur=\z@\@ecfor\c@ef:=#4\do{\maxim@m{\v@lmax}{\c@ef pt}{-\c@ef pt}%
    \ifdim\v@lmax>\v@leur\v@leur=\v@lmax\fi}%
    \ifdim\v@leur<\p@\f@ctech=\@M\else\ifdim\v@leur<\t@n\p@\f@ctech=\@m\else%
    \ifdim\v@leur<\c@nt\p@\f@ctech=\c@nt\else\ifdim\v@leur<\@m\p@\f@ctech=\t@n\else%
    \f@ctech=\@ne\fi\fi\fi\fi%
    \def\listec@ef{0}%
    \@ecfor\c@ef:=#4\do{\sc@lec@nvRI{\c@ef pt}\edef\listec@ef{\listec@ef,\the\s@mme}}%
    \extrairelepremi@r\c@ef\de\listec@ef\figptbary#1:#2[#3;\listec@ef]}\ignorespaces\fi}
\ctr@ld@f\def\sc@lec@nvRI#1{\v@leur=#1\p@rtentiere{\s@mme}{\v@leur}\advance\v@leur-\s@mme\p@%
    \multiply\v@leur\f@ctech\p@rtentiere{\p@rtent}{\v@leur}%
    \multiply\s@mme\f@ctech\advance\s@mme\p@rtent}
\ctr@ln@m\figptcirc
\ctr@ld@f\def\figptcircDD#1:#2:#3;#4(#5){\ifps@cri{\s@uvc@ntr@l\et@tfigptcircDD%
    \c@lptellDD#1:{#2}:#3;#4,#4(#5)\resetc@ntr@l\et@tfigptcircDD}\ignorespaces\fi}
\ctr@ld@f\def\figptcircTD#1:#2:#3,#4,#5;#6(#7){\ifps@cri{\s@uvc@ntr@l\et@tfigptcircTD%
    \setc@ntr@l{2}\c@lExtAxes#3,#4,#5(#6)\figptellP#1:{#2}:#3,-4,-5(#7)%
    \resetc@ntr@l\et@tfigptcircTD}\ignorespaces\fi}
\ctr@ln@m\figptcircumcenter
\ctr@ld@f\def\figptcircumcenterDD#1:#2[#3,#4,#5]{\ifps@cri{\s@uvc@ntr@l\et@tfigptcircumcenterDD%
    \setc@ntr@l{2}\figvectNDD-5[#3,#4]\figptbaryDD-3:[#3,#4;1,1]%
                  \figvectNDD-6[#4,#5]\figptbaryDD-4:[#4,#5;1,1]%
    \resetc@ntr@l{2}\inters@cDD#1:{#2}[-3,-5;-4,-6]%
    \resetc@ntr@l\et@tfigptcircumcenterDD}\ignorespaces\fi}
\ctr@ld@f\def\figptcircumcenterTD#1:#2[#3,#4,#5]{\ifps@cri{\s@uvc@ntr@l\et@tfigptcircumcenterTD%
    \setc@ntr@l{2}\figvectNTD-1[#3,#4,#5]%
    \figvectPTD-3[#3,#4]\figvectNVTD-5[-1,-3]\figptbaryTD-3:[#3,#4;1,1]%
    \figvectPTD-4[#4,#5]\figvectNVTD-6[-1,-4]\figptbaryTD-4:[#4,#5;1,1]%
    \resetc@ntr@l{2}\inters@cTD#1:{#2}[-3,-5;-4,-6]%
    \resetc@ntr@l\et@tfigptcircumcenterTD}\ignorespaces\fi}
\ctr@ln@m\figptcopy
\ctr@ld@f\def\figptcopyDD#1:#2/#3/{\ifps@cri{\Figg@tXY{#3}%
    \Figp@intregDD#1:{#2}(\v@lX,\v@lY)}\ignorespaces\fi}
\ctr@ld@f\def\figptcopyTD#1:#2/#3/{\ifps@cri{\Figg@tXY{#3}%
    \Figp@intregTD#1:{#2}(\v@lX,\v@lY,\v@lZ)}\ignorespaces\fi}
\ctr@ln@m\figptcurvcenter
\ctr@ld@f\def\figptcurvcenterDD#1:#2:#3[#4,#5,#6,#7]{\ifps@cri{\s@uvc@ntr@l\et@tfigptcurvcenterDD%
    \setc@ntr@l{2}\c@lcurvradDD#3[#4,#5,#6,#7]\edef\Sprim@{\repdecn@mb{\result@t}}%
    \figptBezierDD-1::#3[#4,#5,#6,#7]\figpttraDD#1:{#2}=-1/\Sprim@,-5/%
    \resetc@ntr@l\et@tfigptcurvcenterDD}\ignorespaces\fi}
\ctr@ld@f\def\figptcurvcenterTD#1:#2:#3[#4,#5,#6,#7]{\ifps@cri{\s@uvc@ntr@l\et@tfigptcurvcenterTD%
    \setc@ntr@l{2}\figvectDBezierTD -5:1,#3[#4,#5,#6,#7]%
    \figvectDBezierTD -6:2,#3[#4,#5,#6,#7]\vecunit@TD{-5}{-5}%
    \edef\Sprim@{\repdecn@mb{\result@t}}\figvectNVTD-1[-6,-5]%
    \figvectNVTD-5[-5,-1]\c@lproscalTD\v@leur[-6,-5]%
    \invers@{\v@leur}{\v@leur}\v@leur=\Sprim@\v@leur\v@leur=\Sprim@\v@leur%
    \figptBezierTD-1::#3[#4,#5,#6,#7]\edef\Sprim@{\repdecn@mb{\v@leur}}%
    \figpttraTD#1:{#2}=-1/\Sprim@,-5/\resetc@ntr@l\et@tfigptcurvcenterTD}\ignorespaces\fi}
\ctr@ld@f\def\c@lcurvradDD#1[#2,#3,#4,#5]{{\figvectDBezierDD -5:1,#1[#2,#3,#4,#5]%
    \figvectDBezierDD -6:2,#1[#2,#3,#4,#5]\vecunit@DD{-5}{-5}%
    \edef\Sprim@{\repdecn@mb{\result@t}}\figvectNVDD-5[-5]\c@lproscalDD\v@leur[-6,-5]%
    \invers@{\v@leur}{\v@leur}\v@leur=\Sprim@\v@leur\v@leur=\Sprim@\v@leur%
    \global\result@t=\v@leur}}
\ctr@ln@m\figptell
\ctr@ld@f\def\figptellDD#1:#2:#3;#4,#5(#6,#7){\ifps@cri{\s@uvc@ntr@l\et@tfigptell%
    \c@lptellDD#1::#3;#4,#5(#6)\figptrotDD#1:{#2}=#1/#3,#7/%
    \resetc@ntr@l\et@tfigptell}\ignorespaces\fi}
\ctr@ld@f\def\c@lptellDD#1:#2:#3;#4,#5(#6){\c@ssin{\C@}{\S@}{#6}\v@lmin=\C@ pt\v@lmax=\S@ pt%
    \v@lmin=#4\v@lmin\v@lmax=#5\v@lmax%
    \edef\Xc@mp{\repdecn@mb{\v@lmin}}\edef\Yc@mp{\repdecn@mb{\v@lmax}}%
    \setc@ntr@l{2}\figvectC-1(\Xc@mp,\Yc@mp)\figpttraDD#1:{#2}=#3/1,-1/}
\ctr@ld@f\def\figptellP#1:#2:#3,#4,#5(#6){\ifps@cri{\s@uvc@ntr@l\et@tfigptellP%
    \setc@ntr@l{2}\figvectP-1[#3,#4]\figvectP-2[#3,#5]%
    \v@leur=#6pt\c@lptellP{#3}{-1}{-2}\figptcopy#1:{#2}/-3/%
    \resetc@ntr@l\et@tfigptellP}\ignorespaces\fi}
\ctr@ln@m\@ngle
\ctr@ld@f\def\c@lptellP#1#2#3{\edef\@ngle{\repdecn@mb\v@leur}\c@ssin{\C@}{\S@}{\@ngle}%
    \figpttra-3:=#1/\C@,#2/\figpttra-3:=-3/\S@,#3/}
\ctr@ln@m\figptendnormal
\ctr@ld@f\def\figptendnormalDD#1:#2:#3,#4[#5,#6]{\ifps@cri{\s@uvc@ntr@l\et@tfigptendnormal%
    \Figg@tXYa{#5}\Figg@tXY{#6}%
    \advance\v@lX-\v@lXa\advance\v@lY-\v@lYa%
    \setc@ntr@l{2}\Figv@ctCreg-1(\v@lX,\v@lY)\vecunit@{-1}{-1}\Figg@tXY{-1}%
    \delt@=#3\unit@\maxim@m{\delt@}{\delt@}{-\delt@}\edef\l@ngueur{\repdecn@mb{\delt@}}%
    \v@lX=\l@ngueur\v@lX\v@lY=\l@ngueur\v@lY%
    \delt@=\p@\advance\delt@-#4pt\edef\l@ngueur{\repdecn@mb{\delt@}}%
    \figptbaryR-1:[#5,#6;#4,\l@ngueur]\Figg@tXYa{-1}%
    \advance\v@lXa\v@lY\advance\v@lYa-\v@lX%
    \setc@ntr@l{1}\Figp@intregDD#1:{#2}(\v@lXa,\v@lYa)\resetc@ntr@l\et@tfigptendnormal}%
    \ignorespaces\fi}
\ctr@ld@f\def\figptexcenter#1:#2[#3,#4,#5]{\ifps@cri{\let@xte={-}%
    \Figptexinsc@nter#1:#2[#3,#4,#5]}\ignorespaces\fi}
\ctr@ld@f\def\figptincenter#1:#2[#3,#4,#5]{\ifps@cri{\let@xte={}%
    \Figptexinsc@nter#1:#2[#3,#4,#5]}\ignorespaces\fi}
\ctr@ld@f\let\figptinscribedcenter=\figptincenter% pour compatibilite avec anciennes versions
\ctr@ld@f\def\Figptexinsc@nter#1:#2[#3,#4,#5]{%
    \figgetdist\LA@[#4,#5]\figgetdist\LB@[#3,#5]\figgetdist\LC@[#3,#4]%
    \figptbaryR#1:{#2}[#3,#4,#5;\the\let@xte\LA@,\LB@,\LC@]}
\ctr@ln@m\figptinterlineplane
\ctr@ld@f\def\figptinterlineplaneDD{\un@v@ilable{figptinterlineplane}}
\ctr@ld@f\def\figptinterlineplaneTD#1:#2[#3,#4;#5,#6]{\ifps@cri{\s@uvc@ntr@l\et@tfigptinterlineplane%
    \setc@ntr@l{2}\figvectPTD-1[#3,#5]\vecunit@TD{-2}{#6}%
    \r@pPSTD\v@leur[-2,-1,#4]\edef\v@lcoef{\repdecn@mb{\v@leur}}%
    \figpttraTD#1:{#2}=#3/\v@lcoef,#4/\resetc@ntr@l\et@tfigptinterlineplane}\ignorespaces\fi}
\ctr@ln@m\figptorthocenter
\ctr@ld@f\def\figptorthocenterDD#1:#2[#3,#4,#5]{\ifps@cri{\s@uvc@ntr@l\et@tfigptorthocenterDD%
    \setc@ntr@l{2}\figvectNDD-3[#3,#4]\figvectNDD-4[#4,#5]%
    \resetc@ntr@l{2}\inters@cDD#1:{#2}[#5,-3;#3,-4]%
    \resetc@ntr@l\et@tfigptorthocenterDD}\ignorespaces\fi}
\ctr@ld@f\def\figptorthocenterTD#1:#2[#3,#4,#5]{\ifps@cri{\s@uvc@ntr@l\et@tfigptorthocenterTD%
    \setc@ntr@l{2}\figvectNTD-1[#3,#4,#5]%
    \figvectPTD-2[#3,#4]\figvectNVTD-3[-1,-2]%
    \figvectPTD-2[#4,#5]\figvectNVTD-4[-1,-2]%
    \resetc@ntr@l{2}\inters@cTD#1:{#2}[#5,-3;#3,-4]%
    \resetc@ntr@l\et@tfigptorthocenterTD}\ignorespaces\fi}
\ctr@ln@m\figptorthoprojline
\ctr@ld@f\def\figptorthoprojlineDD#1:#2=#3/#4,#5/{\ifps@cri{\s@uvc@ntr@l\et@tfigptorthoprojlineDD%
    \setc@ntr@l{2}\figvectPDD-3[#4,#5]\figvectNVDD-4[-3]\resetc@ntr@l{2}%
    \inters@cDD#1:{#2}[#3,-4;#4,-3]\resetc@ntr@l\et@tfigptorthoprojlineDD}\ignorespaces\fi}
\ctr@ld@f\def\figptorthoprojlineTD#1:#2=#3/#4,#5/{\ifps@cri{\s@uvc@ntr@l\et@tfigptorthoprojlineTD%
    \setc@ntr@l{2}\figvectPTD-1[#4,#3]\figvectPTD-2[#4,#5]\vecunit@TD{-2}{-2}%
    \c@lproscalTD\v@leur[-1,-2]\edef\v@lcoef{\repdecn@mb{\v@leur}}%
    \figpttraTD#1:{#2}=#4/\v@lcoef,-2/\resetc@ntr@l\et@tfigptorthoprojlineTD}\ignorespaces\fi}
\ctr@ln@m\figptorthoprojplane
\ctr@ld@f\def\figptorthoprojplaneDD{\un@v@ilable{figptorthoprojplane}}
\ctr@ld@f\def\figptorthoprojplaneTD#1:#2=#3/#4,#5/{\ifps@cri{\s@uvc@ntr@l\et@tfigptorthoprojplane%
    \setc@ntr@l{2}\figvectPTD-1[#3,#4]\vecunit@TD{-2}{#5}%
    \c@lproscalTD\v@leur[-1,-2]\edef\v@lcoef{\repdecn@mb{\v@leur}}%
    \figpttraTD#1:{#2}=#3/\v@lcoef,-2/\resetc@ntr@l\et@tfigptorthoprojplane}\ignorespaces\fi}
\ctr@ld@f\def\figpthom#1:#2=#3/#4,#5/{\ifps@cri{\s@uvc@ntr@l\et@tfigpthom%
    \setc@ntr@l{2}\figvectP-1[#4,#3]\figpttra#1:{#2}=#4/#5,-1/%
    \resetc@ntr@l\et@tfigpthom}\ignorespaces\fi}
\ctr@ln@m\figptrot
\ctr@ld@f\def\figptrotDD#1:#2=#3/#4,#5/{\ifps@cri{\s@uvc@ntr@l\et@tfigptrotDD%
    \c@ssin{\C@}{\S@}{#5}\setc@ntr@l{2}\figvectPDD-1[#4,#3]\Figg@tXY{-1}%
    \v@lXa=\C@\v@lX\advance\v@lXa-\S@\v@lY%
    \v@lYa=\S@\v@lX\advance\v@lYa\C@\v@lY%
    \Figv@ctCreg-1(\v@lXa,\v@lYa)\figpttraDD#1:{#2}=#4/1,-1/%
    \resetc@ntr@l\et@tfigptrotDD}\ignorespaces\fi}
\ctr@ld@f\def\figptrotTD#1:#2=#3/#4,#5,#6/{\ifps@cri{\s@uvc@ntr@l\et@tfigptrotTD%
    \c@ssin{\C@}{\S@}{#5}%
    \setc@ntr@l{2}\figptorthoprojplaneTD-3:=#4/#3,#6/\figvectPTD-2[-3,#3]%
    \n@rmeucTD\v@leur{-2}\ifdim\v@leur<\Cepsil@n\Figg@tXYa{#3}\else%
    \edef\v@lcoef{\repdecn@mb{\v@leur}}\figvectNVTD-1[#6,-2]%
    \Figg@tXYa{-1}\v@lXa=\v@lcoef\v@lXa\v@lYa=\v@lcoef\v@lYa\v@lZa=\v@lcoef\v@lZa%
    \v@lXa=\S@\v@lXa\v@lYa=\S@\v@lYa\v@lZa=\S@\v@lZa\Figg@tXY{-2}%
    \advance\v@lXa\C@\v@lX\advance\v@lYa\C@\v@lY\advance\v@lZa\C@\v@lZ%
    \Figg@tXY{-3}\advance\v@lXa\v@lX\advance\v@lYa\v@lY\advance\v@lZa\v@lZ\fi%
    \Figp@intregTD#1:{#2}(\v@lXa,\v@lYa,\v@lZa)\resetc@ntr@l\et@tfigptrotTD}\ignorespaces\fi}
\ctr@ln@m\figptsym
\ctr@ld@f\def\figptsymDD#1:#2=#3/#4,#5/{\ifps@cri{\s@uvc@ntr@l\et@tfigptsymDD%
    \resetc@ntr@l{2}\figptorthoprojlineDD-5:=#3/#4,#5/\figvectPDD-2[#3,-5]%
    \figpttraDD#1:{#2}=#3/2,-2/\resetc@ntr@l\et@tfigptsymDD}\ignorespaces\fi}
\ctr@ld@f\def\figptsymTD#1:#2=#3/#4,#5/{\ifps@cri{\s@uvc@ntr@l\et@tfigptsymTD%
    \resetc@ntr@l{2}\figptorthoprojplaneTD-3:=#3/#4,#5/\figvectPTD-2[#3,-3]%
    \figpttraTD#1:{#2}=#3/2,-2/\resetc@ntr@l\et@tfigptsymTD}\ignorespaces\fi}
\ctr@ln@m\figpttra
\ctr@ld@f\def\figpttraDD#1:#2=#3/#4,#5/{\ifps@cri{\Figg@tXYa{#5}\v@lXa=#4\v@lXa\v@lYa=#4\v@lYa%
    \Figg@tXY{#3}\advance\v@lX\v@lXa\advance\v@lY\v@lYa%
    \Figp@intregDD#1:{#2}(\v@lX,\v@lY)}\ignorespaces\fi}
\ctr@ld@f\def\figpttraTD#1:#2=#3/#4,#5/{\ifps@cri{\Figg@tXYa{#5}\v@lXa=#4\v@lXa\v@lYa=#4\v@lYa%
    \v@lZa=#4\v@lZa\Figg@tXY{#3}\advance\v@lX\v@lXa\advance\v@lY\v@lYa%
    \advance\v@lZ\v@lZa\Figp@intregTD#1:{#2}(\v@lX,\v@lY,\v@lZ)}\ignorespaces\fi}
\ctr@ln@m\figpttraC
\ctr@ld@f\def\figpttraCDD#1:#2=#3/#4,#5/{\ifps@cri{\v@lXa=#4\unit@\v@lYa=#5\unit@%
    \Figg@tXY{#3}\advance\v@lX\v@lXa\advance\v@lY\v@lYa%
    \Figp@intregDD#1:{#2}(\v@lX,\v@lY)}\ignorespaces\fi}
\ctr@ld@f\def\figpttraCTD#1:#2=#3/#4,#5,#6/{\ifps@cri{\v@lXa=#4\unit@\v@lYa=#5\unit@\v@lZa=#6\unit@%
    \Figg@tXY{#3}\advance\v@lX\v@lXa\advance\v@lY\v@lYa\advance\v@lZ\v@lZa%
    \Figp@intregTD#1:{#2}(\v@lX,\v@lY,\v@lZ)}\ignorespaces\fi}
\ctr@ld@f\def\figptsaxes#1:#2(#3){\ifps@cri{\an@lys@xes#3,:\ifx\t@xt@\empty%
    \ifTr@isDim\Figpts@xes#1:#2(0,#3,0,#3,0,#3)\else\Figpts@xes#1:#2(0,#3,0,#3)\fi%
    \else\Figpts@xes#1:#2(#3)\fi}\ignorespaces\fi}
\ctr@ln@m\Figpts@xes
\ctr@ld@f\def\Figpts@xesDD#1:#2(#3,#4,#5,#6){%
    \s@mme=#1\figpttraC\the\s@mme:$x$=#2/#4,0/%
    \advance\s@mme\@ne\figpttraC\the\s@mme:$y$=#2/0,#6/}
\ctr@ld@f\def\Figpts@xesTD#1:#2(#3,#4,#5,#6,#7,#8){%
    \s@mme=#1\figpttraC\the\s@mme:$x$=#2/#4,0,0/%
    \advance\s@mme\@ne\figpttraC\the\s@mme:$y$=#2/0,#6,0/%
    \advance\s@mme\@ne\figpttraC\the\s@mme:$z$=#2/0,0,#8/}
\ctr@ld@f\def\figptsmap#1=#2/#3/#4/{\ifps@cri{\s@uvc@ntr@l\et@tfigptsmap%
    \setc@ntr@l{2}\def\list@num{#2}\s@mme=#1%
    \@ecfor\p@int:=\list@num\do{\figvectP-1[#3,\p@int]\Figg@tXY{-1}%
    \pr@dMatV/#4/\figpttra\the\s@mme:=#3/1,-1/\advance\s@mme\@ne}%
    \resetc@ntr@l\et@tfigptsmap}\ignorespaces\fi}
\ctr@ln@m\figptscontrol
\ctr@ld@f\def\figptscontrolDD#1[#2,#3,#4,#5]{\ifps@cri{\s@uvc@ntr@l\et@tfigptscontrolDD\setc@ntr@l{2}%
    \v@lX=\z@\v@lY=\z@\Figtr@nptDD{-5}{#2}\Figtr@nptDD{2}{#5}%
    \divide\v@lX\@vi\divide\v@lY\@vi%
    \Figtr@nptDD{3}{#3}\Figtr@nptDD{-1.5}{#4}\Figp@intregDD-1:(\v@lX,\v@lY)%
    \v@lX=\z@\v@lY=\z@\Figtr@nptDD{2}{#2}\Figtr@nptDD{-5}{#5}%
    \divide\v@lX\@vi\divide\v@lY\@vi\Figtr@nptDD{-1.5}{#3}\Figtr@nptDD{3}{#4}%
    \s@mme=#1\advance\s@mme\@ne\Figp@intregDD\the\s@mme:(\v@lX,\v@lY)%
    \figptcopyDD#1:/-1/\resetc@ntr@l\et@tfigptscontrolDD}\ignorespaces\fi}
\ctr@ld@f\def\figptscontrolTD#1[#2,#3,#4,#5]{\ifps@cri{\s@uvc@ntr@l\et@tfigptscontrolTD\setc@ntr@l{2}%
    \v@lX=\z@\v@lY=\z@\v@lZ=\z@\Figtr@nptTD{-5}{#2}\Figtr@nptTD{2}{#5}%
    \divide\v@lX\@vi\divide\v@lY\@vi\divide\v@lZ\@vi%
    \Figtr@nptTD{3}{#3}\Figtr@nptTD{-1.5}{#4}\Figp@intregTD-1:(\v@lX,\v@lY,\v@lZ)%
    \v@lX=\z@\v@lY=\z@\v@lZ=\z@\Figtr@nptTD{2}{#2}\Figtr@nptTD{-5}{#5}%
    \divide\v@lX\@vi\divide\v@lY\@vi\divide\v@lZ\@vi\Figtr@nptTD{-1.5}{#3}\Figtr@nptTD{3}{#4}%
    \s@mme=#1\advance\s@mme\@ne\Figp@intregTD\the\s@mme:(\v@lX,\v@lY,\v@lZ)%
    \figptcopyTD#1:/-1/\resetc@ntr@l\et@tfigptscontrolTD}\ignorespaces\fi}
\ctr@ld@f\def\Figtr@nptDD#1#2{\Figg@tXYa{#2}\v@lXa=#1\v@lXa\v@lYa=#1\v@lYa%
    \advance\v@lX\v@lXa\advance\v@lY\v@lYa}
\ctr@ld@f\def\Figtr@nptTD#1#2{\Figg@tXYa{#2}\v@lXa=#1\v@lXa\v@lYa=#1\v@lYa\v@lZa=#1\v@lZa%
    \advance\v@lX\v@lXa\advance\v@lY\v@lYa\advance\v@lZ\v@lZa}
\ctr@ld@f\def\figptscontrolcurve#1,#2[#3]{\ifps@cri{\s@uvc@ntr@l\et@tfigptscontrolcurve%
    \def\list@num{#3}\extrairelepremi@r\Ak@\de\list@num%
    \extrairelepremi@r\Ai@\de\list@num\extrairelepremi@r\Aj@\de\list@num%
    \s@mme=#1\figptcopy\the\s@mme:/\Ai@/%
    \setc@ntr@l{2}\figvectP -1[\Ak@,\Aj@]%
    \@ecfor\Ak@:=\list@num\do{\advance\s@mme\@ne\figpttra\the\s@mme:=\Ai@/\curv@roundness,-1/%
       \figvectP -1[\Ai@,\Ak@]\advance\s@mme\@ne\figpttra\the\s@mme:=\Aj@/-\curv@roundness,-1/%
       \advance\s@mme\@ne\figptcopy\the\s@mme:/\Aj@/%
       \edef\Ai@{\Aj@}\edef\Aj@{\Ak@}}\advance\s@mme-#1\divide\s@mme\thr@@%
       \xdef#2{\the\s@mme}%
    \resetc@ntr@l\et@tfigptscontrolcurve}\ignorespaces\fi}
\ctr@ln@m\figptsintercirc
\ctr@ld@f\def\figptsintercircDD#1[#2,#3;#4,#5]{\ifps@cri{\s@uvc@ntr@l\et@tfigptsintercircDD%
    \setc@ntr@l{2}\let\c@lNVintc=\c@lNVintcDD\Figptsintercirc@#1[#2,#3;#4,#5]%    
    \resetc@ntr@l\et@tfigptsintercircDD}\ignorespaces\fi}
\ctr@ld@f\def\figptsintercircTD#1[#2,#3;#4,#5;#6]{\ifps@cri{\s@uvc@ntr@l\et@tfigptsintercircTD%
    \setc@ntr@l{2}\let\c@lNVintc=\c@lNVintcTD\vecunitC@TD[#2,#6]%
    \Figv@ctCreg-3(\v@lX,\v@lY,\v@lZ)\Figptsintercirc@#1[#2,#3;#4,#5]%
    \resetc@ntr@l\et@tfigptsintercircTD}\ignorespaces\fi}
\ctr@ld@f\def\Figptsintercirc@#1[#2,#3;#4,#5]{\figvectP-1[#2,#4]%
    \vecunit@{-1}{-1}\delt@=\result@t\f@ctech=\result@tent%
    \s@mme=#1\advance\s@mme\@ne\figptcopy#1:/#2/\figptcopy\the\s@mme:/#4/%
    \ifdim\delt@=\z@\else%
    \v@lmin=#3\unit@\v@lmax=#5\unit@\v@leur=\v@lmin\advance\v@leur\v@lmax%
    \ifdim\v@leur>\delt@%
    \v@leur=\v@lmin\advance\v@leur-\v@lmax\maxim@m{\v@leur}{\v@leur}{-\v@leur}%
    \ifdim\v@leur<\delt@%
    \divide\v@lmin\f@ctech\divide\v@lmax\f@ctech\divide\delt@\f@ctech%
    \v@lmin=\repdecn@mb{\v@lmin}\v@lmin\v@lmax=\repdecn@mb{\v@lmax}\v@lmax%
    \invers@{\v@leur}{\delt@}\advance\v@lmax-\v@lmin%
    \v@lmax=-\repdecn@mb{\v@leur}\v@lmax\advance\delt@\v@lmax\delt@=.5\delt@%
    \v@lmax=\delt@\multiply\v@lmax\f@ctech%
    \edef\t@ille{\repdecn@mb{\v@lmax}}\figpttra-2:=#2/\t@ille,-1/%
    \delt@=\repdecn@mb{\delt@}\delt@\advance\v@lmin-\delt@%
    \sqrt@{\v@leur}{\v@lmin}\multiply\v@leur\f@ctech\edef\t@ille{\repdecn@mb{\v@leur}}%
    \c@lNVintc\figpttra#1:=-2/-\t@ille,-1/\figpttra\the\s@mme:=-2/\t@ille,-1/\fi\fi\fi}
\ctr@ld@f\def\c@lNVintcDD{\Figg@tXY{-1}\Figv@ctCreg-1(-\v@lY,\v@lX)} % <=> \figvectNVDD-1[-1]
\ctr@ld@f\def\c@lNVintcTD{{\Figg@tXY{-3}\v@lmin=\v@lX\v@lmax=\v@lY\v@leur=\v@lZ%
    \Figg@tXY{-1}\c@lprovec{-3}\vecunit@{-3}{-3}% <=> \figvectNVTD-3[-1,-3]\vecunit@{-3}{-3}
    \Figg@tXY{-1}\v@lmin=\v@lX\v@lmax=\v@lY%
    \v@leur=\v@lZ\Figg@tXY{-3}\c@lprovec{-1}}} % <=> \figvectNVTD-1[-3,-1]
\ctr@ln@m\figptsinterlinell
\ctr@ld@f\def\figptsinterlinellDD#1[#2,#3,#4,#5;#6,#7]{\ifps@cri{\s@uvc@ntr@l\et@tfigptsinterlinellDD%
    \figptcopy#1:/#6/\s@mme=#1\advance\s@mme\@ne\figptcopy\the\s@mme:/#7/%
    \v@lmin=#3\unit@\v@lmax=#4\unit@% a, b
    \setc@ntr@l{2}\figptbaryDD-4:[#6,#7;1,1]\figptsrotDD-3=-4,#7/#2,-#5/% D et rotation
    \Figg@tXY{-3}\Figg@tXYa{#2}\advance\v@lX-\v@lXa\advance\v@lY-\v@lYa% alpha, beta
    \figvectP-1[-3,-2]\Figg@tXYa{-1}\figvectP-3[-4,#7]\Figptsint@rLE{#1}% u1, u2
    \resetc@ntr@l\et@tfigptsinterlinellDD}\ignorespaces\fi}
\ctr@ld@f\def\figptsinterlinellP#1[#2,#3,#4;#5,#6]{\ifps@cri{\s@uvc@ntr@l\et@tfigptsinterlinellP%
    \figptcopy#1:/#5/\s@mme=#1\advance\s@mme\@ne\figptcopy\the\s@mme:/#6/\setc@ntr@l{2}%
    \figvectP-1[#2,#3]\vecunit@{-1}{-1}\v@lmin=\result@t% a
    \figvectP-2[#2,#4]\vecunit@{-2}{-2}\v@lmax=\result@t% b
    \figptbary-4:[#5,#6;1,1]% D
    \figvectP-3[#2,-4]\c@lproscal\v@lX[-3,-1]\c@lproscal\v@lY[-3,-2]% alpha, beta
    \figvectP-3[-4,#6]\c@lproscal\v@lXa[-3,-1]\c@lproscal\v@lYa[-3,-2]% u1, u2
    \Figptsint@rLE{#1}\resetc@ntr@l\et@tfigptsinterlinellP}\ignorespaces\fi}
\ctr@ld@f\def\Figptsint@rLE#1{%
    \getredf@ctDD\f@ctech(\v@lmin,\v@lmax)%
    \getredf@ctDD\p@rtent(\v@lX,\v@lY)\ifnum\p@rtent>\f@ctech\f@ctech=\p@rtent\fi%
    \getredf@ctDD\p@rtent(\v@lXa,\v@lYa)\ifnum\p@rtent>\f@ctech\f@ctech=\p@rtent\fi%
    \divide\v@lmin\f@ctech\divide\v@lmax\f@ctech\divide\v@lX\f@ctech\divide\v@lY\f@ctech%
    \divide\v@lXa\f@ctech\divide\v@lYa\f@ctech%
    \c@rre=\repdecn@mb\v@lXa\v@lmax\mili@u=\repdecn@mb\v@lYa\v@lmin%
    \getredf@ctDD\f@ctech(\c@rre,\mili@u)%
    \c@rre=\repdecn@mb\v@lX\v@lmax\mili@u=\repdecn@mb\v@lY\v@lmin%
    \getredf@ctDD\p@rtent(\c@rre,\mili@u)\ifnum\p@rtent>\f@ctech\f@ctech=\p@rtent\fi%
    \divide\v@lmin\f@ctech\divide\v@lmax\f@ctech\divide\v@lX\f@ctech\divide\v@lY\f@ctech%
    \divide\v@lXa\f@ctech\divide\v@lYa\f@ctech%
    \v@lmin=\repdecn@mb{\v@lmin}\v@lmin\v@lmax=\repdecn@mb{\v@lmax}\v@lmax%
    \edef\G@xde{\repdecn@mb\v@lmin}\edef\P@xde{\repdecn@mb\v@lmax}%
    \c@rre=-\v@lmax\v@leur=\repdecn@mb\v@lY\v@lY\advance\c@rre\v@leur\c@rre=\G@xde\c@rre%
    \v@leur=\repdecn@mb\v@lX\v@lX\v@leur=\P@xde\v@leur\advance\c@rre\v@leur% C
    \v@lmin=\repdecn@mb\v@lYa\v@lmin\v@lmax=\repdecn@mb\v@lXa\v@lmax%
    \mili@u=\repdecn@mb\v@lX\v@lmax\advance\mili@u\repdecn@mb\v@lY\v@lmin% B
    \v@lmax=\repdecn@mb\v@lXa\v@lmax\advance\v@lmax\repdecn@mb\v@lYa\v@lmin% A
    \ifdim\v@lmax>\epsil@n%
    \maxim@m{\v@leur}{\c@rre}{-\c@rre}\maxim@m{\v@lmin}{\mili@u}{-\mili@u}%
    \maxim@m{\v@leur}{\v@leur}{\v@lmin}\maxim@m{\v@lmin}{\v@lmax}{-\v@lmax}%
    \maxim@m{\v@leur}{\v@leur}{\v@lmin}\p@rtentiere{\p@rtent}{\v@leur}\advance\p@rtent\@ne%
    \divide\c@rre\p@rtent\divide\mili@u\p@rtent\divide\v@lmax\p@rtent%
    \delt@=\repdecn@mb{\mili@u}\mili@u\v@leur=\repdecn@mb{\v@lmax}\c@rre%
    \advance\delt@-\v@leur\ifdim\delt@<\z@\else\sqrt@\delt@\delt@%
    \invers@\v@lmax\v@lmax\edef\Uns@rAp{\repdecn@mb\v@lmax}%
    \v@leur=-\mili@u\advance\v@leur-\delt@\v@leur=\Uns@rAp\v@leur%
    \edef\t@ille{\repdecn@mb\v@leur}\figpttra#1:=-4/\t@ille,-3/\s@mme=#1\advance\s@mme\@ne%
    \v@leur=-\mili@u\advance\v@leur\delt@\v@leur=\Uns@rAp\v@leur%
    \edef\t@ille{\repdecn@mb\v@leur}\figpttra\the\s@mme:=-4/\t@ille,-3/\fi\fi}
\ctr@ln@m\figptsorthoprojline
\ctr@ld@f\def\figptsorthoprojlineDD#1=#2/#3,#4/{\ifps@cri{\s@uvc@ntr@l\et@tfigptsorthoprojlineDD%
    \setc@ntr@l{2}\figvectPDD-3[#3,#4]\figvectNVDD-4[-3]\resetc@ntr@l{2}%
    \def\list@num{#2}\s@mme=#1\@ecfor\p@int:=\list@num\do{%
    \inters@cDD\the\s@mme:[\p@int,-4;#3,-3]\advance\s@mme\@ne}%
    \resetc@ntr@l\et@tfigptsorthoprojlineDD}\ignorespaces\fi}
\ctr@ld@f\def\figptsorthoprojlineTD#1=#2/#3,#4/{\ifps@cri{\s@uvc@ntr@l\et@tfigptsorthoprojlineTD%
    \setc@ntr@l{2}\figvectPTD-2[#3,#4]\vecunit@TD{-2}{-2}%
    \def\list@num{#2}\s@mme=#1\@ecfor\p@int:=\list@num\do{%
    \figvectPTD-1[#3,\p@int]\c@lproscalTD\v@leur[-1,-2]%
    \edef\v@lcoef{\repdecn@mb{\v@leur}}\figpttraTD\the\s@mme:=#3/\v@lcoef,-2/%
    \advance\s@mme\@ne}\resetc@ntr@l\et@tfigptsorthoprojlineTD}\ignorespaces\fi}
\ctr@ln@m\figptsorthoprojplane
\ctr@ld@f\def\figptsorthoprojplaneDD{\un@v@ilable{figptsorthoprojplane}}
\ctr@ld@f\def\figptsorthoprojplaneTD#1=#2/#3,#4/{\ifps@cri{\s@uvc@ntr@l\et@tfigptsorthoprojplane%
    \setc@ntr@l{2}\vecunit@TD{-2}{#4}%
    \def\list@num{#2}\s@mme=#1\@ecfor\p@int:=\list@num\do{\figvectPTD-1[\p@int,#3]%
    \c@lproscalTD\v@leur[-1,-2]\edef\v@lcoef{\repdecn@mb{\v@leur}}%
    \figpttraTD\the\s@mme:=\p@int/\v@lcoef,-2/\advance\s@mme\@ne}%
    \resetc@ntr@l\et@tfigptsorthoprojplane}\ignorespaces\fi}
\ctr@ld@f\def\figptshom#1=#2/#3,#4/{\ifps@cri{\s@uvc@ntr@l\et@tfigptshom%
    \setc@ntr@l{2}\def\list@num{#2}\s@mme=#1%
    \@ecfor\p@int:=\list@num\do{\figvectP-1[#3,\p@int]%
    \figpttra\the\s@mme:=#3/#4,-1/\advance\s@mme\@ne}%
    \resetc@ntr@l\et@tfigptshom}\ignorespaces\fi}
\ctr@ln@m\figptsrot
\ctr@ld@f\def\figptsrotDD#1=#2/#3,#4/{\ifps@cri{\s@uvc@ntr@l\et@tfigptsrotDD%
    \c@ssin{\C@}{\S@}{#4}\setc@ntr@l{2}\def\list@num{#2}\s@mme=#1%
    \@ecfor\p@int:=\list@num\do{\figvectPDD-1[#3,\p@int]\Figg@tXY{-1}%
    \v@lXa=\C@\v@lX\advance\v@lXa-\S@\v@lY%
    \v@lYa=\S@\v@lX\advance\v@lYa\C@\v@lY%
    \Figv@ctCreg-1(\v@lXa,\v@lYa)\figpttraDD\the\s@mme:=#3/1,-1/\advance\s@mme\@ne}%
    \resetc@ntr@l\et@tfigptsrotDD}\ignorespaces\fi}
\ctr@ld@f\def\figptsrotTD#1=#2/#3,#4,#5/{\ifps@cri{\s@uvc@ntr@l\et@tfigptsrotTD%
    \c@ssin{\C@}{\S@}{#4}%
    \setc@ntr@l{2}\def\list@num{#2}\s@mme=#1%
    \@ecfor\p@int:=\list@num\do{\figptorthoprojplaneTD-3:=#3/\p@int,#5/%
    \figvectPTD-2[-3,\p@int]%
    \figvectNVTD-1[#5,-2]\n@rmeucTD\v@leur{-2}\edef\v@lcoef{\repdecn@mb{\v@leur}}%
    \Figg@tXYa{-1}\v@lXa=\v@lcoef\v@lXa\v@lYa=\v@lcoef\v@lYa\v@lZa=\v@lcoef\v@lZa%
    \v@lXa=\S@\v@lXa\v@lYa=\S@\v@lYa\v@lZa=\S@\v@lZa\Figg@tXY{-2}%
    \advance\v@lXa\C@\v@lX\advance\v@lYa\C@\v@lY\advance\v@lZa\C@\v@lZ%
    \Figg@tXY{-3}\advance\v@lXa\v@lX\advance\v@lYa\v@lY\advance\v@lZa\v@lZ%
    \Figp@intregTD\the\s@mme:(\v@lXa,\v@lYa,\v@lZa)\advance\s@mme\@ne}%
    \resetc@ntr@l\et@tfigptsrotTD}\ignorespaces\fi}
\ctr@ln@m\figptssym
\ctr@ld@f\def\figptssymDD#1=#2/#3,#4/{\ifps@cri{\s@uvc@ntr@l\et@tfigptssymDD%
    \setc@ntr@l{2}\figvectPDD-3[#3,#4]\Figg@tXY{-3}\Figv@ctCreg-4(-\v@lY,\v@lX)%
    \resetc@ntr@l{2}\def\list@num{#2}\s@mme=#1%
    \@ecfor\p@int:=\list@num\do{\inters@cDD-5:[#3,-3;\p@int,-4]\figvectPDD-2[\p@int,-5]%
    \figpttraDD\the\s@mme:=\p@int/2,-2/\advance\s@mme\@ne}%
    \resetc@ntr@l\et@tfigptssymDD}\ignorespaces\fi}
\ctr@ld@f\def\figptssymTD#1=#2/#3,#4/{\ifps@cri{\s@uvc@ntr@l\et@tfigptssymTD%
    \setc@ntr@l{2}\vecunit@TD{-2}{#4}\def\list@num{#2}\s@mme=#1%
    \@ecfor\p@int:=\list@num\do{\figvectPTD-1[\p@int,#3]%
    \c@lproscalTD\v@leur[-1,-2]\v@leur=2\v@leur\edef\v@lcoef{\repdecn@mb{\v@leur}}%
    \figpttraTD\the\s@mme:=\p@int/\v@lcoef,-2/\advance\s@mme\@ne}%
    \resetc@ntr@l\et@tfigptssymTD}\ignorespaces\fi}
\ctr@ln@m\figptstra
\ctr@ld@f\def\figptstraDD#1=#2/#3,#4/{\ifps@cri{\Figg@tXYa{#4}\v@lXa=#3\v@lXa\v@lYa=#3\v@lYa%
    \def\list@num{#2}\s@mme=#1\@ecfor\p@int:=\list@num\do{\Figg@tXY{\p@int}%
    \advance\v@lX\v@lXa\advance\v@lY\v@lYa%
    \Figp@intregDD\the\s@mme:(\v@lX,\v@lY)\advance\s@mme\@ne}}\ignorespaces\fi}
\ctr@ld@f\def\figptstraTD#1=#2/#3,#4/{\ifps@cri{\Figg@tXYa{#4}\v@lXa=#3\v@lXa\v@lYa=#3\v@lYa%
    \v@lZa=#3\v@lZa\def\list@num{#2}\s@mme=#1\@ecfor\p@int:=\list@num\do{\Figg@tXY{\p@int}%
    \advance\v@lX\v@lXa\advance\v@lY\v@lYa\advance\v@lZ\v@lZa%
    \Figp@intregTD\the\s@mme:(\v@lX,\v@lY,\v@lZ)\advance\s@mme\@ne}}\ignorespaces\fi}
\ctr@ln@m\figptvisilimSL
\ctr@ld@f\def\figptvisilimSLDD{\un@v@ilable{figptvisilimSL}}
\ctr@ld@f\def\figptvisilimSLTD#1:#2[#3,#4;#5,#6]{\ifps@cri{\s@uvc@ntr@l\et@tfigptvisilimSLTD%
    \setc@ntr@l{2}\figvectP-1[#3,#4]\n@rminf{\delt@}{-1}%
    \ifcase\curr@ntproj\v@lX=\cxa@\p@\v@lY=-\p@\v@lZ=\cxb@\p@% Proj cav
    \Figv@ctCreg-2(\v@lX,\v@lY,\v@lZ)\figvectP-3[#5,#6]\figvectNV-1[-2,-3]%
    \or\figvectP-1[#5,#6]\vecunitCV@TD{-1}\v@lmin=\v@lX\v@lmax=\v@lY% Proj ortho
    \v@leur=\v@lZ\v@lX=\cza@\p@\v@lY=\czb@\p@\v@lZ=\czc@\p@\c@lprovec{-1}%
    \or\c@ley@pt{-2}\figvectN-1[#5,#6,-2]\fi% Proj rea
    \edef\Ai@{#3}\edef\Aj@{#4}\figvectP-2[#5,\Ai@]\c@lproscal\v@leur[-1,-2]%
    \ifdim\v@leur>\z@\p@rtent=\@ne\else\p@rtent=\m@ne\fi%
    \figvectP-2[#5,\Aj@]\c@lproscal\v@leur[-1,-2]%
    \ifdim\p@rtent\v@leur>\z@\figptcopy#1:#2/#3/%
    \message{*** \BS@ figptvisilimSL: points are on the same side.}\else%
    \figptcopy-3:/#3/\figptcopy-4:/#4/%
    \loop\figptbary-5:[-3,-4;1,1]\figvectP-2[#5,-5]\c@lproscal\v@leur[-1,-2]%
    \ifdim\p@rtent\v@leur>\z@\figptcopy-3:/-5/\else\figptcopy-4:/-5/\fi%
    \divide\delt@\tw@\ifdim\delt@>\epsil@n\repeat%
    \figptbary#1:#2[-3,-4;1,1]\fi\resetc@ntr@l\et@tfigptvisilimSLTD}\ignorespaces\fi}
\ctr@ld@f\def\c@ley@pt#1{\t@stp@r\ifitis@K\v@lX=\cza@\p@\v@lY=\czb@\p@\v@lZ=\czc@\p@%
    \Figv@ctCreg-1(\v@lX,\v@lY,\v@lZ)\Figp@intreg-2:(\wd\Bt@rget,\ht\Bt@rget,\dp\Bt@rget)%
    \figpttra#1:=-2/-\disob@intern,-1/\else\end\fi}
\ctr@ld@f\def\t@stp@r{\itis@Ktrue\ifnewt@rgetpt\else\itis@Kfalse%
    \message{*** \BS@ figptvisilimXX: target point undefined.}\fi\ifnewdis@b\else%
    \itis@Kfalse\message{*** \BS@ figptvisilimXX: observation distance undefined.}\fi%
    \ifitis@K\else\message{*** This macro must be called after \BS@ psbeginfig or after
    having set the missing parameter(s) with \BS@ figset proj()}\fi}
\ctr@ld@f\def\figscan#1(#2,#3){{\s@uvc@ntr@l\et@tfigscan\@psfgetbb{#1}\if@psfbbfound\else%
    \def\@psfllx{0}\def\@psflly{20}\def\@psfurx{540}\def\@psfury{640}\fi\figscan@{#2}{#3}%
    \resetc@ntr@l\et@tfigscan}\ignorespaces}
\ctr@ld@f\def\figscan@#1#2{%
    \unit@=\@ne bp\setc@ntr@l{2}\figsetmark{}%
    \def\minst@p{20pt}%
    \v@lX=\@psfllx\p@\v@lX=\Sc@leFact\v@lX\r@undint\v@lX\v@lX%
    \v@lY=\@psflly\p@\v@lY=\Sc@leFact\v@lY\ifdim\v@lY>\z@\r@undint\v@lY\v@lY\fi%
    \delt@=\@psfury\p@\delt@=\Sc@leFact\delt@%
    \advance\delt@-\v@lY\v@lXa=\@psfurx\p@\v@lXa=\Sc@leFact\v@lXa\v@leur=\minst@p%
    \edef\valv@lY{\repdecn@mb{\v@lY}}\edef\LgTr@it{\the\delt@}%
    \loop\ifdim\v@lX<\v@lXa\edef\valv@lX{\repdecn@mb{\v@lX}}%
    \figptDD -1:(\valv@lX,\valv@lY)\figwriten -1:\hbox{\vrule height\LgTr@it}(0)%
    \ifdim\v@leur<\minst@p\else\figsetmark{\raise-8bp\hbox{$\scriptscriptstyle\triangle$}}%
    \figwrites -1:\@ffichnb{0}{\valv@lX}(6)\v@leur=\z@\figsetmark{}\fi%
    \advance\v@leur#1pt\advance\v@lX#1pt\repeat%
    \def\minst@p{10pt}%
    \v@lX=\@psfllx\p@\v@lX=\Sc@leFact\v@lX\ifdim\v@lX>\z@\r@undint\v@lX\v@lX\fi%
    \v@lY=\@psflly\p@\v@lY=\Sc@leFact\v@lY\r@undint\v@lY\v@lY%
    \delt@=\@psfurx\p@\delt@=\Sc@leFact\delt@%
    \advance\delt@-\v@lX\v@lYa=\@psfury\p@\v@lYa=\Sc@leFact\v@lYa\v@leur=\minst@p%
    \edef\valv@lX{\repdecn@mb{\v@lX}}\edef\LgTr@it{\the\delt@}%
    \loop\ifdim\v@lY<\v@lYa\edef\valv@lY{\repdecn@mb{\v@lY}}%
    \figptDD -1:(\valv@lX,\valv@lY)\figwritee -1:\vbox{\hrule width\LgTr@it}(0)%
    \ifdim\v@leur<\minst@p\else\figsetmark{$\triangleright$\kern4bp}%
    \figwritew -1:\@ffichnb{0}{\valv@lY}(6)\v@leur=\z@\figsetmark{}\fi%
    \advance\v@leur#2pt\advance\v@lY#2pt\repeat}
\ctr@ld@f\let\figscanI=\figscan
\ctr@ld@f\def\figscan@E#1(#2,#3){{\s@uvc@ntr@l\et@tfigscan@E%
    \Figdisc@rdLTS{#1}{\t@xt@}\pdfximage{\t@xt@}%
    \setbox\Gb@x=\hbox{\pdfrefximage\pdflastximage}%
    \edef\@psfllx{0}\v@lY=-\dp\Gb@x\edef\@psflly{\repdecn@mb{\v@lY}}%
    \edef\@psfurx{\repdecn@mb{\wd\Gb@x}}%
    \v@lY=\dp\Gb@x\advance\v@lY\ht\Gb@x\edef\@psfury{\repdecn@mb{\v@lY}}%
    \figscan@{#2}{#3}\resetc@ntr@l\et@tfigscan@E}\ignorespaces}
\ctr@ld@f\def\figshowpts[#1,#2]{{\figsetmark{$\bullet$}\figsetptname{\bf ##1}%
    \p@rtent=#2\relax\ifnum\p@rtent<\z@\p@rtent=\z@\fi%
    \s@mme=#1\relax\ifnum\s@mme<\z@\s@mme=\z@\fi%
    \loop\ifnum\s@mme<\p@rtent\pt@rvect{\s@mme}%
    \ifitis@K\figwriten{\the\s@mme}:(4pt)\fi\advance\s@mme\@ne\repeat%
    \pt@rvect{\s@mme}\ifitis@K\figwriten{\the\s@mme}:(4pt)\fi}\ignorespaces}
\ctr@ld@f\def\pt@rvect#1{\set@bjc@de{#1}%
    \expandafter\expandafter\expandafter\inqpt@rvec\csname\objc@de\endcsname:}
\ctr@ld@f\def\inqpt@rvec#1#2:{\if#1\C@dCl@spt\itis@Ktrue\else\itis@Kfalse\fi}
\ctr@ld@f\def\figshowsettings{{%
    \immediate\write16{====================================================================}%
    \immediate\write16{ Current settings about:}%
    \immediate\write16{ --- GENERAL ---}%
    \immediate\write16{Scale factor and Unit = \unit@util\space (\the\unit@)
     \space -> \BS@ figinit{ScaleFactorUnit}}%
    \immediate\write16{Update mode = \ifpsupdatem@de yes\else no\fi
     \space-> \BS@ psset(update=yes/no) or \BS@ pssetdefault(update=yes/no)}%
    \immediate\write16{ --- PRINTING ---}%
    \immediate\write16{Implicit point name = \ptn@me{i} \space-> \BS@ figsetptname{Name}}%
    \immediate\write16{Point marker = \the\c@nsymb \space -> \BS@ figsetmark{Mark}}%
    \immediate\write16{Print rounded coordinates = \ifr@undcoord yes\else no\fi
     \space-> \BS@ figsetroundcoord{yes/no}}%
    \immediate\write16{ --- GRAPHICAL (general) ---}%
    \immediate\write16{First-level (or primary) settings:}%
    \immediate\write16{ Color = \curr@ntcolor \space-> \BS@ psset(color=ColorDefinition)}%
    \immediate\write16{ Filling mode = \iffillm@de yes\else no\fi
     \space-> \BS@ psset(fillmode=yes/no)}%
    \immediate\write16{ Line join = \curr@ntjoin \space-> \BS@ psset(join=miter/round/bevel)}%
    \immediate\write16{ Line style = \curr@ntdash \space-> \BS@ psset(dash=Index/Pattern)}%
    \immediate\write16{ Line width = \curr@ntwidth
     \space-> \BS@ psset(width=real in PostScript units)}%
    \immediate\write16{Second-level (or secondary) settings:}%
    \immediate\write16{ Color = \sec@ndcolor \space-> \BS@ psset second(color=ColorDefinition)}%
    \immediate\write16{ Line style = \curr@ntseconddash
     \space-> \BS@ psset second(dash=Index/Pattern)}%
    \immediate\write16{ Line width = \curr@ntsecondwidth
     \space-> \BS@ psset second(width=real in PostScript units)}%
    \immediate\write16{Third-level (or ternary) settings:}%
    \immediate\write16{ Color = \th@rdcolor \space-> \BS@ psset third(color=ColorDefinition)}%
    \immediate\write16{ --- GRAPHICAL (specific) ---}%
    \immediate\write16{Arrow-head:}%
    \immediate\write16{ (half-)Angle = \@rrowheadangle
     \space-> \BS@ psset arrowhead(angle=real in degrees)}%
    \immediate\write16{ Filling mode = \if@rrowhfill yes\else no\fi
     \space-> \BS@ psset arrowhead(fillmode=yes/no)}%
    \immediate\write16{ "Outside" = \if@rrowhout yes\else no\fi
     \space-> \BS@ psset arrowhead(out=yes/no)}%
    \immediate\write16{ Length = \@rrowheadlength
     \if@rrowratio\space(not active)\else\space(active)\fi
     \space-> \BS@ psset arrowhead(length=real in user coord.)}%
    \immediate\write16{ Ratio = \@rrowheadratio
     \if@rrowratio\space(active)\else\space(not active)\fi
     \space-> \BS@ psset arrowhead(ratio=real in [0,1])}%
    \immediate\write16{Curve: Roundness = \curv@roundness
     \space-> \BS@ psset curve(roundness=real in [0,0.5])}%
    \immediate\write16{Mesh: Diagonal = \c@ntrolmesh
     \space-> \BS@ psset mesh(diag=integer in {-1,0,1})}%
    \immediate\write16{Flow chart:}%
    \immediate\write16{ Arrow position = \@rrowp@s
     \space-> \BS@ psset flowchart(arrowposition=real in [0,1])}%
    \immediate\write16{ Arrow reference point = \ifcase\@rrowr@fpt start\else end\fi
     \space-> \BS@ psset flowchart(arrowrefpt = start/end)}%     
    \immediate\write16{ Line type = \ifcase\fclin@typ@ curve\else polygon\fi
     \space-> \BS@ psset flowchart(line=polygon/curve)}%
    \immediate\write16{ Padding = (\Xp@dd, \Yp@dd)
     \space-> \BS@ psset flowchart(padding = real in user coord.)}%
    \immediate\write16{\space\space\space\space(or
     \BS@ psset flowchart(xpadding=real, ypadding=real) )}%
    \immediate\write16{ Radius = \fclin@r@d
     \space-> \BS@ psset flowchart(radius=positive real in user coord.)}%
    \immediate\write16{ Shape = \fcsh@pe
     \space-> \BS@ psset flowchart(shape = rectangle, ellipse or lozenge)}%
    \immediate\write16{ Thickness = \thickn@ss
     \space-> \BS@ psset flowchart(thickness = real in user coord.)}%
    \ifTr@isDim%
    \immediate\write16{ --- 3D to 2D PROJECTION ---}%
    \immediate\write16{Projection : \typ@proj \space-> \BS@ figinit{ScaleFactorUnit, ProjType}}%
    \immediate\write16{Longitude (psi) = \v@lPsi \space-> \BS@ figset proj(psi=real in degrees)}%
    \ifcase\curr@ntproj\immediate\write16{Depth coeff. (Lambda)
     \space = \v@lTheta \space-> \BS@ figset proj(lambda=real in [0,1])}%
    \else\immediate\write16{Latitude (theta)
     \space = \v@lTheta \space-> \BS@ figset proj(theta=real in degrees)}%
    \fi%
    \ifnum\curr@ntproj=\tw@%
    \immediate\write16{Observation distance = \disob@unit
     \space-> \BS@ figset proj(dist=real in user coord.)}%
    \immediate\write16{Target point = \t@rgetpt \space-> \BS@ figset proj(targetpt=pt number)}%
     \v@lX=\ptT@unit@\wd\Bt@rget\v@lY=\ptT@unit@\ht\Bt@rget\v@lZ=\ptT@unit@\dp\Bt@rget%
    \immediate\write16{ Its coordinates are
     (\repdecn@mb{\v@lX}, \repdecn@mb{\v@lY}, \repdecn@mb{\v@lZ})}%
    \fi%
    \fi%
    \immediate\write16{====================================================================}%
    \ignorespaces}}
\ctr@ln@w{newif}\ifitis@vect@r
\ctr@ld@f\def\figvectC#1(#2,#3){{\itis@vect@rtrue\figpt#1:(#2,#3)}\ignorespaces}
\ctr@ld@f\def\Figv@ctCreg#1(#2,#3){{\itis@vect@rtrue\Figp@intreg#1:(#2,#3)}\ignorespaces}
\ctr@ln@m\figvectDBezier
\ctr@ld@f\def\figvectDBezierDD#1:#2,#3[#4,#5,#6,#7]{\ifps@cri{\s@uvc@ntr@l\et@tfigvectDBezierDD%
    \FigvectDBezier@#2,#3[#4,#5,#6,#7]\v@lX=\c@ef\v@lX\v@lY=\c@ef\v@lY%
    \Figv@ctCreg#1(\v@lX,\v@lY)\resetc@ntr@l\et@tfigvectDBezierDD}\ignorespaces\fi}
\ctr@ld@f\def\figvectDBezierTD#1:#2,#3[#4,#5,#6,#7]{\ifps@cri{\s@uvc@ntr@l\et@tfigvectDBezierTD%
    \FigvectDBezier@#2,#3[#4,#5,#6,#7]\v@lX=\c@ef\v@lX\v@lY=\c@ef\v@lY\v@lZ=\c@ef\v@lZ%
    \Figv@ctCreg#1(\v@lX,\v@lY,\v@lZ)\resetc@ntr@l\et@tfigvectDBezierTD}\ignorespaces\fi}
\ctr@ld@f\def\FigvectDBezier@#1,#2[#3,#4,#5,#6]{\setc@ntr@l{2}%
    \edef\T@{#2}\v@leur=\p@\advance\v@leur-#2pt\edef\UNmT@{\repdecn@mb{\v@leur}}%
    \ifnum#1=\tw@\def\c@ef{6}\else\def\c@ef{3}\fi%
    \figptcopy-4:/#3/\figptcopy-3:/#4/\figptcopy-2:/#5/\figptcopy-1:/#6/%
    \l@mbd@un=-4 \l@mbd@de=-\thr@@\p@rtent=\m@ne\c@lDecast%
    \ifnum#1=\tw@\c@lDCDeux{-4}{-3}\c@lDCDeux{-3}{-2}\c@lDCDeux{-4}{-3}\else%
    \l@mbd@un=-4 \l@mbd@de=-\thr@@\p@rtent=-\tw@\c@lDecast%
    \c@lDCDeux{-4}{-3}\fi\Figg@tXY{-4}}
\ctr@ln@m\c@lDCDeux
\ctr@ld@f\def\c@lDCDeuxDD#1#2{\Figg@tXY{#2}\Figg@tXYa{#1}%
    \advance\v@lX-\v@lXa\advance\v@lY-\v@lYa\Figp@intregDD#1:(\v@lX,\v@lY)}
\ctr@ld@f\def\c@lDCDeuxTD#1#2{\Figg@tXY{#2}\Figg@tXYa{#1}\advance\v@lX-\v@lXa%
    \advance\v@lY-\v@lYa\advance\v@lZ-\v@lZa\Figp@intregTD#1:(\v@lX,\v@lY,\v@lZ)}
\ctr@ln@m\figvectN
\ctr@ld@f\def\figvectNDD#1[#2,#3]{\ifps@cri{\Figg@tXYa{#2}\Figg@tXY{#3}%
    \advance\v@lX-\v@lXa\advance\v@lY-\v@lYa%
    \Figv@ctCreg#1(-\v@lY,\v@lX)}\ignorespaces\fi}
\ctr@ld@f\def\figvectNTD#1[#2,#3,#4]{\ifps@cri{\vecunitC@TD[#2,#4]\v@lmin=\v@lX\v@lmax=\v@lY%
    \v@leur=\v@lZ\vecunitC@TD[#2,#3]\c@lprovec{#1}}\ignorespaces\fi}
\ctr@ln@m\figvectNV
\ctr@ld@f\def\figvectNVDD#1[#2]{\ifps@cri{\Figg@tXY{#2}\Figv@ctCreg#1(-\v@lY,\v@lX)}\ignorespaces\fi}
\ctr@ld@f\def\figvectNVTD#1[#2,#3]{\ifps@cri{\vecunitCV@TD{#3}\v@lmin=\v@lX\v@lmax=\v@lY%
    \v@leur=\v@lZ\vecunitCV@TD{#2}\c@lprovec{#1}}\ignorespaces\fi}
\ctr@ln@m\figvectP
\ctr@ld@f\def\figvectPDD#1[#2,#3]{\ifps@cri{\Figg@tXYa{#2}\Figg@tXY{#3}%
    \advance\v@lX-\v@lXa\advance\v@lY-\v@lYa%
    \Figv@ctCreg#1(\v@lX,\v@lY)}\ignorespaces\fi}
\ctr@ld@f\def\figvectPTD#1[#2,#3]{\ifps@cri{\Figg@tXYa{#2}\Figg@tXY{#3}%
    \advance\v@lX-\v@lXa\advance\v@lY-\v@lYa\advance\v@lZ-\v@lZa%
    \Figv@ctCreg#1(\v@lX,\v@lY,\v@lZ)}\ignorespaces\fi}
\ctr@ln@m\figvectU
\ctr@ld@f\def\figvectUDD#1[#2]{\ifps@cri{\n@rmeuc\v@leur{#2}\invers@\v@leur\v@leur%
    \delt@=\repdecn@mb{\v@leur}\unit@\edef\v@ldelt@{\repdecn@mb{\delt@}}%
    \Figg@tXY{#2}\v@lX=\v@ldelt@\v@lX\v@lY=\v@ldelt@\v@lY%
    \Figv@ctCreg#1(\v@lX,\v@lY)}\ignorespaces\fi}
\ctr@ld@f\def\figvectUTD#1[#2]{\ifps@cri{\n@rmeuc\v@leur{#2}\invers@\v@leur\v@leur%
    \delt@=\repdecn@mb{\v@leur}\unit@\edef\v@ldelt@{\repdecn@mb{\delt@}}%
    \Figg@tXY{#2}\v@lX=\v@ldelt@\v@lX\v@lY=\v@ldelt@\v@lY\v@lZ=\v@ldelt@\v@lZ%
    \Figv@ctCreg#1(\v@lX,\v@lY,\v@lZ)}\ignorespaces\fi}
\ctr@ld@f\def\figvisu#1#2#3{\c@ldefproj\initb@undb@x\xdef\figforTeXFigno{\figforTeXnextFigno}%
    \s@mme=\figforTeXnextFigno\advance\s@mme\@ne\xdef\figforTeXnextFigno{\number\s@mme}%
    \setbox\b@xvisu=\hbox{\ifnum\@utoFN>\z@\figinsert{}\gdef\@utoFInDone{0}\fi\ignorespaces#3}%
    \gdef\@utoFInDone{1}\gdef\@utoFN{0}%
    \v@lXa=-\c@@rdYmin\v@lYa=\c@@rdYmax\advance\v@lYa-\c@@rdYmin%
    \v@lX=\c@@rdXmax\advance\v@lX-\c@@rdXmin%
    \setbox#1=\hbox{#2}\v@lY=-\v@lX\maxim@m{\v@lX}{\v@lX}{\wd#1}%
    \advance\v@lY\v@lX\divide\v@lY\tw@\advance\v@lY-\c@@rdXmin%
    \setbox#1=\vbox{\parindent0mm\hsize=\v@lX\vskip\v@lYa%
    \rlap{\hskip\v@lY\smash{\raise\v@lXa\box\b@xvisu}}%
    \def\t@xt@{#2}\ifx\t@xt@\empty\else\medskip\centerline{#2}\fi}\wd#1=\v@lX}
\ctr@ld@f\def\figDecrementFigno{{\xdef\figforTeXnextFigno{\figforTeXFigno}%
    \s@mme=\figforTeXFigno\advance\s@mme\m@ne\xdef\figforTeXFigno{\number\s@mme}}}
\ctr@ln@w{newbox}\Bt@rget\setbox\Bt@rget=\null
\ctr@ln@w{newbox}\BminTD@\setbox\BminTD@=\null
\ctr@ln@w{newbox}\BmaxTD@\setbox\BmaxTD@=\null
\ctr@ln@w{newif}\ifnewt@rgetpt\ctr@ln@w{newif}\ifnewdis@b
\ctr@ld@f\def\b@undb@xTD#1#2#3{%
    \relax\ifdim#1<\wd\BminTD@\global\wd\BminTD@=#1\fi%
    \relax\ifdim#2<\ht\BminTD@\global\ht\BminTD@=#2\fi%
    \relax\ifdim#3<\dp\BminTD@\global\dp\BminTD@=#3\fi%
    \relax\ifdim#1>\wd\BmaxTD@\global\wd\BmaxTD@=#1\fi%
    \relax\ifdim#2>\ht\BmaxTD@\global\ht\BmaxTD@=#2\fi%
    \relax\ifdim#3>\dp\BmaxTD@\global\dp\BmaxTD@=#3\fi}
\ctr@ld@f\def\c@ldefdisob{{\ifdim\wd\BminTD@<\maxdimen\v@leur=\wd\BmaxTD@\advance\v@leur-\wd\BminTD@%
    \delt@=\ht\BmaxTD@\advance\delt@-\ht\BminTD@\maxim@m{\v@leur}{\v@leur}{\delt@}%
    \delt@=\dp\BmaxTD@\advance\delt@-\dp\BminTD@\maxim@m{\v@leur}{\v@leur}{\delt@}%
    \v@leur=5\v@leur\else\v@leur=800pt\fi\c@ldefdisob@{\v@leur}}}
\ctr@ln@m\disob@intern
\ctr@ln@m\disob@
\ctr@ln@m\divf@ctproj
\ctr@ld@f\def\c@ldefdisob@#1{{\v@leur=#1\ifdim\v@leur<\p@\v@leur=800pt\fi%
    \xdef\disob@intern{\repdecn@mb{\v@leur}}%
    \delt@=\ptT@unit@\v@leur\xdef\disob@unit{\repdecn@mb{\delt@}}%
    \f@ctech=\@ne\loop\ifdim\v@leur>\t@n pt\divide\v@leur\t@n\multiply\f@ctech\t@n\repeat%
    \xdef\disob@{\repdecn@mb{\v@leur}}\xdef\divf@ctproj{\the\f@ctech}}%
    \global\newdis@btrue}
\ctr@ln@m\t@rgetpt
\ctr@ld@f\def\c@ldeft@rgetpt{\newt@rgetpttrue\def\t@rgetpt{CenterBoundBox}{%
    \delt@=\wd\BmaxTD@\advance\delt@-\wd\BminTD@\divide\delt@\tw@%
    \v@leur=\wd\BminTD@\advance\v@leur\delt@\global\wd\Bt@rget=\v@leur%
    \delt@=\ht\BmaxTD@\advance\delt@-\ht\BminTD@\divide\delt@\tw@%
    \v@leur=\ht\BminTD@\advance\v@leur\delt@\global\ht\Bt@rget=\v@leur%
    \delt@=\dp\BmaxTD@\advance\delt@-\dp\BminTD@\divide\delt@\tw@%
    \v@leur=\dp\BminTD@\advance\v@leur\delt@\global\dp\Bt@rget=\v@leur}}
\ctr@ln@m\c@ldefproj
\ctr@ld@f\def\c@ldefprojTD{\ifnewt@rgetpt\else\c@ldeft@rgetpt\fi\ifnewdis@b\else\c@ldefdisob\fi}
\ctr@ld@f\def\c@lprojcav{% Projection cavaliere : X = x + y L cos t, Y = z + y L sin t
    \v@lZa=\cxa@\v@lY\advance\v@lX\v@lZa%
    \v@lZa=\cxb@\v@lY\v@lY=\v@lZ\advance\v@lY\v@lZa\ignorespaces}
\ctr@ln@m\v@lcoef
\ctr@ld@f\def\c@lprojrea{% Projection realiste
    \advance\v@lX-\wd\Bt@rget\advance\v@lY-\ht\Bt@rget\advance\v@lZ-\dp\Bt@rget%
    \v@lZa=\cza@\v@lX\advance\v@lZa\czb@\v@lY\advance\v@lZa\czc@\v@lZ%
    \divide\v@lZa\divf@ctproj\advance\v@lZa\disob@ pt\invers@{\v@lZa}{\v@lZa}%
    \v@lZa=\disob@\v@lZa\edef\v@lcoef{\repdecn@mb{\v@lZa}}%
    \v@lXa=\cxa@\v@lX\advance\v@lXa\cxb@\v@lY\v@lXa=\v@lcoef\v@lXa%
    \v@lY=\cyb@\v@lY\advance\v@lY\cya@\v@lX\advance\v@lY\cyc@\v@lZ%
    \v@lY=\v@lcoef\v@lY\v@lX=\v@lXa\ignorespaces}
\ctr@ld@f\def\c@lprojort{% Projection orthogonale
    \v@lXa=\cxa@\v@lX\advance\v@lXa\cxb@\v@lY%
    \v@lY=\cyb@\v@lY\advance\v@lY\cya@\v@lX\advance\v@lY\cyc@\v@lZ%
    \v@lX=\v@lXa\ignorespaces}
\ctr@ld@f\def\Figptpr@j#1:#2/#3/{{\Figg@tXY{#3}\superc@lprojSP%
    \Figp@intregDD#1:{#2}(\v@lX,\v@lY)}\ignorespaces}
\ctr@ln@m\figsetobdist
\ctr@ld@f\def\figsetobdistDD{\un@v@ilable{figsetobdist}}
\ctr@ld@f\def\figsetobdistTD(#1){{\ifcurr@ntPS%
    \immediate\write16{*** \BS@ figsetobdist is ignored inside a
     \BS@ psbeginfig-\BS@ psendfig block.}%
    \else\v@leur=#1\unit@\c@ldefdisob@{\v@leur}\fi}\ignorespaces}
\ctr@ln@m\c@lprojSP
\ctr@ln@m\curr@ntproj
\ctr@ln@m\typ@proj
\ctr@ln@m\superc@lprojSP
\ctr@ld@f\def\Figs@tproj#1{%
    \if#13 \d@faultproj\else\if#1c\d@faultproj%
    \else\if#1o\xdef\curr@ntproj{1}\xdef\typ@proj{orthogonal}%
         \figsetviewTD(\def@ultpsi,\def@ulttheta)%
         \global\let\c@lprojSP=\c@lprojort\global\let\superc@lprojSP=\c@lprojort%
    \else\if#1r\xdef\curr@ntproj{2}\xdef\typ@proj{realistic}%
         \figsetviewTD(\def@ultpsi,\def@ulttheta)%
         \global\let\c@lprojSP=\c@lprojrea\global\let\superc@lprojSP=\c@lprojrea%
    \else\d@faultproj\message{*** Unknown projection. Cavalier projection assumed.}%
    \fi\fi\fi\fi}
\ctr@ld@f\def\d@faultproj{\xdef\curr@ntproj{0}\xdef\typ@proj{cavalier}\figsetviewTD(\def@ultpsi,0.5)%
         \global\let\c@lprojSP=\c@lprojcav\global\let\superc@lprojSP=\c@lprojcav}
\ctr@ln@m\figsettarget
\ctr@ld@f\def\figsettargetDD{\un@v@ilable{figsettarget}}
\ctr@ld@f\def\figsettargetTD[#1]{{\ifcurr@ntPS%
    \immediate\write16{*** \BS@ figsettarget is ignored inside a
     \BS@ psbeginfig-\BS@ psendfig block.}%
    \else\global\newt@rgetpttrue\xdef\t@rgetpt{#1}\Figg@tXY{#1}\global\wd\Bt@rget=\v@lX%
    \global\ht\Bt@rget=\v@lY\global\dp\Bt@rget=\v@lZ\fi}\ignorespaces}
\ctr@ln@m\figsetview
\ctr@ld@f\def\figsetviewDD{\un@v@ilable{figsetview}}
\ctr@ld@f\def\figsetviewTD(#1){\ifcurr@ntPS%
     \immediate\write16{*** \BS@ figsetview is ignored inside a
     \BS@ psbeginfig-\BS@ psendfig block.}\else\Figsetview@#1,:\fi\ignorespaces}
\ctr@ld@f\def\Figsetview@#1,#2:{{\xdef\v@lPsi{#1}\def\t@xt@{#2}%
    \ifx\t@xt@\empty\def\@rgdeux{\v@lTheta}\else\X@rgdeux@#2\fi%
    \c@ssin{\costhet@}{\sinthet@}{#1}\v@lmin=\costhet@ pt\v@lmax=\sinthet@ pt%
    \ifcase\curr@ntproj%
    \v@leur=\@rgdeux\v@lmin\xdef\cxa@{\repdecn@mb{\v@leur}}%
    \v@leur=\@rgdeux\v@lmax\xdef\cxb@{\repdecn@mb{\v@leur}}\v@leur=\@rgdeux pt%
    \relax\ifdim\v@leur>\p@\message{*** Lambda too large ! See \BS@ figset proj() !}\fi%
    \else%
    \v@lmax=-\v@lmax\xdef\cxa@{\repdecn@mb{\v@lmax}}\xdef\cxb@{\costhet@}%
    \ifx\t@xt@\empty\edef\@rgdeux{\def@ulttheta}\fi\c@ssin{\C@}{\S@}{\@rgdeux}%
    \v@lmax=-\S@ pt%
    \v@leur=\v@lmax\v@leur=\costhet@\v@leur\xdef\cya@{\repdecn@mb{\v@leur}}%
    \v@leur=\v@lmax\v@leur=\sinthet@\v@leur\xdef\cyb@{\repdecn@mb{\v@leur}}%
    \xdef\cyc@{\C@}\v@lmin=-\C@ pt%
    \v@leur=\v@lmin\v@leur=\costhet@\v@leur\xdef\cza@{\repdecn@mb{\v@leur}}%
    \v@leur=\v@lmin\v@leur=\sinthet@\v@leur\xdef\czb@{\repdecn@mb{\v@leur}}%
    \xdef\czc@{\repdecn@mb{\v@lmax}}\fi%
    \xdef\v@lTheta{\@rgdeux}}}
\ctr@ld@f\def\def@ultpsi{40}
\ctr@ld@f\def\def@ulttheta{25}
\ctr@ln@m\l@debut
\ctr@ln@m\n@mref
\ctr@ld@f\def\figset#1(#2){\def\t@xt@{#1}\ifx\t@xt@\empty\trtlis@rg{#2}{\Figsetwr@te}% write
    \else\keln@mde#1|%
    \def\n@mref{pr}\ifx\l@debut\n@mref\ifcurr@ntPS% projection
     \immediate\write16{*** \BS@ figset proj(...) is ignored inside a
     \BS@ psbeginfig-\BS@ psendfig block.}\else\trtlis@rg{#2}{\Figsetpr@j}\fi\else%
    \def\n@mref{wr}\ifx\l@debut\n@mref\trtlis@rg{#2}{\Figsetwr@te}\else% write
    \immediate\write16{*** Unknown keyword: \BS@ figset #1(...)}%
    \fi\fi\fi\ignorespaces}
\ctr@ld@f\def\Figsetpr@j#1=#2|{\keln@mtr#1|%
    \def\n@mref{dep}\ifx\l@debut\n@mref\Figsetd@p{#2}\else% depth (lambda)
    \def\n@mref{dis}\ifx\l@debut\n@mref%
     \ifnum\curr@ntproj=\tw@\figsetobdist(#2)\else\Figset@rr\fi\else% dist
    \def\n@mref{lam}\ifx\l@debut\n@mref\Figsetd@p{#2}\else% depth (lambda)
    \def\n@mref{lat}\ifx\l@debut\n@mref\Figsetth@{#2}\else% latitude (theta)
    \def\n@mref{lon}\ifx\l@debut\n@mref\figsetview(#2)\else% longitude (psi)
    \def\n@mref{psi}\ifx\l@debut\n@mref\figsetview(#2)\else% longitude (psi)
    \def\n@mref{tar}\ifx\l@debut\n@mref%
     \ifnum\curr@ntproj=\tw@\figsettarget[#2]\else\Figset@rr\fi\else% target point
    \def\n@mref{the}\ifx\l@debut\n@mref\Figsetth@{#2}\else% latitude (theta)
    \immediate\write16{*** Unknown attribute: \BS@ figset proj(..., #1=...).}%
    \fi\fi\fi\fi\fi\fi\fi\fi}
\ctr@ld@f\def\Figsetd@p#1{\ifnum\curr@ntproj=\z@\figsetview(\v@lPsi,#1)\else\Figset@rr\fi}
\ctr@ld@f\def\Figsetth@#1{\ifnum\curr@ntproj=\z@\Figset@rr\else\figsetview(\v@lPsi,#1)\fi}
\ctr@ld@f\def\Figset@rr{\message{*** \BS@ figset proj(): Attribute "\n@mref" ignored, incompatible
    with current projection}}
\ctr@ld@f\def\initb@undb@xTD{\wd\BminTD@=\maxdimen\ht\BminTD@=\maxdimen\dp\BminTD@=\maxdimen%
    \wd\BmaxTD@=-\maxdimen\ht\BmaxTD@=-\maxdimen\dp\BmaxTD@=-\maxdimen}
\ctr@ln@w{newbox}\Gb@x      % boite a tout faire
\ctr@ln@w{newbox}\Gb@xSC    % boite qui contient le point marker
\ctr@ln@w{newtoks}\c@nsymb  % the point marker
\ctr@ln@w{newif}\ifr@undcoord\ctr@ln@w{newif}\ifunitpr@sent
\ctr@ld@f\def\unssqrttw@{0.707106 }
\ctr@ld@f\def\figAst{\raise-1.15ex\hbox{$\ast$}}
\ctr@ld@f\def\figBullet{\raise-1.15ex\hbox{$\bullet$}}
\ctr@ld@f\def\figCirc{\raise-1.15ex\hbox{$\circ$}}
\ctr@ld@f\def\figDiamond{\raise-1.15ex\hbox{$\diamond$}}%
\ctr@ld@f\def\boxit#1#2{\leavevmode\hbox{\vrule\vbox{\hrule\vglue#1%
    \vtop{\hbox{\kern#1{#2}\kern#1}\vglue#1\hrule}}\vrule}}
\ctr@ld@f\def\centertext#1#2{\vbox{\hsize#1\parindent0cm%
    \leftskip=0pt plus 1fil\rightskip=0pt plus 1fil\parfillskip=0pt{#2}}}
\ctr@ld@f\def\lefttext#1#2{\vbox{\hsize#1\parindent0cm\rightskip=0pt plus 1fil#2}}
\ctr@ld@f\def\c@nterpt{\ignorespaces%
    \kern-.5\wd\Gb@xSC%
    \raise-.5\ht\Gb@xSC\rlap{\hbox{\raise.5\dp\Gb@xSC\hbox{\copy\Gb@xSC}}}%
    \kern .5\wd\Gb@xSC\ignorespaces}
\ctr@ld@f\def\b@undb@xSC#1#2{{\v@lXa=#1\v@lYa=#2%
    \v@leur=\ht\Gb@xSC\advance\v@leur\dp\Gb@xSC%
    \advance\v@lXa-.5\wd\Gb@xSC\advance\v@lYa-.5\v@leur\b@undb@x{\v@lXa}{\v@lYa}%
    \advance\v@lXa\wd\Gb@xSC\advance\v@lYa\v@leur\b@undb@x{\v@lXa}{\v@lYa}}}
\ctr@ln@m\Dist@n
\ctr@ln@m\l@suite
\ctr@ld@f\def\@keldist#1#2{\edef\Dist@n{#2}\y@tiunit{\Dist@n}%
    \ifunitpr@sent#1=\Dist@n\else#1=\Dist@n\unit@\fi}
\ctr@ld@f\def\y@tiunit#1{\unitpr@sentfalse\expandafter\y@tiunit@#1:}
\ctr@ld@f\def\y@tiunit@#1#2:{\ifcat#1a\unitpr@senttrue\else\def\l@suite{#2}%
    \ifx\l@suite\empty\else\y@tiunit@#2:\fi\fi}
\ctr@ln@m\figcoord
\ctr@ld@f\def\figcoordDD#1{{\v@lX=\ptT@unit@\v@lX\v@lY=\ptT@unit@\v@lY%
    \ifr@undcoord\ifcase#1\v@leur=0.5pt\or\v@leur=0.05pt\or\v@leur=0.005pt%
    \or\v@leur=0.0005pt\else\v@leur=\z@\fi%
    \ifdim\v@lX<\z@\advance\v@lX-\v@leur\else\advance\v@lX\v@leur\fi%
    \ifdim\v@lY<\z@\advance\v@lY-\v@leur\else\advance\v@lY\v@leur\fi\fi%
    (\@ffichnb{#1}{\repdecn@mb{\v@lX}},\ifmmode\else\thinspace\fi%
    \@ffichnb{#1}{\repdecn@mb{\v@lY}})}}
\ctr@ld@f\def\@ffichnb#1#2{{\def\@@ffich{\@ffich#1(}\edef\n@mbre{#2}%
    \expandafter\@@ffich\n@mbre)}}
\ctr@ld@f\def\@ffich#1(#2.#3){{#2\ifnum#1>\z@.\fi\def\dig@ts{#3}\s@mme=\z@%
    \loop\ifnum\s@mme<#1\expandafter\@ffichdec\dig@ts:\advance\s@mme\@ne\repeat}}
\ctr@ld@f\def\@ffichdec#1#2:{\relax#1\def\dig@ts{#20}}
\ctr@ld@f\def\figcoordTD#1{{\v@lX=\ptT@unit@\v@lX\v@lY=\ptT@unit@\v@lY\v@lZ=\ptT@unit@\v@lZ%
    \ifr@undcoord\ifcase#1\v@leur=0.5pt\or\v@leur=0.05pt\or\v@leur=0.005pt%
    \or\v@leur=0.0005pt\else\v@leur=\z@\fi%
    \ifdim\v@lX<\z@\advance\v@lX-\v@leur\else\advance\v@lX\v@leur\fi%
    \ifdim\v@lY<\z@\advance\v@lY-\v@leur\else\advance\v@lY\v@leur\fi%
    \ifdim\v@lZ<\z@\advance\v@lZ-\v@leur\else\advance\v@lZ\v@leur\fi\fi%
    (\@ffichnb{#1}{\repdecn@mb{\v@lX}},\ifmmode\else\thinspace\fi%
     \@ffichnb{#1}{\repdecn@mb{\v@lY}},\ifmmode\else\thinspace\fi%
     \@ffichnb{#1}{\repdecn@mb{\v@lZ}})}}
\ctr@ld@f\def\figsetroundcoord#1{\expandafter\Figsetr@undcoord#1:\ignorespaces}
\ctr@ld@f\def\Figsetr@undcoord#1#2:{\if#1n\r@undcoordfalse\else\r@undcoordtrue\fi}
\ctr@ld@f\def\Figsetwr@te#1=#2|{\keln@mun#1|%
    \def\n@mref{m}\ifx\l@debut\n@mref\figsetmark{#2}\else% mark
    \immediate\write16{*** Unknown attribute: \BS@ figset (..., #1=...)}%
    \fi}
\ctr@ld@f\def\figsetmark#1{\c@nsymb={#1}\setbox\Gb@xSC=\hbox{\the\c@nsymb}\ignorespaces}
\ctr@ln@m\ptn@me
\ctr@ld@f\def\figsetptname#1{\def\ptn@me##1{#1}\ignorespaces}
\ctr@ld@f\def\FigWrit@L#1:#2(#3,#4){\ignorespaces\@keldist\v@leur{#3}\@keldist\delt@{#4}%
    \C@rp@r@m\def\list@num{#1}\@ecfor\p@int:=\list@num\do{\FigWrit@pt{\p@int}{#2}}}
\ctr@ld@f\def\FigWrit@pt#1#2{\FigWp@r@m{#1}{#2}\Vc@rrect\figWp@si%
    \ifdim\wd\Gb@xSC>\z@\b@undb@xSC{\v@lX}{\v@lY}\fi\figWBB@x}
\ctr@ld@f\def\FigWp@r@m#1#2{\Figg@tXY{#1}%
    \setbox\Gb@x=\hbox{\def\t@xt@{#2}\ifx\t@xt@\empty\Figg@tT{#1}\else#2\fi}\c@lprojSP}
\ctr@ld@f\let\Vc@rrect=\relax
\ctr@ld@f\let\C@rp@r@m=\relax
\ctr@ld@f\def\figwrite[#1]#2{{\ignorespaces\def\list@num{#1}\@ecfor\p@int:=\list@num\do{%
    \setbox\Gb@x=\hbox{\def\t@xt@{#2}\ifx\t@xt@\empty\Figg@tT{\p@int}\else#2\fi}%
    \Figwrit@{\p@int}}}\ignorespaces}
\ctr@ld@f\def\Figwrit@#1{\Figg@tXY{#1}\c@lprojSP%
    \rlap{\kern\v@lX\raise\v@lY\hbox{\unhcopy\Gb@x}}\v@leur=\v@lY%
    \advance\v@lY\ht\Gb@x\b@undb@x{\v@lX}{\v@lY}\advance\v@lX\wd\Gb@x%
    \v@lY=\v@leur\advance\v@lY-\dp\Gb@x\b@undb@x{\v@lX}{\v@lY}}
\ctr@ld@f\def\figwritec[#1]#2{{\ignorespaces\def\list@num{#1}%
    \@ecfor\p@int:=\list@num\do{\Figwrit@c{\p@int}{#2}}}\ignorespaces}
\ctr@ld@f\def\Figwrit@c#1#2{\FigWp@r@m{#1}{#2}%
    \rlap{\kern\v@lX\raise\v@lY\hbox{\rlap{\kern-.5\wd\Gb@x%
    \raise-.5\ht\Gb@x\hbox{\raise.5\dp\Gb@x\hbox{\unhcopy\Gb@x}}}}}%
    \v@leur=\ht\Gb@x\advance\v@leur\dp\Gb@x%
    \advance\v@lX-.5\wd\Gb@x\advance\v@lY-.5\v@leur\b@undb@x{\v@lX}{\v@lY}%
    \advance\v@lX\wd\Gb@x\advance\v@lY\v@leur\b@undb@x{\v@lX}{\v@lY}}
\ctr@ld@f\def\figwritep[#1]{{\ignorespaces\def\list@num{#1}\setbox\Gb@x=\hbox{\c@nterpt}%
    \@ecfor\p@int:=\list@num\do{\Figwrit@{\p@int}}}\ignorespaces}
\ctr@ld@f\def\figwritew#1:#2(#3){\figwritegcw#1:{#2}(#3,0pt)}
\ctr@ld@f\def\figwritee#1:#2(#3){\figwritegce#1:{#2}(#3,0pt)}
\ctr@ld@f\def\figwriten#1:#2(#3){{\def\Vc@rrect{\v@lZ=\v@leur\advance\v@lZ\dp\Gb@x}%
    \Figwrit@NS#1:{#2}(#3)}\ignorespaces}
\ctr@ld@f\def\figwrites#1:#2(#3){{\def\Vc@rrect{\v@lZ=-\v@leur\advance\v@lZ-\ht\Gb@x}%
    \Figwrit@NS#1:{#2}(#3)}\ignorespaces}
\ctr@ld@f\def\Figwrit@NS#1:#2(#3){\let\figWp@si=\FigWp@siNS\let\figWBB@x=\FigWBB@xNS%
    \FigWrit@L#1:{#2}(#3,0pt)}
\ctr@ld@f\def\FigWp@siNS{\rlap{\kern\v@lX\raise\v@lY\hbox{\rlap{\kern-.5\wd\Gb@x%
    \raise\v@lZ\hbox{\unhcopy\Gb@x}}\c@nterpt}}}
\ctr@ld@f\def\FigWBB@xNS{\advance\v@lY\v@lZ%
    \advance\v@lY-\dp\Gb@x\advance\v@lX-.5\wd\Gb@x\b@undb@x{\v@lX}{\v@lY}%
    \advance\v@lY\ht\Gb@x\advance\v@lY\dp\Gb@x%
    \advance\v@lX\wd\Gb@x\b@undb@x{\v@lX}{\v@lY}}
\ctr@ld@f\def\figwritenw#1:#2(#3){{\let\figWp@si=\FigWp@sigW\let\figWBB@x=\FigWBB@xgWE%
    \def\C@rp@r@m{\v@leur=\unssqrttw@\v@leur\delt@=\v@leur%
    \ifdim\delt@=\z@\delt@=\epsil@n\fi}\let@xte={-}\FigWrit@L#1:{#2}(#3,0pt)}\ignorespaces}
\ctr@ld@f\def\figwritesw#1:#2(#3){{\let\figWp@si=\FigWp@sigW\let\figWBB@x=\FigWBB@xgWE%
    \def\C@rp@r@m{\v@leur=\unssqrttw@\v@leur\delt@=-\v@leur%
    \ifdim\delt@=\z@\delt@=-\epsil@n\fi}\let@xte={-}\FigWrit@L#1:{#2}(#3,0pt)}\ignorespaces}
\ctr@ld@f\def\figwritene#1:#2(#3){{\let\figWp@si=\FigWp@sigE\let\figWBB@x=\FigWBB@xgWE%
    \def\C@rp@r@m{\v@leur=\unssqrttw@\v@leur\delt@=\v@leur%
    \ifdim\delt@=\z@\delt@=\epsil@n\fi}\let@xte={}\FigWrit@L#1:{#2}(#3,0pt)}\ignorespaces}
\ctr@ld@f\def\figwritese#1:#2(#3){{\let\figWp@si=\FigWp@sigE\let\figWBB@x=\FigWBB@xgWE%
    \def\C@rp@r@m{\v@leur=\unssqrttw@\v@leur\delt@=-\v@leur%
    \ifdim\delt@=\z@\delt@=-\epsil@n\fi}\let@xte={}\FigWrit@L#1:{#2}(#3,0pt)}\ignorespaces}
\ctr@ld@f\def\figwritegw#1:#2(#3,#4){{\let\figWp@si=\FigWp@sigW\let\figWBB@x=\FigWBB@xgWE%
    \let@xte={-}\FigWrit@L#1:{#2}(#3,#4)}\ignorespaces}
\ctr@ld@f\def\figwritege#1:#2(#3,#4){{\let\figWp@si=\FigWp@sigE\let\figWBB@x=\FigWBB@xgWE%
    \let@xte={}\FigWrit@L#1:{#2}(#3,#4)}\ignorespaces}
\ctr@ld@f\def\FigWp@sigW{\v@lXa=\z@\v@lYa=\ht\Gb@x\advance\v@lYa\dp\Gb@x%
    \ifdim\delt@>\z@\relax%
    \rlap{\kern\v@lX\raise\v@lY\hbox{\rlap{\kern-\wd\Gb@x\kern-\v@leur%
          \raise\delt@\hbox{\raise\dp\Gb@x\hbox{\unhcopy\Gb@x}}}\c@nterpt}}%
    \else\ifdim\delt@<\z@\relax\v@lYa=-\v@lYa%
    \rlap{\kern\v@lX\raise\v@lY\hbox{\rlap{\kern-\wd\Gb@x\kern-\v@leur%
          \raise\delt@\hbox{\raise-\ht\Gb@x\hbox{\unhcopy\Gb@x}}}\c@nterpt}}%
    \else\v@lXa=-.5\v@lYa%
    \rlap{\kern\v@lX\raise\v@lY\hbox{\rlap{\kern-\wd\Gb@x\kern-\v@leur%
          \raise-.5\ht\Gb@x\hbox{\raise.5\dp\Gb@x\hbox{\unhcopy\Gb@x}}}\c@nterpt}}%
    \fi\fi}
\ctr@ld@f\def\FigWp@sigE{\v@lXa=\z@\v@lYa=\ht\Gb@x\advance\v@lYa\dp\Gb@x%
    \ifdim\delt@>\z@\relax%
    \rlap{\kern\v@lX\raise\v@lY\hbox{\c@nterpt\kern\v@leur%
          \raise\delt@\hbox{\raise\dp\Gb@x\hbox{\unhcopy\Gb@x}}}}%
    \else\ifdim\delt@<\z@\relax\v@lYa=-\v@lYa%
    \rlap{\kern\v@lX\raise\v@lY\hbox{\c@nterpt\kern\v@leur%
          \raise\delt@\hbox{\raise-\ht\Gb@x\hbox{\unhcopy\Gb@x}}}}%
    \else\v@lXa=-.5\v@lYa%
    \rlap{\kern\v@lX\raise\v@lY\hbox{\c@nterpt\kern\v@leur%
          \raise-.5\ht\Gb@x\hbox{\raise.5\dp\Gb@x\hbox{\unhcopy\Gb@x}}}}%
    \fi\fi}
\ctr@ld@f\def\FigWBB@xgWE{\advance\v@lY\delt@%
    \advance\v@lX\the\let@xte\v@leur\advance\v@lY\v@lXa\b@undb@x{\v@lX}{\v@lY}%
    \advance\v@lX\the\let@xte\wd\Gb@x\advance\v@lY\v@lYa\b@undb@x{\v@lX}{\v@lY}}
\ctr@ld@f\def\figwritegcw#1:#2(#3,#4){{\let\figWp@si=\FigWp@sigcW\let\figWBB@x=\FigWBB@xgcWE%
    \let@xte={-}\FigWrit@L#1:{#2}(#3,#4)}\ignorespaces}
\ctr@ld@f\def\figwritegce#1:#2(#3,#4){{\let\figWp@si=\FigWp@sigcE\let\figWBB@x=\FigWBB@xgcWE%
    \let@xte={}\FigWrit@L#1:{#2}(#3,#4)}\ignorespaces}
\ctr@ld@f\def\FigWp@sigcW{\rlap{\kern\v@lX\raise\v@lY\hbox{\rlap{\kern-\wd\Gb@x\kern-\v@leur%
     \raise-.5\ht\Gb@x\hbox{\raise\delt@\hbox{\raise.5\dp\Gb@x\hbox{\unhcopy\Gb@x}}}}%
     \c@nterpt}}}
\ctr@ld@f\def\FigWp@sigcE{\rlap{\kern\v@lX\raise\v@lY\hbox{\c@nterpt\kern\v@leur%
    \raise-.5\ht\Gb@x\hbox{\raise\delt@\hbox{\raise.5\dp\Gb@x\hbox{\unhcopy\Gb@x}}}}}}
\ctr@ld@f\def\FigWBB@xgcWE{\v@lZ=\ht\Gb@x\advance\v@lZ\dp\Gb@x%
    \advance\v@lX\the\let@xte\v@leur\advance\v@lY\delt@\advance\v@lY.5\v@lZ%
    \b@undb@x{\v@lX}{\v@lY}%
    \advance\v@lX\the\let@xte\wd\Gb@x\advance\v@lY-\v@lZ\b@undb@x{\v@lX}{\v@lY}}
\ctr@ld@f\def\figwritebn#1:#2(#3){{\def\Vc@rrect{\v@lZ=\v@leur}\Figwrit@NS#1:{#2}(#3)}\ignorespaces}
\ctr@ld@f\def\figwritebs#1:#2(#3){{\def\Vc@rrect{\v@lZ=-\v@leur}\Figwrit@NS#1:{#2}(#3)}\ignorespaces}
\ctr@ld@f\def\figwritebw#1:#2(#3){{\let\figWp@si=\FigWp@sibW\let\figWBB@x=\FigWBB@xbWE%
    \let@xte={-}\FigWrit@L#1:{#2}(#3,0pt)}\ignorespaces}
\ctr@ld@f\def\figwritebe#1:#2(#3){{\let\figWp@si=\FigWp@sibE\let\figWBB@x=\FigWBB@xbWE%
    \let@xte={}\FigWrit@L#1:{#2}(#3,0pt)}\ignorespaces}
\ctr@ld@f\def\FigWp@sibW{\rlap{\kern\v@lX\raise\v@lY\hbox{\rlap{\kern-\wd\Gb@x\kern-\v@leur%
          \hbox{\unhcopy\Gb@x}}\c@nterpt}}}
\ctr@ld@f\def\FigWp@sibE{\rlap{\kern\v@lX\raise\v@lY\hbox{\c@nterpt\kern\v@leur%
          \hbox{\unhcopy\Gb@x}}}}
\ctr@ld@f\def\FigWBB@xbWE{\v@lZ=\ht\Gb@x\advance\v@lZ\dp\Gb@x%
    \advance\v@lX\the\let@xte\v@leur\advance\v@lY\ht\Gb@x\b@undb@x{\v@lX}{\v@lY}%
    \advance\v@lX\the\let@xte\wd\Gb@x\advance\v@lY-\v@lZ\b@undb@x{\v@lX}{\v@lY}}
\ctr@ln@w{newread}\frf@g  \ctr@ln@w{newwrite}\fwf@g
\ctr@ln@w{newif}\ifcurr@ntPS
\ctr@ln@w{newif}\ifps@cri
\ctr@ln@w{newif}\ifUse@llipse
\ctr@ln@w{newif}\ifpsdebugmode \psdebugmodefalse 
\ctr@ln@w{newif}\ifPDFm@ke
\ifx\pdfliteral\undefined\else\ifnum\pdfoutput>\z@\PDFm@ketrue\fi\fi
\ctr@ld@f\def\initPDF@rDVI{%
\ifPDFm@ke
 \let\figscan=\figscan@E
 \let\newGr@FN=\newGr@FNPDF
 \ctr@ld@f\def\c@mcurveto{c}
 \ctr@ld@f\def\c@mfill{f}
 \ctr@ld@f\def\c@mgsave{q}
 \ctr@ld@f\def\c@mgrestore{Q}
 \ctr@ld@f\def\c@mlineto{l}
 \ctr@ld@f\def\c@mmoveto{m}
 \ctr@ld@f\def\c@msetgray{g}     \ctr@ld@f\def\c@msetgrayStroke{G}
 \ctr@ld@f\def\c@msetcmykcolor{k}\ctr@ld@f\def\c@msetcmykcolorStroke{K}
 \ctr@ld@f\def\c@msetrgbcolor{rg}\ctr@ld@f\def\c@msetrgbcolorStroke{RG}
 \ctr@ld@f\def\d@fprimarC@lor{\curr@ntcolor\space\curr@ntcolorc@md%
               \space\curr@ntcolor\space\curr@ntcolorc@mdStroke}
 \ctr@ld@f\def\d@fsecondC@lor{\sec@ndcolor\space\sec@ndcolorc@md%
               \space\sec@ndcolor\space\sec@ndcolorc@mdStroke}
 \ctr@ld@f\def\d@fthirdC@lor{\th@rdcolor\space\th@rdcolorc@md%
              \space\th@rdcolor\space\th@rdcolorc@mdStroke}
 \ctr@ld@f\def\c@msetdash{d}
 \ctr@ld@f\def\c@msetlinejoin{j}
 \ctr@ld@f\def\c@msetlinewidth{w}
 \ctr@ld@f\def\f@gclosestroke{\immediate\write\fwf@g{s}}
 \ctr@ld@f\def\f@gfill{\immediate\write\fwf@g{\fillc@md}}% Voir la def de \fillc@md ****
 \ctr@ld@f\def\f@gnewpath{}
 \ctr@ld@f\def\f@gstroke{\immediate\write\fwf@g{S}}
\else
 \let\figinsertE=\figinsert
 \let\newGr@FN=\newGr@FNDVI
 \ctr@ld@f\def\c@mcurveto{curveto}
 \ctr@ld@f\def\c@mfill{fill}
 \ctr@ld@f\def\c@mgsave{gsave}
 \ctr@ld@f\def\c@mgrestore{grestore}
 \ctr@ld@f\def\c@mlineto{lineto}
 \ctr@ld@f\def\c@mmoveto{moveto}
 \ctr@ld@f\def\c@msetgray{setgray}          \ctr@ld@f\def\c@msetgrayStroke{}
 \ctr@ld@f\def\c@msetcmykcolor{setcmykcolor}\ctr@ld@f\def\c@msetcmykcolorStroke{}
 \ctr@ld@f\def\c@msetrgbcolor{setrgbcolor}  \ctr@ld@f\def\c@msetrgbcolorStroke{}
 \ctr@ld@f\def\d@fprimarC@lor{\curr@ntcolor\space\curr@ntcolorc@md}
 \ctr@ld@f\def\d@fsecondC@lor{\sec@ndcolor\space\sec@ndcolorc@md}
 \ctr@ld@f\def\d@fthirdC@lor{\th@rdcolor\space\th@rdcolorc@md}
 \ctr@ld@f\def\c@msetdash{setdash}
 \ctr@ld@f\def\c@msetlinejoin{setlinejoin}
 \ctr@ld@f\def\c@msetlinewidth{setlinewidth}
 \ctr@ld@f\def\f@gclosestroke{\immediate\write\fwf@g{closepath\space stroke}}
 \ctr@ld@f\def\f@gfill{\immediate\write\fwf@g{\fillc@md}}
 \ctr@ld@f\def\f@gnewpath{\immediate\write\fwf@g{newpath}}
 \ctr@ld@f\def\f@gstroke{\immediate\write\fwf@g{stroke}}
\fi}
\ctr@ld@f\def\c@pypsfile#1#2{\c@pyfil@{\immediate\write#1}{#2}}
\ctr@ld@f\def\Figinclud@PDF#1#2{\openin\frf@g=#1\pdfliteral{q #2 0 0 #2 0 0 cm}%
    \c@pyfil@{\pdfliteral}{\frf@g}\pdfliteral{Q}\closein\frf@g}
\ctr@ln@w{newif}\ifmored@ta
\ctr@ln@m\bl@nkline
\ctr@ld@f\def\c@pyfil@#1#2{\def\bl@nkline{\par}{\catcode`\%=12
    \loop\ifeof#2\mored@tafalse\else\mored@tatrue\immediate\read#2 to\tr@c
    \ifx\tr@c\bl@nkline\else#1{\tr@c}\fi\fi\ifmored@ta\repeat}}
\ctr@ld@f\def\keln@mun#1#2|{\def\l@debut{#1}\def\l@suite{#2}}
\ctr@ld@f\def\keln@mde#1#2#3|{\def\l@debut{#1#2}\def\l@suite{#3}}
\ctr@ld@f\def\keln@mtr#1#2#3#4|{\def\l@debut{#1#2#3}\def\l@suite{#4}}
\ctr@ld@f\def\keln@mqu#1#2#3#4#5|{\def\l@debut{#1#2#3#4}\def\l@suite{#5}}
\ctr@ld@f\let\@psffilein=\frf@g % file to \read
\ctr@ln@w{newif}\if@psffileok    % continue looking for the bounding box?
\ctr@ln@w{newif}\if@psfbbfound   % success?
\ctr@ln@w{newif}\if@psfverbose   % report what you're making?
\@psfverbosetrue
\ctr@ln@m\@psfllx \ctr@ln@m\@psflly
\ctr@ln@m\@psfurx \ctr@ln@m\@psfury
\ctr@ln@m\resetcolonc@tcode
\ctr@ld@f\def\@psfgetbb#1{\global\@psfbbfoundfalse%
\global\def\@psfllx{0}\global\def\@psflly{0}%
\global\def\@psfurx{30}\global\def\@psfury{30}%
\openin\@psffilein=#1\relax
\ifeof\@psffilein\errmessage{I couldn't open #1, will ignore it}\else
   \edef\resetcolonc@tcode{\catcode`\noexpand\:\the\catcode`\:\relax}%
   {\@psffileoktrue \chardef\other=12
    \def\do##1{\catcode`##1=\other}\dospecials \catcode`\ =10 \resetcolonc@tcode
    \loop
       \read\@psffilein to \@psffileline
       \ifeof\@psffilein\@psffileokfalse\else
          \expandafter\@psfaux\@psffileline:. \\%
       \fi
   \if@psffileok\repeat
   \if@psfbbfound\else
    \if@psfverbose\message{No bounding box comment in #1; using defaults}\fi\fi
   }\closein\@psffilein\fi}%
\ctr@ln@m\@psfbblit
\ctr@ln@m\@psfpercent
{\catcode`\%=12 \global\let\@psfpercent=%\global\def\@psfbblit{%BoundingBox}}%
\ctr@ln@m\@psfaux
\long\def\@psfaux#1#2:#3\\{\ifx#1\@psfpercent
   \def\testit{#2}\ifx\testit\@psfbblit
      \@psfgrab #3 . . . \\%
      \@psffileokfalse
      \global\@psfbbfoundtrue
   \fi\else\ifx#1\par\else\@psffileokfalse\fi\fi}%
\ctr@ld@f\def\@psfempty{}%
\ctr@ld@f\def\@psfgrab #1 #2 #3 #4 #5\\{%
\global\def\@psfllx{#1}\ifx\@psfllx\@psfempty
      \@psfgrab #2 #3 #4 #5 .\\\else
   \global\def\@psflly{#2}%
   \global\def\@psfurx{#3}\global\def\@psfury{#4}\fi}%
\ctr@ld@f\def\PSwrit@cmd#1#2#3{{\Figg@tXY{#1}\c@lprojSP\b@undb@x{\v@lX}{\v@lY}%
    \v@lX=\ptT@ptps\v@lX\v@lY=\ptT@ptps\v@lY%
    \immediate\write#3{\repdecn@mb{\v@lX}\space\repdecn@mb{\v@lY}\space#2}}}
\ctr@ld@f\def\PSwrit@cmdS#1#2#3#4#5{{\Figg@tXY{#1}\c@lprojSP\b@undb@x{\v@lX}{\v@lY}%
    \global\result@t=\v@lX\global\result@@t=\v@lY%
    \v@lX=\ptT@ptps\v@lX\v@lY=\ptT@ptps\v@lY%
    \immediate\write#3{\repdecn@mb{\v@lX}\space\repdecn@mb{\v@lY}\space#2}}%
    \edef#4{\the\result@t}\edef#5{\the\result@@t}}
\ctr@ld@f\def\psaltitude#1[#2,#3,#4]{{\ifcurr@ntPS\ifps@cri%
    \PSc@mment{psaltitude Square Dim=#1, Triangle=[#2 / #3,#4]}%
    \s@uvc@ntr@l\et@tpsaltitude\resetc@ntr@l{2}\figptorthoprojline-5:=#2/#3,#4/%
    \figvectP -1[#3,#4]\n@rminf{\v@leur}{-1}\vecunit@{-3}{-1}%
    \figvectP -1[-5,#3]\n@rminf{\v@lmin}{-1}\figvectP -2[-5,#4]\n@rminf{\v@lmax}{-2}%
    \ifdim\v@lmin<\v@lmax\s@mme=#3\else\v@lmax=\v@lmin\s@mme=#4\fi%
    \figvectP -4[-5,#2]\vecunit@{-4}{-4}\delt@=#1\unit@%
    \edef\t@ille{\repdecn@mb{\delt@}}\figpttra-1:=-5/\t@ille,-3/%
    \figptstra-3=-5,-1/\t@ille,-4/\psline[#2,-5]\s@uvdash{\typ@dash}%
    \pssetdash{\defaultdash}\psline[-1,-2,-3]\pssetdash{\typ@dash}%
    \ifdim\v@leur<\v@lmax\Pss@tsecondSt\psline[-5,\the\s@mme]\Psrest@reSt\fi%
    \PSc@mment{End psaltitude}\resetc@ntr@l\et@tpsaltitude\fi\fi}}
\ctr@ld@f\def\Ps@rcerc#1;#2(#3,#4){\ellBB@x#1;#2,#2(#3,#4,0)%
    \f@gnewpath{\delt@=#2\unit@\delt@=\ptT@ptps\delt@%
    \BdingB@xfalse%
    \PSwrit@cmd{#1}{\repdecn@mb{\delt@}\space #3\space #4\space arc}{\fwf@g}}}
\ctr@ln@m\psarccirc
\ctr@ld@f\def\psarccircDD#1;#2(#3,#4){\ifcurr@ntPS\ifps@cri%
    \PSc@mment{psarccircDD Center=#1 ; Radius=#2 (Ang1=#3, Ang2=#4)}%
    \iffillm@de\Ps@rcerc#1;#2(#3,#4)%
    \f@gfill%
    \else\Ps@rcerc#1;#2(#3,#4)\f@gstroke\fi%
    \PSc@mment{End psarccircDD}\fi\fi}
\ctr@ld@f\def\psarccircTD#1,#2,#3;#4(#5,#6){{\ifcurr@ntPS\ifps@cri\s@uvc@ntr@l\et@tpsarccircTD%
    \PSc@mment{psarccircTD Center=#1,P1=#2,P2=#3 ; Radius=#4 (Ang1=#5, Ang2=#6)}%
    \setc@ntr@l{2}\c@lExtAxes#1,#2,#3(#4)\psarcellPATD#1,-4,-5(#5,#6)%
    \PSc@mment{End psarccircTD}\resetc@ntr@l\et@tpsarccircTD\fi\fi}}
\ctr@ld@f\def\c@lExtAxes#1,#2,#3(#4){%
    \figvectPTD-5[#1,#2]\vecunit@{-5}{-5}\figvectNTD-4[#1,#2,#3]\vecunit@{-4}{-4}%
    \figvectNVTD-3[-4,-5]\delt@=#4\unit@\edef\r@yon{\repdecn@mb{\delt@}}%
    \figpttra-4:=#1/\r@yon,-5/\figpttra-5:=#1/\r@yon,-3/}
\ctr@ln@m\psarccircP
\ctr@ld@f\def\psarccircPDD#1;#2[#3,#4]{{\ifcurr@ntPS\ifps@cri\s@uvc@ntr@l\et@tpsarccircPDD%
    \PSc@mment{psarccircPDD Center=#1; Radius=#2, [P1=#3, P2=#4]}%
    \Ps@ngleparam#1;#2[#3,#4]\ifdim\v@lmin>\v@lmax\advance\v@lmax\DePI@deg\fi%
    \edef\@ngdeb{\repdecn@mb{\v@lmin}}\edef\@ngfin{\repdecn@mb{\v@lmax}}%
    \psarccirc#1;\r@dius(\@ngdeb,\@ngfin)%
    \PSc@mment{End psarccircPDD}\resetc@ntr@l\et@tpsarccircPDD\fi\fi}}
\ctr@ld@f\def\psarccircPTD#1;#2[#3,#4,#5]{{\ifcurr@ntPS\ifps@cri\s@uvc@ntr@l\et@tpsarccircPTD%
    \PSc@mment{psarccircPTD Center=#1; Radius=#2, [P1=#3, P2=#4, P3=#5]}%
    \setc@ntr@l{2}\c@lExtAxes#1,#3,#5(#2)\psarcellPP#1,-4,-5[#3,#4]%
    \PSc@mment{End psarccircPTD}\resetc@ntr@l\et@tpsarccircPTD\fi\fi}}
\ctr@ld@f\def\Ps@ngleparam#1;#2[#3,#4]{\setc@ntr@l{2}%
    \figvectPDD-1[#1,#3]\vecunit@{-1}{-1}\Figg@tXY{-1}\arct@n\v@lmin(\v@lX,\v@lY)%
    \figvectPDD-2[#1,#4]\vecunit@{-2}{-2}\Figg@tXY{-2}\arct@n\v@lmax(\v@lX,\v@lY)%
    \v@lmin=\rdT@deg\v@lmin\v@lmax=\rdT@deg\v@lmax%
    \v@leur=#2pt\maxim@m{\mili@u}{-\v@leur}{\v@leur}%
    \edef\r@dius{\repdecn@mb{\mili@u}}}
\ctr@ld@f\def\Ps@rcercBz#1;#2(#3,#4){\Ps@rellBz#1;#2,#2(#3,#4,0)}
\ctr@ld@f\def\Ps@rellBz#1;#2,#3(#4,#5,#6){%
    \ellBB@x#1;#2,#3(#4,#5,#6)\BdingB@xfalse%
    \c@lNbarcs{#4}{#5}\v@leur=#4pt\setc@ntr@l{2}\figptell-13::#1;#2,#3(#4,#6)%
    \f@gnewpath\PSwrit@cmd{-13}{\c@mmoveto}{\fwf@g}%
    \s@mme=\z@\bcl@rellBz#1;#2,#3(#6)\BdingB@xtrue}
\ctr@ld@f\def\bcl@rellBz#1;#2,#3(#4){\relax%
    \ifnum\s@mme<\p@rtent\advance\s@mme\@ne%
    \advance\v@leur\delt@\edef\@ngle{\repdecn@mb\v@leur}\figptell-14::#1;#2,#3(\@ngle,#4)%
    \advance\v@leur\delt@\edef\@ngle{\repdecn@mb\v@leur}\figptell-15::#1;#2,#3(\@ngle,#4)%
    \advance\v@leur\delt@\edef\@ngle{\repdecn@mb\v@leur}\figptell-16::#1;#2,#3(\@ngle,#4)%
    \figptscontrolDD-18[-13,-14,-15,-16]%
    \PSwrit@cmd{-18}{}{\fwf@g}\PSwrit@cmd{-17}{}{\fwf@g}%
    \PSwrit@cmd{-16}{\c@mcurveto}{\fwf@g}%
    \figptcopyDD-13:/-16/\bcl@rellBz#1;#2,#3(#4)\fi}
\ctr@ld@f\def\Ps@rell#1;#2,#3(#4,#5,#6){\ellBB@x#1;#2,#3(#4,#5,#6)%
    \f@gnewpath{\v@lmin=#2\unit@\v@lmin=\ptT@ptps\v@lmin%
    \v@lmax=#3\unit@\v@lmax=\ptT@ptps\v@lmax\BdingB@xfalse%
    \PSwrit@cmd{#1}%
    {#6\space\repdecn@mb{\v@lmin}\space\repdecn@mb{\v@lmax}\space #4\space #5\space ellipse}{\fwf@g}}%
    \global\Use@llipsetrue}
\ctr@ln@m\psarcell
\ctr@ld@f\def\psarcellDD#1;#2,#3(#4,#5,#6){{\ifcurr@ntPS\ifps@cri%
    \PSc@mment{psarcellDD Center=#1 ; XRad=#2, YRad=#3 (Ang1=#4, Ang2=#5, Inclination=#6)}%
    \iffillm@de\Ps@rell#1;#2,#3(#4,#5,#6)%
    \f@gfill%
    \else\Ps@rell#1;#2,#3(#4,#5,#6)\f@gstroke\fi%
    \PSc@mment{End psarcellDD}\fi\fi}}
\ctr@ld@f\def\psarcellTD#1;#2,#3(#4,#5,#6){{\ifcurr@ntPS\ifps@cri\s@uvc@ntr@l\et@tpsarcellTD%
    \PSc@mment{psarcellTD Center=#1 ; XRad=#2, YRad=#3 (Ang1=#4, Ang2=#5, Inclination=#6)}%
    \setc@ntr@l{2}\figpttraC -8:=#1/#2,0,0/\figpttraC -7:=#1/0,#3,0/%
    \figvectC -4(0,0,1)\figptsrot -8=-8,-7/#1,#6,-4/\psarcellPATD#1,-8,-7(#4,#5)%
    \PSc@mment{End psarcellTD}\resetc@ntr@l\et@tpsarcellTD\fi\fi}}
\ctr@ln@m\psarcellPA
\ctr@ld@f\def\psarcellPADD#1,#2,#3(#4,#5){{\ifcurr@ntPS\ifps@cri\s@uvc@ntr@l\et@tpsarcellPADD%
    \PSc@mment{psarcellPADD Center=#1,PtAxis1=#2,PtAxis2=#3 (Ang1=#4, Ang2=#5)}%
    \setc@ntr@l{2}\figvectPDD-1[#1,#2]\vecunit@DD{-1}{-1}\v@lX=\ptT@unit@\result@t%
    \edef\XR@d{\repdecn@mb{\v@lX}}\Figg@tXY{-1}\arct@n\v@lmin(\v@lX,\v@lY)%
    \v@lmin=\rdT@deg\v@lmin\edef\Inclin@{\repdecn@mb{\v@lmin}}%
    \figgetdist\YR@d[#1,#3]\psarcellDD#1;\XR@d,\YR@d(#4,#5,\Inclin@)%
    \PSc@mment{End psarcellPADD}\resetc@ntr@l\et@tpsarcellPADD\fi\fi}}
\ctr@ld@f\def\psarcellPATD#1,#2,#3(#4,#5){{\ifcurr@ntPS\ifps@cri\s@uvc@ntr@l\et@tpsarcellPATD%
    \PSc@mment{psarcellPATD Center=#1,PtAxis1=#2,PtAxis2=#3 (Ang1=#4, Ang2=#5)}%
    \iffillm@de\Ps@rellPATD#1,#2,#3(#4,#5)%
    \f@gfill%
    \else\Ps@rellPATD#1,#2,#3(#4,#5)\f@gstroke\fi%
    \PSc@mment{End psarcellPATD}\resetc@ntr@l\et@tpsarcellPATD\fi\fi}}
\ctr@ld@f\def\Ps@rellPATD#1,#2,#3(#4,#5){\let\c@lprojSP=\relax%
    \setc@ntr@l{2}\figvectPTD-1[#1,#2]\figvectPTD-2[#1,#3]\c@lNbarcs{#4}{#5}%
    \v@leur=#4pt\c@lptellP{#1}{-1}{-2}\Figptpr@j-5:/-3/%
    \f@gnewpath\PSwrit@cmdS{-5}{\c@mmoveto}{\fwf@g}{\X@un}{\Y@un}%
    \edef\C@nt@r{#1}\s@mme=\z@\bcl@rellPATD}
\ctr@ld@f\def\bcl@rellPATD{\relax%
    \ifnum\s@mme<\p@rtent\advance\s@mme\@ne%
    \advance\v@leur\delt@\c@lptellP{\C@nt@r}{-1}{-2}\Figptpr@j-4:/-3/%
    \advance\v@leur\delt@\c@lptellP{\C@nt@r}{-1}{-2}\Figptpr@j-6:/-3/%
    \advance\v@leur\delt@\c@lptellP{\C@nt@r}{-1}{-2}\Figptpr@j-3:/-3/%
    \v@lX=\z@\v@lY=\z@\Figtr@nptDD{-5}{-5}\Figtr@nptDD{2}{-3}%
    \divide\v@lX\@vi\divide\v@lY\@vi%
    \Figtr@nptDD{3}{-4}\Figtr@nptDD{-1.5}{-6}\v@lmin=\v@lX\v@lmax=\v@lY%
    \v@lX=\z@\v@lY=\z@\Figtr@nptDD{2}{-5}\Figtr@nptDD{-5}{-3}%
    \divide\v@lX\@vi\divide\v@lY\@vi\Figtr@nptDD{-1.5}{-4}\Figtr@nptDD{3}{-6}%
    \BdingB@xfalse%
    \Figp@intregDD-4:(\v@lmin,\v@lmax)\PSwrit@cmdS{-4}{}{\fwf@g}{\X@de}{\Y@de}%
    \Figp@intregDD-4:(\v@lX,\v@lY)\PSwrit@cmdS{-4}{}{\fwf@g}{\X@tr}{\Y@tr}%
    \BdingB@xtrue\PSwrit@cmdS{-3}{\c@mcurveto}{\fwf@g}{\X@qu}{\Y@qu}%
    \B@zierBB@x{1}{\Y@un}(\X@un,\X@de,\X@tr,\X@qu)%
    \B@zierBB@x{2}{\X@un}(\Y@un,\Y@de,\Y@tr,\Y@qu)%
    \edef\X@un{\X@qu}\edef\Y@un{\Y@qu}\figptcopyDD-5:/-3/\bcl@rellPATD\fi}
\ctr@ld@f\def\c@lNbarcs#1#2{%
    \delt@=#2pt\advance\delt@-#1pt\maxim@m{\v@lmax}{\delt@}{-\delt@}%
    \v@leur=\v@lmax\divide\v@leur45 \p@rtentiere{\p@rtent}{\v@leur}\advance\p@rtent\@ne%
    \s@mme=\p@rtent\multiply\s@mme\thr@@\divide\delt@\s@mme}
\ctr@ld@f\def\psarcellPP#1,#2,#3[#4,#5]{{\ifcurr@ntPS\ifps@cri\s@uvc@ntr@l\et@tpsarcellPP%
    \PSc@mment{psarcellPP Center=#1,PtAxis1=#2,PtAxis2=#3 [Point1=#4, Point2=#5]}%
    \setc@ntr@l{2}\figvectP-2[#1,#3]\vecunit@{-2}{-2}\v@lmin=\result@t%
    \invers@{\v@lmax}{\v@lmin}%
    \figvectP-1[#1,#2]\vecunit@{-1}{-1}\v@leur=\result@t%
    \v@leur=\repdecn@mb{\v@lmax}\v@leur\edef\AsB@{\repdecn@mb{\v@leur}}% a/b
    \c@lAngle{#1}{#4}{\v@lmin}\edef\@ngdeb{\repdecn@mb{\v@lmin}}%
    \c@lAngle{#1}{#5}{\v@lmax}\ifdim\v@lmin>\v@lmax\advance\v@lmax\DePI@deg\fi%
    \edef\@ngfin{\repdecn@mb{\v@lmax}}\psarcellPA#1,#2,#3(\@ngdeb,\@ngfin)%
    \PSc@mment{End psarcellPP}\resetc@ntr@l\et@tpsarcellPP\fi\fi}}
\ctr@ld@f\def\c@lAngle#1#2#3{\figvectP-3[#1,#2]%
    \c@lproscal\delt@[-3,-1]\c@lproscal\v@leur[-3,-2]%
    \v@leur=\AsB@\v@leur\arct@n#3(\delt@,\v@leur)#3=\rdT@deg#3}
\ctr@ln@w{newif}\if@rrowratio\@rrowratiotrue
\ctr@ln@w{newif}\if@rrowhfill
\ctr@ln@w{newif}\if@rrowhout
\ctr@ld@f\def\Psset@rrowhe@d#1=#2|{\keln@mun#1|%
    \def\n@mref{a}\ifx\l@debut\n@mref\pssetarrowheadangle{#2}\else% angle
    \def\n@mref{f}\ifx\l@debut\n@mref\pssetarrowheadfill{#2}\else% fillmode
    \def\n@mref{l}\ifx\l@debut\n@mref\pssetarrowheadlength{#2}\else% length
    \def\n@mref{o}\ifx\l@debut\n@mref\pssetarrowheadout{#2}\else% out
    \def\n@mref{r}\ifx\l@debut\n@mref\pssetarrowheadratio{#2}\else% ratio
    \immediate\write16{*** Unknown attribute: \BS@ psset arrowhead(..., #1=...)}%
    \fi\fi\fi\fi\fi}
\ctr@ln@m\@rrowheadangle
\ctr@ln@m\C@AHANG \ctr@ln@m\S@AHANG \ctr@ln@m\UNSS@N
\ctr@ld@f\def\pssetarrowheadangle#1{\edef\@rrowheadangle{#1}{\c@ssin{\C@}{\S@}{#1}%
    \xdef\C@AHANG{\C@}\xdef\S@AHANG{\S@}\v@lmax=\S@ pt%
    \invers@{\v@leur}{\v@lmax}\maxim@m{\v@leur}{\v@leur}{-\v@leur}%
    \xdef\UNSS@N{\the\v@leur}}}
\ctr@ld@f\def\pssetarrowheadfill#1{\expandafter\set@rrowhfill#1:}
\ctr@ld@f\def\set@rrowhfill#1#2:{\if#1n\@rrowhfillfalse\else\@rrowhfilltrue\fi}
\ctr@ld@f\def\pssetarrowheadout#1{\expandafter\set@rrowhout#1:}
\ctr@ld@f\def\set@rrowhout#1#2:{\if#1n\@rrowhoutfalse\else\@rrowhouttrue\fi}
\ctr@ln@m\@rrowheadlength
\ctr@ld@f\def\pssetarrowheadlength#1{\edef\@rrowheadlength{#1}\@rrowratiofalse}
\ctr@ln@m\@rrowheadratio
\ctr@ld@f\def\pssetarrowheadratio#1{\edef\@rrowheadratio{#1}\@rrowratiotrue}
\ctr@ln@m\defaultarrowheadlength
\ctr@ld@f\def\psresetarrowhead{%
    \pssetarrowheadangle{\defaultarrowheadangle}%
    \pssetarrowheadfill{\defaultarrowheadfill}%
    \pssetarrowheadout{\defaultarrowheadout}%
    \pssetarrowheadratio{\defaultarrowheadratio}%
    \d@fm@cdim\defaultarrowheadlength{\defaulth@rdahlength}% Valeur par defaut...
    \pssetarrowheadlength{\defaultarrowheadlength}}
\ctr@ld@f\def\defaultarrowheadratio{0.1}
\ctr@ld@f\def\defaultarrowheadangle{20}
\ctr@ld@f\def\defaultarrowheadfill{no}
\ctr@ld@f\def\defaultarrowheadout{no}
\ctr@ld@f\def\defaulth@rdahlength{8pt}
\ctr@ln@m\psarrow
\ctr@ld@f\def\psarrowDD[#1,#2]{{\ifcurr@ntPS\ifps@cri\s@uvc@ntr@l\et@tpsarrow%
    \PSc@mment{psarrowDD [Pt1,Pt2]=[#1,#2]}\pssetfillmode{no}%
    \psarrowheadDD[#1,#2]\setc@ntr@l{2}\psline[#1,-3]%
    \PSc@mment{End psarrowDD}\resetc@ntr@l\et@tpsarrow\fi\fi}}
\ctr@ld@f\def\psarrowTD[#1,#2]{{\ifcurr@ntPS\ifps@cri\s@uvc@ntr@l\et@tpsarrowTD%
    \PSc@mment{psarrowTD [Pt1,Pt2]=[#1,#2]}\resetc@ntr@l{2}%
    \Figptpr@j-5:/#1/\Figptpr@j-6:/#2/\let\c@lprojSP=\relax\psarrowDD[-5,-6]%
    \PSc@mment{End psarrowTD}\resetc@ntr@l\et@tpsarrowTD\fi\fi}}
\ctr@ln@m\psarrowhead
\ctr@ld@f\def\psarrowheadDD[#1,#2]{{\ifcurr@ntPS\ifps@cri\s@uvc@ntr@l\et@tpsarrowheadDD%
    \if@rrowhfill\def\@hangle{-\@rrowheadangle}\else\def\@hangle{\@rrowheadangle}\fi%
    \if@rrowratio%
    \if@rrowhout\def\@hratio{-\@rrowheadratio}\else\def\@hratio{\@rrowheadratio}\fi%
    \PSc@mment{psarrowheadDD Ratio=\@hratio, Angle=\@hangle, [Pt1,Pt2]=[#1,#2]}%
    \Ps@rrowhead\@hratio,\@hangle[#1,#2]%
    \else%
    \if@rrowhout\def\@hlength{-\@rrowheadlength}\else\def\@hlength{\@rrowheadlength}\fi%
    \PSc@mment{psarrowheadDD Length=\@hlength, Angle=\@hangle, [Pt1,Pt2]=[#1,#2]}%
    \Ps@rrowheadfd\@hlength,\@hangle[#1,#2]%
    \fi%
    \PSc@mment{End psarrowheadDD}\resetc@ntr@l\et@tpsarrowheadDD\fi\fi}}
\ctr@ld@f\def\psarrowheadTD[#1,#2]{{\ifcurr@ntPS\ifps@cri\s@uvc@ntr@l\et@tpsarrowheadTD%
    \PSc@mment{psarrowheadTD [Pt1,Pt2]=[#1,#2]}\resetc@ntr@l{2}%
    \Figptpr@j-5:/#1/\Figptpr@j-6:/#2/\let\c@lprojSP=\relax\psarrowheadDD[-5,-6]%
    \PSc@mment{End psarrowheadTD}\resetc@ntr@l\et@tpsarrowheadTD\fi\fi}}
\ctr@ld@f\def\Ps@rrowhead#1,#2[#3,#4]{\v@leur=#1\p@\maxim@m{\v@leur}{\v@leur}{-\v@leur}%
    \ifdim\v@leur>\Cepsil@n{% Arrow is not degenerated
    \PSc@mment{ps@rrowhead Ratio=#1, Angle=#2, [Pt1,Pt2]=[#3,#4]}\v@leur=\UNSS@N%
    \v@leur=\curr@ntwidth\v@leur\v@leur=\ptpsT@pt\v@leur\delt@=.5\v@leur% = width / (2 sin(Angle))
    \setc@ntr@l{2}\figvectPDD-3[#4,#3]%
    \Figg@tXY{-3}\v@lX=#1\v@lX\v@lY=#1\v@lY\Figv@ctCreg-3(\v@lX,\v@lY)%
    \vecunit@{-4}{-3}\mili@u=\result@t%
    \ifdim#2pt>\z@\v@lXa=-\C@AHANG\delt@%
     \edef\c@ef{\repdecn@mb{\v@lXa}}\figpttraDD-3:=-3/\c@ef,-4/\fi%
    \edef\c@ef{\repdecn@mb{\delt@}}%
    \v@lXa=\mili@u\v@lXa=\C@AHANG\v@lXa%
    \v@lYa=\ptpsT@pt\p@\v@lYa=\curr@ntwidth\v@lYa\v@lYa=\sDcc@ngle\v@lYa%
    \advance\v@lXa-\v@lYa\gdef\sDcc@ngle{0}%
    \ifdim\v@lXa>\v@leur\edef\c@efendpt{\repdecn@mb{\v@leur}}%
    \else\edef\c@efendpt{\repdecn@mb{\v@lXa}}\fi%
    \Figg@tXY{-3}\v@lmin=\v@lX\v@lmax=\v@lY%
    \v@lXa=\C@AHANG\v@lmin\v@lYa=\S@AHANG\v@lmax\advance\v@lXa\v@lYa%
    \v@lYa=-\S@AHANG\v@lmin\v@lX=\C@AHANG\v@lmax\advance\v@lYa\v@lX%
    \setc@ntr@l{1}\Figg@tXY{#4}\advance\v@lX\v@lXa\advance\v@lY\v@lYa%
    \setc@ntr@l{2}\Figp@intregDD-2:(\v@lX,\v@lY)%
    \v@lXa=\C@AHANG\v@lmin\v@lYa=-\S@AHANG\v@lmax\advance\v@lXa\v@lYa%
    \v@lYa=\S@AHANG\v@lmin\v@lX=\C@AHANG\v@lmax\advance\v@lYa\v@lX%
    \setc@ntr@l{1}\Figg@tXY{#4}\advance\v@lX\v@lXa\advance\v@lY\v@lYa%
    \setc@ntr@l{2}\Figp@intregDD-1:(\v@lX,\v@lY)%
    \ifdim#2pt<\z@\fillm@detrue\psline[-2,#4,-1]% fill
    \else\figptstraDD-3=#4,-2,-1/\c@ef,-4/\psline[-2,-3,-1]\fi% no fill
    \ifdim#1pt>\z@\figpttraDD-3:=#4/\c@efendpt,-4/\else\figptcopyDD-3:/#4/\fi%
    \PSc@mment{End ps@rrowhead}}\fi}
\ctr@ld@f\def\sDcc@ngle{0}% Initialisation
\ctr@ld@f\def\Ps@rrowheadfd#1,#2[#3,#4]{{%
    \PSc@mment{ps@rrowheadfd Length=#1, Angle=#2, [Pt1,Pt2]=[#3,#4]}%
    \setc@ntr@l{2}\figvectPDD-1[#3,#4]\n@rmeucDD{\v@leur}{-1}\v@leur=\ptT@unit@\v@leur%
    \invers@{\v@leur}{\v@leur}\v@leur=#1\v@leur\edef\R@tio{\repdecn@mb{\v@leur}}%
    \Ps@rrowhead\R@tio,#2[#3,#4]\PSc@mment{End ps@rrowheadfd}}}
\ctr@ln@m\psarrowBezier
\ctr@ld@f\def\psarrowBezierDD[#1,#2,#3,#4]{{\ifcurr@ntPS\ifps@cri\s@uvc@ntr@l\et@tpsarrowBezierDD%
    \PSc@mment{psarrowBezierDD Control points=#1,#2,#3,#4}\setc@ntr@l{2}%
    \if@rrowratio\c@larclengthDD\v@leur,10[#1,#2,#3,#4]\else\v@leur=\z@\fi%
    \Ps@rrowB@zDD\v@leur[#1,#2,#3,#4]%
    \PSc@mment{End psarrowBezierDD}\resetc@ntr@l\et@tpsarrowBezierDD\fi\fi}}
\ctr@ld@f\def\psarrowBezierTD[#1,#2,#3,#4]{{\ifcurr@ntPS\ifps@cri\s@uvc@ntr@l\et@tpsarrowBezierTD%
    \PSc@mment{psarrowBezierTD Control points=#1,#2,#3,#4}\resetc@ntr@l{2}%
    \Figptpr@j-7:/#1/\Figptpr@j-8:/#2/\Figptpr@j-9:/#3/\Figptpr@j-10:/#4/%
    \let\c@lprojSP=\relax\ifnum\curr@ntproj<\tw@\psarrowBezierDD[-7,-8,-9,-10]%
    \else\f@gnewpath\PSwrit@cmd{-7}{\c@mmoveto}{\fwf@g}%
    \if@rrowratio\c@larclengthDD\mili@u,10[-7,-8,-9,-10]\else\mili@u=\z@\fi%
    \p@rtent=\NBz@rcs\advance\p@rtent\m@ne\subB@zierTD\p@rtent[#1,#2,#3,#4]%
    \f@gstroke%
    \advance\v@lmin\p@rtent\delt@% Initialized in \subB@zierTD
    \v@leur=\v@lmin\advance\v@leur0.33333 \delt@\edef\unti@rs{\repdecn@mb{\v@leur}}%
    \v@leur=\v@lmin\advance\v@leur0.66666 \delt@\edef\deti@rs{\repdecn@mb{\v@leur}}%
    \figptcopyDD-8:/-10/\c@lsubBzarc\unti@rs,\deti@rs[#1,#2,#3,#4]%
    \figptcopyDD-8:/-4/\figptcopyDD-9:/-3/\Ps@rrowB@zDD\mili@u[-7,-8,-9,-10]\fi%
    \PSc@mment{End psarrowBezierTD}\resetc@ntr@l\et@tpsarrowBezierTD\fi\fi}}
\ctr@ld@f\def\c@larclengthDD#1,#2[#3,#4,#5,#6]{{\p@rtent=#2\figptcopyDD-5:/#3/%
    \delt@=\p@\divide\delt@\p@rtent\c@rre=\z@\v@leur=\z@\s@mme=\z@%
    \loop\ifnum\s@mme<\p@rtent\advance\s@mme\@ne\advance\v@leur\delt@%
    \edef\T@{\repdecn@mb{\v@leur}}\figptBezierDD-6::\T@[#3,#4,#5,#6]%
    \figvectPDD-1[-5,-6]\n@rmeucDD{\mili@u}{-1}\advance\c@rre\mili@u%
    \figptcopyDD-5:/-6/\repeat\global\result@t=\ptT@unit@\c@rre}#1=\result@t}
\ctr@ld@f\def\Ps@rrowB@zDD#1[#2,#3,#4,#5]{{\pssetfillmode{no}%
    \if@rrowratio\delt@=\@rrowheadratio#1\else\delt@=\@rrowheadlength pt\fi%
    \v@leur=\C@AHANG\delt@\edef\R@dius{\repdecn@mb{\v@leur}}%
    \FigptintercircB@zDD-5::0,\R@dius[#5,#4,#3,#2]%
    \pssetarrowheadlength{\repdecn@mb{\delt@}}\psarrowheadDD[-5,#5]%
    \let\n@rmeuc=\n@rmeucDD\figgetdist\R@dius[#5,-3]%
    \FigptintercircB@zDD-6::0,\R@dius[#5,#4,#3,#2]%
    \figptBezierDD-5::0.33333[#5,#4,#3,#2]\figptBezierDD-3::0.66666[#5,#4,#3,#2]%
    \figptscontrolDD-5[-6,-5,-3,#2]\psBezierDD1[-6,-5,-4,#2]}}
\ctr@ln@m\psarrowcirc
\ctr@ld@f\def\psarrowcircDD#1;#2(#3,#4){{\ifcurr@ntPS\ifps@cri\s@uvc@ntr@l\et@tpsarrowcircDD%
    \PSc@mment{psarrowcircDD Center=#1 ; Radius=#2 (Ang1=#3,Ang2=#4)}%
    \pssetfillmode{no}\Pscirc@rrowhead#1;#2(#3,#4)%
    \setc@ntr@l{2}\figvectPDD -4[#1,-3]\vecunit@{-4}{-4}%
    \Figg@tXY{-4}\arct@n\v@lmin(\v@lX,\v@lY)%
    \v@lmin=\rdT@deg\v@lmin\v@leur=#4pt\advance\v@leur-\v@lmin%
    \maxim@m{\v@leur}{\v@leur}{-\v@leur}%
    \ifdim\v@leur>\DemiPI@deg\relax\ifdim\v@lmin<#4pt\advance\v@lmin\DePI@deg%
    \else\advance\v@lmin-\DePI@deg\fi\fi\edef\ar@ngle{\repdecn@mb{\v@lmin}}%
    \ifdim#3pt<#4pt\psarccirc#1;#2(#3,\ar@ngle)\else\psarccirc#1;#2(\ar@ngle,#3)\fi%
    \PSc@mment{End psarrowcircDD}\resetc@ntr@l\et@tpsarrowcircDD\fi\fi}}
\ctr@ld@f\def\psarrowcircTD#1,#2,#3;#4(#5,#6){{\ifcurr@ntPS\ifps@cri\s@uvc@ntr@l\et@tpsarrowcircTD%
    \PSc@mment{psarrowcircTD Center=#1,P1=#2,P2=#3 ; Radius=#4 (Ang1=#5, Ang2=#6)}%
    \resetc@ntr@l{2}\c@lExtAxes#1,#2,#3(#4)\let\c@lprojSP=\relax%
    \figvectPTD-11[#1,-4]\figvectPTD-12[#1,-5]\c@lNbarcs{#5}{#6}%
    \if@rrowratio\v@lmax=\degT@rd\v@lmax\edef\D@lpha{\repdecn@mb{\v@lmax}}\fi%
    \advance\p@rtent\m@ne\mili@u=\z@%
    \v@leur=#5pt\c@lptellP{#1}{-11}{-12}\Figptpr@j-9:/-3/%
    \f@gnewpath\PSwrit@cmdS{-9}{\c@mmoveto}{\fwf@g}{\X@un}{\Y@un}%
    \edef\C@nt@r{#1}\s@mme=\z@\bcl@rcircTD\f@gstroke%
    \advance\v@leur\delt@\c@lptellP{#1}{-11}{-12}\Figptpr@j-5:/-3/%
    \advance\v@leur\delt@\c@lptellP{#1}{-11}{-12}\Figptpr@j-6:/-3/%
    \advance\v@leur\delt@\c@lptellP{#1}{-11}{-12}\Figptpr@j-10:/-3/%
    \figptscontrolDD-8[-9,-5,-6,-10]%
    \if@rrowratio\c@lcurvradDD0.5[-9,-8,-7,-10]\advance\mili@u\result@t%
    \maxim@m{\mili@u}{\mili@u}{-\mili@u}\mili@u=\ptT@unit@\mili@u%
    \mili@u=\D@lpha\mili@u\advance\p@rtent\@ne\divide\mili@u\p@rtent\fi%
    \Ps@rrowB@zDD\mili@u[-9,-8,-7,-10]%
    \PSc@mment{End psarrowcircTD}\resetc@ntr@l\et@tpsarrowcircTD\fi\fi}}
\ctr@ld@f\def\bcl@rcircTD{\relax%
    \ifnum\s@mme<\p@rtent\advance\s@mme\@ne%
    \advance\v@leur\delt@\c@lptellP{\C@nt@r}{-11}{-12}\Figptpr@j-5:/-3/%
    \advance\v@leur\delt@\c@lptellP{\C@nt@r}{-11}{-12}\Figptpr@j-6:/-3/%
    \advance\v@leur\delt@\c@lptellP{\C@nt@r}{-11}{-12}\Figptpr@j-10:/-3/%
    \figptscontrolDD-8[-9,-5,-6,-10]\BdingB@xfalse%
    \PSwrit@cmdS{-8}{}{\fwf@g}{\X@de}{\Y@de}\PSwrit@cmdS{-7}{}{\fwf@g}{\X@tr}{\Y@tr}%
    \BdingB@xtrue\PSwrit@cmdS{-10}{\c@mcurveto}{\fwf@g}{\X@qu}{\Y@qu}%
    \if@rrowratio\c@lcurvradDD0.5[-9,-8,-7,-10]\advance\mili@u\result@t\fi%
    \B@zierBB@x{1}{\Y@un}(\X@un,\X@de,\X@tr,\X@qu)%
    \B@zierBB@x{2}{\X@un}(\Y@un,\Y@de,\Y@tr,\Y@qu)%
    \edef\X@un{\X@qu}\edef\Y@un{\Y@qu}\figptcopyDD-9:/-10/\bcl@rcircTD\fi}
\ctr@ld@f\def\Pscirc@rrowhead#1;#2(#3,#4){{%
    \PSc@mment{pscirc@rrowhead Center=#1 ; Radius=#2 (Ang1=#3,Ang2=#4)}%
    \v@leur=#2\unit@\edef\s@glen{\repdecn@mb{\v@leur}}\v@lY=\z@\v@lX=\v@leur%
    \resetc@ntr@l{2}\Figv@ctCreg-3(\v@lX,\v@lY)\figpttraDD-5:=#1/1,-3/%
    \figptrotDD-5:=-5/#1,#4/%
    \figvectPDD-3[#1,-5]\Figg@tXY{-3}\v@leur=\v@lX%
    \ifdim#3pt<#4pt\v@lX=\v@lY\v@lY=-\v@leur\else\v@lX=-\v@lY\v@lY=\v@leur\fi%
    \Figv@ctCreg-3(\v@lX,\v@lY)\vecunit@{-3}{-3}%
    \if@rrowratio\v@leur=#4pt\advance\v@leur-#3pt\maxim@m{\mili@u}{-\v@leur}{\v@leur}%
    \mili@u=\degT@rd\mili@u\v@leur=\s@glen\mili@u\edef\s@glen{\repdecn@mb{\v@leur}}%
    \mili@u=#2\mili@u\mili@u=\@rrowheadratio\mili@u\else\mili@u=\@rrowheadlength pt\fi%
    \figpttraDD-6:=-5/\s@glen,-3/\v@leur=#2pt\v@leur=2\v@leur%
    \invers@{\v@leur}{\v@leur}\c@rre=\repdecn@mb{\v@leur}\mili@u% = sin = L/(2R)
    \mili@u=\c@rre\mili@u=\repdecn@mb{\c@rre}\mili@u%
    \v@leur=\p@\advance\v@leur-\mili@u% \v@leur = cos*cos
    \invers@{\mili@u}{2\v@leur}\delt@=\c@rre\delt@=\repdecn@mb{\mili@u}\delt@%
    \xdef\sDcc@ngle{\repdecn@mb{\delt@}}% sin/(2*cos*cos) used in \Ps@rrowhead
    \sqrt@{\mili@u}{\v@leur}\arct@n\v@leur(\mili@u,\c@rre)%
    \v@leur=\rdT@deg\v@leur% \cor@ngle = atan(L/sqrt(4R*R-L*L))
    \ifdim#3pt<#4pt\v@leur=-\v@leur\fi%
    \if@rrowhout\v@leur=-\v@leur\fi\edef\cor@ngle{\repdecn@mb{\v@leur}}%
    \figptrotDD-6:=-6/-5,\cor@ngle/\psarrowheadDD[-6,-5]%
    \PSc@mment{End pscirc@rrowhead}}}
\ctr@ln@m\psarrowcircP
\ctr@ld@f\def\psarrowcircPDD#1;#2[#3,#4]{{\ifcurr@ntPS\ifps@cri%
    \PSc@mment{psarrowcircPDD Center=#1; Radius=#2, [P1=#3,P2=#4]}%
    \s@uvc@ntr@l\et@tpsarrowcircPDD\Ps@ngleparam#1;#2[#3,#4]%
    \ifdim\v@leur>\z@\ifdim\v@lmin>\v@lmax\advance\v@lmax\DePI@deg\fi%
    \else\ifdim\v@lmin<\v@lmax\advance\v@lmin\DePI@deg\fi\fi%
    \edef\@ngdeb{\repdecn@mb{\v@lmin}}\edef\@ngfin{\repdecn@mb{\v@lmax}}%
    \psarrowcirc#1;\r@dius(\@ngdeb,\@ngfin)%
    \PSc@mment{End psarrowcircPDD}\resetc@ntr@l\et@tpsarrowcircPDD\fi\fi}}
\ctr@ld@f\def\psarrowcircPTD#1;#2[#3,#4,#5]{{\ifcurr@ntPS\ifps@cri\s@uvc@ntr@l\et@tpsarrowcircPTD%
    \PSc@mment{psarrowcircPTD Center=#1; Radius=#2, [P1=#3,P2=#4,P3=#5]}%
    \figgetangleTD\@ngfin[#1,#3,#4,#5]\v@leur=#2pt%
    \maxim@m{\mili@u}{-\v@leur}{\v@leur}\edef\r@dius{\repdecn@mb{\mili@u}}%
    \ifdim\v@leur<\z@\v@lmax=\@ngfin pt\advance\v@lmax-\DePI@deg%
    \edef\@ngfin{\repdecn@mb{\v@lmax}}\fi\psarrowcircTD#1,#3,#5;\r@dius(0,\@ngfin)%
    \PSc@mment{End psarrowcircPTD}\resetc@ntr@l\et@tpsarrowcircPTD\fi\fi}}
\ctr@ld@f\def\psaxes#1(#2){{\ifcurr@ntPS\ifps@cri\s@uvc@ntr@l\et@tpsaxes%
    \PSc@mment{psaxes Origin=#1 Range=(#2)}\an@lys@xes#2,:\resetc@ntr@l{2}%
    \ifx\t@xt@\empty\ifTr@isDim\ps@xes#1(0,#2,0,#2,0,#2)\else\ps@xes#1(0,#2,0,#2)\fi%
    \else\ps@xes#1(#2)\fi\PSc@mment{End psaxes}\resetc@ntr@l\et@tpsaxes\fi\fi}}
\ctr@ld@f\def\an@lys@xes#1,#2:{\def\t@xt@{#2}}
\ctr@ln@m\ps@xes
\ctr@ld@f\def\ps@xesDD#1(#2,#3,#4,#5){%
    \figpttraC-5:=#1/#2,0/\figpttraC-6:=#1/#3,0/\psarrowDD[-5,-6]%
    \figpttraC-5:=#1/0,#4/\figpttraC-6:=#1/0,#5/\psarrowDD[-5,-6]}
\ctr@ld@f\def\ps@xesTD#1(#2,#3,#4,#5,#6,#7){%
    \figpttraC-7:=#1/#2,0,0/\figpttraC-8:=#1/#3,0,0/\psarrowTD[-7,-8]%
    \figpttraC-7:=#1/0,#4,0/\figpttraC-8:=#1/0,#5,0/\psarrowTD[-7,-8]%
    \figpttraC-7:=#1/0,0,#6/\figpttraC-8:=#1/0,0,#7/\psarrowTD[-7,-8]}
\ctr@ln@m\newGr@FN
\ctr@ld@f\def\newGr@FNPDF#1{\s@mme=\Gr@FNb\advance\s@mme\@ne\xdef\Gr@FNb{\number\s@mme}}
\ctr@ld@f\def\newGr@FNDVI#1{\newGr@FNPDF{}\xdef#1{\jobname GI\Gr@FNb.anx}}
\ctr@ld@f\def\psbeginfig#1{\newGr@FN\DefGIfilen@me\gdef\@utoFN{0}%
    \def\t@xt@{#1}\relax\ifx\t@xt@\empty\psupdatem@detrue%
    \gdef\@utoFN{1}\Psb@ginfig\DefGIfilen@me\else\expandafter\Psb@ginfigNu@#1 :\fi}
\ctr@ld@f\def\Psb@ginfigNu@#1 #2:{\def\t@xt@{#1}\relax\ifx\t@xt@\empty\def\t@xt@{#2}%
    \ifx\t@xt@\empty\psupdatem@detrue\gdef\@utoFN{1}\Psb@ginfig\DefGIfilen@me%
    \else\Psb@ginfigNu@#2:\fi\else\Psb@ginfig{#1}\fi}
\ctr@ln@m\PSfilen@me \ctr@ln@m\auxfilen@me
\ctr@ld@f\def\Psb@ginfig#1{\ifcurr@ntPS\else%
    \edef\PSfilen@me{#1}\edef\auxfilen@me{\jobname.anx}%
    \ifpsupdatem@de\ps@critrue\else\openin\frf@g=\PSfilen@me\relax%
    \ifeof\frf@g\ps@critrue\else\ps@crifalse\fi\closein\frf@g\fi%
    \curr@ntPStrue\c@ldefproj\expandafter\setupd@te\defaultupdate:%
    \ifps@cri\initb@undb@x%
    \immediate\openout\fwf@g=\auxfilen@me\initpss@ttings\fi%
    \fi}
\ctr@ld@f\def\Gr@FNb{0}
\ctr@ld@f\def\figforTeXFileno{\Gr@FNb}
\ctr@ld@f\def\figforTeXFigno{0 }
\ctr@ld@f\def\figforTeXnextFigno{1 }
\ctr@ld@f\edef\DefGIfilen@me{\jobname GI.anx}
\ctr@ld@f\def\initpss@ttings{\psreset{arrowhead,curve,first,flowchart,mesh,second,third}%
    \Use@llipsefalse}
\ctr@ld@f\def\B@zierBB@x#1#2(#3,#4,#5,#6){{\c@rre=\t@n\epsil@n% Do not reduce this value
    \v@lmax=#4\advance\v@lmax-#5\v@lmax=\thr@@\v@lmax\advance\v@lmax#6\advance\v@lmax-#3%
    \mili@u=#4\mili@u=-\tw@\mili@u\advance\mili@u#3\advance\mili@u#5%
    \v@lmin=#4\advance\v@lmin-#3\maxim@m{\v@leur}{-\v@lmax}{\v@lmax}%
    \maxim@m{\delt@}{-\mili@u}{\mili@u}\maxim@m{\v@leur}{\v@leur}{\delt@}%
    \maxim@m{\delt@}{-\v@lmin}{\v@lmin}\maxim@m{\v@leur}{\v@leur}{\delt@}%
    \ifdim\v@leur>\c@rre\invers@{\v@leur}{\v@leur}\edef\Uns@rM@x{\repdecn@mb{\v@leur}}%
    \v@lmax=\Uns@rM@x\v@lmax\mili@u=\Uns@rM@x\mili@u\v@lmin=\Uns@rM@x\v@lmin%
    \maxim@m{\v@leur}{-\v@lmax}{\v@lmax}\ifdim\v@leur<\c@rre%
    \maxim@m{\v@leur}{-\mili@u}{\mili@u}\ifdim\v@leur<\c@rre\else%
    \invers@{\mili@u}{\mili@u}\v@leur=-0.5\v@lmin%
    \v@leur=\repdecn@mb{\mili@u}\v@leur\m@jBBB@x{\v@leur}{#1}{#2}(#3,#4,#5,#6)\fi%
    \else\delt@=\repdecn@mb{\mili@u}\mili@u\v@leur=\repdecn@mb{\v@lmax}\v@lmin%
    \advance\delt@-\v@leur\ifdim\delt@<\z@\else\invers@{\v@lmax}{\v@lmax}%
    \edef\Uns@rAp{\repdecn@mb{\v@lmax}}\sqrt@{\delt@}{\delt@}%
    \v@leur=-\mili@u\advance\v@leur\delt@\v@leur=\Uns@rAp\v@leur%
    \m@jBBB@x{\v@leur}{#1}{#2}(#3,#4,#5,#6)%
    \v@leur=-\mili@u\advance\v@leur-\delt@\v@leur=\Uns@rAp\v@leur%
    \m@jBBB@x{\v@leur}{#1}{#2}(#3,#4,#5,#6)\fi\fi\fi}}
\ctr@ld@f\def\m@jBBB@x#1#2#3(#4,#5,#6,#7){{\relax\ifdim#1>\z@\ifdim#1<\p@%
    \edef\T@{\repdecn@mb{#1}}\v@lX=\p@\advance\v@lX-#1\edef\UNmT@{\repdecn@mb{\v@lX}}%
    \v@lX=#4\v@lY=#5\v@lZ=#6\v@lXa=#7\v@lX=\UNmT@\v@lX\advance\v@lX\T@\v@lY%
    \v@lY=\UNmT@\v@lY\advance\v@lY\T@\v@lZ\v@lZ=\UNmT@\v@lZ\advance\v@lZ\T@\v@lXa%
    \v@lX=\UNmT@\v@lX\advance\v@lX\T@\v@lY\v@lY=\UNmT@\v@lY\advance\v@lY\T@\v@lZ%
    \v@lX=\UNmT@\v@lX\advance\v@lX\T@\v@lY%
    \ifcase#2\or\v@lY=#3\or\v@lY=\v@lX\v@lX=#3\fi\b@undb@x{\v@lX}{\v@lY}\fi\fi}}
\ctr@ld@f\def\PsB@zier#1[#2]{{\f@gnewpath%
    \s@mme=\z@\def\list@num{#2,0}\extrairelepremi@r\p@int\de\list@num%
    \PSwrit@cmdS{\p@int}{\c@mmoveto}{\fwf@g}{\X@un}{\Y@un}\p@rtent=#1\bclB@zier}}
\ctr@ld@f\def\bclB@zier{\relax%
    \ifnum\s@mme<\p@rtent\advance\s@mme\@ne\BdingB@xfalse%
    \extrairelepremi@r\p@int\de\list@num\PSwrit@cmdS{\p@int}{}{\fwf@g}{\X@de}{\Y@de}%
    \extrairelepremi@r\p@int\de\list@num\PSwrit@cmdS{\p@int}{}{\fwf@g}{\X@tr}{\Y@tr}%
    \BdingB@xtrue%
    \extrairelepremi@r\p@int\de\list@num\PSwrit@cmdS{\p@int}{\c@mcurveto}{\fwf@g}{\X@qu}{\Y@qu}%
    \B@zierBB@x{1}{\Y@un}(\X@un,\X@de,\X@tr,\X@qu)%
    \B@zierBB@x{2}{\X@un}(\Y@un,\Y@de,\Y@tr,\Y@qu)%
    \edef\X@un{\X@qu}\edef\Y@un{\Y@qu}\bclB@zier\fi}
\ctr@ln@m\psBezier
\ctr@ld@f\def\psBezierDD#1[#2]{\ifcurr@ntPS\ifps@cri%
    \PSc@mment{psBezierDD N arcs=#1, Control points=#2}%
    \iffillm@de\PsB@zier#1[#2]%
    \f@gfill%
    \else\PsB@zier#1[#2]\f@gstroke\fi%
    \PSc@mment{End psBezierDD}\fi\fi}
\ctr@ln@m\et@tpsBezierTD% ou doubler les {}
\ctr@ld@f\def\psBezierTD#1[#2]{\ifcurr@ntPS\ifps@cri\s@uvc@ntr@l\et@tpsBezierTD%
    \PSc@mment{psBezierTD N arcs=#1, Control points=#2}%
    \iffillm@de\PsB@zierTD#1[#2]%
    \f@gfill%
    \else\PsB@zierTD#1[#2]\f@gstroke\fi%
    \PSc@mment{End psBezierTD}\resetc@ntr@l\et@tpsBezierTD\fi\fi}
\ctr@ld@f\def\PsB@zierTD#1[#2]{\ifnum\curr@ntproj<\tw@\PsB@zier#1[#2]\else\PsB@zier@TD#1[#2]\fi}
\ctr@ld@f\def\PsB@zier@TD#1[#2]{{\f@gnewpath%
    \s@mme=\z@\def\list@num{#2,0}\extrairelepremi@r\p@int\de\list@num%
    \let\c@lprojSP=\relax\setc@ntr@l{2}\Figptpr@j-7:/\p@int/%
    \PSwrit@cmd{-7}{\c@mmoveto}{\fwf@g}%
    \loop\ifnum\s@mme<#1\advance\s@mme\@ne\extrairelepremi@r\p@intun\de\list@num%
    \extrairelepremi@r\p@intde\de\list@num\extrairelepremi@r\p@inttr\de\list@num%
    \subB@zierTD\NBz@rcs[\p@int,\p@intun,\p@intde,\p@inttr]\edef\p@int{\p@inttr}\repeat}}
\ctr@ld@f\def\subB@zierTD#1[#2,#3,#4,#5]{\delt@=\p@\divide\delt@\NBz@rcs\v@lmin=\z@%
    {\Figg@tXY{-7}\edef\X@un{\the\v@lX}\edef\Y@un{\the\v@lY}%
    \s@mme=\z@\loop\ifnum\s@mme<#1\advance\s@mme\@ne%
    \v@leur=\v@lmin\advance\v@leur0.33333 \delt@\edef\unti@rs{\repdecn@mb{\v@leur}}%
    \v@leur=\v@lmin\advance\v@leur0.66666 \delt@\edef\deti@rs{\repdecn@mb{\v@leur}}%
    \advance\v@lmin\delt@\edef\trti@rs{\repdecn@mb{\v@lmin}}%
    \figptBezierTD-8::\trti@rs[#2,#3,#4,#5]\Figptpr@j-8:/-8/%
    \c@lsubBzarc\unti@rs,\deti@rs[#2,#3,#4,#5]\BdingB@xfalse%
    \PSwrit@cmdS{-4}{}{\fwf@g}{\X@de}{\Y@de}\PSwrit@cmdS{-3}{}{\fwf@g}{\X@tr}{\Y@tr}%
    \BdingB@xtrue\PSwrit@cmdS{-8}{\c@mcurveto}{\fwf@g}{\X@qu}{\Y@qu}%
    \B@zierBB@x{1}{\Y@un}(\X@un,\X@de,\X@tr,\X@qu)%
    \B@zierBB@x{2}{\X@un}(\Y@un,\Y@de,\Y@tr,\Y@qu)%
    \edef\X@un{\X@qu}\edef\Y@un{\Y@qu}\figptcopyDD-7:/-8/\repeat}}
\ctr@ld@f\def\NBz@rcs{2}
\ctr@ld@f\def\c@lsubBzarc#1,#2[#3,#4,#5,#6]{\figptBezierTD-5::#1[#3,#4,#5,#6]%
    \figptBezierTD-6::#2[#3,#4,#5,#6]\Figptpr@j-4:/-5/\Figptpr@j-5:/-6/%
    \figptscontrolDD-4[-7,-4,-5,-8]}
\ctr@ln@m\pscirc
\ctr@ld@f\def\pscircDD#1(#2){\ifcurr@ntPS\ifps@cri\PSc@mment{pscircDD Center=#1 (Radius=#2)}%
    \psarccircDD#1;#2(0,360)\PSc@mment{End pscircDD}\fi\fi}
\ctr@ld@f\def\pscircTD#1,#2,#3(#4){\ifcurr@ntPS\ifps@cri%
    \PSc@mment{pscircTD Center=#1,P1=#2,P2=#3 (Radius=#4)}%
    \psarccircTD#1,#2,#3;#4(0,360)\PSc@mment{End pscircTD}\fi\fi}
\ctr@ln@m\p@urcent
{\catcode`\%=12\gdef\p@urcent{%}}
\ctr@ld@f\def\PSc@mment#1{\ifpsdebugmode\immediate\write\fwf@g{\p@urcent\space#1}\fi}
\ctr@ln@m\acc@louv \ctr@ln@m\acc@lfer
{\catcode`\[=1\catcode`\{=12\gdef\acc@louv[{}}
{\catcode`\]=2\catcode`\}=12\gdef\acc@lfer{}]]
\ctr@ld@f\def\PSdict@{\ifUse@llipse%
    \immediate\write\fwf@g{/ellipsedict 9 dict def ellipsedict /mtrx matrix put}%
    \immediate\write\fwf@g{/ellipse \acc@louv ellipsedict begin}%
    \immediate\write\fwf@g{ /endangle exch def /startangle exch def}%
    \immediate\write\fwf@g{ /yrad exch def /xrad exch def}%
    \immediate\write\fwf@g{ /rotangle exch def /y exch def /x exch def}%
    \immediate\write\fwf@g{ /savematrix mtrx currentmatrix def}%
    \immediate\write\fwf@g{ x y translate rotangle rotate xrad yrad scale}%
    \immediate\write\fwf@g{ 0 0 1 startangle endangle arc}%
    \immediate\write\fwf@g{ savematrix setmatrix end\acc@lfer def}%
    \fi\PShe@der{EndProlog}}
\ctr@ld@f\def\Pssetc@rve#1=#2|{\keln@mun#1|%
    \def\n@mref{r}\ifx\l@debut\n@mref\pssetroundness{#2}\else% roundness
    \immediate\write16{*** Unknown attribute: \BS@ psset curve(..., #1=...)}%
    \fi}
\ctr@ln@m\curv@roundness
\ctr@ld@f\def\pssetroundness#1{\edef\curv@roundness{#1}}
\ctr@ld@f\def\defaultroundness{0.2} % Valeur par defaut
\ctr@ln@m\pscurve
\ctr@ld@f\def\pscurveDD[#1]{{\ifcurr@ntPS\ifps@cri\PSc@mment{pscurveDD Points=#1}%
    \s@uvc@ntr@l\et@tpscurveDD%
    \iffillm@de\Psc@rveDD\curv@roundness[#1]%
    \f@gfill%
    \else\Psc@rveDD\curv@roundness[#1]\f@gstroke\fi%
    \PSc@mment{End pscurveDD}\resetc@ntr@l\et@tpscurveDD\fi\fi}}
\ctr@ld@f\def\pscurveTD[#1]{{\ifcurr@ntPS\ifps@cri%
    \PSc@mment{pscurveTD Points=#1}\s@uvc@ntr@l\et@tpscurveTD\let\c@lprojSP=\relax%
    \iffillm@de\Psc@rveTD\curv@roundness[#1]%
    \f@gfill%
    \else\Psc@rveTD\curv@roundness[#1]\f@gstroke\fi%
    \PSc@mment{End pscurveTD}\resetc@ntr@l\et@tpscurveTD\fi\fi}}
\ctr@ld@f\def\Psc@rveDD#1[#2]{%
    \def\list@num{#2}\extrairelepremi@r\Ak@\de\list@num%
    \extrairelepremi@r\Ai@\de\list@num\extrairelepremi@r\Aj@\de\list@num%
    \f@gnewpath\PSwrit@cmdS{\Ai@}{\c@mmoveto}{\fwf@g}{\X@un}{\Y@un}%
    \setc@ntr@l{2}\figvectPDD -1[\Ak@,\Aj@]%
    \@ecfor\Ak@:=\list@num\do{\figpttraDD-2:=\Ai@/#1,-1/\BdingB@xfalse%
       \PSwrit@cmdS{-2}{}{\fwf@g}{\X@de}{\Y@de}%
       \figvectPDD -1[\Ai@,\Ak@]\figpttraDD-2:=\Aj@/-#1,-1/%
       \PSwrit@cmdS{-2}{}{\fwf@g}{\X@tr}{\Y@tr}\BdingB@xtrue%
       \PSwrit@cmdS{\Aj@}{\c@mcurveto}{\fwf@g}{\X@qu}{\Y@qu}%
       \B@zierBB@x{1}{\Y@un}(\X@un,\X@de,\X@tr,\X@qu)%
       \B@zierBB@x{2}{\X@un}(\Y@un,\Y@de,\Y@tr,\Y@qu)%
       \edef\X@un{\X@qu}\edef\Y@un{\Y@qu}\edef\Ai@{\Aj@}\edef\Aj@{\Ak@}}}
\ctr@ld@f\def\Psc@rveTD#1[#2]{\ifnum\curr@ntproj<\tw@\Psc@rvePPTD#1[#2]\else\Psc@rveCPTD#1[#2]\fi}
\ctr@ld@f\def\Psc@rvePPTD#1[#2]{\setc@ntr@l{2}%
    \def\list@num{#2}\extrairelepremi@r\Ak@\de\list@num\Figptpr@j-5:/\Ak@/%
    \extrairelepremi@r\Ai@\de\list@num\Figptpr@j-3:/\Ai@/%
    \extrairelepremi@r\Aj@\de\list@num\Figptpr@j-4:/\Aj@/%
    \f@gnewpath\PSwrit@cmdS{-3}{\c@mmoveto}{\fwf@g}{\X@un}{\Y@un}%
    \figvectPDD -1[-5,-4]%
    \@ecfor\Ak@:=\list@num\do{\Figptpr@j-5:/\Ak@/\figpttraDD-2:=-3/#1,-1/%
       \BdingB@xfalse\PSwrit@cmdS{-2}{}{\fwf@g}{\X@de}{\Y@de}%
       \figvectPDD -1[-3,-5]\figpttraDD-2:=-4/-#1,-1/%
       \PSwrit@cmdS{-2}{}{\fwf@g}{\X@tr}{\Y@tr}\BdingB@xtrue%
       \PSwrit@cmdS{-4}{\c@mcurveto}{\fwf@g}{\X@qu}{\Y@qu}%
       \B@zierBB@x{1}{\Y@un}(\X@un,\X@de,\X@tr,\X@qu)%
       \B@zierBB@x{2}{\X@un}(\Y@un,\Y@de,\Y@tr,\Y@qu)%
       \edef\X@un{\X@qu}\edef\Y@un{\Y@qu}\figptcopyDD-3:/-4/\figptcopyDD-4:/-5/}}
\ctr@ld@f\def\Psc@rveCPTD#1[#2]{\setc@ntr@l{2}%
    \def\list@num{#2}\extrairelepremi@r\Ak@\de\list@num%
    \extrairelepremi@r\Ai@\de\list@num\extrairelepremi@r\Aj@\de\list@num%
    \Figptpr@j-7:/\Ai@/%
    \f@gnewpath\PSwrit@cmd{-7}{\c@mmoveto}{\fwf@g}%
    \figvectPTD -9[\Ak@,\Aj@]%
    \@ecfor\Ak@:=\list@num\do{\figpttraTD-10:=\Ai@/#1,-9/%
       \figvectPTD -9[\Ai@,\Ak@]\figpttraTD-11:=\Aj@/-#1,-9/%
       \subB@zierTD\NBz@rcs[\Ai@,-10,-11,\Aj@]\edef\Ai@{\Aj@}\edef\Aj@{\Ak@}}}
\ctr@ld@f\def\psendfig{\ifcurr@ntPS\ifps@cri\immediate\closeout\fwf@g%
    \immediate\openout\fwf@g=\PSfilen@me\relax%
    \ifPDFm@ke\PSBdingB@x\else%
    \immediate\write\fwf@g{\p@urcent\string!PS-Adobe-2.0 EPSF-2.0}%
    \PShe@der{Creator\string: TeX (fig4tex.tex)}%
    \PShe@der{Title\string: \PSfilen@me}%
    \PShe@der{CreationDate\string: \the\day/\the\month/\the\year}%
    \PSBdingB@x%
    \PShe@der{EndComments}\PSdict@\fi%
    \immediate\write\fwf@g{\c@mgsave}%
    \openin\frf@g=\auxfilen@me\c@pypsfile\fwf@g\frf@g\closein\frf@g%
    \immediate\write\fwf@g{\c@mgrestore}%
    \PSc@mment{End of file.}\immediate\closeout\fwf@g%
    \immediate\openout\fwf@g=\auxfilen@me\immediate\closeout\fwf@g%
    \immediate\write16{File \PSfilen@me\space created.}\fi\fi\curr@ntPSfalse\ps@critrue}
\ctr@ld@f\def\PShe@der#1{\immediate\write\fwf@g{\p@urcent\p@urcent#1}}
\ctr@ld@f\def\PSBdingB@x{{\v@lX=\ptT@ptps\c@@rdXmin\v@lY=\ptT@ptps\c@@rdYmin%
     \v@lXa=\ptT@ptps\c@@rdXmax\v@lYa=\ptT@ptps\c@@rdYmax%
     \PShe@der{BoundingBox\string: \repdecn@mb{\v@lX}\space\repdecn@mb{\v@lY}%
     \space\repdecn@mb{\v@lXa}\space\repdecn@mb{\v@lYa}}}}
\ctr@ld@f\def\psfcconnect[#1]{{\ifcurr@ntPS\ifps@cri\PSc@mment{psfcconnect Points=#1}%
    \pssetfillmode{no}\s@uvc@ntr@l\et@tpsfcconnect\resetc@ntr@l{2}%
    \fcc@nnect@[#1]\resetc@ntr@l\et@tpsfcconnect\PSc@mment{End psfcconnect}\fi\fi}}
\ctr@ld@f\def\fcc@nnect@[#1]{\let\N@rm=\n@rmeucDD\def\list@num{#1}%
    \extrairelepremi@r\Ai@\de\list@num\edef\pr@m{\Ai@}\v@leur=\z@\p@rtent=\@ne\c@llgtot%
    \ifcase\fclin@typ@\edef\list@num{[\pr@m,#1,\Ai@}\expandafter\pscurve\list@num]%
    \else\ifdim\fclin@r@d\p@>\z@\Pslin@conge[#1]\else\psline[#1]\fi\fi%
    \v@leur=\@rrowp@s\v@leur\edef\list@num{#1,\Ai@,0}%
    \extrairelepremi@r\Ai@\de\list@num\mili@u=\epsil@n\c@llgpart%
    \advance\mili@u-\epsil@n\advance\mili@u-\delt@\advance\v@leur-\mili@u%
    \ifcase\fclin@typ@\invers@\mili@u\delt@%
    \ifnum\@rrowr@fpt>\z@\advance\delt@-\v@leur\v@leur=\delt@\fi%
    \v@leur=\repdecn@mb\v@leur\mili@u\edef\v@lt{\repdecn@mb\v@leur}%
    \extrairelepremi@r\Ak@\de\list@num%
    \figvectPDD-1[\pr@m,\Aj@]\figpttraDD-6:=\Ai@/\curv@roundness,-1/%
    \figvectPDD-1[\Ak@,\Ai@]\figpttraDD-7:=\Aj@/\curv@roundness,-1/%
    \delt@=\@rrowheadlength\p@\delt@=\C@AHANG\delt@\edef\R@dius{\repdecn@mb{\delt@}}%
    \ifcase\@rrowr@fpt%
    \FigptintercircB@zDD-8::\v@lt,\R@dius[\Ai@,-6,-7,\Aj@]\psarrowheadDD[-5,-8]\else%
    \FigptintercircB@zDD-8::\v@lt,\R@dius[\Aj@,-7,-6,\Ai@]\psarrowheadDD[-8,-5]\fi%
    \else\advance\delt@-\v@leur%
    \p@rtentiere{\p@rtent}{\delt@}\edef\C@efun{\the\p@rtent}%
    \p@rtentiere{\p@rtent}{\v@leur}\edef\C@efde{\the\p@rtent}%
    \figptbaryDD-5:[\Ai@,\Aj@;\C@efun,\C@efde]\ifcase\@rrowr@fpt%
    \delt@=\@rrowheadlength\unit@\delt@=\C@AHANG\delt@\edef\t@ille{\repdecn@mb{\delt@}}%
    \figvectPDD-2[\Ai@,\Aj@]\vecunit@{-2}{-2}\figpttraDD-5:=-5/\t@ille,-2/\fi%
    \psarrowheadDD[\Ai@,-5]\fi}
\ctr@ld@f\def\c@llgtot{\@ecfor\Aj@:=\list@num\do{\figvectP-1[\Ai@,\Aj@]\N@rm\delt@{-1}%
    \advance\v@leur\delt@\advance\p@rtent\@ne\edef\Ai@{\Aj@}}}
\ctr@ld@f\def\c@llgpart{\extrairelepremi@r\Aj@\de\list@num\figvectP-1[\Ai@,\Aj@]\N@rm\delt@{-1}%
    \advance\mili@u\delt@\ifdim\mili@u<\v@leur\edef\pr@m{\Ai@}\edef\Ai@{\Aj@}\c@llgpart\fi}
\ctr@ld@f\def\Pslin@conge[#1]{\ifnum\p@rtent>\tw@{\def\list@num{#1}%
    \extrairelepremi@r\Ai@\de\list@num\extrairelepremi@r\Aj@\de\list@num%
    \figptcopy-6:/\Ai@/\figvectP-3[\Ai@,\Aj@]\vecunit@{-3}{-3}\v@lmax=\result@t%
    \@ecfor\Ak@:=\list@num\do{\figvectP-4[\Aj@,\Ak@]\vecunit@{-4}{-4}%
    \minim@m\v@lmin\v@lmax\result@t\v@lmax=\result@t%
    \det@rm\delt@[-3,-4]\maxim@m\mili@u{\delt@}{-\delt@}\ifdim\mili@u>\Cepsil@n%
    \ifdim\delt@>\z@\figgetangleDD\Angl@[\Aj@,\Ak@,\Ai@]\else%
    \figgetangleDD\Angl@[\Aj@,\Ai@,\Ak@]\fi%
    \v@leur=\PI@deg\advance\v@leur-\Angl@\p@\divide\v@leur\tw@%
    \edef\Angl@{\repdecn@mb\v@leur}\c@ssin{\C@}{\S@}{\Angl@}\v@leur=\fclin@r@d\unit@%
    \v@leur=\S@\v@leur\mili@u=\C@\p@\invers@\mili@u\mili@u%
    \v@leur=\repdecn@mb{\mili@u}\v@leur%
    \minim@m\v@leur\v@leur\v@lmin\edef\t@ille{\repdecn@mb{\v@leur}}%
    \figpttra-5:=\Aj@/-\t@ille,-3/\psline[-6,-5]\figpttra-6:=\Aj@/\t@ille,-4/%
    \figvectNVDD-3[-3]\figvectNVDD-8[-4]\inters@cDD-7:[-5,-3;-6,-8]%
    \ifdim\delt@>\z@\psarccircP-7;\fclin@r@d[-5,-6]\else\psarccircP-7;\fclin@r@d[-6,-5]\fi%
    \else\psline[-6,\Aj@]\figptcopy-6:/\Aj@/\fi% Points alignes
    \edef\Ai@{\Aj@}\edef\Aj@{\Ak@}\figptcopy-3:/-4/}\psline[-6,\Aj@]}\else\psline[#1]\fi}
\ctr@ld@f\def\psfcnode[#1]#2{{\ifcurr@ntPS\ifps@cri\PSc@mment{psfcnode Points=#1}%
    \s@uvc@ntr@l\et@tpsfcnode\resetc@ntr@l{2}%
    \def\t@xt@{#2}\ifx\t@xt@\empty\def\g@tt@xt{\setbox\Gb@x=\hbox{\Figg@tT{\p@int}}}%
    \else\def\g@tt@xt{\setbox\Gb@x=\hbox{#2}}\fi%
    \v@lmin=\h@rdfcXp@dd\advance\v@lmin\Xp@dd\unit@\multiply\v@lmin\tw@%
    \v@lmax=\h@rdfcYp@dd\advance\v@lmax\Yp@dd\unit@\multiply\v@lmax\tw@%
    \Figv@ctCreg-8(\unit@,-\unit@)\def\list@num{#1}%
    \delt@=\curr@ntwidth bp\divide\delt@\tw@%
    \fcn@de\PSc@mment{End psfcnode}\resetc@ntr@l\et@tpsfcnode\fi\fi}}
\ctr@ld@f\def\d@butn@de{\g@tt@xt\v@lX=\wd\Gb@x%
    \v@lY=\ht\Gb@x\advance\v@lY\dp\Gb@x\advance\v@lX\v@lmin\advance\v@lY\v@lmax}
\ctr@ld@f\def\fcn@deE{%
    \@ecfor\p@int:=\list@num\do{\d@butn@de\v@lX=\unssqrttw@\v@lX\v@lY=\unssqrttw@\v@lY%
    \ifdim\thickn@ss\p@>\z@% Shadow
    \v@lXa=\v@lX\advance\v@lXa\delt@\v@lXa=\ptT@unit@\v@lXa\edef\XR@d{\repdecn@mb\v@lXa}%
    \v@lYa=\v@lY\advance\v@lYa\delt@\v@lYa=\ptT@unit@\v@lYa\edef\YR@d{\repdecn@mb\v@lYa}%
    \arct@n\v@leur(\v@lXa,\v@lYa)\v@leur=\rdT@deg\v@leur\edef\@nglde{\repdecn@mb\v@leur}%
    {\c@lptellDD-2::\p@int;\XR@d,\YR@d(\@nglde)}% \v@lmin & \v@lmax modified in \c@lptellDD
    \advance\v@leur-\PI@deg\edef\@nglun{\repdecn@mb\v@leur}%
    {\c@lptellDD-3::\p@int;\XR@d,\YR@d(\@nglun)}%
    \figptstra-6=-3,-2,\p@int/\thickn@ss,-8/\pssetfillmode{yes}\us@secondC@lor%
    \psline[-2,-3,-6,-5]\psarcell-4;\XR@d,\YR@d(\@nglun,\@nglde,0)\fi% End shadow
    \v@lX=\ptT@unit@\v@lX\v@lY=\ptT@unit@\v@lY%
    \edef\XR@d{\repdecn@mb\v@lX}\edef\YR@d{\repdecn@mb\v@lY}%
    \pssetfillmode{yes}\us@thirdC@lor\psarcell\p@int;\XR@d,\YR@d(0,360,0)%
    \pssetfillmode{no}\us@primarC@lor\psarcell\p@int;\XR@d,\YR@d(0,360,0)}}
\ctr@ld@f\def\fcn@deL{\delt@=\ptT@unit@\delt@\edef\t@ille{\repdecn@mb\delt@}%
    \@ecfor\p@int:=\list@num\do{\Figg@tXYa{\p@int}\d@butn@de%
    \ifdim\v@lX>\v@lY\itis@Ktrue\else\itis@Kfalse\fi%
    \advance\v@lXa-\v@lX\Figp@intreg-1:(\v@lXa,\v@lYa)%
    \advance\v@lXa\v@lX\advance\v@lYa-\v@lY\Figp@intreg-2:(\v@lXa,\v@lYa)%
    \advance\v@lXa\v@lX\advance\v@lYa\v@lY\Figp@intreg-3:(\v@lXa,\v@lYa)%
    \advance\v@lXa-\v@lX\advance\v@lYa\v@lY\Figp@intreg-4:(\v@lXa,\v@lYa)%
    \ifdim\thickn@ss\p@>\z@\Figg@tXYa{\p@int}\pssetfillmode{yes}\us@secondC@lor% Shadow
    \c@lpt@xt{-1}{-4}\c@lpt@xt@\v@lXa\v@lYa\v@lX\v@lY\c@rre\delt@%
    \Figp@intregDD-9:(\v@lZ,\v@lYa)\Figp@intregDD-11:(\v@lZa,\v@lYa)%
    \c@lpt@xt{-4}{-3}\c@lpt@xt@\v@lYa\v@lXa\v@lY\v@lX\delt@\c@rre%
    \Figp@intregDD-12:(\v@lXa,\v@lZ)\Figp@intregDD-10:(\v@lXa,\v@lZa)%
    \ifitis@K\figptstra-7=-9,-10,-11/\thickn@ss,-8/\psline[-9,-11,-5,-6,-7]\else%
    \figptstra-7=-10,-11,-12/\thickn@ss,-8/\psline[-10,-12,-5,-6,-7]\fi\fi% End shadow
    \pssetfillmode{yes}\us@thirdC@lor\psline[-1,-2,-3,-4]%
    \pssetfillmode{no}\us@primarC@lor\psline[-1,-2,-3,-4,-1]}}
\ctr@ld@f\def\c@lpt@xt#1#2{\figvectN-7[#1,#2]\vecunit@{-7}{-7}\figpttra-5:=#1/\t@ille,-7/%
    \figvectP-7[#1,#2]\Figg@tXY{-7}\c@rre=\v@lX\delt@=\v@lY\Figg@tXY{-5}}
\ctr@ld@f\def\c@lpt@xt@#1#2#3#4#5#6{\v@lZ=#6\invers@{\v@lZ}{\v@lZ}\v@leur=\repdecn@mb{#5}\v@lZ%
    \v@lZ=#2\advance\v@lZ-#4\mili@u=\repdecn@mb{\v@leur}\v@lZ%
    \v@lZ=#3\advance\v@lZ\mili@u\v@lZa=-\v@lZ\advance\v@lZa\tw@#1}
\ctr@ld@f\def\fcn@deR{\@ecfor\p@int:=\list@num\do{\Figg@tXYa{\p@int}\d@butn@de%
    \advance\v@lXa-0.5\v@lX\advance\v@lYa-0.5\v@lY\Figp@intreg-1:(\v@lXa,\v@lYa)%
    \advance\v@lXa\v@lX\Figp@intreg-2:(\v@lXa,\v@lYa)%
    \advance\v@lYa\v@lY\Figp@intreg-3:(\v@lXa,\v@lYa)%
    \advance\v@lXa-\v@lX\Figp@intreg-4:(\v@lXa,\v@lYa)%
    \ifdim\thickn@ss\p@>\z@\pssetfillmode{yes}\us@secondC@lor% Shadow
    \Figv@ctCreg-5(-\delt@,-\delt@)\figpttra-9:=-1/1,-5/%
    \Figv@ctCreg-5(\delt@,-\delt@)\figpttra-10:=-2/1,-5/%
    \Figv@ctCreg-5(\delt@,\delt@)\figpttra-11:=-3/1,-5/%
    \figptstra-7=-9,-10,-11/\thickn@ss,-8/\psline[-9,-11,-5,-6,-7]\fi% End shadow
    \pssetfillmode{yes}\us@thirdC@lor\psline[-1,-2,-3,-4]%
    \pssetfillmode{no}\us@primarC@lor\psline[-1,-2,-3,-4,-1]}}
\ctr@ln@m\@rrowp@s
\ctr@ln@m\Xp@dd     \ctr@ln@m\Yp@dd
\ctr@ln@m\fclin@r@d \ctr@ln@m\thickn@ss
\ctr@ld@f\def\Pssetfl@wchart#1=#2|{\keln@mtr#1|%
    \def\n@mref{arr}\ifx\l@debut\n@mref\expandafter\keln@mtr\l@suite|%
     \def\n@mref{owp}\ifx\l@debut\n@mref\edef\@rrowp@s{#2}\else% arrowposition
     \def\n@mref{owr}\ifx\l@debut\n@mref\setfcr@fpt#2|\else% arrowrefpt
     \immediate\write16{*** Unknown attribute: \BS@ psset flowchart(..., #1=...)}%
     \fi\fi\else%
    \def\n@mref{lin}\ifx\l@debut\n@mref\setfccurv@#2|\else% line
    \def\n@mref{pad}\ifx\l@debut\n@mref\edef\Xp@dd{#2}\edef\Yp@dd{#2}\else% padding
    \def\n@mref{rad}\ifx\l@debut\n@mref\edef\fclin@r@d{#2}\else% connection radius
    \def\n@mref{sha}\ifx\l@debut\n@mref\setfcshap@#2|\else% shape
    \def\n@mref{thi}\ifx\l@debut\n@mref\edef\thickn@ss{#2}\else% thickness
    \def\n@mref{xpa}\ifx\l@debut\n@mref\edef\Xp@dd{#2}\else% xpadding
    \def\n@mref{ypa}\ifx\l@debut\n@mref\edef\Yp@dd{#2}\else% ypadding
    \immediate\write16{*** Unknown attribute: \BS@ psset flowchart(..., #1=...)}%
    \fi\fi\fi\fi\fi\fi\fi\fi}
\ctr@ln@m\@rrowr@fpt \ctr@ln@m\fclin@typ@
\ctr@ld@f\def\setfcr@fpt#1#2|{\if#1e\def\@rrowr@fpt{1}\else\def\@rrowr@fpt{0}\fi}
\ctr@ld@f\def\setfccurv@#1#2|{\if#1c\def\fclin@typ@{0}\else\def\fclin@typ@{1}\fi}
\ctr@ln@m\h@rdfcXp@dd \ctr@ln@m\h@rdfcYp@dd
\ctr@ln@m\fcn@de \ctr@ln@m\fcsh@pe
\ctr@ld@f\def\setfcshap@#1#2|{%
    \if#1e\let\fcn@de=\fcn@deE\def\h@rdfcXp@dd{4pt}\def\h@rdfcYp@dd{4pt}%
     \edef\fcsh@pe{ellipse}\else%
    \if#1l\let\fcn@de=\fcn@deL\def\h@rdfcXp@dd{4pt}\def\h@rdfcYp@dd{4pt}%
     \edef\fcsh@pe{lozenge}\else%
          \let\fcn@de=\fcn@deR\def\h@rdfcXp@dd{6pt}\def\h@rdfcYp@dd{6pt}%
     \edef\fcsh@pe{rectangle}\fi\fi}
\ctr@ld@f\def\psline[#1]{{\ifcurr@ntPS\ifps@cri\PSc@mment{psline Points=#1}%
    \let\pslign@=\Pslign@P\Pslin@{#1}\PSc@mment{End psline}\fi\fi}}
\ctr@ld@f\def\pslineF#1{{\ifcurr@ntPS\ifps@cri\PSc@mment{pslineF Filename=#1}%
    \let\pslign@=\Pslign@F\Pslin@{#1}\PSc@mment{End pslineF}\fi\fi}}
\ctr@ld@f\def\pslineC(#1){{\ifcurr@ntPS\ifps@cri\PSc@mment{pslineC}%
    \let\pslign@=\Pslign@C\Pslin@{#1}\PSc@mment{End pslineC}\fi\fi}}
\ctr@ld@f\def\Pslin@#1{\iffillm@de\pslign@{#1}%
    \f@gfill%
    \else\pslign@{#1}\ifx\derp@int\premp@int%
    \f@gclosestroke%
    \else\f@gstroke\fi\fi}
\ctr@ld@f\def\Pslign@P#1{\def\list@num{#1}\extrairelepremi@r\p@int\de\list@num%
    \edef\premp@int{\p@int}\f@gnewpath%
    \PSwrit@cmd{\p@int}{\c@mmoveto}{\fwf@g}%
    \@ecfor\p@int:=\list@num\do{\PSwrit@cmd{\p@int}{\c@mlineto}{\fwf@g}%
    \edef\derp@int{\p@int}}}
\ctr@ld@f\def\Pslign@F#1{\s@uvc@ntr@l\et@tPslign@F\setc@ntr@l{2}\openin\frf@g=#1\relax%
    \ifeof\frf@g\message{*** File #1 not found !}\end\else%
    \read\frf@g to\tr@c\edef\premp@int{\tr@c}\expandafter\extr@ctCF\tr@c:%
    \f@gnewpath\PSwrit@cmd{-1}{\c@mmoveto}{\fwf@g}%
    \loop\read\frf@g to\tr@c\ifeof\frf@g\mored@tafalse\else\mored@tatrue\fi%
    \ifmored@ta\expandafter\extr@ctCF\tr@c:\PSwrit@cmd{-1}{\c@mlineto}{\fwf@g}%
    \edef\derp@int{\tr@c}\repeat\fi\closein\frf@g\resetc@ntr@l\et@tPslign@F}
\ctr@ln@m\extr@ctCF
\ctr@ld@f\def\extr@ctCFDD#1 #2:{\v@lX=#1\unit@\v@lY=#2\unit@\Figp@intregDD-1:(\v@lX,\v@lY)}
\ctr@ld@f\def\extr@ctCFTD#1 #2 #3:{\v@lX=#1\unit@\v@lY=#2\unit@\v@lZ=#3\unit@%
    \Figp@intregTD-1:(\v@lX,\v@lY,\v@lZ)}
\ctr@ld@f\def\Pslign@C#1{\s@uvc@ntr@l\et@tPslign@C\setc@ntr@l{2}%
    \def\list@num{#1}\extrairelepremi@r\p@int\de\list@num%
    \edef\premp@int{\p@int}\f@gnewpath%
    \expandafter\Pslign@C@\p@int:\PSwrit@cmd{-1}{\c@mmoveto}{\fwf@g}%
    \@ecfor\p@int:=\list@num\do{\expandafter\Pslign@C@\p@int:%
    \PSwrit@cmd{-1}{\c@mlineto}{\fwf@g}\edef\derp@int{\p@int}}%
    \resetc@ntr@l\et@tPslign@C}
\ctr@ld@f\def\Pslign@C@#1 #2:{{\def\t@xt@{#1}\ifx\t@xt@\empty\Pslign@C@#2:% Discard leading spaces
    \else\extr@ctCF#1 #2:\fi}}
\ctr@ln@m\c@ntrolmesh
\ctr@ld@f\def\Pssetm@sh#1=#2|{\keln@mun#1|%
    \def\n@mref{d}\ifx\l@debut\n@mref\pssetmeshdiag{#2}\else% diag
    \immediate\write16{*** Unknown attribute: \BS@ psset mesh(..., #1=...)}%
    \fi}
\ctr@ld@f\def\pssetmeshdiag#1{\edef\c@ntrolmesh{#1}}
\ctr@ld@f\def\defaultmeshdiag{0}    % Valeur par defaut
\ctr@ld@f\def\psmesh#1,#2[#3,#4,#5,#6]{{\ifcurr@ntPS\ifps@cri%
    \PSc@mment{psmesh N1=#1, N2=#2, Quadrangle=[#3,#4,#5,#6]}%
    \s@uvc@ntr@l\et@tpsmesh\Pss@tsecondSt\setc@ntr@l{2}%
    \ifnum#1>\@ne\Psmeshp@rt#1[#3,#4,#5,#6]\fi%
    \ifnum#2>\@ne\Psmeshp@rt#2[#4,#5,#6,#3]\fi%
    \ifnum\c@ntrolmesh>\z@\Psmeshdi@g#1,#2[#3,#4,#5,#6]\fi%
    \ifnum\c@ntrolmesh<\z@\Psmeshdi@g#2,#1[#4,#5,#6,#3]\fi\Psrest@reSt%
    \psline[#3,#4,#5,#6,#3]\PSc@mment{End psmesh}\resetc@ntr@l\et@tpsmesh\fi\fi}}
\ctr@ld@f\def\Psmeshp@rt#1[#2,#3,#4,#5]{{\l@mbd@un=\@ne\l@mbd@de=#1\loop%
    \ifnum\l@mbd@un<#1\advance\l@mbd@de\m@ne\figptbary-1:[#2,#3;\l@mbd@de,\l@mbd@un]%
    \figptbary-2:[#5,#4;\l@mbd@de,\l@mbd@un]\psline[-1,-2]\advance\l@mbd@un\@ne\repeat}}
\ctr@ld@f\def\Psmeshdi@g#1,#2[#3,#4,#5,#6]{\figptcopy-2:/#3/\figptcopy-3:/#6/%
    \l@mbd@un=\z@\l@mbd@de=#1\loop\ifnum\l@mbd@un<#1%
    \advance\l@mbd@un\@ne\advance\l@mbd@de\m@ne\figptcopy-1:/-2/\figptcopy-4:/-3/%
    \figptbary-2:[#3,#4;\l@mbd@de,\l@mbd@un]%
    \figptbary-3:[#6,#5;\l@mbd@de,\l@mbd@un]\Psmeshdi@gp@rt#2[-1,-2,-3,-4]\repeat}
\ctr@ld@f\def\Psmeshdi@gp@rt#1[#2,#3,#4,#5]{{\l@mbd@un=\z@\l@mbd@de=#1\loop%
    \ifnum\l@mbd@un<#1\figptbary-5:[#2,#5;\l@mbd@de,\l@mbd@un]%
    \advance\l@mbd@de\m@ne\advance\l@mbd@un\@ne%
    \figptbary-6:[#3,#4;\l@mbd@de,\l@mbd@un]\psline[-5,-6]\repeat}}
\ctr@ln@m\psnormal
\ctr@ld@f\def\psnormalDD#1,#2[#3,#4]{{\ifcurr@ntPS\ifps@cri%
    \PSc@mment{psnormal Length=#1, Lambda=#2 [Pt1,Pt2]=[#3,#4]}%
    \s@uvc@ntr@l\et@tpsnormal\resetc@ntr@l{2}\figptendnormal-6::#1,#2[#3,#4]%
    \figptcopyDD-5:/-1/\psarrow[-5,-6]%
    \PSc@mment{End psnormal}\resetc@ntr@l\et@tpsnormal\fi\fi}}
\ctr@ld@f\def\psreset#1{\trtlis@rg{#1}{\Psreset@}}
\ctr@ld@f\def\Psreset@#1|{\keln@mde#1|%
    \def\n@mref{ar}\ifx\l@debut\n@mref\psresetarrowhead\else% arrowhead
    \def\n@mref{cu}\ifx\l@debut\n@mref\psset curve(roundness=\defaultroundness)\else% curve
    \def\n@mref{fi}\ifx\l@debut\n@mref\psset (color=\defaultcolor,dash=\defaultdash,%
         fill=\defaultfill,join=\defaultjoin,width=\defaultwidth)\else% primary settings
    \def\n@mref{fl}\ifx\l@debut\n@mref\psset flowchart(arrowp=\defaultfcarrowposition,%
	arrowr=\defaultfcarrowrefpt,line=\defaultfcline,xpadd=\defaultfcxpadding,%
	ypadd=\defaultfcypadding,radius=\defaultfcradius,shape=\defaultfcshape,%
	thick=\defaultfcthickness)\else% flow chart
    \def\n@mref{me}\ifx\l@debut\n@mref\psset mesh(diag=\defaultmeshdiag)\else% mesh
    \def\n@mref{se}\ifx\l@debut\n@mref\psresetsecondsettings\else% secondary
    \def\n@mref{th}\ifx\l@debut\n@mref\psset third(color=\defaultthirdcolor)\else% ternary
    \immediate\write16{*** Unknown keyword #1 (\BS@ psreset).}%
    \fi\fi\fi\fi\fi\fi\fi}
\ctr@ld@f\def\psset#1(#2){\def\t@xt@{#1}\ifx\t@xt@\empty\trtlis@rg{#2}{\Pssetf@rst}% primary settings
    \else\keln@mde#1|%
    \def\n@mref{ar}\ifx\l@debut\n@mref\trtlis@rg{#2}{\Psset@rrowhe@d}\else% arrow-head
    \def\n@mref{cu}\ifx\l@debut\n@mref\trtlis@rg{#2}{\Pssetc@rve}\else% curve
    \def\n@mref{fi}\ifx\l@debut\n@mref\trtlis@rg{#2}{\Pssetf@rst}\else% primary settings
    \def\n@mref{fl}\ifx\l@debut\n@mref\trtlis@rg{#2}{\Pssetfl@wchart}\else% flow chart
    \def\n@mref{me}\ifx\l@debut\n@mref\trtlis@rg{#2}{\Pssetm@sh}\else% mesh
    \def\n@mref{se}\ifx\l@debut\n@mref\trtlis@rg{#2}{\Pssets@cond}\else% secondary settings
    \def\n@mref{th}\ifx\l@debut\n@mref\trtlis@rg{#2}{\Pssetth@rd}\else% ternary settings
    \immediate\write16{*** Unknown keyword: \BS@ psset #1(...)}%
    \fi\fi\fi\fi\fi\fi\fi\fi}
\ctr@ld@f\def\pssetdefault#1(#2){\ifcurr@ntPS\immediate\write16{*** \BS@ pssetdefault is ignored
    inside a \BS@ psbeginfig-\BS@ psendfig block.}%
    \immediate\write16{*** It must be called before \BS@ psbeginfig.}\else%
    \def\t@xt@{#1}\ifx\t@xt@\empty\trtlis@rg{#2}{\Pssd@f@rst}\else\keln@mde#1|%
    \def\n@mref{ar}\ifx\l@debut\n@mref\trtlis@rg{#2}{\Pssd@@rrowhe@d}\else% arrow-head
    \def\n@mref{cu}\ifx\l@debut\n@mref\trtlis@rg{#2}{\Pssd@c@rve}\else% curve
    \def\n@mref{fi}\ifx\l@debut\n@mref\trtlis@rg{#2}{\Pssd@f@rst}\else% primary settings
    \def\n@mref{fl}\ifx\l@debut\n@mref\trtlis@rg{#2}{\Pssd@fl@wchart}\else% flow chart
    \def\n@mref{me}\ifx\l@debut\n@mref\trtlis@rg{#2}{\Pssd@m@sh}\else% mesh
    \def\n@mref{se}\ifx\l@debut\n@mref\trtlis@rg{#2}{\Pssd@s@cond}\else% secondary settings
    \def\n@mref{th}\ifx\l@debut\n@mref\trtlis@rg{#2}{\Pssd@th@rd}\else% ternary settings
    \immediate\write16{*** Unknown keyword: \BS@ pssetdefault #1(...)}%
    \fi\fi\fi\fi\fi\fi\fi\fi\initpss@ttings\fi}
\ctr@ld@f\def\Pssd@f@rst#1=#2|{\keln@mun#1|%
    \def\n@mref{c}\ifx\l@debut\n@mref\edef\defaultcolor{#2}\else% color
    \def\n@mref{d}\ifx\l@debut\n@mref\edef\defaultdash{#2}\else% dash
    \def\n@mref{f}\ifx\l@debut\n@mref\edef\defaultfill{#2}\else% fillmode
    \def\n@mref{j}\ifx\l@debut\n@mref\edef\defaultjoin{#2}\else% line join
    \def\n@mref{u}\ifx\l@debut\n@mref\edef\defaultupdate{#2}\pssetupdate{#2}\else% update
    \def\n@mref{w}\ifx\l@debut\n@mref\edef\defaultwidth{#2}\else% line width
    \immediate\write16{*** Unknown attribute: \BS@ pssetdefault (..., #1=...)}%
    \fi\fi\fi\fi\fi\fi}
\ctr@ld@f\def\Pssd@@rrowhe@d#1=#2|{\keln@mun#1|%
    \def\n@mref{a}\ifx\l@debut\n@mref\edef\defaultarrowheadangle{#2}\else% angle
    \def\n@mref{f}\ifx\l@debut\n@mref\edef\defaultarrowheadangle{#2}\else% fillmode
    \def\n@mref{l}\ifx\l@debut\n@mref\y@tiunit{#2}\ifunitpr@sent%
     \edef\defaulth@rdahlength{#2}\else\edef\defaulth@rdahlength{#2pt}%
     \message{*** \BS@ pssetdefault (..., #1=#2, ...) : unit is missing, pt is assumed.}%
     \fi\else% length
    \def\n@mref{o}\ifx\l@debut\n@mref\edef\defaultarrowheadout{#2}\else% out
    \def\n@mref{r}\ifx\l@debut\n@mref\edef\defaultarrowheadratio{#2}\else% ratio
    \immediate\write16{*** Unknown attribute: \BS@ pssetdefault arrowhead(..., #1=...)}%
    \fi\fi\fi\fi\fi}
\ctr@ld@f\def\Pssd@c@rve#1=#2|{\keln@mun#1|%
    \def\n@mref{r}\ifx\l@debut\n@mref\edef\defaultroundness{#2}\else%
    \immediate\write16{*** Unknown attribute: \BS@ pssetdefault curve(..., #1=...)}%
    \fi}
\ctr@ld@f\def\Pssd@fl@wchart#1=#2|{\keln@mtr#1|%
    \def\n@mref{arr}\ifx\l@debut\n@mref\expandafter\keln@mtr\l@suite|%
     \def\n@mref{owp}\ifx\l@debut\n@mref\edef\defaultfcarrowposition{#2}\else% arrowposition
     \def\n@mref{owr}\ifx\l@debut\n@mref\edef\defaultfcarrowrefpt{#2}\else% arrowrefpt
     \immediate\write16{*** Unknown attribute: \BS@ pssetdefault flowchart(..., #1=...)}%
     \fi\fi\else%
    \def\n@mref{lin}\ifx\l@debut\n@mref\edef\defaultfcline{#2}\else% line
    \def\n@mref{pad}\ifx\l@debut\n@mref\edef\defaultfcxpadding{#2}%
                    \edef\defaultfcypadding{#2}\else% padding
    \def\n@mref{rad}\ifx\l@debut\n@mref\edef\defaultfcradius{#2}\else% connection radius
    \def\n@mref{sha}\ifx\l@debut\n@mref\edef\defaultfcshape{#2}\else% shape
    \def\n@mref{thi}\ifx\l@debut\n@mref\edef\defaultfcthickness{#2}\else% thickness
    \def\n@mref{xpa}\ifx\l@debut\n@mref\edef\defaultfcxpadding{#2}\else% xpadding
    \def\n@mref{ypa}\ifx\l@debut\n@mref\edef\defaultfcypadding{#2}\else% ypadding
    \immediate\write16{*** Unknown attribute: \BS@ pssetdefault flowchart(..., #1=...)}%
    \fi\fi\fi\fi\fi\fi\fi\fi}
\ctr@ld@f\def\defaultfcarrowposition{0.5}%\ctr@ld@f\let\defaultfcarrowpos=\defaultfcarrowposition
\ctr@ld@f\def\defaultfcarrowrefpt{start}
\ctr@ld@f\def\defaultfcline{polygon}
\ctr@ld@f\def\defaultfcradius{0}
\ctr@ld@f\def\defaultfcshape{rectangle}
\ctr@ld@f\def\defaultfcthickness{0}%\ctr@ld@f\let\defaultfcthick=\defaultfcthickness
\ctr@ld@f\def\defaultfcxpadding{0}%\ctr@ld@f\let\defaultfcxpad=\defaultfcxpadding
\ctr@ld@f\def\defaultfcypadding{0}%\ctr@ld@f\let\defaultfcypad=\defaultfcypadding
\ctr@ld@f\def\Pssd@m@sh#1=#2|{\keln@mun#1|%
    \def\n@mref{d}\ifx\l@debut\n@mref\edef\defaultmeshdiag{#2}\else%
    \immediate\write16{*** Unknown attribute: \BS@ pssetdefault mesh(..., #1=...)}%
    \fi}
\ctr@ld@f\def\Pssd@s@cond#1=#2|{\keln@mun#1|%
    \def\n@mref{c}\ifx\l@debut\n@mref\edef\defaultsecondcolor{#2}\else%
    \def\n@mref{d}\ifx\l@debut\n@mref\edef\defaultseconddash{#2}\else%
    \def\n@mref{w}\ifx\l@debut\n@mref\edef\defaultsecondwidth{#2}\else%
    \immediate\write16{*** Unknown attribute: \BS@ pssetdefault second(..., #1=...)}%
    \fi\fi\fi}
\ctr@ld@f\def\Pssd@th@rd#1=#2|{\keln@mun#1|%
    \def\n@mref{c}\ifx\l@debut\n@mref\edef\defaultthirdcolor{#2}\else%
    \immediate\write16{*** Unknown attribute: \BS@ pssetdefault third(..., #1=...)}%
    \fi}
\ctr@ln@w{newif}\iffillm@de
\ctr@ld@f\def\pssetfillmode#1{\expandafter\setfillm@de#1:}
\ctr@ld@f\def\setfillm@de#1#2:{\if#1n\fillm@defalse\else\fillm@detrue\fi}
\ctr@ld@f\def\defaultfill{no}     % Valeur par defaut
\ctr@ln@w{newif}\ifpsupdatem@de
\ctr@ld@f\def\pssetupdate#1{\ifcurr@ntPS\immediate\write16{*** \BS@ pssetupdate is ignored inside a
     \BS@ psbeginfig-\BS@ psendfig block.}%
    \immediate\write16{*** It must be called before \BS@ psbeginfig.}%
    \else\expandafter\setupd@te#1:\fi}
\ctr@ld@f\def\setupd@te#1#2:{\if#1n\psupdatem@defalse\else\psupdatem@detrue\fi}
\ctr@ld@f\def\defaultupdate{no}     % Valeur par defaut
\ctr@ln@m\curr@ntcolor \ctr@ln@m\curr@ntcolorc@md
\ctr@ld@f\def\Pssetc@lor#1{\ifps@cri\result@tent=\@ne\expandafter\c@lnbV@l#1 :%
    \def\curr@ntcolor{}\def\curr@ntcolorc@md{}%
    \ifcase\result@tent\or\pssetgray{#1}\or\or\pssetrgb{#1}\or\pssetcmyk{#1}\fi\fi}
\ctr@ln@m\curr@ntcolorc@mdStroke
\ctr@ld@f\def\pssetcmyk#1{\ifps@cri\def\curr@ntcolor{#1}\def\curr@ntcolorc@md{\c@msetcmykcolor}%
    \def\curr@ntcolorc@mdStroke{\c@msetcmykcolorStroke}%
    \ifcurr@ntPS\PSc@mment{pssetcmyk Color=#1}\us@primarC@lor\fi\fi}
\ctr@ld@f\def\pssetrgb#1{\ifps@cri\def\curr@ntcolor{#1}\def\curr@ntcolorc@md{\c@msetrgbcolor}%
    \def\curr@ntcolorc@mdStroke{\c@msetrgbcolorStroke}%
    \ifcurr@ntPS\PSc@mment{pssetrgb Color=#1}\us@primarC@lor\fi\fi}
\ctr@ld@f\def\pssetgray#1{\ifps@cri\def\curr@ntcolor{#1}\def\curr@ntcolorc@md{\c@msetgray}%
    \def\curr@ntcolorc@mdStroke{\c@msetgrayStroke}%
    \ifcurr@ntPS\PSc@mment{pssetgray Gray level=#1}\us@primarC@lor\fi\fi}
\ctr@ln@m\fillc@md
\ctr@ld@f\def\us@primarC@lor{\immediate\write\fwf@g{\d@fprimarC@lor}%
    \let\fillc@md=\prfillc@md}
\ctr@ld@f\def\prfillc@md{\d@fprimarC@lor\space\c@mfill}
\ctr@ld@f\def\defaultcolor{0}       % Valeur par defaut
\ctr@ld@f\def\c@lnbV@l#1 #2:{\def\t@xt@{#1}\relax\ifx\t@xt@\empty\c@lnbV@l#2:% Discard leading spaces
    \else\c@lnbV@l@#1 #2:\fi}
\ctr@ld@f\def\c@lnbV@l@#1 #2:{\def\t@xt@{#2}\ifx\t@xt@\empty%
    \def\t@xt@{#1}\ifx\t@xt@\empty\advance\result@tent\m@ne\fi% Discard trailing spaces
    \else\advance\result@tent\@ne\c@lnbV@l@#2:\fi}
\ctr@ld@f\def\Blackcmyk{0 0 0 1}
\ctr@ld@f\def\Whitecmyk{0 0 0 0}
\ctr@ld@f\def\Cyancmyk{1 0 0 0}
\ctr@ld@f\def\Magentacmyk{0 1 0 0}
\ctr@ld@f\def\Yellowcmyk{0 0 1 0}
\ctr@ld@f\def\Redcmyk{0 1 1 0}
\ctr@ld@f\def\Greencmyk{1 0 1 0}
\ctr@ld@f\def\Bluecmyk{1 1 0 0}
\ctr@ld@f\def\Graycmyk{0 0 0 0.50}
\ctr@ld@f\def\BrickRedcmyk{0 0.89 0.94 0.28} % PANTONE 1805
\ctr@ld@f\def\Browncmyk{0 0.81 1 0.60} % PANTONE 1615
\ctr@ld@f\def\ForestGreencmyk{0.91 0 0.88 0.12} % PANTONE 349
\ctr@ld@f\def\Goldenrodcmyk{ 0 0.10 0.84 0} % PANTONE 109
\ctr@ld@f\def\Marooncmyk{0 0.87 0.68 0.32} % PANTONE 201
\ctr@ld@f\def\Orangecmyk{0 0.61 0.87 0} % PANTONE ORANGE-021
\ctr@ld@f\def\Purplecmyk{0.45 0.86 0 0} % PANTONE PURPLE
\ctr@ld@f\def\RoyalBluecmyk{1. 0.50 0 0} % No PANTONE match
\ctr@ld@f\def\Violetcmyk{0.79 0.88 0 0} % PANTONE VIOLET
\ctr@ld@f\def\Blackrgb{0 0 0}
\ctr@ld@f\def\Whitergb{1 1 1}
\ctr@ld@f\def\Redrgb{1 0 0}
\ctr@ld@f\def\Greenrgb{0 1 0}
\ctr@ld@f\def\Bluergb{0 0 1}
\ctr@ld@f\def\Cyanrgb{0 1 1}
\ctr@ld@f\def\Magentargb{1 0 1}
\ctr@ld@f\def\Yellowrgb{1 1 0}
\ctr@ld@f\def\Grayrgb{0.5 0.5 0.5}
\ctr@ld@f\def\Chocolatergb{0.824 0.412 0.118}
\ctr@ld@f\def\DarkGoldenrodrgb{0.722 0.525 0.043}
\ctr@ld@f\def\DarkOrangergb{1 0.549 0}
\ctr@ld@f\def\Firebrickrgb{0.698 0.133 0.133}
\ctr@ld@f\def\ForestGreenrgb{0.133 0.545 0.133}
\ctr@ld@f\def\Goldrgb{1 0.843 0}
\ctr@ld@f\def\HotPinkrgb{1 0.412 0.706}
\ctr@ld@f\def\Maroonrgb{0.690 0.188 0.376}
\ctr@ld@f\def\Pinkrgb{1 0.753 0.796}
\ctr@ld@f\def\RoyalBluergb{0.255 0.412 0.882}
\ctr@ld@f\def\Pssetf@rst#1=#2|{\keln@mun#1|%
    \def\n@mref{c}\ifx\l@debut\n@mref\Pssetc@lor{#2}\else% color
    \def\n@mref{d}\ifx\l@debut\n@mref\pssetdash{#2}\else% dash
    \def\n@mref{f}\ifx\l@debut\n@mref\pssetfillmode{#2}\else% fillmode
    \def\n@mref{j}\ifx\l@debut\n@mref\pssetjoin{#2}\else% line join
    \def\n@mref{u}\ifx\l@debut\n@mref\pssetupdate{#2}\else% update
    \def\n@mref{w}\ifx\l@debut\n@mref\pssetwidth{#2}\else% line width
    \immediate\write16{*** Unknown attribute: \BS@ psset (..., #1=...)}%
    \fi\fi\fi\fi\fi\fi}
\ctr@ln@m\curr@ntdash
\ctr@ld@f\def\s@uvdash#1{\edef#1{\curr@ntdash}}
\ctr@ld@f\def\defaultdash{1}        % Valeur par defaut (numero sans espace)
\ctr@ld@f\def\pssetdash#1{\ifps@cri\edef\curr@ntdash{#1}\ifcurr@ntPS\expandafter\Pssetd@sh#1 :\fi\fi}
\ctr@ld@f\def\Pssetd@shI#1{\PSc@mment{pssetdash Index=#1}\ifcase#1%
    \or\immediate\write\fwf@g{[] 0 \c@msetdash}%         Index=1
    \or\immediate\write\fwf@g{[6 2] 0 \c@msetdash}%      Index=2
    \or\immediate\write\fwf@g{[4 2] 0 \c@msetdash}%      Index=3
    \or\immediate\write\fwf@g{[2 2] 0 \c@msetdash}%      Index=4
    \or\immediate\write\fwf@g{[1 2] 0 \c@msetdash}%      Index=5
    \or\immediate\write\fwf@g{[2 4] 0 \c@msetdash}%      Index=6
    \or\immediate\write\fwf@g{[3 5] 0 \c@msetdash}%      Index=7
    \or\immediate\write\fwf@g{[3 3] 0 \c@msetdash}%      Index=8
    \or\immediate\write\fwf@g{[3 5 1 5] 0 \c@msetdash}%  Index=9
    \or\immediate\write\fwf@g{[6 4 2 4] 0 \c@msetdash}%  Index=10
    \fi}
\ctr@ld@f\def\Pssetd@sh#1 #2:{{\def\t@xt@{#1}\ifx\t@xt@\empty\Pssetd@sh#2:% Discard leading spaces
    \else\def\t@xt@{#2}\ifx\t@xt@\empty\Pssetd@shI{#1}\else\s@mme=\@ne\def\debutp@t{#1}%
    \an@lysd@sh#2:\ifodd\s@mme\edef\debutp@t{\debutp@t\space\finp@t}\def\finp@t{0}\fi%
    \PSc@mment{pssetdash Pattern=#1 #2}%
    \immediate\write\fwf@g{[\debutp@t] \finp@t\space\c@msetdash}\fi\fi}}
\ctr@ld@f\def\an@lysd@sh#1 #2:{\def\t@xt@{#2}\ifx\t@xt@\empty\def\finp@t{#1}\else%
    \edef\debutp@t{\debutp@t\space#1}\advance\s@mme\@ne\an@lysd@sh#2:\fi}
\ctr@ln@m\curr@ntwidth
\ctr@ld@f\def\s@uvwidth#1{\edef#1{\curr@ntwidth}}
\ctr@ld@f\def\defaultwidth{0.4}     % Valeur par defaut
\ctr@ld@f\def\pssetwidth#1{\ifps@cri\edef\curr@ntwidth{#1}\ifcurr@ntPS%
    \PSc@mment{pssetwidth Width=#1}\immediate\write\fwf@g{#1 \c@msetlinewidth}\fi\fi}
\ctr@ln@m\curr@ntjoin
\ctr@ld@f\def\pssetjoin#1{\ifps@cri\edef\curr@ntjoin{#1}\ifcurr@ntPS\expandafter\Pssetj@in#1:\fi\fi}
\ctr@ld@f\def\Pssetj@in#1#2:{\PSc@mment{pssetjoin join=#1}%
    \if#1r\def\t@xt@{1}\else\if#1b\def\t@xt@{2}\else\def\t@xt@{0}\fi\fi%
    \immediate\write\fwf@g{\t@xt@\space\c@msetlinejoin}}
\ctr@ld@f\def\defaultjoin{miter}   % Valeur par defaut
\ctr@ld@f\def\Pssets@cond#1=#2|{\keln@mun#1|%
    \def\n@mref{c}\ifx\l@debut\n@mref\Pssets@condcolor{#2}\else%
    \def\n@mref{d}\ifx\l@debut\n@mref\pssetseconddash{#2}\else%
    \def\n@mref{w}\ifx\l@debut\n@mref\pssetsecondwidth{#2}\else%
    \immediate\write16{*** Unknown attribute: \BS@ psset second(..., #1=...)}%
    \fi\fi\fi}
\ctr@ln@m\curr@ntseconddash
\ctr@ld@f\def\pssetseconddash#1{\edef\curr@ntseconddash{#1}}
\ctr@ld@f\def\defaultseconddash{4}  % Valeur par defaut (numero sans espace)
\ctr@ln@m\curr@ntsecondwidth
\ctr@ld@f\def\pssetsecondwidth#1{\edef\curr@ntsecondwidth{#1}}
\ctr@ld@f\edef\defaultsecondwidth{\defaultwidth} % Valeur par defaut
\ctr@ld@f\def\psresetsecondsettings{%
    \pssetseconddash{\defaultseconddash}\pssetsecondwidth{\defaultsecondwidth}%
    \Pssets@condcolor{\defaultsecondcolor}}
\ctr@ln@m\sec@ndcolor \ctr@ln@m\sec@ndcolorc@md
\ctr@ld@f\def\Pssets@condcolor#1{\ifps@cri\result@tent=\@ne\expandafter\c@lnbV@l#1 :%
    \def\sec@ndcolor{}\def\sec@ndcolorc@md{}%
    \ifcase\result@tent\or\pssetsecondgray{#1}\or\or\pssetsecondrgb{#1}%
    \or\pssetsecondcmyk{#1}\fi\fi}
\ctr@ln@m\sec@ndcolorc@mdStroke
\ctr@ld@f\def\pssetsecondcmyk#1{\def\sec@ndcolor{#1}\def\sec@ndcolorc@md{\c@msetcmykcolor}%
    \def\sec@ndcolorc@mdStroke{\c@msetcmykcolorStroke}}
\ctr@ld@f\def\pssetsecondrgb#1{\def\sec@ndcolor{#1}\def\sec@ndcolorc@md{\c@msetrgbcolor}%
    \def\sec@ndcolorc@mdStroke{\c@msetrgbcolorStroke}}
\ctr@ld@f\def\pssetsecondgray#1{\def\sec@ndcolor{#1}\def\sec@ndcolorc@md{\c@msetgray}%
    \def\sec@ndcolorc@mdStroke{\c@msetgrayStroke}}
\ctr@ld@f\def\us@secondC@lor{\immediate\write\fwf@g{\d@fsecondC@lor}%
    \let\fillc@md=\sdfillc@md}
\ctr@ld@f\def\sdfillc@md{\d@fsecondC@lor\space\c@mfill}
\ctr@ld@f\edef\defaultsecondcolor{\defaultcolor} % Valeur par defaut
\ctr@ld@f\def\Pss@tsecondSt{%
    \s@uvdash{\typ@dash}\pssetdash{\curr@ntseconddash}%
    \s@uvwidth{\typ@width}\pssetwidth{\curr@ntsecondwidth}\us@secondC@lor}
\ctr@ld@f\def\Psrest@reSt{\pssetwidth{\typ@width}\pssetdash{\typ@dash}\us@primarC@lor}
\ctr@ld@f\def\Pssetth@rd#1=#2|{\keln@mun#1|%
    \def\n@mref{c}\ifx\l@debut\n@mref\Pssetth@rdcolor{#2}\else%
    \immediate\write16{*** Unknown attribute: \BS@ psset third(..., #1=...)}%
    \fi}
\ctr@ln@m\th@rdcolor \ctr@ln@m\th@rdcolorc@md
\ctr@ld@f\def\Pssetth@rdcolor#1{\ifps@cri\result@tent=\@ne\expandafter\c@lnbV@l#1 :%
    \def\th@rdcolor{}\def\th@rdcolorc@md{}%
    \ifcase\result@tent\or\Pssetth@rdgray{#1}\or\or\Pssetth@rdrgb{#1}%
    \or\Pssetth@rdcmyk{#1}\fi\fi}
\ctr@ln@m\th@rdcolorc@mdStroke
\ctr@ld@f\def\Pssetth@rdcmyk#1{\def\th@rdcolor{#1}\def\th@rdcolorc@md{\c@msetcmykcolor}%
    \def\th@rdcolorc@mdStroke{\c@msetcmykcolorStroke}}
\ctr@ld@f\def\Pssetth@rdrgb#1{\def\th@rdcolor{#1}\def\th@rdcolorc@md{\c@msetrgbcolor}%
    \def\th@rdcolorc@mdStroke{\c@msetrgbcolorStroke}}
\ctr@ld@f\def\Pssetth@rdgray#1{\def\th@rdcolor{#1}\def\th@rdcolorc@md{\c@msetgray}%
    \def\th@rdcolorc@mdStroke{\c@msetgrayStroke}}
\ctr@ld@f\def\us@thirdC@lor{\immediate\write\fwf@g{\d@fthirdC@lor}%
    \let\fillc@md=\thfillc@md}
\ctr@ld@f\def\thfillc@md{\d@fthirdC@lor\space\c@mfill}
\ctr@ld@f\def\defaultthirdcolor{1}  % Valeur par defaut
\ctr@ld@f\def\pstrimesh#1[#2,#3,#4]{{\ifcurr@ntPS\ifps@cri%
    \PSc@mment{pstrimesh Type=#1, Triangle=[#2,#3,#4]}%
    \s@uvc@ntr@l\et@tpstrimesh\ifnum#1>\@ne\Pss@tsecondSt\setc@ntr@l{2}%
    \Pstrimeshp@rt#1[#2,#3,#4]\Pstrimeshp@rt#1[#3,#4,#2]%
    \Pstrimeshp@rt#1[#4,#2,#3]\Psrest@reSt\fi\psline[#2,#3,#4,#2]%
    \PSc@mment{End pstrimesh}\resetc@ntr@l\et@tpstrimesh\fi\fi}}
\ctr@ld@f\def\Pstrimeshp@rt#1[#2,#3,#4]{{\l@mbd@un=\@ne\l@mbd@de=#1\loop\ifnum\l@mbd@de>\@ne%
    \advance\l@mbd@de\m@ne\figptbary-1:[#2,#3;\l@mbd@de,\l@mbd@un]%
    \figptbary-2:[#2,#4;\l@mbd@de,\l@mbd@un]\psline[-1,-2]%
    \advance\l@mbd@un\@ne\repeat}}
\initpr@lim\initpss@ttings\initPDF@rDVI% Initialisation preliminaire
\ctr@ln@w{newbox}\figBoxA
\ctr@ln@w{newbox}\figBoxB
\ctr@ln@w{newbox}\figBoxC
\catcode`\@=12

\title[A derivation of the linear Boltzmann equation]{A geometric derivation of the linear Boltzmann equation for a particle
interacting with a Gaussian random field}

\author{Sébastien Breteaux}

\selectlanguage{french}%

\address{IRMAR, UMR-CNRS 6625, Université de Rennes 1, campus de Beaulieu,
35042 Rennes Cedex, France. ENS de Cachan, Antenne de Bretagne, Campus
de Ker Lann, Av. R. Schuman, 35170 Bruz, France.}

\selectlanguage{english}%

\email{sebastien.breteaux@ens-cachan.org}
\begin{abstract}
In this article the linear Boltzmann equation is derived for a particle
interacting with a Gaussian random field, in the weak coupling limit,
with renewal in time of the random field. The initial data can be
chosen arbitrarily. The proof is geometric and involves coherent states
and semi-classical calculus.
\end{abstract}



\maketitle

\section{Introduction}

In this article we derive the linear Boltzmann equation for a particle
interacting with a translation invariant centered Gaussian random
field. The evolution of this particle is described by the Liouville
- Von Neumann equation with a Hamiltonian~$-\Delta_{x}+\mathcal{V}_{\omega}^{h}(x)$,
where the potential depends on a random parameter~$\omega$. In the
weak coupling limit, the dependence of the random potential with respect
to~$h$ is~$\mathcal{V}_{\omega}^{h}=\sqrt{h}\mathcal{V}_{\omega}$,
where $h$ represents the ratio between the microscopic and macroscopic
scales. We consider the limit $h\to0$. In the case of a Gaussian
random field the weak coupling limit and the low density limit agree.
Through an isomorphism between the Gaussian space~$L^{2}(\Omega_{\mathbb{P}},\mathbb{P};\mathbb{C})$
associated with~$L^{2}(\mathbb{R}^{d};\mathbb{R})$ and the symmetric
Fock space~$\Gamma L^{2}(\mathbb{R}^{d})$ associated with~$L^{2}(\mathbb{R}^{d};\mathbb{C})$,
multiplication by~$\mathcal{V}_{\omega}(x)$ corresponds to the field
operator~$\sqrt{2}\Phi(V(x-\cdot))$ for some function~$V$. We
can thus express the Hamiltonian in the Fock space and approximate
the dynamics by an explicitly solvable one whose solutions are coherent
states. The geometric idea behind the computations is due to the fact
that the initial state is the vacuum, and we can thus expect that
for short times the system is approximately in a coherent state whose
parameter moves slightly in the phase space. This parameter in the
(infinite dimensional) phase space then gives the important information
in the limit~$h\to0$. The computations done with this solution allow
us to recover the dual linear Boltzmann equation for short times for
the observables. A renewal of the random field allows us to reach
long times.

The derivation of the linear Boltzmann equation has been studied for
both classical and quantum microscopic models. In the classical case
Gallavotti~\cite{PhysRev.185.308} provided a derivation of the linear
Boltzmann equation for Green functions in the case of a Lorentz gas.
Later Spohn~\cite{RevModPhys.52.569} presented a review of different
classical microscopic models and of kinetic equations obtained as
limits of these models, with emphasis on the approximate Markovian
behaviour of the microscopic dynamics (some quantum models were also
studied). Boldrighini, Bunimovich and Sina\u\i~\cite{MR725107}
gave a derivation of the linear Boltzmann equation for the density
of particles in the case of the Lorentz model. In the quantum case,
Spohn derived in~\cite{MR0471824} the radiative transport equation
in the spatially homogenous case. Later Ho, Landau and Wilkins studied
in~\cite{MR1223523} the weak coupling limit of a Fermi gas in a
translation invariant Gaussian potential (and other random potentials).
Their proofs made use of combinatorics and graph techniques. In the
case of a particle interacting with a Gaussian random field (the setting
of this article) Erd\H{o}s and Yau~\cite{MR1744001} removed the
small time restriction, and also generalized the initial data to WKB
states, using methods with graph expansions. Developements of that
method by Chen~\cite{MR2165532} and Erd\H{o}s, Salmhofer and Yau~\cite{MR2413135,MR2283953}
did not require a Gaussian form for the random field but still supposed
an initial state of the WKB form. The linear Boltzmann equation was
derived in the radiative transport limit by Bal, Papanicolaou and
Ryzhik~\cite{MR1888863} in the quantum case, and by Poupaud and
Vasseur~\cite{MR1996779} in the classical case using a potential
stochastic in time. This assumption automatically ensures that there
is no self-correlation in the paths of the particles and simplifies
the problem. Later Bechouche, Poupaud and Soler~\cite{MR2205910}
used similar techniques to get a model for collisions at the quantum
level and obtain a kind of quantum linear Boltzmann equation. For
these stochastic methods the initial state can be arbitrary but the
potential is almost surely bounded, which excludes Gaussian or Poissonian
random fields.

\subsection*{Remarks}

Our derivation is given in the case of a Gaussian random field but
other random fields could be considered with the same type of methods,
for example a Poissonian random field. Note that the weak coupling
and low density limit do not then agree.

Our approach allows initial states to be arbitrary, contrary to WKB
initial states.

The framework of quantum field theory allows to see how geometry
in phase space is involved. We use the viewpoint of Ammari and Nier~\cite{MR2465733}
but in a case that is not in the framework chosen by the authors.
Indeed we are not dealing with a mean field limit and the introduction
of a parameter~$\varepsilon$ is an artifact that allows us to keep
track of the importance of the different terms.\foreignlanguage{french}{}
We thus adopt a different viewpoint from the graph expansions or the
stochastic viewpoint adopted in other works on the subject, and this
allows us to keep track of the geometry.

However, we cannot as of yet reach times of order 1 like in~\cite{MR1744001,MR2165532,MR2283953,MR2413135}.
As we do not get the approximate Markovian behaviour in a satisfying
way, we need to introduce a renewal of the random potential. Attal
and Pautrat in~\cite{MR2205464} and Attal and Joye in~\cite{MR2312948}
deal in a more sophisticated way with interactions defined piecewise
in time. Other Ansätze may give a better approximation of the solution
to the initial problem and give the Markovian behaviour of the evolution.

One of the important tools in our derivation of the linear Boltzmann
equation is the use of \emph{a priori} estimates to show that we do
not lose too much mass in the measures during our approximations.
The mass conservation and positivity properties of the linear Boltzmann
equation then allow us to complete the proof.

Our result holds in dimension~$d\geq3$ as dispersion inequalities
for the free Schrödinger group provide the time integrability needed
for some expressions. It may be possible to reach the limit case of
dimension $d=2$.

\subsection*{Outline of the article}

In Section~\ref{sec:Model-and-result}, we describe the quantum model,
state the main result and give the structure of the proof. We then
recall some facts about the linear Boltzmann equation in Section~\ref{par:Boltzmann}.
We specify the link between the Gaussian random field and the symmetric
Fock space in Section~\ref{par:From-stochastics-to-Fock} and thus
obtain a new expression for the dynamics. We study an approximate
dynamics in Section~\ref{par:An-approximated-equation}. We use this
explicit solution to compute the measurement of an observable for
short times in Section~\ref{par:Calculus-approximated}. We control
the error involved in this approximation in Section~\ref{par:Comparison-original-approximated}.
And finally, we combine these results to complete the proof in Section~\ref{par:Gluing-estimates}.

\selectlanguage{french}%

\selectlanguage{english}%

\section{Model and result\label{sec:Model-and-result}}

\subsection{The model}

Let $\omega\in\Omega_{\mathbb{P}}$ be a random parameter and $x\in\mathbb{R}^{d}$
($d\geq1$) a space parameter. Let \index{random field2@$\mathcal{V}_{\omega}^{h}\left(x\right)$ random field}$\mathcal{V}_{\omega}^{h}\left(x\right)$
the translation invariant centered Gaussian random field with mean
zero and covariance~$hG\left(x-x^{\prime}\right)$, such that~$\hat{G}=|\hat{V}|^{2}$
with~$\hat{V}\in\mathcal{S}\left(\mathbb{R}^{d};\mathbb{R}\right)$.
We consider the Liouville - Von Neumann equation \begin{equation}
ih\partial_{t}\rho_{t,\omega}=[H_{\omega}^{h},\rho_{t,\omega}]\,,\qquad H_{\omega}^{h}=-\Delta_{x}+\mathcal{V}_{\omega}^{h}(x)\,,\label{eq:Schrodinger-stoch-initiale}\end{equation}
with an initial condition~$\rho_{0,\omega}^{h}=\rho_{0}^{h}$ in
the set of states on $L_{x}^{2}=L^{2}(\mathbb{R}_{x}^{d};\mathbb{C})$
(\emph{i.e.} the subset of the non-negative trace class operators
$\mathcal{L}_{1}^{+}(L_{x}^{2})$ whose trace is $1$). Note that
$[A,B]$ denotes the commutator~$AB-BA$ of two operators.

We now introduce the renewal of the random field. We fix a time~$T$,
an integer~$N$ and set~$\Delta t=T/N$. For a state~$\rho$ on
$L_{x}^{2}$, let\index{state1@$\rho_{t}^{h},\,\rho_{N,\Delta t}^{h}\left(\bar{\omega}_{N}\right),\,\rho_{N,\Delta t}^{h}$}\begin{align}
\mathcal{G}_{t}^{h}(\rho) & =\int e^{-i\frac{t}{h}H_{h,\omega}}\,\rho\, e^{i\frac{t}{h}H_{h,\omega}}\diff\mathbb{P}(\omega)\,,\label{eq:rho-h-t}\\
\rho_{t}^{h} & =\mathcal{G}_{t}^{h}(\rho)\,,\label{eq:rho-h-N-Deltat-omega}\\
\rho_{N,\Delta t}^{h} & =(\mathcal{G}_{\Delta t}^{h})^{N}(\rho)\,.\label{eq:rho-h-N-Deltat}\end{align}
With~$t_{k}=k\Delta t$, the dynamics is defined piecewise on the
intervals~$\left[t_{k-1},t_{k}\right]$ by the Hamiltonians $H_{h,\omega_{k}}=-\Delta_{x}+\mathcal{V}_{h,\omega_{k}}(x)$
with independent random fields~$\mathcal{V}_{h,\omega_{k}}$, $\omega_{k}$
in copies of~$\Omega_{\mathbb{P}}$. Thus we get, for an initial
data~$\rho_{0}\in\mathcal{L}_{+}^{1}(L_{x}^{2})$, that the system
is in the state $\rho_{N,\Delta t}^{h}$ at time~$T$.

\subsection{The main result}

Let~$b\in\mathcal{C}_{0}^{\infty}(\mathbb{R}_{x,\xi}^{2d})$. The
measure of the observable~$b^{W}\!(hx,D_{x})$ in a state~$\rho$
on~$L_{x}^{2}$ is given by\[
m_{h}(b,\rho)=\Tr\!\big[b^{W}\!(hx,D_{x})\rho\big]\,,\]
where the \index{Weyl quantization@$b^{W}(hx,D_{x})$, Weyl quantization}\emph{Weyl
quantization} (see for example Martinez's book~\cite{MR1872698})
is defined by\[
b^{W}\negthickspace(hx,D_{x})u\,(x)=(2\pi)^{-d}\int_{\mathbb{R}_{x',\xi}^{2d}}e^{i(x-x^{\prime}).\xi}b\big(h\tfrac{x+x'}{2},\xi\big)\, u(x^{\prime})\diff x^{\prime}\diff\xi\,.\]
Semiclassical measures (and microlocal defect measures) have been
studied by, among others, Gérard~\cite{MR1131589,MR1135919}, Burq~\cite{MR1627111},
Gérard, Markowich, Mauser and Poupaud \cite{MR1438151,MR1721376}
and Lions and Paul~\cite{MR1251718}. Let us quote Theorem~\ref{thm:mesure-semi-classique},
which is a direct consequence of a theorem which can be found in~\cite{MR1627111}
(with $(\rho^{h})$ replacing $(|u_{k}\rangle\langle u_{k}|)$ for
weakly convergent sequence~$(u_{k})$ of $L_{x}^{2}$).
\begin{thm}
\label{thm:mesure-semi-classique}Let~$(\rho^{h})_{h\in(0,h_{0}]}$,
$h_{0}>0$ be a family of states on~$L_{x}^{2}$. There exist a sequence~$h_{k}\to0$
and a non-negative measure~$\mu$ on~$\mathbb{R}_{x,\xi}^{2d}$
such that\[
\forall b\in\mathcal{C}_{0}^{\infty}(\mathbb{R}_{x,\xi}^{2d})\,,\quad\lim_{n\to+\infty}m_{h_{k}}(b,\rho^{h_{k}})=\int_{\mathbb{R}_{x,\xi}^{2d}}b\,\diff\mu\,.\]
The measure~$\mu$ is called a \emph{semiclassical measure} (or Wigner
measure) associated with the sequence~$(\rho^{h_{k}})$. Let~$\mathcal{M}(\rho^{h},h\in(0,h_{0}])$
be the set of such measures. If this set is a singleton~$\{\mu\}$
then the family~$(\rho^{h})$ is said to be \emph{pure} and associated
with~$\mu$.
\end{thm}
By a simple sequence extraction out of the range of the parameter
$h$, the family can always be assumed to be pure. For evolution problems
the fact that the sequence extraction can be performed uniformly for
all times is a property to be proved.

We can now state the main theorem of this article.
\begin{thm}
\label{thm:Main_theorem}Assume~$d\geq3$. Let~$\Delta t=h^{\alpha}$,
$N=N^{h}=T/h^{\alpha}$, $\alpha\in(\frac{3}{4},1)$.

Assume that~$(\rho^{h})_{h\in(0,h_{0}]}$ is pure and associated
with~$\mu_{0}$ such that~$\mu_{0}(\mathbb{R}_{x}^{d}\times\mathbb{R}_{\xi}^{d*})=1$. 

Then~$(\rho_{N,\Delta t}^{h})_{h\in(0,h_{0}]}$ is pure and associated
with~$\mu_{T}$, where $(\mu_{t})_{t}$ solves the linear Boltzmann
equation\begin{equation}
\partial_{t}\mu_{t}(x,\xi)+2\xi.\partial_{x}\mu_{t}(x,\xi)=\negthickspace\int\negthickspace\sigma(\xi,\xi^{\prime})\,\delta\bigl(|\xi|^{2}-|\xi^{\prime}|^{2}\bigr)(\mu_{t}(x,\xi^{\prime})-\mu_{t}(x,\xi))\diff\xi^{\prime}\label{eq:linear-boltzmann-equation}\end{equation}
with the initial condition $\mu_{t=0}=\mu_{0}$ and~$\sigma(\xi,\xi^{\prime})=2\pi|\hat{V}(\xi-\xi^{\prime})|^{2}$.
\end{thm}
The Fourier transform on $\mathbb{R}^{d}$ is here $\hat{u}(\xi)=\mathcal{F}u(\xi)=\int_{\mathbb{R}_{x}^{d}}e^{-ix.\xi}u(x)\diff x$.

\begin{proof}[Sketch of the Proof.]
Let~$\mu_{T}$ in~$\mathcal{M}(\rho_{N,\Delta t}^{h},h\in(0,h_{0}])$.
We denote by~$\mathcal{B}(t)$ (resp. $\mathcal{B}^{T}\!(t)$) the
flow associated with the (resp. dual) Boltzmann equation~(\ref{eq:linear-boltzmann-equation}),
see Section~\ref{par:Boltzmann}. For any non-negative~$b$ in~$\mathcal{C}_{0}^{\infty}(\mathbb{R}_{x}^{d}\times\mathbb{R}_{\xi}^{d*})$
we shall prove
\begin{enumerate}
\item $\int_{\mathbb{R}_{x}^{d}\times\mathbb{R}_{\xi}^{d*}}\! b\diff\mu_{T}\geq\liminf_{h\to0}\Tr[\rho_{N,\Delta t}^{h}\, b^{W}\!(hx,D_{x})]$
by the definition of~$\mu_{T}$,
\item \label{enu:Technical_part_general_proof}$\liminf_{h\to0}\Tr[\rho_{N,\Delta t}^{h}\, b^{W}\!(hx,D_{x})]\geq\int_{\mathbb{R}_{x}^{d}\times\mathbb{R}_{\xi}^{d*}}(\mathcal{B}^{T}(T)b)\diff\mu_{0}$
(see Remark~\ref{rem:technical-step}),
\item $\int_{\mathbb{R}_{x}^{d}\times\mathbb{R}_{\xi}^{d*}}(\mathcal{B}^{T}(T)b)\diff\mu_{0}=\int_{\mathbb{R}_{x}^{d}\times\mathbb{R}_{\xi}^{d*}}b\,\diff\,(\mathcal{B}(T)\mu_{0})$
by the definition of~$\mathcal{B}(T)$.
\end{enumerate}
From these statements, the lower bound\[
\int_{\mathbb{R}_{x}^{d}\times\mathbb{R}_{\xi}^{d*}}b\,\diff\mu_{T}\geq\int_{\mathbb{R}_{x}^{d}\times\mathbb{R}_{\xi}^{d*}}b\,\diff\left(\mathcal{B}(T)\mu_{0}\right)\]
follows. Since this inequality holds for any non-negative~$b$ from
the set of smooth functions with compact support~$\mathcal{C}_{0}^{\infty}(\mathbb{R}_{x}^{d}\times\mathbb{R}_{\xi}^{d*})$,
which is dense in the set of continuous functions vanishing at {}``infinity''~$\mathcal{C}_{\infty}^{0}(\mathbb{R}_{x}^{d}\times\mathbb{R}_{\xi}^{d*})$,
whose dual is the set of Radon measures~$\mathcal{M}_{b}(\mathbb{R}_{x}^{d}\times\mathbb{R}_{\xi}^{d*})$,
we get\[
\left.\mu_{T}\right|_{\mathbb{R}_{x}^{d}\times\mathbb{R}_{\xi}^{d*}}\geq\left.\mathcal{B}(T)\mu_{0}\right|_{\mathbb{R}_{x}^{d}\times\mathbb{R}_{\xi}^{d*}}\,.\]
But we also have~$\mathcal{B}(T)\mu_{0}(\mathbb{R}_{x}^{d}\times\mathbb{R}_{\xi}^{d*})=1$
from the mass conservation property of the linear Boltzmann equation
and~$\mu_{T}(\mathbb{R}_{x}^{d}\times\mathbb{R}_{\xi}^{d})\leq1$
from the properties of semiclassical measures. So, necessarily,\[
\mu_{T}\big(\mathbb{R}_{x}^{d}\times\mathbb{R}_{\xi}^{d*}\big)=1\,,\qquad\mu_{T}\big(\mathbb{R}_{x}^{d}\times\left\{ 0\right\} _{\xi}\big)=0\]
and~$\mu_{T}=\mathcal{B}(T)\mu_{0}$. Hence the result.\end{proof}
\begin{rem}
\label{rem:technical-step}Step~\ref{enu:Technical_part_general_proof}
is the technical part and requires various estimates developed in
this article.
\end{rem}

\begin{rem}
Let us justify the scaling in the Weyl quantization. Physically the
parameter~$h$ is the quotient of the microscopic scale over the
macroscopic scale, either in time or in position. Thus if we consider
an observable~$b(X,\Xi)$ varying on a macroscopic scale, the corresponding
observable on the microscopic scale will be~$b(hx,\xi)$.

The scaling of the random field according to the covariance~$hG(x-x^{\prime})$
is done on a mesoscopic scale imposed by the kinetic regime. In microscopic
variables, consider a particle moving among obstacles with a velocity~$v\propto1$
and a distance of interaction~$R\propto1$. During a time~$T$ the
particle sweeps a volume of order~$vTR^{d-1}$. In the kinetic regime
it is assumed that during a long microscopic time~$T=t/h$ with~$t\propto1$
the macroscopic time, the average particle encounters a number~$\propto1$
of obstacles. We denote by~$\rho$ the density of obstacles and thus
obtain~$\rho=1/vTR^{d-1}\propto h$. To get this density of obstacles
we need the distance between two nearest obstacles to be of order~$h^{-1/d}$.

Thus we consider a Schrödinger equation of the form\[
i\partial_{T}\psi=-\Delta_{x}\psi+\mathcal{V}_{\omega}^{h}(x)\,\psi\,,\]
that is,\[
ih\partial_{t}\psi=-\Delta_{x}\psi+\mathcal{V}_{\omega}^{h}(x)\,\psi\,.\]
A translation invariant Gaussian random field of covariance~$G(x-x^{\prime})$,
$\hat{G}=|\hat{V}|^{2}$, is of the form~$V*W_{\omega}$, where~$W_{\omega}$
is the spatial white noise and~$V$ describes the interaction potential.
In the kinetic regime the obstacles are spread at the mesoscopic scale~$h^{1/d}$.
Only the white noise~$W_{\omega}^{h}$ is rescaled (and not~$V$)
according to \[
\forall\varphi\in\mathcal{S}(\mathbb{R}^{d};\mathbb{R})\,,\quad\int\varphi(h^{1/d}x)\, W_{\omega}^{h}(x)\diff x=\int\varphi(x)\, W_{\omega}(x)\diff x\,,\]
\emph{i.e.,~}$W_{\omega}^{h}(x)=hW_{\omega}(h^{1/d}x)$. Thus we
get~$\mathcal{V}_{\omega}^{h}=hV*W_{\omega}(h^{1/d}\cdot)$ and~$G^{h}=hG$.

To prove Theorem~\ref{thm:Main_theorem} we first consider the case
without the renewal of the stochastics, \emph{i.e.},~$N=1$ for short
times in Sections~\ref{par:An-approximated-equation}, \ref{par:Calculus-approximated},
\ref{par:Comparison-original-approximated} and then glue together
the estimates obtained this way~$N$ times for~$N$ {}``big''
in Section~\ref{par:Gluing-estimates}. To simplify the problem of
finding estimates for short times we approximate the equation by a
simpler one which is solved and studied in Section~\ref{par:An-approximated-equation}.
In Section~\ref{par:Calculus-approximated}, using the solution to
the approximated equation, we carry out explicit computations which
give rise to the different terms of the dual linear Boltzmann equation.
Then we control the error between the solutions of the approximated
equation and the exact equation in Section~\ref{par:Comparison-original-approximated}.
All these computations are done within the framework of quantum field
theory. This allows us
\begin{itemize}
\item to use conveniently the geometric content of coherent states,
\item to keep track of the different orders of importance of the different
terms by using the Wick quantization with a parameter~$\varepsilon$.
\end{itemize}
We expose the correspondence between the stochastic and Fock space
viewpoints in Section~\ref{par:From-stochastics-to-Fock}.
\end{rem}

\begin{rem}
Our initial data~$(\rho^{h})_{h\in(0,h_{0}]}$ are assumed to belong
to~$\mathcal{L}_{1}^{+}L_{x}^{2}$ with $\Tr\rho^{h}=1$. We thus
make estimates for states~$\rho$ in~$\mathcal{L}_{1}^{+}L_{x}^{2}$,
with~$\Tr\rho=1$ with constants independent of~$\rho$.
\end{rem}

\section{\label{par:Boltzmann}The linear Boltzmann equation}

Information on the linear Boltzmann equation can be found in the books
of Dautray and Lions~\cite{MR792484,MR902802} or Reed and Simon~\cite{MR529429}. 

In Section~\ref{par:Boltzmann} the set of values of functions is~$\mathbb{R}$
when nothing is precised. We assume that~$\sigma\in\mathcal{C}^{\infty}(\mathbb{R}_{\xi}^{d}\times\mathbb{R}_{\xi^{\prime}}^{d})$
and~$\sigma\geq0$.

\subsection{Formal definition}

Since all the objects we use are diagonal in~$\left|\xi\right|$,
the following notations are convenient.

\noindent \textbf{Notation:} Let~$0<r<r^{\prime}<+\infty$, we define
the Sobolev spaces\[
H^{n}[r,r^{\prime}]=H^{n}(\mathbb{R}_{x}^{d}\times A_{\xi}[r,r^{\prime}])\]
where~$A_{\xi}[r,r^{\prime}]$ is the annulus~$\{\xi\in\mathbb{R}^{d},\,|\xi|\in(r,r^{\prime})\}$
in the variable~$\xi$. When there is no ambiguity we write~$A_{\xi}$
for~$A_{\xi}[r,r^{\prime}]$. We also write~$L^{2}[r,r^{\prime}]$
for~$H^{0}[r,r^{\prime}]$.
\begin{defn}
The \emph{linear Boltzmann equation} is formally the equation, with
initial condition~$\mu_{t=0}=\mu_{0}$,\[
\partial_{t}\mu=\{\mu,\left|\xi\right|^{2}\}+Q\mu\,,\]
where the \emph{collision operator}\index{collision operator1@$Q,\, Q_{-},\, Q_{+}$,~collision operator}~$Q$
is defined for~$b\in L^{2}[r,r^{\prime}]$ by\begin{equation}
Qb=Q_{+}b-Q_{-}b\,,\label{eq:noyau-de-collision}\end{equation}
with\index{sigma@$\sigma\left(\xi,\xi^{\prime}\right)$}\begin{align*}
Q_{+}b(x,\xi) & =\int_{\mathbb{R}_{\xi'}^{d}}b(x,\xi^{\prime})\,\sigma(\xi,\xi^{\prime})\,\delta\big(\left|\xi\right|^{2}-\left|\xi^{\prime}\right|^{2}\big)\diff\xi^{\prime}\,,\\
Q_{-}b(x,\xi) & =b(x,\xi)\,\int_{\mathbb{R}_{\xi'}^{d}}\sigma(\xi,\xi^{\prime})\,\delta\big(\left|\xi\right|^{2}-\left|\xi^{\prime}\right|^{2}\big)\diff\xi^{\prime}\,.\end{align*}
The dual linear Boltzmann equation with initial condition~$b_{t=0}=b_{0}$
is \[
\partial_{t}b=-\{b,\left|\xi\right|^{2}\}+Qb=2\xi.\partial_{x}b+Qb\,.\]
\end{defn}
\begin{rem}
For a given~$\xi$ the integrals in the collision operator only involve
the values of $\sigma(\xi,\left|\xi\right|\omega)$ and~$b(x,\left|\xi\right|\omega)$
for~$\omega\in\mathbb{S}^{d-1}$.
\end{rem}
We show in Section~\ref{sub:Properties_lin_boltz} that the dual
linear Boltzmann equation is solved by a group~$(\mathcal{B}^{T}(t))_{t\in\mathbb{R}}$
of operators on~$\mathcal{C}_{\infty}^{0}(\mathbb{R}_{x}^{d}\times\mathbb{R}_{\xi}^{d*})$
and in Section~\ref{sec:rigorous-linear-Boltzmann} that it defines
by duality a group~$(\mathcal{B}(t))_{t\in\mathbb{R}}$ of operators
on~$\mathcal{M}_{b}(\mathbb{R}_{x}^{d}\times\mathbb{R}_{\xi}^{d*})$.

\subsection{Properties\label{sub:Properties_lin_boltz}}

We recall here the main properties of the dual linear Boltzmann equation.
(The arguments are the same as for the linear Boltzmann equation.)

 We begin by solving the dual linear Boltzmann equation in~$L^{2}[r,r^{\prime}]$
in the sense of semigroups.
\begin{prop}
\label{pro:Boltzmann-estimation-sobolev}Let~$0<r<r^{\prime}<+\infty$.

The operator
\begin{itemize}
\item $2\xi.\partial_{x}$ generates a strongly continuous contraction semigroup
on~$L^{2}[r,r^{\prime}]$.
\item $Q$ is well defined and bounded on~$H^{n}[r,r^{\prime}]$, with\[
\left\Vert Q\right\Vert _{\mathcal{L}\left(H^{n}\left[r,r^{\prime}\right]\right)}\leq C_{d}\sup_{\left|\alpha\right|\leq n}\left\Vert \partial^{\alpha}\sigma\right\Vert _{\infty,A_{\xi}^{2}\left[r,r^{\prime}\right]}\,.\]
The group of space-translations~$(e^{2t\xi.\partial_{x}})_{t}$ preserves~$H^{n}[r,r^{\prime}]$.
\item $2\xi.\partial_{x}+Q$ generates a semigroup~$(\mathcal{B}^{T}(t))_{t\geq0}$
bounded by~$\exp(t\left\Vert Q\right\Vert _{\mathcal{L}\left(L^{2}\left[r,r^{\prime}\right]\right)})$
since~$Q$ is bounded on~$L^{2}[r,r^{\prime}]$.
\end{itemize}
The strongly continuous group $(\mathcal{B}^{T}(t))_{t\geq0}$ preserves
\begin{enumerate}
\item \label{enu:Boltzman-Sobolev}the Sobolev spaces~$H^{n}[r,r^{\prime}]$,
for~$n\in\mathbb{N}$,
\item \label{enu:Boltzman-compact}the set of functions with compact support,
\item \label{enu:Boltzmann-D}the set of infinitely differentiable functions
with compact support in~$\mathbb{R}_{x}^{d}\times A_{\xi}[r,r^{\prime}]$,
$\mathcal{C}_{0}^{\infty}(\mathbb{R}_{x}^{d}\times A_{\xi}[r,r^{\prime}])$,
\item \label{enu:Boltzmann-positive}the set of non-negative functions,
for~$t\geq0$.
\end{enumerate}
\end{prop}
\begin{proof}
The properties of generation of groups are clear.

Point~(\ref{enu:Boltzman-Sobolev}) is a consequence of Proposition~\ref{pro:Boltzmann-estimation-sobolev}.

Point~(\ref{enu:Boltzman-compact}) follows from the Trotter approximation\[
\mathcal{B}^{T}(t)=\lim_{n\to\infty}\big(e^{2\frac{t}{n}\xi.\partial_{x}}e^{\frac{t}{n}Q}\big)^{n}\,,\]
the fact that~$Q$ is {}``local'' in $\left(x,\left|\xi\right|\right)$,
and that the speed of propagation of the space-translations is finite
when~$\xi\in A_{\xi}[r,r^{\prime}]$.

Point~(\ref{enu:Boltzmann-D}) follows from~(\ref{enu:Boltzman-Sobolev}),
(\ref{enu:Boltzman-compact}) and\[
\mathcal{C}_{0}^{\infty}(\mathbb{R}_{x}^{d}\times A_{\xi}[r,r^{\prime}])=\bigcap_{n=0}^{\infty}H^{n}\left[r,r^{\prime}\right]\bigcap\left\{ f,\,\Supp f\,\mbox{compact}\right\} \,.\]

Point~(\ref{enu:Boltzmann-positive}) follows from both the Trotter
approximation\[
\mathcal{B}^{T}(t)=\lim_{n\to\infty}\big(e^{2\frac{t}{n}\xi.\partial_{x}}e^{\frac{t}{n}Q_{+}}e^{-\frac{t}{n}Q_{-}}\big)^{n}\]
and the fact that~$e^{2\frac{t}{n}\xi.\partial_{x}}$ preserves the
non-negative functions as a translation, $e^{\frac{t}{n}Q_{+}}$ preserves
the non-negative functions for~$t\geq0$ because~$Q_{+}$ does,
$e^{-\frac{t}{n}Q_{-}}$ preserves the non-negative functions as a
multiplication operator by a positive function. 
\end{proof}
Since~$\mathcal{C}_{0}^{\infty}(\mathbb{R}_{x}^{d}\times A_{\xi})\subset D(2\xi.\partial_{x})$
we can give the following result.
\begin{prop}
For all~$b_{0}\in\mathcal{C}_{0}^{\infty}(\mathbb{R}_{x}^{d}\times A_{\xi})$,
$b_{t}=\mathcal{B}^{T}(t)b_{0}$ is the unique solution in~$\mathcal{C}^{1}(\mathbb{R}^{+};L^{2}[r,r^{\prime}])\cap\mathcal{C}^{0}(\mathbb{R}^{+};D(2\xi.\partial_{x}))$
to the Dual linear Boltzmann equation such that~$b_{t=0}=b_{0}$.
Moreover~$\forall t\in\mathbb{R},\, b_{t}\in\mathcal{C}_{0}^{\infty}(\mathbb{R}_{x}^{d}\times A_{\xi})$.
If~$b_{0}$ is non-negative, then $\forall t\geq0$, $b_{t}$ is
non-negative.
\end{prop}

\subsection{\label{sec:rigorous-linear-Boltzmann}The linear Boltzmann equation}

The continuous functions vanishing at infinity and the Radon measures
on a locally compact, Hausdorff space~$X$ are denoted by \index{continuous vanishing@$\mathcal{C}_{\infty}^{0}\left(X;\mathbb{R}\right)$,
continuous functions vanishing at infinity}\index{measure@$\mathcal{M}_{b}\left(X;\mathbb{R}\right)$, Radon measures}\begin{align*}
\mathcal{C}_{\infty}^{0}(X) & =\{f\in\mathcal{C}^{0}(X)\,,\,\forall\varepsilon>0\,,\,\exists K\,\mbox{compact s.t.}\,\forall x\notin K\,,\,\left|f\left(x\right)\right|<\varepsilon\}\,,\\
\mathcal{M}_{b}(X) & =(\mathcal{C}_{\infty}^{0}(X))^{\prime}\,.\end{align*}

\begin{prop}
The semigroup~$(\mathcal{B}^{T}(t))_{t\geq0}$ defined on~$\mathcal{C}_{0}^{\infty}(\mathbb{R}_{x}^{d}\times\mathbb{R}_{\xi}^{d*})$
extends to a strongly continuous group on~$(\mathcal{C}_{\infty}^{0}(\mathbb{R}_{x}^{d}\times\mathbb{R}_{\xi}^{d*}),\left\Vert \cdot\right\Vert _{\infty})$
and defines by duality a (weak$^{*}$ continuous) group~$\mathcal{B}(t)$
on~$\mathcal{M}_{b}(\mathbb{R}_{x}^{d}\times\mathbb{R}_{\xi}^{d*})$. \end{prop}
\begin{proof}
Using a partition of the unity, $\mathcal{B}^{T}(t)$ extends to~$\mathcal{C}^{\infty}(\mathbb{R}_{x}^{d}\times\mathbb{R}_{\xi}^{d*})$.
Since~$\mathcal{B}^{T}(t)$ is positive, we have~$\mathcal{B}^{T}(t)(\left\Vert b\right\Vert _{\infty}\pm b)\geq0$
for all~$b$ in~$\mathcal{C}_{0}^{\infty}(\mathbb{R}_{x}^{d}\times\mathbb{R}_{\xi}^{d*})$
and so $\|\mathcal{B}^{T}\!(t)\, b\|_{\infty}\leq\|b\|_{\infty}$.
The group $\mathcal{B}^{T}(t)$ thus extends continuously to~$\mathcal{C}_{\infty}^{0}(\mathbb{R}_{x}^{d}\times\mathbb{R}_{\xi}^{d*})$.\end{proof}
\begin{defn}
The \emph{linear Boltzmann group}\index{Boltzmann group@$\mathcal{B}\left(t\right)$, linear Boltzmann group}~$(\mathcal{B}(t))$
is defined on~$\mathcal{M}_{b}(\mathbb{R}_{x}^{d}\times\mathbb{R}_{\xi}^{d*})$
by duality: let~$\mu\in\mathcal{M}_{b}(\mathbb{R}_{x}^{d}\times\mathbb{R}_{\xi}^{d*})$,
then, for any~$t\in\mathbb{R}$, \[
\forall b\in\mathcal{C}_{\infty}^{0}(\mathbb{R}_{x}^{d}\times\mathbb{R}_{\xi}^{d*}),\qquad\langle\mathcal{B}(t)\mu,b\rangle=\langle\mu,\mathcal{B}^{T}(t)b\rangle\,.\]

\end{defn}

\subsection{A Trotter-type approximation}

This Section provides a result in the spirit of Trotter's approximation
$(e^{A/N}e^{B/N})^{N}\to e^{A+B}$ useful to deal with the renewal
of the stochasticity.
\begin{prop}
\label{pro:trotter-boltzmann}Let~$b\in\mathcal{C}_{0}^{\infty}(\mathbb{R}_{x,\xi}^{2d})$,
$T>0$ and~$n\in\mathbb{N}$. There are constants~$C_{n,Q}$ and~$C_{T,b}$
such that for all~$N\in\mathbb{N}^{*}$\[
\mathcal{N}_{n}\big(e^{T(2\xi.\partial_{x}+Q)}b-\big(e^{\frac{T}{N}Q}e^{\frac{T}{N}2\xi.\partial_{x}}\big)^{N}b\big)\leq e^{T(2n+C_{n,Q})}C_{T,b}\frac{T^{2}}{N}\,.\]
where, for~$n\in\mathbb{N}$,\index{norm@$\mathcal{N}_{n}\left(b\right)$}
$\mathcal{N}_{n}(b):=\sup_{|\alpha|\leq n}\left\Vert \partial^{\alpha}b\right\Vert _{\infty}$.\end{prop}
\begin{notation}
Let~$Q_{t}=e^{t2\xi.\partial_{x}}Qe^{-t2\xi.\partial_{x}}\in\mathcal{L}(L^{2}[r,r^{\prime}])$\index{collision operator4@$Q_{t}$, approximated collision operator},
\emph{i.e.}~$Q_{t}=Q_{+,t}-Q_{-}$ with\[
Q_{+,t}b(x,\xi)=\int_{\mathbb{R}_{\xi^{\prime}}^{d}}\negthickspace\sigma(\xi,\xi^{\prime})\,\delta\big(|\xi|^{2}-|\xi^{\prime}|^{2}\big)\, b(x-2t(\xi^{\prime}-\xi),\xi^{\prime})\diff\xi^{\prime}\,.\]
Let also~$Q_{-,t}=Q_{-}$ to have consistent notations in the sequel.

Let~$G_{Q}(t,t_{0})$ be the dynamical system associated with the
one parameter family~$(Q_{t})$ in $\mathcal{C}(\mathbb{R};\mathcal{L}(L^{2}[r,r^{\prime}]))$
given by\[
\left\{ \begin{aligned}\partial_{t}b_{t} & =Q_{t}\, b_{t}\\
b_{t=t_{0}} & =b_{0}\in L_{x,\xi}^{2}\end{aligned}
\right.\,,\qquad b_{t}=G_{Q}\left(t,t_{0}\right)b_{0}\,.\]
Note the relation $\mathcal{B}^{T}(t)=G_{Q}(t,0)e^{2t\xi.\partial_{x}}=e^{2t\xi.\partial_{x}}G_{Q}(0,-t)$.

For~$b\in\mathcal{C}_{0}^{\infty}(\mathbb{R}_{x}^{d}\times A_{\xi}[r,r^{\prime}])$,
let\[
\mathcal{N}_{n}(Q)=\sup_{b\neq0}\frac{\mathcal{N}_{n}(Qb)}{\mathcal{N}_{n}b}\;\,\mbox{and}\;\,\mathcal{N}_{n+1,n}(s,Q-Q_{s})=\sup_{b\neq0}\frac{\mathcal{N}_{n}\big((Q-Q_{s})b\big)}{s\left(1+2\left|s\right|\right)^{n}\mathcal{N}_{n+1}b}\,.\]
\end{notation}
\begin{lem}
\label{lem:proprietes-N}For any~$n\in\mathbb{N}$,~$s\geq0$ and~$b\in\mathcal{C}_{0}^{\infty}(\mathbb{R}_{x}^{d}\times A_{\xi}[r,r^{\prime}])$,
there exist constants~$C_{1}$, and~$C_{2}$ depending on~$d$,
$r$ and~$r^{\prime}$ such that
\begin{enumerate}
\item $\mathcal{N}_{n}(Q)\leq C_{1}$,
\item $\mathcal{N}_{n+1,n}(t,Q-Q_{t})\leq C_{2}$,
\item $\mathcal{N}_{n}(e^{2t\xi.\partial_{x}}b)\leq(1+2\left|t\right|)^{n}\,\mathcal{N}_{n}(b)$.
\end{enumerate}
\end{lem}
\begin{proof}
The first point is clear from the integral expression of~$Qb$.

For the second point differentiate and estimate the integral formula
for~$b\left(x-2t\xi,\xi\right)-b\left(x,\xi\right)$, with~$\left|\alpha\right|\leq n$,\begin{align*}
\big|\partial^{\alpha}\big(b(x-2t\xi,\xi)-b(x,\xi)\big)\big| & \leq\int_{0}^{t}\big|\partial^{\alpha}\big(2\xi.\partial_{x}b(x-2s\xi,\xi)\big)\big|\diff s\\
 & \leq2\left|\xi\right|t\left(1+2t\right)^{n}\mathcal{N}_{n+1}(b)\,.\end{align*}

The last point results from~$\left(e^{2t\xi.\partial_{x}}b\right)\negthinspace(x,\xi)=b(x+2t\xi,\xi)$.\end{proof}
\begin{lem}
\label{lem:boltzman-estimate-N}Let~$b,\,\tilde{b}\in\mathcal{C}_{0}^{\infty}\left(\mathbb{R}_{x}^{d}\times A_{\xi}\left[r,r^{\prime}\right]\right)$,
then for all $t\geq0$,\[
e^{tQ}\tilde{b}-G_{Q}(t,0)b=e^{tQ}(\tilde{b}-b)+\int_{0}^{t}e^{(t-s)Q}(Q-Q_{s})G_{Q}(s,0)b\diff s\]
and we have the estimate\begin{multline*}
\mathcal{N}_{n}\bigl(e^{tQ}\tilde{b}-G_{Q}(t,0)b\bigr)\leq e^{t\mathcal{N}_{n}Q}\,\mathcal{N}_{n}(\tilde{b}-b)\\
+t^{2}(1+2t)^{n}e^{t\mathcal{N}_{n}Q}\sup_{s\in[0,t]}\bigl\{\mathcal{N}_{n+1,n}(s,Q-Q_{s})\,\mathcal{N}_{n+1}\bigl(G_{Q}(s,0)\bigr)\bigr\}\,\mathcal{N}_{n+1}(b)\,.\end{multline*}
\end{lem}
\begin{proof}
The equality is clear once we have computed that both sides satisfy
the equation \[
\partial_{t}\Delta_{t}=Q\Delta_{t}+(Q-Q_{t})G_{Q}(t,0)b\,.\]
The inequality then follows from Lemma~\ref{lem:proprietes-N}.
\end{proof}

\begin{proof}[Proof of Proposition~\ref{pro:trotter-boltzmann}]
We fix~$N$ and forget the~$N$'s in the notations concerning~$\tilde{b}$.
We set~$b_{t}=\mathcal{B}^{T}\left(t\right)b$ and define~$\tilde{b}_{t}$
piecewise on~$\left[0,T\right]$ by setting~$t_{k}=\frac{kT}{N}$,
$\tilde{b}_{t_{k}}=\big(e^{\frac{T}{N}Q}e^{\frac{T}{N}2\xi.\partial_{x}}\big)^{k}b_{0}$
and, for $t\in[t_{k},t_{k+1})$, $\tilde{b}_{t}=e^{(t-t_{k})Q}e^{(t-t_{k})2\xi.\partial_{x}}\tilde{b}_{t_{k}}$.
Let $\delta_{k}=\mathcal{N}_{n}\big(b_{t_{k}}-\tilde{b}_{t_{k}}\big)$;
we get\[
e^{\frac{T}{N}Q}e^{\frac{T}{N}2\xi.\partial_{x}}\tilde{b}_{t_{k}}-e^{\frac{T}{N}\left(2\xi.\partial_{x}+Q\right)}b_{t_{k}}=e^{\frac{T}{N}Q}e^{\frac{T}{N}2\xi.\partial_{x}}\tilde{b}_{t_{k}}-G_{Q}\big({\textstyle \frac{T}{N}},0\big)e^{\frac{T}{N}2\xi.\partial_{x}}b_{t_{k}}\]
and we can then use Lemma~\ref{lem:boltzman-estimate-N} to obtain\begin{align*}
\delta_{k+1} & \leq e^{\frac{T}{N}\mathcal{N}_{n}Q}\,{\textstyle \left(1+2\frac{T}{N}\right)^{n}}\,\delta_{k}+{\textstyle \left(\frac{T}{N}\right)^{2}}\,{\textstyle \left(1+2\frac{T}{N}\right)^{n}}\, e^{\frac{T}{N}\mathcal{N}_{n}Q}\\
 & \qquad\qquad\sup_{s\in\left[t_{k},t_{k+1}\right]}\mathcal{N}_{n+1,n}\big(s-t_{k},Q-Q_{s-t_{k}}\big)\\
 & \qquad\qquad\sup_{s\in\left[t_{k},t_{k+1}\right]}\mathcal{N}_{n+1}\big(G_{Q}(s-t_{k},0)\, e^{\frac{T}{N}2\xi.\partial_{x}}b_{t_{k}}\big)\\
 & \leq e^{\frac{T}{N}\mathcal{N}_{n}Q}e^{2\frac{nT}{N}}\delta_{k}+{\textstyle \left(\frac{T}{N}\right)^{2}}e^{\frac{T}{N}\mathcal{N}_{n}Q}C_{N,T}\end{align*}
 where we introduced\begin{multline*}
C_{N,T,b}={\textstyle \left(1+2\frac{T}{N}\right)^{n}}\sup_{s\in\left[0,T/N\right]}\mathcal{N}_{n+1,n}(s,Q-Q_{s})\\
\sup_{k\in\left\{ 0,\dots,N-1\right\} }\sup_{s\in\left[0,T/N\right]}\mathcal{N}_{n+1}\big(G_{Q}(s,0)\, e^{-\frac{T}{N}Q}b_{t_{k+1}}\big)\,.\end{multline*}
Then we get the recursive formula\[
\delta_{k+1}\leq e^{\frac{T}{N}(2n+\mathcal{N}_{n}Q)}\,\delta_{k}+C_{N,T,b}\,{\textstyle \left(\frac{T}{N}\right)^{2}}\, e^{\frac{T}{N}\mathcal{N}_{n}Q}\]
so that\[
\delta_{N}\leq e^{T\,(2n+\mathcal{N}_{n}Q)}\, C_{N,T,b}{\textstyle \frac{T^{2}}{N}}\,.\]
The only thing remaining is to observe that~$C_{N,T,b}\leq C_{T,b}$,
with \[
C_{T,b}:=\left(1+2T\right)^{n}\sup_{s\in\left[0,T\right]}\mathcal{N}_{n+1,n}(s,Q-Q_{s})\sup_{s_{j}\in\left[0,T\right]}\mathcal{N}_{n+1}\big(G_{Q}(s_{1},0)\, e^{-s_{2}Q}b_{s_{3}}\big)\]
and for a fixed~$T$ this quantity~$C_{T,b}$ is finite, so that
we get the result.
\end{proof}

\section{\label{par:From-stochastics-to-Fock}From stochastics to the Fock
space}

\subsection{The second quantization}

The method of second quantization is exposed in the books of Berezin~\cite{MR0208930}
and Bratteli and Robinson~\cite{MR1441540}, an introduction to quantum
field theory and second quantization can be found in the book of Folland~\cite{MR2436991}.
The series of articles of Ginibre and Velo~\cite{MR530915,MR539736,MR602197,MR605198}
uses this framework with a small parameter to handle classical or
mean field limits by extending the Hepp method~\cite{MR0332046}.
We use the notation and framework of articles of Ammari and Nier~\cite{MR2465733,MR2513969}
to handle the second quantization with a small parameter. For the
convenience of the reader we expose briefly this framework.

Most of the operators on the Fock space in this article arise as Wick
quantizations of polynomials.
\begin{defn}
Let~$(\mathcal{H},\langle\cdot,\cdot\rangle)$ be a complex separable
Hilbert space (the scalar product is $\mathbb{C}$-antilinear with
respect to the left variable). The symmetric tensor product is denoted
by~$\vee$. The \emph{polynomials} with variable in~$\mathcal{H}$
are the finite linear combinations of monomials~$Q:\mathcal{H}\to\mathbb{C}$
of the form\[
Q\left(z\right)=\langle z^{\vee q},\tilde{Q}z^{\vee p}\rangle\]
where~$p,q\in\mathbb{N}$, $\tilde{Q}\in\mathcal{L}(\mathcal{H}^{\vee p},\mathcal{H}^{\vee q})$
and~$\langle\cdot,\cdot\rangle$ denotes the scalar product on~$\mathcal{H}^{\vee q}$.
The set of such polynomials is denoted by~$\mathcal{P}(\mathcal{H})$.

The symmetric \emph{Fock space}\index{Fock space@$\Gamma\mathcal{H},\,\Gamma_{n}\mathcal{H},\,\Gamma_{F}\mathcal{H}$,
Fock space} associated to~$\mathcal{H}$ is\[
\Gamma\mathcal{H}=\bigoplus_{n=0}^{\infty}\Gamma_{n}\mathcal{H}\]
with~$\Gamma_{n}\mathcal{H}=\mathcal{H}^{\vee n}$ the Hilbert completed
$n$-th symmetric power of~$\mathcal{H}$ and the sum is completed,
the set of \emph{finite particle vectors}~$\Gamma_{F}\mathcal{H}$
is defined as the Fock space but with an algebraic sum.

Let $\varepsilon>0$. The \emph{Wick quantization}\index{Wick quantization@$Q^{Wick}$, Wick quantization}
of a polynomial is defined as the linear combination of the Wick quantizations
of its monomials, and for a monomial~$Q$ we define~$Q^{Wick}:\Gamma_{F}\mathcal{H}\to\Gamma_{F}\mathcal{H}$
as the linear operator which vanishes on $\mathcal{H}^{\vee n}$ for
$n<p$ and for $n\geq0$\[
\left.Q^{Wick}\right|_{\mathcal{H}^{\vee n+p}}={\textstyle \frac{\sqrt{(n+p)!(n+q)!}}{n!}}\,\varepsilon^{\frac{p+q}{2}}\,(\tilde{Q}\vee\Id_{\mathcal{H}^{\vee n}})\in\mathcal{L}\big(\mathcal{H}^{\vee n+p},\mathcal{H}^{\vee n+q}\big)\,.\]
The field operator\index{field operator@$\Phi_{\varepsilon}\left(f\right)$,~field operator}~$\Phi_{\varepsilon}(f)$
($f\in\mathcal{H}$) is the closure of the essentially self-adjoint
operator~$(\langle z,f\rangle+\langle f,z\rangle)^{Wick}/\sqrt{2}$.
Using the Weyl operator $W(f)=\exp(i\Phi_{\varepsilon}(f))$\index{Weyl operator@$W(f)$, Weyl operator}
the coherent state~$E(f)=W\big(\frac{\sqrt{2}}{i\varepsilon}f\big)\,\Omega$
can be defined, where $\Omega=(1,0,0,\dots)\in\Gamma\mathcal{H}$
is the empty state. The Weyl operators satisfy the relation\[
W(f)\, W(g)=e^{-\frac{i\varepsilon}{2}\Im\left\langle f,g\right\rangle }W(f+g)\,.\]
The second quantization $\diff\Gamma_{\!\varepsilon}(A)$\index{Fock2@$\diff\Gamma_{\varepsilon}\left(A\right)$}
of a self-adjoint operator~$A$ on~$\mathcal{H}$ is\[
\left.\diff\Gamma_{\!\varepsilon}(A)\right|_{D(A)^{\vee n,\mbox{alg}}}=\varepsilon\left(A\otimes\Id_{\mathcal{H}}\otimes\cdots\otimes\Id_{\mathcal{H}}+\cdots+\Id_{\mathcal{H}}\otimes\cdots\otimes\Id_{\mathcal{H}}\otimes A\right)\]
and for a unitary~$U$ on~$\mathcal{H}$, the unitary operator~$\Gamma(U)$\index{Fock2@$\Gamma(U)$}
on~$\Gamma\mathcal{H}$ is defined by\[
\left.\Gamma(U)\right|_{\mathcal{H}^{\vee n}}=U^{\vee n}=U\otimes\cdots\otimes U\]
 and thus~$\Gamma(e^{itA})=\exp\big(\frac{it}{\varepsilon}\diff\Gamma_{\!\varepsilon}(A)\big)$.
\end{defn}

\subsection{The expression of the dynamic in the Fock space}

The relation between Gaussian random processes and the Fock space
is treated in the books of Simon~\cite{MR0489552} and Glimm and
Jaffe~\cite{MR887102}, we recall a theorem about this relation.
\begin{thm}
Let $\mathcal{V}^{h}(x)$ be the centered, translation invariant,
gaussian random field with covariance $hG(x-y)$ such that $\hat{G}=|\hat{V}|^{2}$
for some $V\in\mathcal{S}(\mathbb{R}^{d};\mathbb{R})$. The symmetric
Fock space~$\Gamma L^{2}(\mathbb{R}^{d};\mathbb{C})$ is unitarily
equivalent to~$L^{2}(\Omega_{\mathbb{P}},\mathbb{P};\mathbb{C})$
under a unitary~$D:\Gamma\mathcal{H}_{\mathbb{C}}\to L^{2}(\Omega_{\mathbb{P}},\mathbb{P};\mathbb{C})$
such that
\begin{itemize}
\item $D\,\Omega=1$,
\item $D\,\sqrt{2h}\Phi_{1}(\tau_{x}V)\, D^{-1}=\mathcal{V}^{h}(x)$, with~$\mathcal{V}^{h}(x)$
seen as a multiplication operator on~$L^{2}(\Omega_{\mathbb{P}},\mathbb{P};\mathbb{C})$.
\end{itemize}
\end{thm}
For Hilbert spaces~$\mathcal{H}$ and~$\mathcal{H}^{\prime}$, $\Tr_{\mathcal{H}^{\prime}}[A]$\index{trace@$\Tr_{\mathcal{H}}$, partial trace}
denotes the partial trace of an operator~$A\in\mathcal{L}_{1}(\mathcal{H}\otimes\mathcal{H}^{\prime})$,
$\Tr_{\mathcal{H}}[\Tr_{\mathcal{H}^{\prime}}[A]B]=\Tr_{\mathcal{H}\otimes\mathcal{H}^{\prime}}[A(B\otimes I_{\mathcal{H}^{\prime}})]$,
$\forall B\in\mathcal{}$ on~$\mathcal{H}\otimes\mathcal{H}^{\prime}$.
\begin{prop}
Let ~$H_{h}=-\Delta_{x}+\sqrt{2h}\Phi_{1}(\tau_{x}V)$, with $\tau_{x}f(y)=f(y-x)$
for~$x\in\mathbb{R}^{d}$ and $f\in L_{y}^{2}$. Then\[
\mathcal{G}_{t}^{h}(\rho)=\Tr{}_{\Gamma L_{y}^{2}}\big[e^{-i\frac{t}{h}H_{h}}\,\rho\otimes|\Omega\rangle\langle\Omega|\, e^{i\frac{t}{h}H_{h}}\big]\,.\]
\end{prop}
\begin{proof}
In the stochastic presentation we can express the integral in~$\omega$
in the definition of~$\mathcal{G}_{t}^{h}$ as a partial trace\begin{align*}
\mathcal{G}_{t}^{h}(\rho) & =\int e^{-i\frac{t}{h}H_{h,\omega}}\rho1(\omega)1(\omega)e^{i\frac{t}{h}H_{h,\omega}}\diff\mathbb{P}(\omega)\,.\\
 & =\Tr{}_{L^{2}(\Omega_{\mathbb{P}},\mathbb{P})}\Big[\int^{\oplus}e^{-i\frac{t}{h}H_{h,\omega}}\diff\mathbb{P}(\omega)\rho\otimes|1\rangle\langle1|\int^{\oplus}e^{i\frac{t}{h}H_{h,\omega^{\prime}}}\diff\mathbb{P}(\omega^{\prime})\Big]\,.\end{align*}
Using the isomorphism $\mathcal{U}:=\Id_{L_{x}^{2}}\otimes D:L_{x}^{2}\otimes\Gamma L_{y}^{2}\to L_{x}^{2}\otimes L^{2}(\Omega_{\mathbb{P}},\mathbb{P})$
we get\[
\mathcal{U}^{*}\int^{\oplus}e^{-i\frac{\Delta t}{h}H_{h,\omega}}\diff\mathbb{P}(\omega)\:\mathcal{U}=e^{-i\frac{\Delta t}{h}H_{h}}\,,\quad\mbox{and}\quad\mathcal{U}^{*}\,\rho\otimes\left|1\right\rangle \left\langle 1\right|\:\mathcal{U}=\rho\otimes\left|\Omega\right\rangle \left\langle \Omega\right|\]
with~$H_{h}:=\mathcal{U}^{*}\:(\int^{\oplus}H_{h,\omega}\diff\mathbb{P}(\omega))\:\mathcal{U}=-\Delta_{x}+\sqrt{2h}\Phi_{1}(\tau_{x}V)$.
Hence the result.
\end{proof}

\subsection{Existence of the dynamic}

We show that the dynamic of the system is well defined.  Since we
work with a fixed~$h>0$ the value of~$h$ is here irrelevant and
we set~$h=1$ in this section to clarify our exposition. We write
for short
\begin{itemize}
\item $-\Delta_{x}$ for the operator~$-\Delta_{x}\otimes\Id_{\Gamma L_{y}^{2}}$,
\item $N$ for the operator~$\Id_{L_{x}^{2}}\otimes N$ with~$N=\diff\Gamma_{1}(\Id_{L_{y}^{2}})$
the number operator on~$\Gamma L_{y}^{2}$ and
\item $\Phi_{1}(\tau_{\cdot}V)$ the operator on~$L^{2}(\mathbb{R}^{d};\Gamma L_{y}^{2})\simeq L^{2}(\mathbb{R}^{d})\otimes\Gamma L_{y}^{2}$
defined by $u\mapsto\Phi_{1}(\tau_{\cdot}V)u$ with $[\Phi_{1}(\tau_{\cdot}V)u](x):=[\Phi_{1}(\tau_{x}V)][u(x)]$.\end{itemize}
\begin{prop}
If~$V$ belongs to the Sobolev space~$H^{2}(\mathbb{R}^{d})$, then\[
H=-\Delta_{x}+\sqrt{2}\Phi_{1}(\tau_{\cdot}V)\,,\]
is essentially self-adjoint on~$D':=\mathcal{C}_{0}^{\infty}(\mathbb{R}^{d})\otimes^{\text{alg}}\Gamma_{F}L_{y}^{2}$
and its closure is essentially self-adjoint on any other core for~$N'=\Id-\Delta_{x}+N$.\end{prop}
\begin{proof}
We still denote by~$N^{\prime}$ the closure of the essentially self-adjoint
operator~$N^{\prime}$ defined on~$D^{\prime}$. Then~$D'$ is
a core for this operator. We remark that~$N'\geq I$ on~$D'$ and
thus also on~$D\left(N'\right)$ as~$D'$ is a core for~$N'$.

We verify the two estimates needed for Nelson's commutator theorem
(see the book of Reed and Simon~\cite{MR0493420}). Let~$u\in D^{\prime}$,
then\begin{align*}
\left\Vert Hu\right\Vert _{L_{x}^{2}\otimes\Gamma L_{y}^{2}} & \leq\left\Vert -\Delta_{x}u\right\Vert _{L_{x}^{2}\otimes\Gamma L_{y}^{2}}+2\left\Vert V\right\Vert _{L^{2}}\left\Vert \sqrt{N+1}u\right\Vert _{L_{x}^{2}\otimes\Gamma L_{y}^{2}}\,,\\
 & \leq\left(1+2\left\Vert V\right\Vert _{L^{2}}\right)\left\Vert N^{\prime}u\right\Vert _{L_{x}^{2}\otimes\Gamma L_{y}^{2}}\,.\end{align*}
In the sense of quadratic forms\begin{align*}
[H,N^{\prime}] & =\sqrt{2}[\Phi(\tau_{\cdot}V),-\Delta_{x}+N]\\
 & =\sqrt{2}\Phi(\tau_{\cdot}\nabla V).\nabla_{x}+\sqrt{2}\Phi(\tau_{\cdot}\Delta V)+(a^{*}(\tau_{\cdot}V)-a(\tau_{\cdot}V))\end{align*}
so that $\left|\langle Hu,N'u\rangle-\langle N'u,Hu\rangle\right|\leq6\|V\|_{H^{2}}\|N^{\prime1/2}u\|^{2}$
which achieves the proof.
\end{proof}

\subsection{The scaling for field operators}

The~$\varepsilon$ parameter is an intermediate scale which allows
to easily identify the graduation in Wick powers. We set~\index{Ad@$\Ad$}$\Ad\{A\}[B]=ABA^{-1}$.
Let~$\left(D_{\varepsilon}f\right)\left(y\right)=\varepsilon^{-d/2}f\left(\frac{y}{\varepsilon}\right)$
and\index{H_{varepsilon}@$H_{\varepsilon}$}\[
H_{h,\varepsilon}=\Ad\left\{ \Id_{L_{x}^{2}}\otimes\Gamma D_{\varepsilon}\right\} \left[H_{h}\right]=-\Delta_{x}\otimes I_{\Gamma_{y}}+\sqrt{2h}\Phi\big(\varepsilon^{-d/2}V(x-\tfrac{y}{\varepsilon})\big)\,.\]

\section{\label{par:An-approximated-equation}An approximated equation and
its solution}

\subsection{Space translation in the fields and Fourier transform}
\begin{notation}
For an object~$X=\left(X_{1},\dots,X_{d}\right)$ with~$d$ components,
like~$\xi\in\mathbb{R}^{d}$, $D_{x}=\left(\partial_{x_{1}},\dots,\partial_{x_{d}}\right)$
or~$\diff\Gamma_{\varepsilon}(D_{y})$, let~$X^{.2}:=X_{1}^{2}+\cdots+X_{d}^{2}\,.$
\end{notation}
We want to work with a field operator with no dependence in~$x$.
Then we recall that the translation~$\tau_{x}$ of~$x$ can be written
as~$e^{-ix.D_{y}}$ and thus \[
\Gamma\big(e^{i\varepsilon x.D_{y}}\big)\, H_{h,\varepsilon}\,\Gamma\big(e^{-i\varepsilon x.D_{y}}\big)=(\diff\Gamma_{\!\varepsilon}(D_{y})-D_{x})^{.2}+\sqrt{2}\Phi_{\varepsilon}\big({\textstyle \varepsilon^{-d/2}\sqrt{\frac{h}{\varepsilon}}V\left(-\tfrac{y}{\varepsilon}\right)}\big)\]
where we use the $\varepsilon$-dependent operator~$\diff\Gamma_{\!\varepsilon}$.
A conjugation by the Fourier transform in both the particle and the
field variables yields\index{hat{H}_{varepsilon}@$\hat{H}_{\varepsilon}$}
a new expression for the Hamiltonian, and an approximated version\begin{align*}
\hat{H}_{h,\varepsilon} & =\xi^{.2}-\diff\Gamma_{\!\varepsilon}(2\xi.\eta)+\diff\Gamma_{\!\varepsilon}(\eta)^{.2}+\sqrt{2}\Phi_{\varepsilon}(f_{h,\varepsilon})\,,\\
\hat{H}_{h,\varepsilon}^{app} & =\xi^{.2}+\diff\Gamma_{\!\varepsilon}\left(\varepsilon\eta^{.2}-2\xi.\eta\right)+\sqrt{2}\Phi_{\varepsilon}\left(f_{h,\varepsilon}\right)\,,\end{align*}
with~$f_{h,\varepsilon}(\eta)=\varepsilon^{d/2}\sqrt{\frac{h}{\varepsilon}}\hat{V}(-\varepsilon\eta)$,
\emph{i.e.}~$\hat{H}_{h,\varepsilon}=Q_{h,\varepsilon}^{Wick}$ and~$\hat{H}_{h,\varepsilon}^{app}=Q_{h,\varepsilon}^{app,Wick}$
with \begin{align*}
Q_{h,\varepsilon}(z) & =\xi^{.2}+\langle z,(\varepsilon\eta^{.2}-2\xi.\eta)z\rangle+\langle z,\eta z\rangle^{.2}+2\Re\langle z,f_{h,\varepsilon}\rangle\,,\\
Q_{h,\varepsilon}^{app}(z) & =\xi^{.2}+\langle z,(\varepsilon\eta^{.2}-2\xi.\eta)z\rangle\phantom{+\langle z,\eta z\rangle^{.2}}+2\Re\langle z,f_{h,\varepsilon}\rangle\,.\end{align*}
Note that in the approximated Hamiltonian we neglect the quartic
part $\langle z,\eta z\rangle^{.2}$. The evolution associated with
the approximated Hamiltonian\index{hat{H}_{h,varepsilon}^{app}@$\hat{H}_{h,\varepsilon}^{app}$}
is explicitely solvable.
\begin{defn}
For~$\rho\in\mathcal{L}_{1}\left(L_{x}^{2}\right)$, let\index{state2@$\boldsymbol{\rho}_{t},\,\hat{\boldsymbol{\rho}}_{t},\,\rho_{t}^{\varepsilon},\,\boldsymbol{\rho}_{t}^{app},\,\boldsymbol{\hat{\rho}}_{t}^{app},\rho_{t}^{\varepsilon,app}$}\begin{align*}
\boldsymbol{\rho}_{t} & =\Ad\big\{ e^{-i\frac{t}{\varepsilon}H_{h,\varepsilon}}\big\}\left[\rho\otimes|\Omega\rangle\langle\Omega|\right]\,, & \boldsymbol{\rho}_{t}^{app} & =\Ad\big\{ e^{-i\frac{t}{\varepsilon}H_{h,\varepsilon}^{app}}\big\}\left[\rho\otimes|\Omega\rangle\langle\Omega|\right]\,,\\
\hat{\boldsymbol{\rho}}_{t} & =\Ad\big\{ e^{-i\frac{t}{\varepsilon}\hat{H}_{h,\varepsilon}}\big\}\left[\hat{\rho}\otimes|\Omega\rangle\langle\Omega|\right]\,, & \hat{\boldsymbol{\rho}}_{t}^{app} & =\Ad\big\{ e^{-i\frac{t}{\varepsilon}\hat{H}_{h,\varepsilon}^{app}}\big\}\left[\hat{\rho}\otimes|\Omega\rangle\langle\Omega|\right]\,,\\
\rho_{t}^{\varepsilon} & =\Tr{}_{\Gamma L_{x}^{2}}[\boldsymbol{\rho}_{t}]\,, & \rho_{t}^{\varepsilon,app} & =\Tr{}_{\Gamma L_{x}^{2}}[\boldsymbol{\rho}_{t}^{app}]\,.\end{align*}

\end{defn}
This definition is consistent with the previous one given for~$\rho_{t}^{h}$
as~$\rho_{t}^{h}=\rho_{\frac{\varepsilon}{h}t}^{\varepsilon}$ and
the dilatation acts only in the Fock space part of~$L_{x}^{2}\otimes\Gamma L_{y}^{2}$.

\subsection{The solution of the approximated equation}
\begin{prop}
\label{pro:Approximated_solution}For~$\psi_{0}\in L_{x}^{2}$\index{hat{Psi}_{t},,hat{Psi}_{t}^{app}@$\hat{\Psi}_{t},\,\hat{\Psi}_{t}^{app}$}\index{z@$z_{t}$}\index{omega@$\omega_{t}$},
let\begin{align}
\hat{\Psi}_{h,\varepsilon,t}=e^{-i\frac{t}{\varepsilon}\hat{H}_{h,\varepsilon}}\Omega & \otimes\hat{\psi}_{0}\qquad\mbox{and}\qquad\hat{\Psi}_{h,\varepsilon,t}^{app}=e^{-i\frac{t}{\varepsilon}\hat{H}_{h,\varepsilon}^{app}}\Omega\otimes\hat{\psi}_{0}\,,\\
z_{h,\varepsilon,t} & =-i\int_{0}^{t}e^{-i\frac{s}{\varepsilon}(\varepsilon^{2}\eta^{.2}-2\xi.\varepsilon\eta)}f_{h,\varepsilon}\diff s\\
\omega_{h,\varepsilon,t} & =t\xi^{2}+\int_{0}^{t}\Re\langle z_{s},f_{h,\varepsilon}\rangle\diff s\,.\end{align}

Then
\begin{enumerate}
\item \label{enu:formule_Psi}$\hat{\Psi}_{h,\varepsilon,t}^{app}=e^{-i\frac{\omega_{h,\varepsilon,t}}{\varepsilon}}W\big({\textstyle \frac{\sqrt{2}}{i\varepsilon}}z_{h,\varepsilon,t}\big)\,\Omega\otimes\hat{\psi}_{0}$.
\item \label{pro:estimation-z_t}There is a constant~$C_{G,d}$ depending
on~$G$ and the dimension~$d$ such that\textup{\[
\left\Vert \left|\eta\right|^{\nu}z_{h,\varepsilon,t}\right\Vert _{L_{\eta}^{2}}\leq C_{G,d}(\tfrac{ht}{\varepsilon})^{1/2}\,\varepsilon^{1/2-\nu}\,.\]
}
\item \label{enu:llPsi-Psi_appll}Let~$T_{0}>0$. There is a constant~$C_{T_{0},G,d}$
such that for~$\frac{ht}{\varepsilon}\leq T_{0}$,\[
\big\|\hat{\Psi}_{h,\varepsilon,t}-\hat{\Psi}_{h,\varepsilon,t}^{app}\big\|\leq C_{T_{0},G,d}\big(\tfrac{ht}{\varepsilon}/\sqrt{h}\big)^{2}\,.\]

\item \label{enu:Number_estimate_Psi}For both~$\hat{\Psi}_{h,\varepsilon,t}^{\sharp}=\hat{\Psi}_{h,\varepsilon,t}$
and~$\hat{\Psi}_{h,\varepsilon,t}^{\sharp}=\hat{\Psi}_{h,\varepsilon,t}^{app}$\[
\big\|(\varepsilon+N_{\varepsilon})^{1/2}\hat{\Psi}_{h,\varepsilon,t}^{\sharp}\big\|\leq C_{d}\Big(\sqrt{\varepsilon}+\sqrt{\tfrac{t}{2}\tfrac{ht}{\varepsilon}\|\hat{G}\|_{L^{1}}}\Big)\,.\]

\end{enumerate}
\end{prop}

First we get rid of the quadratic part $\diff\Gamma_{\!\varepsilon}$.
Let 
\begin{itemize}
\item $\tilde{\hat{\Psi}}_{h,\varepsilon,t}=e^{i\frac{t}{\varepsilon}\xi^{.2}}e^{i\frac{t}{\varepsilon}\diff\Gamma_{\!\varepsilon}(\varepsilon\eta^{.2}-2\xi.\eta)}\hat{\Psi}_{h,\varepsilon,t}\;\mbox{and}\;\tilde{\hat{\Psi}}_{h,\varepsilon,t}^{app}=e^{i\frac{t}{\varepsilon}\xi^{.2}}e^{i\frac{t}{\varepsilon}\diff\Gamma_{\!\varepsilon}(\varepsilon\eta^{.2}-2\xi.\eta)}\hat{\Psi}_{h,\varepsilon,t}^{app}$,
\item $\tilde{f}_{h,\varepsilon,t}=e^{i\frac{t}{\varepsilon}(\varepsilon^{2}\eta^{.2}-2\xi.\varepsilon\eta)}f_{h,\varepsilon}$,
\item $\tilde{z}_{h,\varepsilon,t}=-i\int_{0}^{t}\tilde{f}_{h,\varepsilon,s}\diff s$,
\item $\tilde{\omega}_{h,\varepsilon,t}=\int_{0}^{t}\Re\langle\tilde{z}_{h,\varepsilon,s},\tilde{f}_{h,\varepsilon,s}\rangle\diff s$.
\end{itemize}
It is then enough to prove the results with the objects with a $\sim$
sign.
\begin{lem}
\label{lem:Suppression_partie_quadratique}Then $\tilde{\hat{\Psi}}_{t}$
(resp. $\tilde{\hat{\Psi}}_{t}^{app}$) is solution of the equation\[
i\varepsilon\partial_{t}\tilde{\hat{\Psi}}_{h,\varepsilon,t}=\tilde{Q}_{h,\varepsilon}^{Wick}\tilde{\hat{\Psi}}_{h,\varepsilon,t}\quad\mbox{(resp.}\quad i\varepsilon\partial_{t}\tilde{\hat{\Psi}}_{h,\varepsilon,t}^{app}=\tilde{Q}_{h,\varepsilon}^{app,Wick}\tilde{\hat{\Psi}}_{h,\varepsilon,t}^{app}\,\mbox{)}\]
with initial condition $\Omega\otimes\hat{\psi}_{0}$, $\tilde{Q}_{h,\varepsilon,t}(z)=2\Re\langle z,\tilde{f}_{h,\varepsilon,t}\rangle+\left\langle z,\eta z\right\rangle ^{.2}$
(resp. $\tilde{Q}_{h,\varepsilon,t}^{app}(z)=2\Re\langle z,\tilde{f}_{h,\varepsilon,t}\rangle$).
\end{lem}
The function $\tilde{z}_{h,\varepsilon,t}$ is the solution of $i\partial_{t}\tilde{z}_{h,\varepsilon,t}=\partial_{\bar{z}}\tilde{Q}_{h,\varepsilon,t}^{app}(\tilde{z}_{h,\varepsilon,t})=\tilde{f}_{h,\varepsilon,t}$,
with initial condition~$\tilde{z}_{h,\varepsilon,0}=0$
\begin{proof}
Indeed \begin{align*}
i\varepsilon\partial_{t}\tilde{\hat{\Psi}}_{t} & =i\varepsilon\partial_{t}[e^{i\frac{t}{\varepsilon}\xi^{.2}}e^{i\frac{t}{\varepsilon}\diff\Gamma_{\!\varepsilon}(\varepsilon\eta^{.2}-2\xi.\eta)}\hat{\Psi}_{t}]\\
 & =e^{i\frac{t}{\varepsilon}\xi^{.2}}e^{i\frac{t}{\varepsilon}\diff\Gamma_{\!\varepsilon}(\varepsilon\eta^{.2}-2\xi.\eta)}[2\Re\langle z,f\rangle+\langle z,\eta z\rangle^{.2}]^{Wick}\hat{\Psi}_{t}\\
 & =[2\Re\langle z,e^{it(\varepsilon\eta^{.2}-2\xi.\eta)}f\rangle+\langle z,\eta z\rangle^{.2}]^{Wick}e^{i\frac{t}{\varepsilon}\xi^{.2}}e^{i\frac{t}{\varepsilon}d\Gamma_{\varepsilon}(\varepsilon\eta^{.2}-2\xi.\eta)}\hat{\Psi}_{t}\\
 & =\tilde{Q}_{t}^{Wick}\tilde{\hat{\Psi}}_{t}\,.\end{align*}
And we can proceed analogously with $\tilde{\Psi}_{t}^{app}$.
\end{proof}

\begin{proof}[Proof of Proposition~\ref{pro:Approximated_solution}]
Point~(\ref{enu:formule_Psi}) follows from applying $i\varepsilon\partial_{t}$
to the right hand side:\begin{align*}
i & \varepsilon\partial_{t}e^{-i\frac{\tilde{\omega}_{t}}{\varepsilon}}W\big({\textstyle \frac{\sqrt{2}}{i\varepsilon}}\tilde{z}_{t}\big)\,\Omega\otimes\hat{\psi}_{0}\\
 & =\Big(\partial_{t}\tilde{\omega}-i\varepsilon\frac{i\varepsilon}{2}\Im\big\langle{\textstyle \frac{\sqrt{2}}{i\varepsilon}}\tilde{z}_{t},-\tfrac{\sqrt{2}}{\varepsilon}\tilde{f}_{t}\big\rangle+i\varepsilon i\Phi_{\varepsilon}\big(-\tfrac{\sqrt{2}}{\varepsilon}\tilde{f}_{t}\big)\Big)e^{-i\frac{\omega_{t}}{\varepsilon}}W\big(\tfrac{\sqrt{2}}{i\varepsilon}\tilde{z}_{t}\big)\,\Omega\otimes\hat{\psi}_{0}\\
 & =\Big(\partial_{t}\tilde{\omega}-\Im\big\langle\tfrac{1}{i}\tilde{z}_{t},\tilde{f}_{t}\big\rangle+\sqrt{2}\Phi_{\varepsilon}(\tilde{f}_{t})\Big)\tilde{\hat{\Psi}}_{t}^{app}\end{align*}
since $\frac{1}{t}\langle\varphi,[W(z+tu)-W(z)]\psi\rangle\xrightarrow[t\to0]{}\langle\varphi,[-\frac{i\varepsilon}{2}\Im\langle z,u\rangle+i\Phi_{\varepsilon}(u)]W(z)\psi\rangle$.

For Point~(\ref{pro:estimation-z_t}) we compute \[
\big\||\eta|^{\nu}\tilde{z}_{h,\varepsilon,t}\big\|_{L_{\eta}^{2}}^{2}=\int_{0}^{t}\int_{0}^{t}\int_{\mathbb{R}_{\eta}^{d}}e^{i\frac{s-s'}{\varepsilon}(\varepsilon^{2}\eta^{.2}-2\xi.\varepsilon\eta)}|\eta|^{2\nu}|f_{h,\varepsilon}(\eta)|^{2}\diff\eta\diff s\diff s'\,.\]
Note that the internal integral is uniformly bounded by~$C_{G}\,\varepsilon^{-2\nu}\frac{h}{\varepsilon}$.
The change of variable $\eta^{\prime}=\varepsilon\eta-\xi$ gives\begin{multline*}
\int_{\mathbb{R}_{\eta}^{d}}e^{i\frac{s-s'}{\varepsilon}(\varepsilon^{2}\eta^{.2}-2\xi.\varepsilon\eta)}|\eta|^{2\nu}|f_{h,\varepsilon}(\eta)|^{2}\diff\eta\\
=\varepsilon^{-2\nu}\tfrac{h}{\varepsilon}e^{-i\frac{s-s'}{\varepsilon}\xi^{.2}}\int_{\mathbb{R}_{\eta}^{d}}e^{i\frac{s-s'}{\varepsilon}\eta^{.2}}|\eta+\xi|^{2\nu}\hat{G}(\eta+\xi)\diff\eta\end{multline*}
as $f_{h,\varepsilon}(\eta)=\varepsilon^{d/2}\sqrt{\frac{h}{\varepsilon}}\hat{V}(-\varepsilon\eta)$
and~$\frac{h}{\varepsilon}\,\hat{G}(\varepsilon\eta)\,\varepsilon^{d}=\left|f_{h,\varepsilon}(\eta)\right|^{2}$.
For~$s\neq s^{\prime}$\begin{align*}
\big|\int_{\mathbb{R}_{\eta}^{d}}e^{i\frac{s-s'}{\varepsilon}\eta^{.2}}|\eta+\xi|^{2\nu}\hat{G}(\eta+\xi)\diff\eta\big| & =\big(\tfrac{\pi\varepsilon}{s'-s}\big)^{d/2}\big\|\mathcal{F}(\eta\mapsto|\eta+\xi|^{2\nu}\hat{G}(\eta+\xi))\big\|_{L^{1}}\\
 & =\big(\tfrac{\pi\varepsilon}{s'-s}\big)^{d/2}\big\|\mathcal{F}(\eta\mapsto|\eta|^{2\nu}\hat{G}(\eta))\big\|_{L^{1}}\end{align*}
The bound\begin{align*}
\big\||\eta|^{\nu}\tilde{z}_{h,\varepsilon,t}\big\|_{L_{\eta}^{2}}^{2} & \leq C_{G}\tfrac{h}{\varepsilon}\varepsilon^{-2\nu}\int_{0}^{t}\int_{0}^{t}\min\big\{\big(\tfrac{\pi\varepsilon}{s'-s}\big)^{d/2},1\big\}\diff s\diff s'\\
 & \leq C_{G}\tfrac{h}{\varepsilon}\varepsilon^{-2\nu}\big[\pi^{d/2}\varepsilon^{d/2}\int_{|s-s'|\geq2\delta,s,s^{\prime}\in[0,t]}\tfrac{\diff s\diff s'}{(s'-s)^{d/2}}+2\sqrt{2}t\delta\big]\\
 & \leq C_{G}\tfrac{h}{\varepsilon}\varepsilon^{-2\nu}\big[\pi^{d/2}\varepsilon^{d/2}2^{d/4}2\sqrt{2}t\tfrac{2}{d-2}\delta^{1-d/2}+2\sqrt{2}t\delta\big]\end{align*}
is optimal when~$\delta=\varepsilon$.

For Point~(\ref{enu:llPsi-Psi_appll}), let~$\Delta\tilde{Q}_{t}\left(z\right)=\left\langle z,\eta z\right\rangle ^{.2}$.
First we remark that\[
\Delta\tilde{\hat{\Psi}}_{h,\varepsilon,t}=-\frac{i}{\varepsilon}\int_{0}^{t}e^{-i\frac{t-s}{\varepsilon}\tilde{Q}_{h,\varepsilon}^{Wick}}\Delta\tilde{Q}^{Wick}\tilde{\hat{\Psi}}_{h,\varepsilon,s}^{app}\diff s\,.\]
Since $i\varepsilon\partial_{t}\Delta\tilde{\hat{\Psi}}_{t}=\tilde{Q}^{Wick}\Delta\tilde{\hat{\Psi}}_{t}+\Delta\tilde{Q}^{Wick}\tilde{\hat{\Psi}}_{t}^{app}$
and that the integral expression on the right satisfies the same differential
equation. The difference $\Delta\tilde{\hat{\Psi}}_{h,\varepsilon,t}$
can then be controlled as\[
\big\|\Delta\tilde{\hat{\Psi}}_{h,\varepsilon,t}\big\|\leq\frac{1}{\varepsilon}\int_{0}^{t}\big\|\Delta\tilde{Q}^{Wick}\, E(\tilde{z}_{h,\varepsilon,s})\big\|\diff s\,.\]
The relation~$\langle E(z),R^{Wick}E(z)\rangle=R(z)$ with $R^{Wick}=(\Delta\tilde{Q}^{Wick})^{*}\Delta\tilde{Q}^{Wick}$
gives\begin{multline*}
\Symb\!{}^{Wick}\big(\big[(\langle z,\eta z\rangle^{.2})^{Wick}\big]^{2}\big)\\
=(\langle z,\eta z\rangle^{.2})^{2}+4\varepsilon(\langle z,\eta z\rangle.\langle\eta z|)(|\eta z\rangle.\langle z,\eta z\rangle)+2\varepsilon^{2}(\langle\eta z|^{.\otimes2})(|\eta z\rangle^{.\otimes2})\,,\end{multline*}
using the estimate in Point~(\ref{pro:estimation-z_t}), we obtain
that\[
\big\|\Delta\tilde{Q}^{Wick}\, E(\tilde{z}_{h,\varepsilon,t})\big\|^{2}\leq C_{T_{0},G,d}\big((\tfrac{ht}{\varepsilon})^{4}+4\varepsilon\left(\tfrac{ht}{\varepsilon}\right)^{2}\tfrac{ht}{\varepsilon^{2}}+2\varepsilon^{2}(\tfrac{ht}{\varepsilon^{2}})^{2}\big)\]
which gives the result for~$\frac{ht}{\varepsilon}\leq T_{0}$.

For Point~(\ref{enu:Number_estimate_Psi}), let~$\gamma_{t}=\|(\varepsilon+N_{\varepsilon})^{1/2}\hat{\Psi}_{t}^{\sharp}\|$,
then\[
i\varepsilon\partial_{t}(\gamma_{t}^{2})=\big\langle\hat{\Psi}_{t}^{\sharp},[\Phi_{\varepsilon}(f_{h,\varepsilon}),N_{\varepsilon}]\hat{\Psi}_{t}^{\sharp}\big\rangle\]
with~$f_{h,\varepsilon}=\sqrt{\frac{h}{\varepsilon}}\varepsilon^{d/2}\overline{\hat{V}(\varepsilon\eta)}$,
since~$\xi$ and~$\diff\Gamma_{\!\varepsilon}(\eta)$ commute with~$N_{\varepsilon}=\diff\Gamma_{\!\varepsilon}(Id)$.
We get\[
\left[a_{\varepsilon}(f_{h,\varepsilon}),\diff\Gamma_{\!\varepsilon}(1)\right]=i\partial_{s}\left.\left[\Gamma(e^{i\varepsilon s})\, a_{\varepsilon}(f_{h,\varepsilon})\,\Gamma(e^{-i\varepsilon s})\right]\right|_{s=0}=a_{\varepsilon}(\varepsilon f_{h,\varepsilon})\,.\]
The other term of the commutator can be computed analogously (but~$a_{\varepsilon}(\cdot)$
is~$\mathbb{C}$-antilinear whereas~$a_{\varepsilon}^{*}(\cdot)$
is~$\mathbb{C}$-linear). Introducing this relation into the differential
equation and taking the modulus, we get\[
\left|i\varepsilon\partial_{t}(\gamma_{t}^{2})\right|\leq\sqrt{2}^{-1}\big\|\hat{\Psi}_{t}^{\sharp}\big\|\,\big(\big\| a_{\varepsilon}(\varepsilon f_{h,\varepsilon})\,\hat{\Psi}_{t}^{\sharp}\big\|+\big\| a_{\varepsilon}^{*}(\varepsilon f_{h,\varepsilon})\,\hat{\Psi}_{t}^{\sharp}\big\|\big)\,.\]
But\[
\big\| a_{\varepsilon}^{*}(\varepsilon f_{h,\varepsilon})\,\hat{\Psi}_{t}^{\sharp}\big\|^{2}\leq\|\varepsilon f_{h,\varepsilon}\|_{L_{\xi}^{2}}^{2}\big\langle\hat{\Psi}_{t}^{\sharp},(\varepsilon+N_{\varepsilon})\hat{\Psi}_{t}^{\sharp}\big\rangle\]
and the same estimate holds for annihilation operators. Using~$\|\hat{G}\|_{L^{1}}=\frac{h}{\varepsilon}\|f_{h,\varepsilon}\|_{L_{\xi}^{2}}^{2}$,
we finally get a differential inequality for the function~$\gamma_{t}$\[
2\varepsilon\gamma_{t}\partial_{t}\gamma_{t}\leq\left|i\varepsilon\partial_{t}\left(\gamma_{t}^{2}\right)\right|\leq\sqrt{2\varepsilon h\|\hat{G}\|_{L^{1}}}\gamma_{t}\,.\]
The result follows by dividing by~$2\varepsilon\gamma_{t}$ and integrating
in time, since~$\gamma_{0}=C_{d}\sqrt{\varepsilon}$.
\end{proof}

\section[Measure of an observable for the approximated dynamics]{\label{par:Calculus-approximated}Measure of an observable at a
mesoscopic scale for the approximated dynamics}

\subsection{Result}

In this section we make the connection between the microscopic dynamic
and the linear Boltzmann equation.
\begin{prop}
\label{pro:Result-calculus-approximation}Let~$\alpha\in[0,1)$ and
assume~$h^{\alpha}\leq\frac{ht}{\varepsilon}\leq1$. Let~$b\in\mathcal{C}_{0}^{\infty}(\mathbb{R}_{x}^{d}\times\mathbb{R}_{\xi}^{d*})$
and~$\rho\in\mathcal{L}_{1}^{+}L_{x}^{2}$, $\Tr\rho\leq1$ such
that the kernel of $\hat{\rho}=\Ad\left\{ \mathcal{F}_{x}\right\} \left[\rho\right]$
has a bounded support. Introduce the symbol~$b_{t}=e^{tQ}e^{2t\xi.\partial_{x}}b$
where~$Q$ is the collision operator introduced in Equation~\ref{eq:noyau-de-collision}
with here~$\sigma(\xi,\xi^{\prime})=2\pi\hat{G}(\xi^{\prime}-\xi)=2\pi|\hat{V}(\xi-\xi^{\prime})|^{2}$.
The inequality\[
m_{h}(b,\rho_{t}^{\varepsilon,app})\geq m_{h}\big(b_{\frac{ht}{\varepsilon}},\rho\big)-\mathcal{E}_{\ref{par:Calculus-approximated}}\]
then holds with~$\mathcal{E}_{\ref{par:Calculus-approximated}}=C_{b,\mu}\frac{ht}{\varepsilon}\big(\frac{ht}{\varepsilon}+h+\big[h(\frac{ht}{\varepsilon})^{-1}\big]^{d/2-1}+h^{\mu(d,\alpha)}\big)$
for some constant~$C_{b,\mu}>0$ and~$\mu\left(d,\alpha\right)>0$.\end{prop}
\begin{rem}
\label{rem:result-calculus-approximation-indep-x}This result also
holds with~$b$ a symbol in~$\mathcal{C}_{0}^{\infty}(\mathbb{R}_{\xi}^{d*};\mathbb{C})$.
The proof is the same as for Proposition~\ref{pro:Result-calculus-approximation},
with the symplectic Fourier transform~$\mathcal{F}^{\sigma}$ replaced
by the usual Fourier transform. The special case when~$b\left(\xi\right)=b_{1}(|\xi|^{2})$
is of particular interest and the symbol~$b_{t}$ in the previous
statement does not depend on~$t$.
\end{rem}
Proposition~\ref{pro:Result-calculus-approximation} is a by-product
of the following stronger result.
\begin{prop}
\label{pro:side-result-calculus-approx}Let~$b_{s}\in\mathcal{C}^{1}(\mathbb{R};\mathcal{C}_{0}^{\infty}(\mathbb{R}_{x,\xi}^{2d}))$
such that for some~$R>1$, and for all~$s$, $\Supp_{\xi}b_{s}\subset B_{R}\setminus B_{R^{-1}}$.
Let~$\rho\in\mathcal{L}_{1}^{+}L_{x}^{2}$, $\Tr\rho\leq1$ such
that the kernel of~$\hat{\rho}=\Ad\{\mathcal{F}_{x}\}[\rho]$ has
a bounded support. Then\begin{multline*}
m_{h}\big(b_{\frac{ht}{\varepsilon}},\rho_{t}^{\varepsilon,app}\big)\\
\geq m_{h}(b,\rho)-\tfrac{i}{\varepsilon}\int_{0}^{t}\negmedspace m_{h}\big(i\varepsilon\partial_{s}b_{s}-ih\{b_{s},\xi^{.2}\}+ihQ_{-\frac{ht}{\varepsilon}}b_{s},\rho_{s}^{\varepsilon,app}\big)\diff s-\mathcal{E}_{\ref{par:Calculus-approximated}}\,.\end{multline*}
\end{prop}
\begin{rem}
The conservation of the support in~$\xi$ is important and is provided
by the properties of the dual linear Boltzmann equation in the application
of this proposition.\end{rem}
\begin{proof}[Proof that Proposition~\ref{pro:side-result-calculus-approx} implies
Proposition~\ref{pro:Result-calculus-approximation}.]
Since one can make mistakes between the notations of those two propositions
we use notations with tildes, $\tilde{b}$ for Proposition~\ref{pro:Result-calculus-approximation}
and without tildes for Proposition~\ref{pro:side-result-calculus-approx}.
Thus we want\[
\tilde{b}=b_{\frac{ht}{\varepsilon}}\,,\qquad\tilde{b}_{\frac{ht}{\varepsilon}}=b\,.\]
Denote by~$\tilde{G}(t,t_{0})$ the dynamical system associated with~$(-2\xi.\partial_{x}-Q_{-t})_{t}$
given by\[
\left\{ \begin{aligned}\partial_{t}b_{t} & =\left(-2\xi.\partial_{x}-Q_{-t}\right)b_{t}\\
b_{t=t_{0}} & =b_{0}\end{aligned}
\right.\,,\qquad b_{t}=\tilde{G}(t,t_{0})\, b_{0}\,.\]
To have a vanishing term for~$b$ in the integral we require~$b_{ht/\varepsilon}=\tilde{G}(\frac{ht}{\varepsilon},0)b$,
so that with~$\tilde{b}_{ht/\varepsilon}=\tilde{G}(0,-\frac{ht}{\varepsilon})\tilde{b}$,
we will get the expected result. The only thing remaining to prove
is~$\tilde{G}(0,-t)=e^{tQ}e^{2t\xi.\partial_{x}}\,.$ It is equivalent
to show that~$e^{2t\xi.\partial_{x}}\tilde{G}\left(t,0\right)=e^{-tQ}$,
which is clear by derivation and using that~$Q_{t}=e^{t2\xi.\partial_{x}}Qe^{-t2\xi.\partial_{x}}$.
\end{proof}

\subsection{Expression of the measure of an observable for the approximated equation}

We carry out an explicit computation using only the approximated equation.
\begin{notation}
Let~$\sigma(X_{1},X_{2})=\xi_{1}.x_{2}-x_{1}.\xi_{2}$ ($X_{j}=(x_{j},\xi_{j})\in\mathbb{R}_{x,\xi}^{2d}$)
be the standard symplectic form on~$\mathbb{R}_{x,\xi}^{2d}$.

Let~$X^{\prime}=(x^{\prime},\xi^{\prime})\in\mathbb{R}_{x,\xi}^{2d}$,
\index{Weyl operator2@$\tau_{X}^{h}$, Weyl operator}the Weyl operators
on~$L_{x}^{2}$ are defined by\[
\tau_{X'}^{h}=\big(e^{-i\sigma(\cdot,X')}\big)^{W}(hx,D_{x})=e^{-i\sigma(\cdot,X')^{W}(hx,D_{x})}=e^{i(\xi'\cdot hx-x'\cdot D_{x})}\,,\]
their Fourier transform is denoted by~$\hat{\tau}_{P}^{h}:=\Ad\left\{ \mathcal{F}_{x}\right\} \left[\tau_{P}^{h}\right]$.\index{Weyl operator3@$\hat{\tau}_{P}^{h}$, Weyl operator}
Note that the formula\[
\hat{\tau}_{X_{1}}^{h}\hat{\tau}_{X_{2}}^{h}=e^{\frac{i}{2}h\sigma(X_{1},X_{2})}\hat{\tau}_{X_{1}+X_{2}}^{h}=e^{ih\sigma(X_{1},X_{2})}\hat{\tau}_{X_{2}}^{h}\hat{\tau}_{X_{1}}^{h}\]
holds.

The symplectic Fourier transform\index{Fourier transform, symplectic@$\mathcal{F}^{\sigma}$, symplectic
Fourier transform}~$\mathcal{F}^{\sigma}$ on~$L^{2}(\mathbb{R}_{x,\xi}^{2d};\mathbb{C})$
is, with~$\dbar X=\diff X/(2\pi)^{d}$\index{dbar@$\dbar X$}, \[
\mathcal{F}^{\sigma}b(X)=\int_{\mathbb{R}^{2d}}e^{-i\sigma(X,X')}b(X')\dbar X'\,.\]
Note that $(\mathcal{F}^{\sigma})^{-1}=\mathcal{F}^{\sigma}$.\end{notation}
\begin{prop}
\label{pro:integral-expression-for-m}Let~$b$ be a symbol in~$\mathcal{C}_{0}^{\infty}(\mathbb{R}_{x,\xi}^{2d})$
and~$\rho\in\mathcal{L}_{1}^{+}L_{x}^{2}$, $\Tr\rho\leq1$, then\[
m_{h}(b,\rho_{t}^{\varepsilon,app})\,=\iiint\mathcal{F}^{\sigma}b(P)\, K_{P}(\xi_{1},\xi_{2})\,\hat{\tau}_{P}^{h}(\xi_{2},\xi_{1})\,\diff\xi_{1}\diff\xi_{2}\,\dbar P\]
where\index{omega2@$\left[\omega\right]$}\index{phi1@$\left[\varphi\right]_{1}$}\index{phi2@$\left[\varphi,p_{x}\right]_{2}$}
$K_{P}(\xi_{1},\xi_{2})=e^{-i\frac{\omega_{h,\varepsilon,t}^{\xi_{1}}-\omega_{h,\varepsilon,t}^{\xi_{2}}}{\varepsilon}}\hat{\rho}(\xi_{1},\xi_{2})e^{-\tfrac{1}{2\varepsilon}\big(|z_{h,\varepsilon,t}^{\xi_{1}}|^{2}+|z_{h,\varepsilon,t}^{\xi_{2}}|^{2}+2\langle z_{h,\varepsilon,t}^{\xi_{2}},e^{ip_{x}\cdot\varepsilon\eta}z_{h,\varepsilon,t}^{\xi_{1}}\rangle\big)}$.

\label{pro:id_m=00003Dm+ih...}Let~$b_{t}\in\mathcal{C}^{1}(\mathbb{R};\mathcal{C}_{0}^{\infty}(\mathbb{R}_{x,\xi}^{2d}))$,
then\[
i\varepsilon\partial_{t}m_{h}(b_{t},\rho_{t}^{\varepsilon,app})=m_{h}(i\varepsilon\partial_{t}b_{t},\rho_{t}^{\varepsilon,app})+ih\left(m_{\left\{ ,\right\} }-m_{-}+m_{+}\right)\,.\]
where, for~$\kappa=\left\{ ,\right\} ,-,+$ we define\index{m2@$m_{\left\{ \cdot,\cdot\right\} },\, m_{+},\, m_{-}$}\[
m_{\kappa}=\int_{\mathbb{R}_{P}^{2d}}\mathcal{F}^{\sigma}b(P)\Tr\left[\hat{\boldsymbol{\rho}}_{t}^{app}\,\Gamma(e^{ip_{x}\cdot\varepsilon\eta})\,\mathscr{A}_{\kappa,P}\right]\dbar P\]
with the operators~$\mathscr{A}_{\kappa,P}$ defined by their kernels,\index{A@$\mathscr{A}_{\left\{ ,\right\} },\,\mathscr{A}_{-},\,\mathscr{A}_{+}$}
for $j=1,2$, by\begin{align}
\mathscr{A}_{\left\{ ,\right\} ,P} & =\mathscr{A}_{\left\{ ,\right\} ,P}^{1}-\mathscr{A}_{\left\{ ,\right\} ,P}^{2}\,, & ih\,\mathscr{A}_{\left\{ ,\right\} ,P}^{j}(\xi_{1},\xi_{2}) & =\hat{\tau}_{P}^{h}(\xi_{2},\xi_{1})\,\partial_{t}\omega_{t}^{\xi_{j}}\,,\label{eq:def-operateur-A{}P}\\
\mathscr{A}_{-,P} & =\mathscr{A}_{-,P}^{1}+\mathscr{A}_{-,P}^{2}\,, & ih\,\mathscr{A}_{-,P}^{j}(\xi_{1},\xi_{2}) & =\hat{\tau}_{P}^{h}(\xi_{2},\xi_{1})\, i\partial_{t}\tfrac{1}{2}|z_{t}^{\xi_{j}}|^{2}\,,\label{eq:def-operateur-A-P}\\
 &  & ih\,\mathscr{A}_{+,P}(\xi_{1},\xi_{2}) & =\hat{\tau}_{P}^{h}(\xi_{2},\xi_{1})\, i\partial_{t}\left[\varphi,p_{x}\right]_{2}\,,\label{eq:def-operateur-A+P}\end{align}
with~$[\varphi,p_{x}]_{2}=\langle z_{t}^{\xi_{2}},e^{ip_{x}\cdot\varepsilon\eta}z_{t}^{\xi_{1}}\rangle$.

\end{prop}
The indexes~$\left\{ ,\right\} $, $-$ and~$+$ are chosen to recall
the terms of the linear Boltzmann equation, $\left\{ ,\right\} $
corresponding to~$\{\xi^{2},\cdot\}$, $+$ to~$Q_{+}$ and~$-$
to~$Q_{-}$.
\begin{rem}
Each of those terms~$m_{\kappa}$ is shown in the sequel to be of
the form $m_{\kappa}=m(c_{\kappa},\rho_{t}^{\varepsilon,app})+\Delta_{\kappa}$
where~$\Delta_{\kappa}$ denotes a {}``small'' error term.

\end{rem}
\begin{proof}
Since~$b^{W}(hx,D_{x})=\int\mathcal{F}^{\sigma}b(P)\,\tau_{P}^{h}\,\dbar P$,
we have for~$\rho\in\mathcal{L}_{1}^{+}$\[
m_{h}(b,\rho)=\int\mathcal{F}^{\sigma}b(P)\Tr\negthinspace\left[\tau_{P}^{h}\,\rho\right]\,\dbar P\,.\]
From~$e^{i\varepsilon x.\lambda}\,\tau_{P}^{h}\, e^{-i\varepsilon x.\lambda}=e^{i\varepsilon\lambda.p_{x}}\,\tau_{P}^{h}$
and taking~$\lambda$ as the spectral parameter of~$\diff\Gamma_{\!\varepsilon}(D_{y})$,
$\Gamma(e^{i\varepsilon x.D_{y}})\,\tau_{P}^{h}\,\Gamma(e^{-i\varepsilon x.D_{y}})=\Gamma(e^{ip_{x}\cdot\varepsilon D_{y}})\,\tau_{P}^{h}$
and after conjugating with the Fourier transforms, we obtain\[
\Ad\left\{ \left(\mathcal{F}_{x}\otimes\Gamma\mathcal{F}_{y}\right)\Gamma(e^{i\varepsilon xD_{y}})\right\} \big[\tau_{P}^{h}\big]=\Gamma(e^{ip_{x}\cdot\varepsilon\eta})\,\hat{\tau}_{P}^{h}\,.\]
Thus, by translating and Fourier transforming we get the expression\[
m_{h}(b,\rho_{t}^{\varepsilon,app})=\int\mathcal{F}^{\sigma}b(P)\Tr\negthinspace\left[\hat{\boldsymbol{\rho}}_{t}^{app}\Gamma(e^{ip_{x}\cdot\varepsilon\eta})\,\hat{\tau}_{P}^{h}\right]\dbar P\,.\]
It then remains to compute the kernel~$K_{P}$ of the operator~$\Tr_{\Gamma L_{\eta}^{2}}[\hat{\boldsymbol{\rho}}_{t}^{app}\,\Gamma(e^{ip_{x}\cdot\varepsilon\eta})]$
on~$L_{\xi}^{2}$. Using~$\hat{\rho}\otimes\left|\Omega\right\rangle \left\langle \Omega\right|=\int_{\xi_{1}}^{\oplus}\int_{\xi_{2}}^{\oplus}\hat{\rho}(\xi_{1},\xi_{2})\,\left|\Omega\right\rangle \left\langle \Omega\right|\diff\xi_{1}\diff\xi_{2}$
we get\begin{multline*}
\Tr_{\Gamma L_{\eta}^{2}}\big[\hat{\boldsymbol{\rho}}_{t}^{app}\,\Gamma(e^{ip_{x}\cdot\varepsilon\eta})\big]\\
=\Tr_{\Gamma L_{\eta}^{2}}\Big[\int_{\mathbb{R}_{\xi_{1}}^{d}}^{\oplus}\int_{\mathbb{R}_{\xi_{2}}^{d}}^{\oplus}\big|E(z_{t}^{\xi_{1}})\big\rangle\big\langle E(z_{t}^{\xi_{2}})\big|e^{-i\frac{\omega_{t}^{\xi_{1}}}{\varepsilon}}e^{i\frac{\omega_{t}^{\xi_{2}}}{\varepsilon}}\hat{\rho}(\xi_{1},\xi_{2})\diff\xi_{1}\diff\xi_{2}\Gamma(e^{ip_{x}\cdot\varepsilon\eta})\Big]\end{multline*}
and we obtain the kernel\begin{align*}
K_{P}(\xi_{1},\xi_{2}) & =e^{-i\frac{\omega_{t}^{\xi_{1}}-\omega_{t}^{\xi_{2}}}{\varepsilon}}\hat{\rho}(\xi_{1},\xi_{2})\big\langle E(z_{t}^{\xi_{2}})\big|\,\Gamma(e^{ip_{x}\cdot\varepsilon\eta})\,\big|E(z_{t}^{\xi_{1}})\big\rangle\end{align*}
which brings the expected expression using the calculus on coherent
states.

For the formula for the derivative\begin{multline*}
i\varepsilon\partial_{t}m_{h}(b,\rho_{t}^{\varepsilon,app})=\iiint\Big[\mathcal{F}^{\sigma}i\varepsilon\partial_{t}b(P)\\
+\mathcal{F}^{\sigma}b(P)\big\{\partial_{t}\big(\omega_{t}^{\xi_{1}}-\omega_{t}^{\xi_{2}}\big)-i\tfrac{1}{2}\partial_{t}\big(|z_{t}^{\xi_{1}}|^{2}+|z_{t}^{\xi_{2}}|^{2}\big)+i\partial_{t}[\varphi,p_{x}]_{2}\big\}\Big]\\
K_{P}(\xi_{1},\xi_{2})\hat{\tau}_{P}^{h}(\xi_{2},\xi_{1})\diff\xi_{1}\diff\xi_{2}\,\dbar P\end{multline*}
and so it suffices to observe that for~$\kappa=\left\{ ,\right\} $,
$-$, $+$,\begin{align*}
\Tr\big[ & \hat{\boldsymbol{\rho}}_{t}^{app}\,\Gamma(e^{ip_{x}\cdot\varepsilon\eta})\,\mathscr{A}_{\kappa,P}\big]\\
 & =\iint\hat{\rho}(\xi_{1},\xi_{2})\big\langle E(z_{t}^{\xi_{2}})\big|\,\Gamma(e^{ip_{x}\cdot\varepsilon\eta})\,\big|E(z_{t}^{\xi_{1}})\big\rangle e^{-i\frac{\omega_{t}^{\xi_{1}}-\omega_{t}^{\xi_{2}}}{\varepsilon}}\mathscr{A}_{\kappa,P}(\xi_{1},\xi_{2})\diff\xi_{1}\diff\xi_{2}\\
 & =\iint\mathscr{A}_{\kappa,P}(\xi_{1},\xi_{2})\, K_{P}(\xi_{1},\xi_{2})\diff\xi_{1}\diff\xi_{2}.\end{align*}
which is the expected result.
\end{proof}

\subsection{Two estimates}

We need estimates to get rid of the term~$\Gamma(e^{ip_{x}\cdot\varepsilon\eta})$
and control errors on the operators~$\mathcal{A}_{P}$.
\begin{prop}
\label{pro:Gamma-approx-Id}Let~$\mathscr{A}_{P}$ be a $P$-dependent
family of operators in~$\mathcal{L}(L_{\xi}^{2})$. Then there exists
a constant~$C_{G,d}$ such that\[
\langle P\rangle^{-k}\,\big|\Tr\big[\hat{\boldsymbol{\rho}}_{t}^{app}\,\big(\Gamma(e^{ip_{x}\cdot\varepsilon\eta})-\Id\big)\,\mathscr{A}_{P}\big]\big|\leq C_{G,d}\tfrac{ht}{\varepsilon}\sup_{P\in\mathbb{R}^{2d}}\langle P\rangle^{-k}\|\mathscr{A}_{P}\|_{\mathcal{L}(L_{\xi}^{2})}\]
and\begin{multline*}
\Big|\int_{\mathbb{R}_{P}^{2d}}\mathcal{F}^{\sigma}b(P)\Tr\big[\hat{\boldsymbol{\rho}}_{t}^{app}\,\big(\Gamma(e^{ip_{x}\cdot\varepsilon\eta})-\Id\big)\,\mathscr{A}_{P}\big]\dbar P\Big|\\
\leq C_{G,d}\tfrac{ht}{\varepsilon}\big\|\left\langle \cdot\right\rangle ^{k}\mathcal{F}^{\sigma}b\big\|_{L_{P}^{1}}\sup_{P}\langle P\rangle^{-k}\big\|\mathscr{A}_{P}\big\|_{\mathcal{L}\left(L_{\xi}^{2}\right)}\,.\end{multline*}

\end{prop}
This can be proved in two steps.
\begin{rem}
It suffices to prove this property with $\rho=\left|\psi\right\rangle \left\langle \psi\right|$
with a $\hat{\psi}$ with bounded support as any~$\rho\in\mathcal{L}_{1}^{+}L_{x}^{2}$,
$\Tr\rho=1$ the decomposition $\rho=\sum_{j\geq0}\lambda_{j}|\psi_{j}\rangle\langle\psi_{j}|$
holds with positive~$\lambda_{j}$'s and~$\sum_{j}\lambda_{j}=1$,
and\[
\Supp\hat{\rho}(\xi,\xi^{\prime})\subset B_{M}^{2}\Leftrightarrow\forall j,\:\Supp\hat{\psi}_{j}\subset B_{M}\,.\]
\end{rem}
\begin{proof}
\label{lem:Gamma-Id}For~$\hat{\Psi}$ be a normed vector in~$L_{\xi}^{2}\otimes\Gamma L_{\eta}^{2}$\[
\big|\Tr\big[|\hat{\Psi}\rangle\langle\hat{\Psi}|\big(\Gamma(e^{ip_{x}\cdot\varepsilon\eta})-\Id\big)\mathscr{A}_{P}\big]\big|\leq\big\|\big(\Gamma(e^{ip_{x}\cdot\varepsilon\eta})-\Id\big)\hat{\Psi}\big\|\,\big\|\mathscr{A}_{P}\big\|_{\mathcal{L}(L_{\xi}^{2})}.\]

For~$\hat{\Psi}=\hat{\Psi}_{h,\varepsilon,t}^{app}$ associated with
$\psi$, the calculus on coherent states gives\begin{multline*}
\big\|\big(\Gamma(e^{ip_{x}\cdot\varepsilon\eta})-\Id\big)\hat{\Psi}_{h,\varepsilon,t}^{app}\big\|^{2}=\sup_{\xi}\big\| E(e^{ip_{x}\cdot\varepsilon\eta}z_{h,\varepsilon,t}^{\xi})-E(z_{h,\varepsilon,t}^{\xi})\big\|^{2}\\
=\sup_{\xi}2\big(1-\cos\big(\tfrac{1}{\varepsilon}\Im\langle e^{ip_{x}.\varepsilon\eta}z_{h,\varepsilon,t}^{\xi},z_{h,\varepsilon,t}^{\xi}\rangle\big)\big)\leq C_{G,d}^{2}(\tfrac{ht}{\varepsilon})^{2}\,,\end{multline*}
where the inequality follows form~$\left|1-\cos t\right|\leq t^{2}/2$
and the estimates on~$\|z_{t}\|$. We then get the second result
by an integration.\end{proof}
\begin{prop}
\label{pro:controle-erreur-sans-gamma}Let~$\mathscr{E}_{P}$ be
a $P$-dependent family of operators in~$\mathcal{L}(L_{\xi}^{2})$
and~$\hat{\boldsymbol{\rho}}$ be a state on~$L_{\xi}^{2}\otimes\Gamma L_{\eta}^{2}$.
Then for any integer~$k$ (with possibly infinite quantities)\[
\Big|\int_{\mathbb{R}_{P}^{2d}}\mathcal{F}^{\sigma}b(P)\left|\Tr\left[\hat{\boldsymbol{\rho}}\,\mathscr{E}_{P}\right]\right|\dbar P\Big|\leq\big\|\!\left\langle \cdot\right\rangle ^{k}\mathcal{F}^{\sigma}b\big\|_{L_{P}^{1}}\,\sup_{P}\big\|\left\langle P\right\rangle ^{-k}\mathscr{E}_{P}\big\|_{\mathcal{L}\left(L_{\xi}^{2}\right)}\,.\]

\end{prop}

\subsection{The transport term~$m_{\left\{ ,\right\} }$}

The result of this section is the following.
\begin{prop}
\label{pro:term1}Let~$\rho\in\mathcal{L}_{1}^{+}L_{x}^{2}$, $\Tr\rho\leq1$
and~$b\in\mathcal{C}_{0}^{\infty}(\mathbb{R}_{x,\xi}^{2d})$ such
that $\Supp\hat{\rho}(\xi,\xi^{\prime})\subset B_{R}^{2}$ , and~$\Supp_{\xi}b\subset B_{R}$
for some~$R>0$ then\index{Delta1@$\Delta_{\left\{ \cdot,\cdot\right\} }$}
\[
m_{\left\{ ,\right\} }=m(-\{b,\xi^{.2}\},t)\big)+\Delta_{\left\{ ,\right\} }\]
with~$|\Delta_{\left\{ ,\right\} }|\leq C_{G,R,b}\,(\frac{ht}{\varepsilon}+h+(\frac{\varepsilon}{t})^{d/2})$.\end{prop}
\begin{rem}
We can introduce a cutoff function~$\chi_{R}\in\mathcal{C}_{0}^{\infty}(\mathbb{R}_{\xi}^{d})$
such that $\chi_{R}(B_{R})=\left\{ 1\right\} $, $\chi_{R}(\mathbb{R}_{\xi}^{d}\setminus B_{R+1})=\left\{ 0\right\} $
and~$\chi_{R}(\mathbb{R}_{\xi}^{d})\subset\left[0,1\right]$.
\end{rem}
Proposition~\ref{pro:term1} is proved by doing a succession of approximations.
The error terms~$\Delta_{\left\{ ,\right\} ,j}$, $j=1,2,3$ are
given by the approximation process (where we write shortly~$b^{W}$
for~$b^{W}(-hD_{\xi},\xi)$)\begin{align*}
m_{\left\{ ,\right\} } & =\int\mathcal{F}^{\sigma}b(P)\Tr\Big[\hat{\boldsymbol{\rho}}_{t}^{app}\,\Gamma(e^{ip_{x}\cdot\varepsilon\eta})\,\tfrac{1}{ih}\left[\hat{\tau}_{P}^{h}\,,\chi_{R}\partial_{t}\omega\times\right]\Big]\dbar P\\
 & =\Tr\Big[\hat{\boldsymbol{\rho}}_{t}^{app}\,\tfrac{1}{ih}\left[b^{W}\,,\chi_{R}\partial_{t}\omega\times\right]\Big]\dbar P+\Delta_{\left\{ ,\right\} ,1}\displaybreak[0]\\
 & =\int\mathcal{F}^{\sigma}\big(-\{b\,,\xi^{.2}\}\big)\left(P\right)\Tr\Big[\hat{\boldsymbol{\rho}}_{t}^{app}\,\hat{\tau}_{P}^{h}\Big]\dbar P+{\textstyle \sum_{j=1}^{2}}\Delta_{\left\{ ,\right\} ,j}\displaybreak[0]\\
 & =m\big(-\{b\,,\xi^{.2}\},t\big)+{\textstyle \sum_{j=1}^{3}}\Delta_{\left\{ ,\right\} ,j}\,.\end{align*}
where we used that~$\mathscr{A}_{\left\{ ,\right\} ,P}=\frac{1}{ih}[\hat{\tau}_{P}^{h}\,,\,\partial_{t}\omega\times]$
and where the quantities~$\Delta_{\left\{ ,\right\} ,j}$ are defined
by\begin{align*}
\Delta_{\left\{ ,\right\} ,1} & =\int\mathcal{F}^{\sigma}b(P)\,\Tr\Big[\hat{\boldsymbol{\rho}}_{t}^{app}\left(\Gamma(e^{ip_{x}\cdot\varepsilon\eta})-\Id\right)\tfrac{1}{ih}\left[\hat{\tau}_{P}^{h}\,,\chi_{R}\partial_{t}\omega\times\right]\Big]\dbar P\,,\\
\Delta_{\left\{ ,\right\} ,2} & =\Tr\Big[\hat{\boldsymbol{\rho}}_{t}^{app}\,\tfrac{1}{ih}\big(\big[b\,,\chi_{R}\partial_{t}\omega\times\big]-\tfrac{h}{i}\{b\,,\chi_{R}\xi^{.2}\}^{W}\big)\Big]\dbar P\,,\\
\Delta_{\left\{ ,\right\} ,3} & =\int\mathcal{F}^{\sigma}\!\big(\!-\{b\,,\xi^{.2}\}\big)(P)\,\Tr\Big[\hat{\boldsymbol{\rho}}_{t}^{app}\big(\Id-\Gamma(e^{ip_{x}\cdot\varepsilon\eta})\big)\hat{\tau}_{P}^{h}\Big]\dbar P\,.\end{align*}

\begin{prop}
\label{pro:Delta{}=00003DsumDelta{}j}With the hypotheses and notations
of Proposition~\ref{pro:term1}, for some integer~$k$,
\begin{enumerate}
\item \label{enu:Delta{}1}$|\Delta_{\left\{ ,\right\} ,1}|\leq2\frac{ht}{\varepsilon}\|\left\langle \cdot\right\rangle ^{k}\mathcal{F}^{\sigma}b\|_{L_{P}^{1}}\,\mathcal{O}\big(1+h+[h(\frac{ht}{\varepsilon})^{-1}]^{d/2-1}\big)$,
\item \label{enu:Delta{}2}$|\Delta_{\left\{ ,\right\} ,2}|\leq\big(\left\Vert \mathcal{F}^{\sigma}b\right\Vert _{L_{P}^{1}}+\|\left\langle \cdot\right\rangle ^{k}\mathcal{F}^{\sigma}b\|_{L_{P}^{1}}\big)\,\mathcal{O}\big(h+(\frac{\varepsilon}{t})^{\frac{d}{2}-1}\big)$,
\item \label{enu:Delta{}3}$|\Delta_{\left\{ ,\right\} ,3}|\leq\frac{ht}{\varepsilon}\|\mathcal{F}^{\sigma}\{b,\xi^{.2}\}\|_{L_{P}^{1}}$.
\end{enumerate}
\end{prop}

\begin{proof}[Proof of Proposition~\ref{pro:Delta{}=00003DsumDelta{}j}.]
Point~\ref{enu:Delta{}1} is a result of Proposition~\ref{pro:Gamma-approx-Id}
and Lemma~\ref{lem:[]=00003D{}}.

For Point~\ref{enu:Delta{}2}\[
\Delta_{\left\{ ,\right\} ,2}=\int_{\mathbb{R}_{P}^{2d}}\mathcal{F}^{\sigma}b(P)\,\Tr\Bigl[\hat{\boldsymbol{\rho}}_{t}^{app}\,\frac{1}{ih}\Bigl(\bigl[\hat{\tau}_{P}^{h}\,,\chi_{R}\partial_{t}\omega_{h,\varepsilon,t}\times\bigr]-\frac{h}{i}\bigl\{\hat{\tau}_{P}^{h}\,,\chi_{R}\xi^{.2}\bigr\}^{W}\Bigr)\Bigr]\dbar P\]
so that Lemma~\ref{lem:[]=00003D{}} and Proposition~\ref{pro:controle-erreur-sans-gamma}
give the estimation.

Point~\ref{enu:Delta{}3} is an application of Proposition~\ref{pro:Gamma-approx-Id}.
\end{proof}

\begin{lem}
\label{lem:[]=00003D{}}We have, for some integer~$k$,\[
[\hat{\tau}_{P}^{h},\!\chi_{R}\partial_{t}\omega_{h,\varepsilon,t}\times]=-ih\{e^{i\sigma(P,X)},\chi_{R}\xi^{2}\}^{W}\!(-hD_{\xi},\xi)+h\mathcal{O}\big(\langle P\rangle^{k}h+(\tfrac{\varepsilon}{t})^{\frac{d}{2}-1}\big)\,.\]
and in particular~$\left\Vert \left[\hat{\tau}_{P}^{h},\chi_{R}\partial_{t}\omega\times\right]\right\Vert _{\mathcal{L}\left(L_{\xi}^{2}\right)}\leq\left\langle P\right\rangle ^{k}\mathcal{O}(h).$\end{lem}
\begin{proof}[{Proof of Lemma~\ref{lem:[]=00003D{}}.}]
First observe that the time derivative of~$\omega$ is given by\[
\partial_{t}\omega_{h,\varepsilon,t}=\xi^{.2}+\Re\big\langle z_{h,\varepsilon,t}^{\xi},f_{h,\varepsilon}\big\rangle=\xi^{.2}-h\Im\int_{0}^{t/\varepsilon}\negthickspace\negthickspace\int_{\mathbb{R}_{\eta}^{d}}e^{is(\eta^{.2}-2\xi.\eta)}\hat{G}(\eta)\diff\eta\diff s\]
once we replace~$f_{h,\varepsilon}$ by its expression in terms
of~$\hat{V}$, use~$\hat{G}=|\hat{V}|^{2}$ and make a change of
variable. By setting\[
R(u,\xi):=\chi_{R}(\xi)\Im\lim_{M\to+\infty}\int_{u}^{M}\int_{\mathbb{R}_{\eta}^{d}}e^{is\left(\eta^{.2}-2\xi.\eta\right)}\,\hat{G}(\eta)\diff\eta\diff s\]
we get $\chi_{R}\partial_{t}\omega=\chi_{R}(\xi)\xi^{.2}-hR(0,\xi)+hR(\tfrac{t}{\varepsilon},\xi)$.
The part in~$\xi^{.2}$ gives the only relevant contribution\[
[\hat{\tau}_{P}^{h},\chi_{R}\xi^{.2}\times]=-ih\{e^{i\sigma(P,X)},\chi_{R}\xi^{.2}\times\}^{Weyl}+\langle P\rangle^{k}\mathcal{O}_{h\to0}(h^{2})\,.\]
One of the other parts can be estimated without using the commutator
structure\[
\big\|[\hat{\tau}_{P}^{h},R({\textstyle \frac{t}{\varepsilon}},\xi)\times]\big\|_{\mathcal{L}(L_{\xi}^{2})}\leq2\|\hat{\tau}_{P}^{h}\|_{\mathcal{L}(L_{\xi}^{2})}\left\Vert R({\textstyle \frac{t}{\varepsilon}},\xi)\right\Vert _{L_{\xi}^{\infty}}\leq C(\tfrac{\varepsilon}{t})^{\frac{d}{2}-1}\]
since\[
\int_{\mathbb{R}_{\eta}^{d}}e^{is\left(\eta^{.2}-2\xi.\eta\right)}\,\hat{G}(\eta)\diff\eta=e^{-is\xi^{.2}}\int_{\mathbb{R}_{x}^{d}}G(x)\, e^{-ix.\xi}\,\big(\tfrac{2\pi}{|s|}\big)^{d/2}e^{id\sign s\frac{\pi}{4}}e^{\frac{x^{2}}{2is}}\diff x\]
whose modulus is bounded by~$(\tfrac{2\pi}{|s|})^{d/2}\left\Vert G\right\Vert _{L^{1}}$.

Since~$R(0,\cdot)$ is in~$\mathcal{C}_{0}^{\infty}(\mathbb{R}_{\xi}^{d})$
we can apply the symbolic calculus\[
[\hat{\tau}_{P}^{h},hR(0,\xi)\times]=-ih^{2}\{e^{i\sigma(P,X)}\,,R(0,\xi)\}^{W}(-hD_{\xi},\xi)+\mathcal{O}(h^{2}\langle P\rangle^{k})\big)\]
where for some integer~$k$,\[
\big\|\{e^{i\sigma(P,X)}\,,R(0,\xi)\times\}^{W}(-hD_{\xi},\xi)\big\|_{\mathcal{L}(L_{\xi}^{2})}=\langle P\rangle^{k}\mathcal{O}_{h\to0}(1)\,,\]
which concludes the proof of the lemma.
\end{proof}

\subsection{The collision terms $m_{-}$ and $m_{+}$}
\begin{prop}
Let~$b\in\mathcal{C}_{0}^{\infty}(\mathbb{R}_{x}^{d}\times\mathbb{R}_{\xi}^{d*})$
and~$\rho\in\mathcal{L}_{1}^{+}L_{x}^{2}$, $\Tr\rho\leq1$ such
that for some~$R>0$, $\Supp_{\xi}b\subset B_{R}-B_{1/R}$ and $\Supp\hat{\rho}(\xi,\xi^{\prime})\subset B_{R}^{2}$.
Then\index{Delta2, Delta3@$\Delta_{+},\,\Delta_{-}$}\[
m_{\pm}=m(Q_{\pm,t}(b),t)+\Delta_{\pm}\]
and for any~$\alpha\in[0,1)$, there are constants $\mu=\mu(d,\alpha)>0$
and~$C_{R,b,G,d,\alpha,\mu}>0$, such that for~$h^{\alpha}\leq\frac{th}{\varepsilon}\leq1$,\[
\left|\Delta_{\pm}\right|\leq C_{R,b,G,\mu}\big(\tfrac{ht}{\varepsilon}+h^{\mu}\big)\,.\]

\end{prop}
\textbf{Notation} \label{def:kappa-c-c-zeta-Q+--Q+-zeta}For~$\zeta>0$,
$r\in\mathbb{R}$ and~$P\in\mathbb{R}_{p_{x},p_{\xi}}^{2d}$, set\index{kappa zeta@$\kappa^{\zeta}$}\index{c@$\mathfrak{c},\,\mathfrak{c}^{\zeta},\,\mathfrak{c}_{P,t}^{\zeta}$},
with $\kappa^{\zeta}(r)=\frac{1}{\pi}\frac{\zeta}{r^{2}+\zeta^{2}}$,\begin{align*}
\mathfrak{c}(\xi) & =2\pi\int_{\mathbb{R}_{\eta}^{d}}\hat{G}(\eta+\xi)\,\delta(\eta^{.2}-\xi^{.2})\diff\eta\,,\\
\mathfrak{c}^{\zeta}(\xi) & =2\pi\int_{\mathbb{R}_{\eta}^{d}}\hat{G}(\eta)\,\kappa^{\zeta}(\eta^{.2}-2\xi.\eta)\diff\eta\,,\\
\mathfrak{c}_{P,t}^{\zeta}(x,\xi) & =2\pi\int_{\mathbb{R}_{\eta}^{d}}\hat{G}(\eta)\, e^{i\sigma(P,(-2t\eta,-\eta))}\kappa^{\zeta}(\eta^{.2}-2\xi.\eta)\diff\eta\,.\end{align*}
Associate with these functions the operators defined for~$b\in\mathcal{C}_{0}^{\infty}(\mathbb{R}_{x}^{d}\times\mathbb{R}_{\xi}^{d*})$
by\index{collision approx@$Q_{-}^{\zeta},\, Q_{+,t}^{\zeta}$}\begin{align*}
Q_{-}^{\zeta}(b) & =\mathfrak{c}^{\zeta}\, b\,,\qquad Q_{-}(b)=\mathfrak{c}\, b\,,\\
Q_{+,t}^{\zeta}b(x,\xi) & =\int_{\mathbb{R}_{P}^{2d}}\mathcal{F}^{\sigma}b(P)\, e^{i\sigma(P,X)}\,\mathfrak{c}_{P,t}^{\zeta}(x,\xi)\dbar P\,.\end{align*}

\begin{prop}
For~$d\geq3$, and $h^{\alpha}\leq\frac{ht}{\varepsilon}\leq1$,\[
m_{\pm}=m(Q_{\pm,t}(b),t)+\sum_{k=1}^{4}\Delta_{\pm,k}\]
with
\begin{itemize}
\item $\left|\Delta_{\pm,1}\right|\leq\frac{ht}{\varepsilon}C_{d}\max\bigl\{\|\hat{G}\|_{L^{1}},\|G\|_{L^{1}}\bigr\}\left\Vert \mathcal{F}^{\sigma}b\right\Vert _{L_{P}^{1}}$\textup{,}
\item $\left|\Delta_{\pm,2}\right|\leq C_{\alpha,\beta,\nu,G,d}h^{\nu}$,
\item $\left|\Delta_{\pm,3}\right|\leq\zeta^{\gamma}\mathcal{N}_{k(d)}(b)\, C_{d,G,C,\gamma}$
for~$\gamma\in(0,1)$,
\item $\left|\Delta_{\pm,4}\right|\leq\frac{ht}{\varepsilon}\bigl\Vert\mathcal{F}^{\sigma}\bigl(Q_{\pm,\frac{ht}{\varepsilon}}(b)\bigr)\bigr\Vert_{L_{P}^{1}}$
\end{itemize}
for some~$\nu,\,\beta>0$ with~$\zeta=h^{\beta}$.

\end{prop}
This result will be proved in the next paragraphs by considering successively
all the error terms. These error terms~$\Delta_{\pm,j}$, $j=1,\dots,4$
are given by the following approximation process (where we write shortly
$B^{W}$ for $B^{W}(-hD_{\xi},\xi)$)\begin{align*}
m_{\pm} & =\int\mathcal{F}^{\sigma}b(P)\,\Tr\big[\hat{\boldsymbol{\rho}}_{t}^{app}\,\mathscr{A}_{\pm,P}\big]\dbar P+\Delta_{\pm,1}\\
 & =\int\mathcal{F}^{\sigma}b(P)\,\Tr\big[\hat{\boldsymbol{\rho}}_{t}^{app}\,\big(\mathfrak{c}_{\pm,P}^{\zeta}e^{i\sigma(P,\cdot)}\big)^{W}\big]\dbar P+{\textstyle \sum_{j=1}^{2}}\Delta_{\pm,j}\displaybreak[0]\\
 & =\int\mathcal{F}^{\sigma}\big(Q_{\pm,\frac{ht}{\varepsilon}}b\big)(P)\,\Tr\big[\hat{\boldsymbol{\rho}}_{t}^{app}\,\hat{\tau}_{P}^{h}\big]\dbar P+{\textstyle \sum_{j=1}^{3}}\Delta_{\pm,j}\displaybreak[0]\\
 & =m(Q_{\pm,\frac{ht}{\varepsilon}}b,t)+{\textstyle \sum_{j=1}^{4}}\Delta_{\pm,j}\,.\end{align*}
The error terms~$\Delta_{\pm,j}$ are thus given by\begin{align}
\Delta_{\pm,1} & =\int\mathcal{F}^{\sigma}b\left(P\right)\Tr\Big[\hat{\boldsymbol{\rho}}_{t}^{app}\left(\Gamma(e^{ip_{x}\cdot\varepsilon\eta})-\Id\right)\mathscr{A}_{\pm,P}\Big]\dbar P\,,\label{eq:def-Delta+-1}\\
\Delta_{\pm,2} & =\int\mathcal{F}^{\sigma}b\left(P\right)\Tr\Big[\hat{\boldsymbol{\rho}}_{t}^{app}\big(\mathscr{A}_{\pm,P}-\big(\mathfrak{c}_{\pm,P}^{\zeta}e^{i\sigma(P,\cdot)}\big)^{W}\big)\Big]\dbar P\,,\label{eq:def-Delta+-2}\\
\Delta_{\pm,3} & =\Tr\Big[\hat{\boldsymbol{\rho}}_{t}^{app}\big(Q_{\pm,\frac{ht}{\varepsilon}}^{\zeta}b-Q_{\pm,\frac{ht}{\varepsilon}}b\big)^{W}\Big]\,,\label{eq:def-Delta+-3}\\
\Delta_{\pm,4} & =\int\mathcal{F}^{\sigma}\big(Q_{\pm,\frac{ht}{\varepsilon}}b\big)(P)\,\Tr\Big[\hat{\boldsymbol{\rho}}_{t}^{app}\big(\Id-\Gamma(e^{ip_{x}\cdot\varepsilon\eta})\big)\hat{\tau}_{P}^{h}\Big]\dbar P\,,\label{eq:def-Delta+-4}\end{align}
since~$\hat{\tau}_{P}^{h}=\left(e^{i\sigma(P,\cdot)}\right)^{W}$,\[
\big(Q_{\pm,\frac{ht}{\varepsilon}}^{\zeta}b\big)^{W}=\int_{\mathbb{R}_{P}^{2d}}\mathcal{F}^{\sigma}b(P)\,\big(\mathfrak{c}_{\pm,P}^{\zeta}e^{i\sigma(P,\cdot)}\big)^{W}\dbar P\,,\]
and the same relation holds without~$\zeta$ and\[
\int_{P}\mathcal{F}^{\sigma}\big(Q_{\pm,\frac{ht}{\varepsilon}}b\big)(P)\,\Tr\left[\hat{\boldsymbol{\rho}}_{t}^{app}\,\Gamma(e^{ip_{x}\cdot\varepsilon\eta})\,\hat{\tau}_{P}^{h}\right]\dbar P=m\big(Q_{\pm,\frac{ht}{\varepsilon}}b,t\big)\,.\]

The term~$\Delta_{\pm,4}$ can be estimated right away using Proposition~\ref{pro:Gamma-approx-Id}.

\subsubsection{Computation of the operators~$\mathscr{A}_{\pm,P}$}

We recall that the operators~$\mathscr{A}_{\pm,P}$ are defined by
their kernels in Equations~(\ref{eq:def-operateur-A{}P}), (\ref{eq:def-operateur-A-P}),
(\ref{eq:def-operateur-A+P}). 
\begin{prop}
\label{pro:expression-A-j}The operators~$\mathscr{A}_{-,j}$ can
be expressed as\begin{align*}
\mathscr{A}_{-,P}^{1} & =\int_{\mathbb{R}_{\eta}^{d}}\int_{0}^{t/\varepsilon}\hat{\tau}_{P}^{h}\circ\Re\big(e^{is(\eta^{.2}-2\xi.\eta)}\big)\negthickspace\times\,\hat{G}(\eta)\diff s\diff\eta\,,\\
\mathscr{A}_{-,P}^{2} & =\int_{\mathbb{R}_{\eta}^{d}}\int_{0}^{t/\varepsilon}\Re\big(e^{is(\eta^{.2}-2\xi.\eta)}\big)\negthickspace\times\,\circ\,\hat{\tau}_{P}^{h}\,\hat{G}(\eta)\diff s\diff\eta\,.\end{align*}
The operator $\mathscr{A}_{+,P}$ can be decomposed as $\mathscr{A}_{+,P}=\mathscr{A}_{+,P}^{1}+\mathscr{A}_{+,P}^{2}$
with\begin{align*}
\mathscr{A}_{+,P}^{1} & =\int_{0}^{t/\varepsilon}\int_{\mathbb{R}_{\eta}^{d}}e^{-i\sigma\left(P,(2\frac{t}{\varepsilon}\eta,\eta)\right)}\hat{\tau}_{P}^{h}\circ e^{-is(\eta^{.2}-2\xi.\eta)}\hat{G}(\eta)\diff\eta\diff s\,,\\
\mathscr{A}_{+,P}^{2} & =\int_{0}^{t/\varepsilon}\int_{\mathbb{R}_{\eta}^{d}}e^{-i\sigma\left(P,(2\frac{t}{\varepsilon}\eta,\eta)\right)}e^{is(\eta^{.2}-2\xi.\eta)}\circ\hat{\tau}_{P}^{h}\,\hat{G}(\eta)\diff\eta\diff s\,.\end{align*}
\end{prop}
\begin{proof}
Computing the time derivative of~$\tfrac{1}{2}|z_{h,\varepsilon,t}|^{2}$
brings\[
\partial_{t}\tfrac{1}{2}|z_{h,\varepsilon,t}|^{2}=h\Re\int_{\mathbb{R}_{\eta}^{d}}\int_{0}^{t/\varepsilon}e^{is(\eta^{.2}-2\xi_{j}.\eta)}\,\hat{G}(\eta)\diff s\diff\eta\,.\]
From the definition of~$\mathscr{A}_{-,j,P}$ in terms of their kernel,
we get\[
ih\,\mathscr{A}_{-,P}^{1}=i\hat{\tau}_{P}^{h}\circ(\partial_{t}\tfrac{1}{2}|z_{t}^{\xi}|^{2})\,,\qquad ih\,\mathscr{A}_{-,P}^{2}=i(\partial_{t}\tfrac{1}{2}|z_{t}^{\xi}|^{2})\circ\hat{\tau}_{P}^{h}\,,\]
hence the result for~$\mathscr{A}_{-}^{j}$.

The time derivative of~$[\varphi,p_{x}]_{2}$ is \begin{align*}
\partial_{t}[\varphi,p_{x}]_{2} & =\phantom{+}h\int_{\mathbb{R}_{\eta}^{d}}\int_{0}^{t/\varepsilon}e^{ip_{x}\cdot\eta}e^{is\left(\eta^{.2}-2\xi_{1}.\eta\right)}e^{-i\frac{t}{\varepsilon}\left(\eta^{.2}-2\xi_{2}.\eta\right)}\diff s\,\hat{G}(\eta)\diff\eta\\
 & \phantom{=}+h\int_{\mathbb{R}_{\eta}^{d}}\int_{0}^{t/\varepsilon}e^{ip_{x}\cdot\eta}e^{i\frac{t}{\varepsilon}\left(\eta^{.2}-2\xi_{1}.\eta\right)}e^{-is\left(\eta^{.2}-2\xi_{2}.\eta\right)}\diff s\,\hat{G}(\eta)\diff\eta\,.\end{align*}
We now focus on the first term (analogous computations give the second
term). The definition of~$\mathscr{A}_{+,P}$ in terms of their kernel
gives then \begin{align*}
\mathscr{A}_{+,P}^{1} & =\int_{0}^{t/\varepsilon}\int_{\mathbb{R}_{\eta}^{d}}e^{ip_{x}\cdot\eta}e^{-i\frac{t}{\varepsilon}\left(\eta^{.2}-2\xi.\eta\right)}\circ\hat{\tau}_{P}^{h}\circ e^{is\left(\eta^{.2}-2\xi.\eta\right)}\,\hat{G}(\eta)\diff\eta\diff s\,,\end{align*}
The relation~$e^{2i\frac{t}{\varepsilon}\xi.\eta}\circ\hat{\tau}_{P}^{h}=e^{-2i\frac{t}{\varepsilon}p_{\xi}\eta}\hat{\tau}_{P}^{h}\circ e^{2i\frac{t}{\varepsilon}\xi.\eta}$
brings the result up to a change of variable.
\end{proof}

Thus we get six different terms (four for the~$\mathscr{A}_{-}$
terms due to the real parts and two for the~$\mathscr{A}_{+}$ terms)
with a very similar structure. In order to avoid repeating analogous
calculations several times we introduce the following notations.
\begin{notation}
Set \index{A@$\mathscr{A}_{\vec{\mu}}^{j},\,\mathscr{B}_{\vec{\mu}}^{j},\,\mathscr{C}_{\vec{\mu}}^{j,\zeta}$}
(by writing shortly~$B^{W}$ for~$B^{W}(-hD_{\xi},\xi)$)\begin{align}
\mathscr{A}_{\vec{\mu}}^{1}(s) & =\int_{\mathbb{R}^{d}}\hat{G}(\eta)\, e^{\mu_{1}i\tilde{\sigma}}\,\hat{\tau}_{P}^{h}\circ e^{-\mu_{2}is(\eta^{.2}-2\xi.\eta)}\diff\eta\,,\label{eq:def-A1-mu}\\
\mathscr{B}_{\vec{\mu}}^{1}(s) & =\int_{\mathbb{R}^{d}}\hat{G}(\eta)\, e^{\mu_{1}i\tilde{\sigma}}\,\hat{\tau}_{(p_{x}-\mu_{2}2s\eta,p_{\xi})}^{h}\, e^{-\mu_{2}is\eta^{.2}}\diff\eta\,,\label{eq:def-B1-mu}\\
\mathscr{C}_{\vec{\mu}}^{1,\zeta} & =\int_{\mathbb{R}^{d}}\hat{G}(\eta)\big(e^{\mu_{1}i\tilde{\sigma}}e^{i\sigma(P,\cdot)}\big)^{W}\frac{\diff\eta}{\zeta+\mu_{2}i\left(\eta^{.2}-2\xi.\eta\right)}\,,\displaybreak[0]\label{eq:def-C1-zeta}\\
\mathscr{A}_{\vec{\mu}}^{2}(s) & =\int_{\mathbb{R}^{d}}\hat{G}(\eta)\, e^{\mu_{1}i\tilde{\sigma}}e^{\mu_{2}is(\eta^{.2}-2\xi.\eta)}\circ\hat{\tau}_{P}^{h}\diff\eta\,,\label{eq:def-A2-mu}\\
\mathscr{B}_{\vec{\mu}}^{2}(s) & =\int_{\mathbb{R}^{d}}\hat{G}(\eta)\, e^{\mu_{1}i\tilde{\sigma}}\,\hat{\tau}_{(p_{x}+\mu_{2}2s\eta,p_{\xi})}^{h}\, e^{\mu_{2}is\eta^{.2}}\diff\eta\,,\label{eq:def-B2-mu}\\
\mathscr{C}_{\vec{\mu}}^{2,\zeta} & =\int_{\mathbb{R}^{d}}\hat{G}(\eta)\big(e^{\mu_{1}i\tilde{\sigma}}e^{i\sigma(P,\cdot)}\big)^{W}\frac{\diff\eta}{\zeta-\mu_{2}i\left(\eta^{.2}-2\xi.\eta\right)}\,,\label{eq:def-C2-zeta}\end{align}
with~$\tilde{\sigma}=\sigma\left(P,(-2h\frac{t}{\varepsilon}\eta,-\eta)\right)$.
The terms~$\mu_{1},\,\mu_{2}$ are chosen to adapt to the cases of
the terms~$m_{\pm}$. More precisely, for~$j=1,\,2$, the previous
quantities become\[
\mathscr{A}_{-}^{j}=\int_{0}^{t/\varepsilon}\big(\mathscr{A}_{0,1}^{j}(s)+\mathscr{A}_{0,-1}^{j}(s)\big)\diff s\qquad\mbox{and}\qquad\mathscr{A}_{+}^{j}=\int_{0}^{t/\varepsilon}\mathscr{A}_{1,1}^{j}(s)\diff s\,.\]

\end{notation}
We first show that the operators~$\mathscr{C}_{\vec{\mu}}^{\zeta}$
are good approximations of the operators~$\mathscr{A}_{\vec{\mu}}=\int_{0}^{t/\varepsilon}\mathscr{A}_{\vec{\mu}}(s)\diff s$
if the parameter~$\zeta$ is well chosen. We use the operators~$\int_{0}^{t/\varepsilon}\mathscr{B}_{\vec{\mu}}(s)\diff s$
as an intermediate step. Then we study the limit of the operators~$\mathscr{C}_{\vec{\mu}}^{\zeta}$,
with a distinction between the cases~$m_{-}$ and~$m_{+}$.

\subsubsection{Estimate of the error terms~$\Delta_{\pm,1}$}
\begin{prop}
For~$d\geq3$,\[
\left|\Delta_{\pm,1}\right|\leq\tfrac{ht}{\varepsilon}C_{d}\max\bigl\{\|\hat{G}\|_{L^{1}},\|G\|_{L^{1}}\bigr\}\left\Vert \mathcal{F}^{\sigma}b\right\Vert _{L_{P}^{1}}\,.\]
\end{prop}
\begin{proof}
The term~$\Delta_{\pm,1}$ was defined in Equation~(\ref{eq:def-Delta+-1}).
This inequality follows from Propositions~\ref{pro:Gamma-approx-Id}
and~\ref{pro:estimation-Aj-mu-(s)} below since~$s\mapsto\min\{1,s^{-d/2}\}$
is integrable on~$\mathbb{R}^{+}$ for~$d\geq3$.
\end{proof}

\begin{prop}
\label{pro:estimation-Aj-mu-(s)}The families of operators~$\mathscr{A}(s)=\mathscr{A}_{\vec{\mu}}^{j}(s)$
satisfy\[
\left\Vert \mathscr{A}(s)\right\Vert _{\mathcal{L}(L_{\xi}^{2})}\leq C_{d}\max\bigl\{\|\hat{G}\|_{L^{1}},\|G\|_{L^{1}}\bigr\}\min\bigl\{1,s^{-d/2}\bigr\}\,.\]
\end{prop}
\begin{proof}
A uniform estimate of Equations~(\ref{eq:def-A1-mu}) and (\ref{eq:def-A2-mu})
yields $\|\mathscr{A}_{\vec{\mu}}^{j}(s)\|{}_{\mathcal{L}(L_{\xi}^{2})}\leq C_{d}\|\hat{G}\|_{L^{1}}$.
In order to obtain the part of the estimate with the dependence in~$s$,
we use the formula\[
\big\|\mathscr{A}_{\vec{\mu}}^{j}(s)\big\|_{\mathcal{L}(L_{\xi}^{2})}=\sup\big\{\big|\big\langle\psi,\mathscr{A}_{\vec{\mu}}^{j}(s)\varphi\big\rangle\big|\,,\,\left\Vert \psi\right\Vert _{L_{\xi}^{2}}=\left\Vert \varphi\right\Vert _{L_{\xi}^{2}}=1\big\}\,.\]
We can then compute, for~$\psi$, $\varphi\in L_{\xi}^{2}$,\begin{align*}
\big\langle\psi,\mathscr{A}_{\vec{\mu}}^{j}(s)\varphi\big\rangle & =\int_{\mathbb{R}_{\eta}^{d}}\big\langle\psi,\hat{G}(\eta)\, e^{i\mu_{1}\tilde{\sigma}}\,\hat{\tau}_{P}^{h}\circ e^{-\mu_{2}is(\eta^{.2}-2\xi.\eta)}\varphi\big\rangle_{\xi}\diff\eta\\
 & =\int_{\mathbb{R}_{\xi}^{d}}\big\langle\hat{G}(\eta)\,\hat{\tau}_{-P}^{h}\psi(\xi),e^{\mu_{1}i\tilde{\sigma}}e^{-\mu_{2}is(\eta^{.2}-2\xi.\eta)}\varphi(\xi)\big\rangle_{\eta}\diff\xi\\
 & =\int_{\mathbb{R}_{\theta}^{d}}\langle\psi_{\theta},\varphi_{\vec{\mu},\theta}\rangle_{\xi}\diff\theta/(2\pi)^{d}\,,\end{align*}
where we defined, for~$\theta\in\mathbb{R}_{\theta}^{d}$,\[
\varphi_{\vec{\mu},\theta}=\int e^{i\theta\eta}e^{\mu_{1}i\tilde{\sigma}}e^{-\mu_{2}is(\eta^{.2}-2\xi.\eta)}\varphi(\xi)\diff\eta\,,\quad\psi_{\theta}=\int e^{i\theta\eta}\,\hat{G}(\eta)\,\hat{\tau}_{-P}^{h}\,\psi(\xi)\diff\eta\,.\]
We first compute\[
\varphi_{\vec{\mu},\theta}(\xi)=(\tfrac{\pi}{s})^{d/2}e^{i\frac{(\theta+\mu_{2}2s\xi+\mu_{1}(2hsp_{\xi}-p_{x}))^{2}}{4\mu_{2}s}}e^{i\frac{\pi}{4}d}\varphi(\xi)\]
where we used the formula $\int e^{-ix\eta}e^{-a\eta^{2}}\diff\eta=\left(\frac{\pi}{a}\right)^{d/2}e^{-x^{2}/4a}$
with~$a=\mu_{2}is$ and~$x=-\left(\theta+\mu_{2}2s\xi+\mu_{1}\left(2hsp_{\xi}-p_{x}\right)\right)$
and so $\|\varphi_{\vec{\mu},\theta}\|_{L^{\infty}(\mathbb{R}_{\theta}^{d};L_{\xi}^{2})}\leq\left(\frac{\pi}{s}\right)^{d/2}\|\varphi\|_{L_{\xi}^{2}}$.
We now observe that \[
\Big\|\int e^{i\theta\eta}\,\hat{G}(\eta)\,\hat{\tau}_{-P}^{h}\diff\eta\Big\|_{L^{1}(\mathbb{R}_{\theta}^{d};\mathcal{L}(L_{\xi}^{2}))}\leq\left(2\pi\right)^{d}\left\Vert G\right\Vert _{L^{1}}\]
so that~$\left\Vert \psi_{\theta}\right\Vert _{L^{1}(\mathbb{R}_{\theta}^{d};L_{\xi}^{2})}\leq C_{d}\left\Vert G\right\Vert _{L^{1}}\left\Vert \psi\right\Vert _{L_{\xi}^{2}}$.
And finally\[
\big|\big\langle\psi,\mathscr{A}_{\vec{\mu}}(s)\varphi\big\rangle\big|\leq C_{d}\,\left\Vert G\right\Vert _{L^{1}}\,(\tfrac{\pi}{s})^{d/2}\,\|\varphi\|_{L_{\xi}^{2}}\,\|\psi\|_{L_{\xi}^{2}}\]
and we obtain the desired result~$\|\mathscr{A}_{\vec{\mu}}(s)\|_{\mathcal{L}(L_{\xi}^{2})}\leq C_{d}\left\Vert G\right\Vert _{L^{1}}s^{-d/2}\,.$
\end{proof}

\subsubsection{Estimate of the error terms~$\Delta_{\pm,2}$}

\begin{prop}
Let~$\alpha\in(0,1]$. There are constants~$\beta=\beta(d,\alpha)\in(0,1)$,~$\nu=\nu(d,\alpha)\in(0,1)$
and~$C=C(\alpha,\beta,\nu,d,G)>0$ such that, for~$h^{\alpha}\leq\frac{th}{\varepsilon}\leq1$,
and~$\zeta=h^{\beta}$,\[
\left|\Delta_{\pm,2}\right|\leq\|\left\langle \cdot\right\rangle ^{k}\mathcal{F}^{\sigma}b\|_{L^{1}}Ch^{\nu}\,.\]

\end{prop}

In order to prove this result we use Proposition~\ref{pro:controle-erreur-sans-gamma}
and thus control~\[
\Big\|\int_{0}^{t/\varepsilon}\mathscr{A}(s)\diff s-\mathscr{C}^{\zeta}\Big\|_{\mathcal{L}(L_{\xi}^{2})}\,.\]
We first give an abstract result and then show that our cases fit
within this framework.
\begin{prop}
\label{pro:Abstract-control}For~$M$, $t$, $\varepsilon$ such
that~$1\leq M\leq\frac{t}{\varepsilon}$. Suppose given $(\mathscr{A}(s))_{s\geq0}$,
$(\mathscr{B}(s))_{s\geq0}$ and~$(\mathscr{C}^{\zeta})_{0<\zeta<1}$
three families of operators in~$\mathcal{L}(L_{\xi}^{2})$ (also
dependent on~$h$ and~$P=(p_{x,}p_{\xi})$) such that for some constants~$C_{\mathscr{A}}$,
$C_{\mathscr{A},\mathscr{B}}$, $C_{\mathscr{B},\mathscr{C}}$, independent
of~$h,\varepsilon,t,P,M,\zeta$,
\begin{enumerate}
\item $\left\Vert \mathscr{A}(s)\right\Vert _{\mathcal{L}(L_{\xi}^{2})}\leq C_{\mathscr{A}}\min\left\{ 1,s^{-d/2}\right\} $,
\item $\left\Vert \mathscr{A}(s)-\mathscr{B}(s)\right\Vert _{\mathcal{L}(L_{\xi}^{2})}\leq C_{\mathscr{A},\mathscr{B}}hs\left|p_{\xi}\right|$,
\item $r_{\zeta,M}(x,\xi)\,:=\Symb^{Weyl}(\int_{0}^{M}\mathscr{B}(s)\, e^{-\zeta s}\diff s-\mathscr{C}^{\zeta})$
satisfies for some~$k=k(d)\in\mathbb{N}$,\textup{\[
\sup_{\left|\alpha\right|\leq k}\|\partial_{x,\xi}^{\alpha}r_{\zeta,M}\|_{L_{x,\xi}^{\infty}}\leq C_{\mathscr{B},\mathscr{C}}\left\langle P\right\rangle ^{k}\big(\tfrac{M}{\zeta}\big)^{k}e^{-\zeta M}\,.\]
}
\end{enumerate}
Then, for~$\zeta M\geq1$,
\begin{enumerate}
\item $\|\int_{0}^{t/\varepsilon}\mathscr{A}(s)\diff s\|_{\mathcal{L}\left(L_{\xi}^{2}\right)}\leq\frac{d}{d-2}C_{\mathscr{A}}$,
\item $\|\int_{0}^{t/\varepsilon}\mathscr{A}(s)\diff s-\int_{0}^{M}\mathscr{A}(s)\diff s\|_{\mathcal{L}(L_{\xi}^{2})}\leq\frac{2}{d-2}C_{\mathscr{A}}M^{1-\frac{d}{2}}$,
\item for~$d\geq3$,\[
\Big\|\int_{0}^{M}\mathscr{A}(s)\big(1-e^{-\zeta s}\big)\diff s\Big\|_{\mathcal{L}(L_{\xi}^{2})}\leq5C_{\mathscr{A}}\zeta^{1/2}\]

\item $\|\int_{0}^{M}\left(\mathscr{A}(s)-\mathscr{B}(s)\right)e^{-\zeta s}\diff s\|_{\mathcal{L}(L_{\xi}^{2})}\leq\frac{1}{2}C_{\mathscr{A},\mathscr{B}}h\zeta^{-2}|p_{\xi}|$,
\item for some integer~$k=k(d)$\textup{, \[
\Big\|\int_{0}^{M}\mathscr{B}(s)\, e^{-\zeta s}\diff s-\mathscr{C}^{\zeta}\Big\|_{\mathcal{L}(L_{\xi}^{2})}\leq C_{d,k'}C_{\mathscr{B},\mathscr{C}}\left\langle P\right\rangle ^{k}\big(\tfrac{M}{\zeta}\big)^{k}e^{-\zeta M}\,.\]
}
\item Let\textup{~$\frac{ht}{\varepsilon}\geq h^{\alpha}$}, $\zeta=h^{\beta}$
with~$\beta\in(0,\frac{1}{2})$ and~$\beta+\alpha<1$, and~$\nu=\nu(\alpha,\beta)<\min\{(1-\alpha)/2,\beta/2,1-2\beta\}$,
we have\[
\Big\|\int_{0}^{\frac{t}{\varepsilon}}\mathscr{A}(s)\diff s-\mathscr{C}^{\zeta}\Big\|_{\mathcal{L}(L_{\xi}^{2})}\leq Ch^{\nu}\]
with~$C=C(\nu,\alpha,\beta,C_{\mathscr{A}},C_{\mathscr{A}\mathscr{B}},C_{\mathscr{B}\mathscr{C}})$.\textup{}
\end{enumerate}
\end{prop}
\begin{proof}
Points~1 and~2 are proved by integration of the first assumed estimate
and using~$1\leq M\leq\frac{t}{\varepsilon}$ for~\emph{2.}

Point~$3$ is proved by integration of the first assumed estimate,
using~$1-e^{-\zeta s}\leq\zeta s$ for~$\zeta s\leq1$ and~$1-e^{-\zeta s}\leq1$
for~$\zeta s\geq1$,\[
\int_{0}^{M}(1-e^{-\zeta s})\min\{1,s^{-d/2}\}\diff s\leq\zeta\int_{0}^{1}s\diff s+\zeta\int_{1}^{1/\zeta}s^{1-\frac{d}{2}}\diff s+\int_{1/\zeta}^{+\infty}s^{-d/2}\diff s\,,\]
which brings the result.

For Point~4, we use the second assumption and~$\int_{0}^{M}se^{-\zeta s}\diff s\leq\zeta^{-2}\int_{0}^{+\infty}ue^{-u}\diff u\,.$

For Point~5, the known estimates for pseudo-differential operators
give\[
\big\| r^{W}(-hD_{\xi},\xi)\big\|\leq C_{k}\sup_{\left|\alpha\right|\leq N_{k}}\big\|\partial_{x,\xi}^{\alpha}r\big\|_{L^{\infty}(\mathbb{R}^{2d})}\,.\]
This and the third hypothesis imply the result.

For Point~6, we would like to choose the ($h$-dependent) parameters~$M$
and~$\zeta$ such that the quantity\[
M^{1-\frac{d}{2}}+\sqrt{\zeta}+h\zeta^{-2}+\big(\tfrac{M}{\zeta}\big)^{k}e^{-\zeta M}\,,\]
is small when~$h$ tends to~$0$ and~$M$ not too big. We choose~$hM=h^{\alpha}$
and~$\zeta=h^{\beta}$ with~$\beta+\alpha<1$, $\alpha,\beta>0$
so that the previous quantity is smaller than\[
h^{\left(1-\alpha\right)\left(\frac{d}{2}-1\right)}+h^{\beta/2}+h^{1-2\beta}+h^{-k\left(1-\alpha+\beta\right)}\exp\!\big(-(h^{\beta+\alpha-1})\big)\,.\]
In order to get a small quantity it suffices to require~$\beta<\frac{1}{2}$.
Then we get an error term whose size is controlled by~$h^{\nu\left(\alpha,\beta\right)}$.\end{proof}
\begin{prop}
\label{pro:application-abstract-control}The families of operators~$\mathscr{A}(s)=\mathscr{A}_{\vec{\mu}}^{j}(s)$,
$\mathscr{B}(s)=\mathscr{B}_{\vec{\mu}}^{j}(s)$ and $\mathscr{C}^{\zeta}=\mathscr{C}_{\vec{\mu}}^{j,\zeta}$
satisfy the hypotheses of Proposition~\ref{pro:Abstract-control}\textup{}
with \[
C_{\mathscr{A}}=C_{d}\max\bigl\{\|\hat{G}\|_{L^{1}},\|G\|_{L^{1}}\bigr\}\,,\; C_{\mathscr{A},\mathscr{B}}=\bigl\Vert\left|\cdot\right|\hat{G}\bigr\Vert_{L^{1}}\,,\; C_{\mathscr{B},\mathscr{C}}=\bigl\Vert\left\langle \cdot\right\rangle ^{k}\hat{G}\bigr\Vert_{L^{1}}\,,\]
 for some integer~$k$.\end{prop}
\begin{proof}
Point~1 is contained in Proposition~\ref{pro:estimation-Aj-mu-(s)}.

We show Point~2 for~$\mathscr{A}_{\vec{\mu}}^{1}$ and~$\mathscr{B}_{\vec{\mu}}^{1}$,
the proof can be adapted to the case of~$\mathscr{A}_{\vec{\mu}}^{2}$
and~$\mathscr{B}_{\vec{\mu}}^{2}$. We observe that \[
\hat{\tau}_{P}^{h}\circ\big(e^{\mu_{2}is2\xi.\eta}\times\big)=e^{-\mu_{2}is\eta hp_{\xi}}\hat{\tau}_{P-(\mu_{2}2s\eta,0)}^{h}\]
and\[
\big(e^{i\sigma(P,X)}e^{\mu_{2}is2\xi.\eta}\big)^{W}\negthickspace(-hD_{\xi},\xi)=\hat{\tau}_{(p_{x}-\mu_{2}2s\eta,p_{\xi})}^{h}\,.\]
Thus we obtain the estimation\[
\big\|\hat{\tau}_{P}^{h}\circ\big(e^{\mu_{2}is2\xi.\eta}\times\big)-\big(e^{i\sigma\left(P,X\right)}e^{\mu_{2}is2\xi.\eta}\big)^{W}\negthickspace(-hD_{\xi},\xi)\big\|_{\mathcal{L}(L_{\xi}^{2})}\leq hs\left|\eta\right|\left|p_{\xi}\right|\]
Since the Weyl symbol of~$\mathscr{B}_{\vec{\mu}}^{1}(s)$ is\[
\frac{1}{2}\int_{\mathbb{R}_{\eta}^{d}}\hat{G}(\eta)\, e^{i\mu_{1}\tilde{\sigma}}e^{i\sigma(P,X)}e^{-\mu_{2}is\left(\eta^{.2}-2\xi.\eta\right)}\diff\eta\]
we get the estimate with~$C_{\mathcal{A},\mathcal{B}}=\int_{\mathbb{R}_{\eta}^{d}}\hat{G}(\eta)\left|\eta\right|\diff\eta$.

For Point~3, the Weyl symbol of $\int_{0}^{M}\mathscr{B}_{\vec{\mu}}^{1}(s)\, e^{-\zeta s}\diff s$
is\begin{align*}
\Symb\!{}^{Weyl}\, & \int_{0}^{M}\mathscr{B}_{\vec{\mu}}^{1}(s)\, e^{-\zeta s}\diff s\\
 & =\int_{\mathbb{R}_{\eta}^{d}}\hat{G}(\eta)\, e^{\mu_{1}i\tilde{\sigma}}e^{i\sigma(P,X)}\Big[\frac{e^{-\mu_{2}is\left(\eta^{.2}-2\xi.\eta\right)-\zeta s}}{-\mu_{2}i\left(\eta^{.2}-2\xi.\eta\right)-\zeta}\Big]_{0}^{M}\diff\eta\\
 & =\Symb\!{}^{Weyl}\,\mathscr{C}_{\vec{\mu}}^{1,\zeta}+r_{\zeta,M}\end{align*}
with\[
r_{\zeta,M}(x,\xi)=-\int_{\mathbb{R}_{\eta}^{d}}\hat{G}(\eta)\, e^{\mu_{1}i\tilde{\sigma}}e^{i\sigma\left(P,X\right)}\frac{e^{-\mu_{2}iM\left(\eta^{.2}-2\xi.\eta\right)-\zeta M}}{\mu_{2}i\left(\eta^{.2}-2\xi.\eta\right)+\zeta}\diff\eta\,.\]
and this expression allows us to get the estimate\[
\left|\partial_{x,\xi}^{\alpha}r_{\zeta,M}(x,\xi)\right|\leq\int_{\mathbb{R}_{\eta}^{d}}\hat{G}(\eta)\,\langle P\rangle^{k}(M\langle\eta\rangle)^{k}\frac{1}{\zeta^{k+1}}e^{-\zeta M}\diff\eta\]
which yields the result with $k+1$ replaced by $k$. The same
proof holds for~$\mathscr{B}_{\vec{\mu}}^{2}(s)$ and~$\mathscr{C}_{\vec{\mu}}^{2,\zeta}$.
\end{proof}

\subsubsection{Estimate of the error term~$\Delta_{-,3}$}
\begin{prop}
\label{pro:estimation-Delta-3}Let~$b\in\mathcal{C}_{0}^{\infty}(\mathbb{R}_{x,\xi}^{2d})$
with~$\Supp_{\xi}b\subset B_{R}\setminus B_{1/R}$ for some~$R>1$.
Let~$\gamma\in(0,1)$. There exists a constant~$C_{G,b,\gamma}>0$
such that, for all~$\zeta>0$,\[
\left|\Delta_{-,3}\right|\leq\zeta^{\gamma}\mathcal{N}_{k}(b)C_{G,b,\gamma}\]
for some integer~$k=k\left(d\right)$ big enough.\end{prop}
\begin{proof}
We recall that\[
\Delta_{-,3}=\Tr\big[\hat{\boldsymbol{\rho}}_{t}^{app}(Q_{-}^{\zeta}b-Qb)^{W}(-hD_{\xi},\xi-\diff\Gamma_{\!\varepsilon}(\eta))\big]\]
so that\[
\left|\Delta_{-,3}\right|\leq\big\|(Q_{-}^{\zeta}b-Q_{-}b)^{W}(-hD_{\xi},\xi-\diff\Gamma_{\!\varepsilon}(\eta))\big\|_{\mathcal{L}(L_{\xi}^{2}\otimes\Gamma L_{\eta}^{2})}\leq C_{k,d}\mathcal{N}_{k}(Q_{-}^{\zeta}b-Q_{-}b)\]
for some integer~$k$ big enough. By recalling~$Q_{-}^{\zeta}(b)=\mathfrak{c}^{\zeta}\, b$
and~$Q_{-}(b)=\mathfrak{c}\, b$ it is then sufficient to prove Lemma~\ref{lem:c-zeta-c}
below.\end{proof}
\begin{lem}
\label{lem:c-zeta-c}\textup{}For any integer~$k$ and~$\gamma$
in~$[0,1)$, a positive constant~$C_{k,\gamma,G,C}$ exists such
that for~$\zeta\in(0,\zeta_{0})$\textup{\[
\sup_{\left|\alpha\right|\leq k}\sup_{\left|\xi\right|\in\left[R^{-1},R\right]}\left|\partial_{\xi}^{\alpha}\left(\mathfrak{c}^{\zeta}-\mathfrak{c}\right)\left(\xi\right)\right|\leq C_{k,\gamma,G,R}\zeta^{\gamma}\,.\]
}\end{lem}
\begin{proof}
With~$\kappa^{\zeta}$, $\mathfrak{c}$, $\mathfrak{c}^{\zeta}$
introduced in Definition~\ref{def:kappa-c-c-zeta-Q+--Q+-zeta}, $\mathfrak{c}^{\zeta}-\mathfrak{c}$
can be expressed as\[
\left(\mathfrak{c}^{\zeta}-\mathfrak{c}\right)\left(\xi\right)=\int_{\mathbb{R}_{\eta}^{d}}\hat{G}(\eta)\,\kappa^{\zeta}(\eta^{.2}-2\xi.\eta)\diff\eta-\int_{\mathbb{R}_{\eta}^{d}}\hat{G}(\xi+\eta)\,\delta\big(\left|\eta\right|^{2}-\left|\xi\right|^{2}\big)\diff\eta\,.\]
We express the first integral as

\begin{align*}
\int_{\mathbb{R}_{\eta}^{d}}\hat{G}(\eta)\,\kappa^{\zeta}\big((\eta-\xi)^{.2}-\xi^{.2}\big)\diff\eta & =\int_{S^{d-1}}\int_{\mathbb{R}_{\rho}}f_{\xi,\omega}\left(r\right)\kappa^{\zeta}(\xi^{.2}-r)\diff r\diff\omega\\
 & =\int_{S^{d-1}}f_{\xi,\omega}*\kappa^{\zeta}(\xi^{.2})\diff\omega\end{align*}
and~$f_{\xi,\omega}(r)\,:=\frac{1}{2}r^{\frac{d-2}{2}}g(\xi+\sqrt{r}\omega)\,1_{[0,+\infty)}(r)$.
The partial derivative\[
\partial_{\xi_{j}}f_{\xi,\omega}(r)=\frac{1}{2}r^{\frac{d-2}{2}}\,\partial_{\xi_{j}}g(\xi+\sqrt{r}\omega)\,1_{[0,+\infty)}(r)\]
has the same form as the function~$f_{\xi}$. Then we observe that\begin{multline*}
\partial_{\xi_{j}}\big(f_{\xi,\omega}*\kappa^{\zeta}-f_{\xi,\omega}\big)\big(\left|\xi\right|^{2}\big)\\
=\left[(\partial_{\xi_{j}}f_{\xi,\omega})*\kappa^{\zeta}-\partial_{\xi_{j}}f_{\xi,\omega}\right]\big(\left|\xi\right|^{2}\big)+\left[\partial_{r}(f_{\xi,\omega}*\kappa^{\zeta}-f_{\xi,\omega})\right]\big(\left|\xi\right|^{2}\big)\,2\xi_{j}\end{multline*}
so that by doing successive derivations it suffices to deal only with
quantities of the form~$\partial_{r}^{k}(\partial_{\xi}^{\beta}f_{\xi,\omega}*\kappa^{\zeta}-\partial_{\xi}^{\beta}f_{\xi,\omega})$
which are in fact of the form~$\partial_{r}^{k}(f*\kappa^{\zeta}-f)$
with~$f$ satisfying the hypotheses of Lemma~\ref{pro:approx-unite-1}
uniformly in~$\omega$ so that we get the expected control, by integration
over~$\omega$.\end{proof}
\begin{lem}
\label{pro:approx-unite-1}Let~$f:\mathbb{R}_{r}\to\mathbb{R}$ continuous,
vanishing on~$\mathbb{R}^{-}$, such that~$\left.f\right|_{\mathbb{R}_{*}^{+}}\in\mathcal{C}^{\infty}(\mathbb{R}_{*}^{+})$
is rapidly decreasing towards~$+\infty$. Let~$0<r_{\min}<r_{\max}$.
Then\textup{\[
\forall\gamma\in(0,1)\,,\quad\exists C_{f,\gamma}\,,\qquad\big\|\left.\partial_{r}^{k}[f*\kappa^{\zeta}-f]\right|_{[r_{\min},r_{\max}]}\big\|_{L^{\infty}}\leq C_{\gamma}\zeta^{\gamma}\,.\]
}\end{lem}
\begin{proof}
We choose~$A$ and~$\Delta r$ such that~$0<A<\Delta r<r_{\min}/2$.
Let~$\chi_{1}$ a $\mathcal{C}^{\infty}$ decreasing function such
that\begin{align*}
\chi_{1}(r) & =1\quad\mbox{if}\quad r\leq A/2\\
 & =0\quad\mbox{if}\quad A\leq r\,.\end{align*}
Let~$f_{1}=\chi_{1}f$ and~$f_{2}=\left(1-\chi_{1}\right)f$ then\[
f*\delta^{\zeta}=f_{1}\underset{\mathcal{E}',\mathcal{C}^{\infty}}{*}\kappa^{\zeta}+f_{2}\underset{\mathcal{S},L^{1}}{*}\kappa^{\zeta}\,.\]
Since~$\partial_{r}^{k}\left(f_{2}*\kappa^{\zeta}\right)=\left(\partial_{r}^{k}f_{2}\right)*\kappa^{\zeta}$,
Lemma~\ref{pro:Approx-by-convolution} gives, for the second term,\[
\left\Vert \left(\partial_{r}^{k}f_{2}\right)*\kappa^{\zeta}-\pi\partial_{r}^{k}f_{2}\right\Vert _{L^{\infty}}\leq C_{\gamma}\big(\big\| f_{2}^{\left(k\right)}\big\|_{\infty}+\big\| f_{2}^{\left(k+1\right)}\big\|_{\infty}\big)\zeta^{\gamma}\,.\]
We are only interested in~$r\in[r_{\min},r_{\max}]$ with~$0<r_{\min}<r_{\max}$
when evaluating $\partial_{r}^{k}(f*\kappa^{\zeta})$. We insert another
cutoff function~$\chi_{2}\in\mathcal{C}_{0}^{\infty}(\mathbb{R})$
such that\begin{alignat*}{3}
\chi_{2}(r) & =0 & \qquad & \mbox{if}\qquad & r & \leq r_{\min}-2\Delta r\\
 & =1 &  & \mbox{if} & r_{\min}-\Delta r\leq r & \leq r_{\max}+\Delta r\\
 & =0 &  & \mbox{if} & r_{\max}+2\Delta r\leq r\end{alignat*}
Then~$f_{1}*\kappa^{\zeta}=f_{1}*\chi_{2}\kappa^{\zeta}+f_{1}*\left(1-\chi_{2}\right)\kappa^{\zeta}$
and our hypotheses on the supports give\begin{align*}
\Supp\{f_{1}*(1-\chi_{2})\kappa^{\zeta}\} & \subset\Supp f_{1}+\Supp(1-\chi_{2})\\
 & \subset\mathbb{R}\setminus[r_{\min}-\Delta r+A,r_{\max}+\Delta r]\,.\end{align*}
Since $A<\Delta r$ we obtain~$\left.\left[f_{1}*\left(1-\chi_{2}\right)\kappa^{\zeta}\right]\right|_{\left[r_{\min},r_{\max}\right]}=0$
and we can restrict ourselves to the computation of~$f_{1}\underset{\mathcal{E}',\mathcal{C}_{0}^{\infty}}{*}\chi_{2}\kappa^{\zeta}$.
More precisely we want to estimate\[
\Big\|\partial_{r}^{k}\big(f_{1}\underset{\mathcal{E}',\mathcal{C}_{0}^{\infty}}{*}\chi_{2}u_{\zeta}\big)\Big|_{[r_{\min},r_{\max}]}\Big\|_{L^{\infty}}\]
since~$\chi_{2}\delta=0$ and thus~$f_{1}\underset{\mathcal{E}',\mathcal{E}'}{*}\chi_{2}\delta=0$.
But the same considerations hold for the supports of the derivatives.
Thus it is sufficient to observe that we have the control\[
\big\| f_{1}\underset{L^{1},\mathcal{C}_{0}^{\infty}}{*}\partial_{r}^{k}(\chi_{2}\kappa^{\zeta})\big\|_{L^{\infty}}\leq\|f_{1}\|_{L^{1}}\|\partial^{k}(\chi_{2}\kappa^{\zeta})\|_{L^{\infty}}\leq\|f_{1}\|_{L^{1}}C_{\chi_{2}}\sup_{r\geq r_{\min}-2\Delta r}|\partial^{k}\kappa^{\zeta}|\]
where the~$\sup$ is controlled by~$C\zeta$ with $C$ only dependent
on~$\Delta r$ and~$r_{\min}$ since\[
2\partial^{k}\kappa^{\zeta}\left(r\right)=i^{k}k!\frac{-\left(ir-\zeta\right)^{k+1}+\left(ir+\zeta\right)^{k+1}}{\left(r^{2}+\zeta^{2}\right)^{k+1}}\,.\]
Consequently\[
\Big\|\partial_{r}^{k}\big[f_{1}\underset{\mathcal{E}',\mathcal{C}_{0}^{\infty}}{*}\chi_{2}\kappa^{\zeta}-f_{1}\underset{\mathcal{E}',\mathcal{E}'}{*}\chi_{2}\delta\big]\Big|_{[r_{\min},r_{\max}]}\Big\|_{L^{\infty}}\leq C\zeta\]
and this ends the proof.
\end{proof}

\subsubsection{Estimate of the error term~$\Delta_{+,3}$}
\begin{rem}
Throughout this section we will make definitions that are dependent
on the value of~$\frac{th}{\varepsilon}$. This will not be a problem
as long as~$\frac{th}{\varepsilon}\leq1$ which will be satisfied
with our choice of $\varepsilon=\varepsilon\left(h\right)\gg h$.\end{rem}
\begin{prop}
Let~$b\in\mathcal{C}_{0}^{\infty}(\mathbb{R}_{x,\xi}^{2d})$ with~$\Supp_{\xi}b\subset B_{R}\setminus B_{1/R}$
for some~$R>1$. Let~$\gamma\in(0,1)$. There exists a constant~$C_{G,R,\gamma}>0$
such that, for all~$\zeta>0$,\[
|\Delta_{+,3}|\leq\zeta^{\gamma}\,\mathcal{N}_{k}(b)\, C_{G,R,\gamma}\]
for some integer~$k=k(d)$ big enough.\end{prop}
\begin{proof}
Since $\Delta_{+,3}=\Tr\big[\hat{\boldsymbol{\rho}}_{t}^{app}\big(Q_{+,\frac{ht}{\varepsilon}}^{\zeta}b-Q_{+,\frac{ht}{\varepsilon}}b\big)^{W}\negthickspace(-hD_{\xi},\xi-\diff\Gamma_{\!\varepsilon}(\eta))\big]$
we get\begin{align*}
\left|\Delta_{+,3}\right| & \leq\big\|\big(Q_{+,\frac{ht}{\varepsilon}}^{\zeta}b-Q_{+,\frac{ht}{\varepsilon}}b\big)^{W}\negthickspace(-hD_{\xi},\xi-\diff\Gamma_{\!\varepsilon}(\eta))\big\|_{\mathcal{L}(L_{\xi}^{2}\otimes\Gamma L_{\eta}^{2})}\\
 & \leq C_{k,d}\,\mathcal{N}_{k}\big(Q_{+,\frac{ht}{\varepsilon}}^{\zeta}b-Q_{+,\frac{ht}{\varepsilon}}b\big)\end{align*}
for some integer~$k=k\left(d\right)$ big enough.

Thus we boil down to prove that for any integer~$k\geq0$ there is
a constant~$C_{k,b,G,\gamma}>0$ such that for any~$\zeta>0$\[
\mathcal{N}_{k}\big(Q_{+,\frac{ht}{\varepsilon}}^{\zeta}b-Q_{+,\frac{ht}{\varepsilon}}b\big)\leq C_{k,G,\gamma}\,\mathcal{N}_{k}(b)\,\zeta^{\gamma}\,.\]
But we have a convenient expression for~$Q_{+,\frac{ht}{\varepsilon}}^{\zeta}$\begin{align*}
Q_{+,\frac{ht}{\varepsilon}}^{\zeta}b(x,\xi) & =2\pi\int_{\mathbb{R}_{\eta}^{d}}\hat{G}(\eta)\, b\big(x-2\tfrac{ht}{\varepsilon}\eta,\xi-\eta\big)\,\kappa^{\zeta}(\eta^{.2}-2\xi.\eta)\diff\eta\\
 & =2\pi\int_{\mathbb{R}_{\eta}^{d}}\hat{G}(\xi-\eta)\, b\big(x-2\tfrac{ht}{\varepsilon}\xi+2\tfrac{ht}{\varepsilon}\eta,\eta\big)\,\kappa^{\zeta}(\eta^{.2}-2\xi.\eta)\diff\eta\\
 & =\pi\int_{\mathbb{S}_{\omega}^{d-1}}\int_{\mathbb{R}_{r}^{+}}\varphi_{\omega}(x,\xi,r)\, K^{\zeta}(r-\xi^{.2})\diff r\diff\omega\,,\end{align*}
with~$\varphi_{\omega}(x,\xi,r)=0$ for~$r\leq0$, and for~$r\geq0$,\begin{equation}
\varphi_{\omega}(x,\xi,r)=\hat{G}(\xi-\sqrt{r}\omega)\, b\big(x-2\tfrac{ht}{\varepsilon}\xi+2\tfrac{ht}{\varepsilon}\sqrt{r}\omega,\sqrt{r}\omega\big)\, r^{d/2-1}\label{eq:phi-omega}\end{equation}
defined for~$\omega\in\mathbb{S}^{d-1}$ and~$x,\,\xi\in\mathbb{R}^{d}$.
We also have a convenient expression for~$Q_{+,\frac{ht}{\varepsilon}}b$
in terms of~$\varphi_{\omega}$,\[
Q_{+,\frac{ht}{\varepsilon}}b\left(x,\xi\right)=\pi\int_{\mathbb{S}_{\omega}^{d-1}}\varphi_{\omega}(x,\xi,\xi^{.2})\diff\omega\,.\]
The conclusion is then given by Lemma~\ref{lem:phi-omega-K-phi-omega}.
\end{proof}

\begin{lem}
\label{lem:phi-omega-K-phi-omega}For any~$\gamma\in(0,1)$, uniformly
in~$\omega\in\mathbb{S}_{\omega}^{d-1}$,\[
\mathcal{N}_{k}\Big(\int_{\mathbb{R}_{r}^{+}}\varphi_{\omega}(x,\xi,r)\,\kappa^{\zeta}(r-\xi^{.2})\diff r-\varphi_{\omega}(x,\xi,\xi^{.2})\Big)\leq C_{k,G,\gamma}\,\zeta^{\gamma}\,.\]
\end{lem}
\begin{proof}
The integral can be expressed as a convolution product\[
\int_{\mathbb{R}_{r}}\varphi_{\omega}(x,\xi,r)\,\kappa^{\zeta}(r-\xi^{.2})\diff r=\big(\varphi(x,\xi,\cdot)*\kappa^{\zeta}\big)(\xi^{.2})\,.\]
Since the derivation behaves well with the difference, \emph{i.e.}\begin{multline*}
\partial_{x}^{\alpha}\partial_{\xi}^{\beta}\left(\big(\varphi_{\omega}(x,\xi,\cdot)*\kappa^{\zeta}\big)(\xi^{.2})-\varphi_{\omega}(x,\xi,\xi^{.2})\right)=\sum_{\alpha^{\prime},\beta^{\prime},\gamma^{\prime}}c_{\alpha^{\prime},\beta^{\prime},\gamma^{\prime}}2^{\left|\gamma^{\prime}\right|}\xi^{\gamma^{\prime}}\times\\
\left[\left(\big(\partial_{x}^{\alpha^{\prime}}\partial_{\xi}^{\beta^{\prime}}\partial_{r}^{\gamma^{\prime}}\varphi_{\omega}\big)(x,\xi,\cdot)*\kappa^{\zeta}\right)(\xi^{.2})-\big(\partial_{x}^{\alpha^{\prime}}\partial_{\xi}^{\beta^{\prime}}\partial_{r}^{\gamma^{\prime}}\varphi_{\omega}\big)(x,\xi,\xi^{.2})\right]\,,\end{multline*}
it suffices to apply Lemma~\ref{pro:Approx-by-convolution}.
\end{proof}
For~$\zeta>0$, and~$r\in\mathbb{R}$, let~$\kappa^{\zeta}(r)=\frac{1}{\pi}\frac{\zeta}{r^{2}+\zeta^{2}}$.
\begin{lem}
\label{pro:Approx-by-convolution}Let~$f$ be a function in the Schwartz
class. Then for any~$\gamma\in(0,1)$, a constant~$C_{\gamma}>0$
exists such that \[
\forall\zeta>0,\,\left\Vert f*\kappa^{\zeta}-f\right\Vert _{L^{\infty}}\leq\max\left\{ \left\Vert f\right\Vert _{\infty},\left\Vert f'\right\Vert _{\infty}\right\} C_{\gamma}\zeta^{\gamma}\,.\]
\end{lem}
\begin{proof}
The formula $f(r_{0}+\zeta r)-f(r_{0})=\zeta r\int_{0}^{1}f'(r_{0}+s\zeta r)\diff s$
and an interpolation with $|f(r_{0}+\zeta r)-f(r_{0})|\leq2\|f\|_{\infty}$
give for~$\gamma\in\left[0,1\right]$,\[
\left|f(r_{0}+\zeta r)-f(r_{0})\right|\leq2\max\left\{ \left\Vert f\right\Vert _{\infty},\left\Vert f'\right\Vert _{\infty}\right\} \zeta^{\gamma}\left|r\right|^{\gamma}\,.\]
So, for~$\gamma\in[0,1)$, \[
\big|\int_{\mathbb{R}}[f(r_{0}+\zeta r)-f(r_{0})]\tfrac{\diff r}{r^{2}+1}\big|\leq\max\{\|f\|_{\infty},\|f'\|_{\infty}\}C_{\gamma}\zeta^{\gamma}\]
which is the expected result.
\end{proof}

\section[Comparison between approximated and exact dynamics]{\label{par:Comparison-original-approximated}Comparisons of the
measures of an observable at a mesoscopic scale for the original and
approximated dynamics}

\begin{rem}
Let~$b\in\mathcal{C}_{0}^{\infty}(\mathbb{R}_{x,\xi}^{2d})$, $\rho\in\mathcal{L}_{1}L_{x}^{2}$
and~$t\geq0$,\begin{align*}
m(b,\rho_{t}^{\varepsilon}) & =\Tr\left[b^{W}\negmedspace\big(\!-hD_{\xi},\xi-\diff\Gamma_{\!\varepsilon}(\eta)\big)\:\hat{\boldsymbol{\rho}}_{t}\right]\,,\\
m(b,\rho_{t}^{\varepsilon,app}) & =\Tr\left[b^{W}\negmedspace\big(\!-hD_{\xi},\xi-\diff\Gamma_{\!\varepsilon}(\eta)\big)\:\hat{\boldsymbol{\rho}}_{t}^{app}\right]\,.\end{align*}
\end{rem}
\begin{defn}
Let~$b\in\mathcal{C}_{0}^{\infty}(\mathbb{R}_{x,\xi}^{2d})$, $\rho\in\mathcal{L}_{1}L_{x}^{2}$
a state, $t\geq0$ and~$\chi\in\mathcal{C}_{0}^{\infty}(\mathbb{R}_{x,\xi}^{2d})$
we define \begin{align*}
m(b,\rho,t,\chi) & =\Tr\left[\chi\big(\!\diff\Gamma_{\!\varepsilon}(\eta)\big)\: b^{W}\negmedspace\big(\!-hD_{\xi},\xi-\diff\Gamma_{\!\varepsilon}(\eta)\big)\:\chi\big(\!\diff\Gamma_{\!\varepsilon}(\eta)\big)\:\hat{\boldsymbol{\rho}}_{t}\right]\\
m^{app}(b,\rho,t,\chi) & =\Tr\left[\chi\big(\!\diff\Gamma_{\!\varepsilon}(\eta)\big)\: b^{W}\negmedspace\big(\!-hD_{\xi},\xi-\diff\Gamma_{\!\varepsilon}(\eta)\big)\:\chi\big(\!\diff\Gamma_{\!\varepsilon}(\eta)\big)\:\hat{\boldsymbol{\rho}}_{t}^{app}\right]\,.\end{align*}

\end{defn}

\begin{prop}
\label{pro:quality-of-approx-m(b,t)}Assume~$\tfrac{ht}{\varepsilon}/\sqrt{h}$.
Let~$b\in\mathcal{C}_{0}^{\infty}(\mathbb{R}_{x,\xi}^{2d})$ non-negative
such that~$\Supp_{\xi}b\subset B_{R}\setminus B_{1/R}$ for some~$R>0$,
$\rho\in\mathcal{L}_{1}^{+}L_{x}^{2}$ with $\Tr\rho\leq1$ and for~$j=1,2$,
$\chi_{j}\in\mathcal{C}_{0}^{\infty}(\mathbb{R}_{\lambda}^{d})$ with
values in~$\left[0,1\right]$, $\chi_{j}(B_{M_{j}})=\left\{ 1\right\} $
for~$M_{1}=3R$ and with~$\chi_{2}(\mathbb{R}^{d}-B_{R+1})=\left\{ 0\right\} $.
There is a constant~$C_{R,b,\chi_{1},\chi_{2}}$ (which does not
depend on~$\rho$) such that\[
m_{h}^{app}\big(b,(\rho_{\chi_{2}})_{t}^{app}\big)-m_{h}(b,\rho_{t})\leq\mathcal{E}_{\ref{par:Comparison-original-approximated}}=C_{R,b,\chi_{1},\chi_{2}}\big(h+(\tfrac{ht}{\varepsilon}/\sqrt{h})^{3}+\mathcal{E}_{\ref{par:Calculus-approximated}}\big)\]
with~$\rho_{\chi_{2}}=\chi_{2}(D_{x})\,\rho\,\chi_{2}(D_{x})$.
\end{prop}
We use the decomposition~$\mathcal{E}_{\ref{par:Comparison-original-approximated}}=\mathcal{E}_{\ref{sec:Step-1:-Truncation}}+\mathcal{E}_{\ref{sec:Step-2:-Comparison}}+\mathcal{E}_{\ref{sec:Step-3:-Release}}$
corresponding to the steps:
\begin{enumerate}
\item $m_{h}(b,\rho_{\chi_{2}},t,\chi_{1})-m_{h}(b,\rho_{t})\leq\mathcal{E}_{\ref{sec:Step-1:-Truncation}}=Ch$, 
\item $m_{h}^{app}(b,\rho_{\chi_{2}},t,\chi_{1})-m_{h}(b,\rho_{\chi_{2}},t,\chi_{1})\leq\mathcal{E}_{\ref{sec:Step-2:-Comparison}}=C(\tfrac{ht}{\varepsilon}/\sqrt{h})^{3}$, 
\item $m_{h}(b,(\rho_{\chi_{2}})_{t}^{app})-m_{h}^{app}(b,\rho_{\chi_{2}},t,\chi_{1})\leq\mathcal{E}_{\ref{sec:Step-3:-Release}}=\mathcal{E}_{\ref{par:Calculus-approximated}}+Ch$. 
\end{enumerate}

\subsection{\label{sec:Step-1:-Truncation}Step 1: Introduction of cutoffs}

We introduce cutoff functions both on the state~$\rho$ and the Wick
observable $b^{W}(-hD_{\xi},\xi-\diff\Gamma_{\varepsilon}(\eta))$.
\begin{prop}
Let~$b\in\mathcal{C}_{0}^{\infty}(\mathbb{R}_{x,\xi}^{2d})$ non-negative
such that~$\Supp_{\xi}b\subset B_{R}$ for some~$R>0$, $\rho\in\mathcal{L}_{1}^{+}L_{x}^{2}$,
$\Tr\rho\leq1$, and, for~$j=1,2,$ $\chi_{j}\in\mathcal{C}_{0}^{\infty}(\mathbb{R}_{\lambda}^{d})$
with values in~$\left[0,1\right]$ and~$\chi_{j}(B_{M_{j}})=\left\{ 1\right\} $
for some~$M_{j}>0$. Then there is a constant~$C_{b,\chi_{1},\chi_{2}}$
such that

\[
m(b,\rho_{\chi_{2}},t,\chi_{1})-m(b,\rho_{t})\leq\mathcal{E}_{\ref{sec:Step-1:-Truncation}}=C_{b,\chi_{1},\chi_{2}}h\]
with~$\rho_{\chi_{2}}=\chi_{2}(D_{x})\circ\rho\circ\chi_{2}(D_{x})$.

\end{prop}
\begin{proof}
Using the functional calculus for the self-adjoint operator~$\diff\Gamma_{\varepsilon}(\eta)$
and since\begin{align*}
b(x,\xi-\lambda) & \geq\chi_{2}(\xi)\, b(x,\xi-\lambda)\,\chi_{1}(\lambda)\,\chi_{2}(\xi)\\
 & \geq\chi_{2}(\xi)\,\sharp^{h}\, b(x,\xi-\lambda)\,\chi_{1}(\lambda)\,\sharp^{h}\,\chi_{2}(\xi)-C_{b,\chi_{1},\chi_{2}}h\end{align*}
holds uniformly in~$\lambda$, we can write\begin{multline*}
b^{W}\!\big(\!-hD_{\xi},\xi-\diff\Gamma_{\!\varepsilon}(\eta)\big)\\
\geq\chi_{2}(\xi)\circ b^{W}\!\big(\!-hD_{\xi},\xi-\diff\Gamma_{\!\varepsilon}(\eta)\big)\,\chi_{1}(\diff\Gamma_{\!\varepsilon}(\eta))\circ\chi_{2}(\xi)-C_{b,\chi_{1},\chi_{2}}h\,.\end{multline*}
And thus\begin{align*}
m(b,\rho_{t}) & =\Tr\left[b^{W}\!\big(\!-hD_{\xi},\xi-\diff\Gamma_{\!\varepsilon}(\eta)\big)\,\hat{\boldsymbol{\rho}}_{t}\right]\\
 & \geq\Tr\left[b^{W}\!\big(\!-hD_{\xi},\xi-\diff\Gamma_{\!\varepsilon}(\eta)\big)\,\chi_{1}(\diff\Gamma(\eta))\,\widehat{\boldsymbol{\rho_{\chi_{2}}}}_{t}\right]-C_{b,\chi_{1},\chi_{2}}h\end{align*}
since~$\left[H_{\varepsilon},\chi_{2}\right]=0$.
\end{proof}

\subsection{\label{sec:Step-2:-Comparison}Step 2: Comparison between truncated
solutions}
\begin{prop}
Suppose~$ $$\frac{ht}{\varepsilon}\leq\sqrt{h}$. Let~$b\in\mathcal{C}_{0}^{\infty}(\mathbb{R}_{x,\xi}^{2d})$
non-negative, $\rho\in\mathcal{L}_{1}^{+}L_{x}^{2}$, $\Tr\rho\leq1$
and~$\chi\in\mathcal{C}_{0}^{\infty}(\mathbb{R}_{\lambda}^{d})$
with values in~$\left[0,1\right]$, and~$\chi(B_{M})=\left\{ 1\right\} $
for some~$M>0$, then there is a constant~$C_{G,b,\chi}$ such that
\[
\left|m(b,\rho,t,\chi)-m^{app}(b,\rho,t,\chi)\right|\leq\mathcal{E}_{\ref{sec:Step-2:-Comparison}}=C_{G,b,\chi}\big(\tfrac{ht}{\varepsilon}/\sqrt{h}\big)^{3}\,.\]

\end{prop}
Set\begin{equation}
b_{\chi}=b\big(-hD_{\xi},\xi-\diff\Gamma_{\!\varepsilon}(\eta)\big)\:\chi(\diff\Gamma_{\!\varepsilon}(\eta))\,.\label{eq:def-b-chi}\end{equation}
We want to control the error when we consider~$\Tr\left[b_{\chi}\,\boldsymbol{\rho}_{t}^{app}\right]$
instead of $\Tr\left[b_{\chi}\,\boldsymbol{\rho}_{t}\right]$ \emph{i.e.}
we want to control~$\Tr\left[b_{\chi}\, u_{t}\right]$ with\begin{equation}
u_{t}=\boldsymbol{\rho}_{t}-\boldsymbol{\rho}_{t}^{app}\,.\label{eq:def-rho-rho-app}\end{equation}
Since $i\varepsilon\partial_{t}\boldsymbol{\rho}_{t}=[H_{\varepsilon},\boldsymbol{\rho}_{t}]$
and $i\varepsilon\partial_{t}\boldsymbol{\rho}_{t}^{app}=[H_{\varepsilon},\boldsymbol{\rho}_{t}^{app}]-[H_{\varepsilon}-H_{\varepsilon}^{app},\boldsymbol{\rho}_{t}^{app}]$,
the difference~$u_{t}$ is solution of the differential equation\[
i\varepsilon\partial_{t}u_{t}=\big[(\xi-\diff\Gamma_{\!\varepsilon}(\eta))^{2},u_{t}\big]+\big[\Phi_{\varepsilon}(f_{h,\varepsilon}),u_{t}\big]-\big[\diff\Gamma_{\!\varepsilon}(\eta)^{2}-\varepsilon\diff\Gamma_{\!\varepsilon}(\eta^{2}),\boldsymbol{\rho}_{t}^{app}\big]\]
with initial data~$u_{t=0}=0$. We can then use the integral expression\[
\Tr\left[b_{\chi}\, u_{t}\right]=-\tfrac{i}{\varepsilon}\int_{0}^{t}\Tr\big[b_{\chi}i\varepsilon\partial_{t}u_{t}\big]\diff s\,.\]

\begin{rem}
Let $\mathcal{H}$ be a Hilbert space. If~$A,\, B\in\mathcal{L}(\mathcal{H})$
and~$C\in\mathcal{L}_{1}(\mathcal{H})$, then the cyclicty of the
trace gives $\Tr\left[A\left[B\,,C\right]\right]=\Tr\left[\left[A\,,B\right]C\right]$.\end{rem}
\begin{lem}
There exists a constant~$C$ independent of~$\chi$ such that for~$b_{\chi}$
and~$u_{t}$ defined by Equations~(\ref{eq:def-b-chi}) and~(\ref{eq:def-rho-rho-app}),\end{lem}
\begin{enumerate}
\item $\left|\frac{1}{\varepsilon}\int_{0}^{t}\Tr\big[b_{\chi}\big[(\xi-\diff\Gamma_{\!\varepsilon}(\eta))^{2}\,,u_{h,\varepsilon,s}\big]\big]\diff s\right|\leq\frac{h}{\varepsilon}\int_{0}^{t}\|u_{h,\varepsilon,s}\|_{\mathcal{L}_{1}}\diff s\leq C\frac{h^{2}t^{3}}{\varepsilon^{3}}$,
\item $\frac{1}{\varepsilon}\int_{0}^{t}\Tr\left[b_{\chi}\left[\diff\Gamma_{\!\varepsilon}(\eta)^{2}-\varepsilon\diff\Gamma_{\!\varepsilon}(\eta^{2})\,,\boldsymbol{\rho}^{app}\right]\right]\diff s=0$,
\item $\left|\frac{1}{\varepsilon}\int_{0}^{t}\Tr\left[b_{\chi}\left[\Phi_{\varepsilon}(f_{h,\varepsilon})\,,u_{s}\right]\right]\diff s\right|\leq C\frac{t^{3}h^{3/2}}{\varepsilon^{7/2}}\Big(\sqrt{\varepsilon}+\sqrt{\frac{t}{2}}\sqrt{\frac{ht}{\varepsilon}}\Big)$.
\end{enumerate}

\begin{proof}
For Point~$1$, let us introduce~$\chi_{1}\succ\chi$ (\emph{i.e.}
$\chi_{1}\in\mathcal{C}_{0}^{\infty}$ with values in $[0,1]$ such
that $\chi_{1}\equiv1$ on $\Supp\chi$) in order to handle only bounded
operators:\begin{align*}
\Tr & \big[b_{\chi}\big[(\xi-\diff\Gamma_{\!\varepsilon}(\eta))^{2},u_{s}\big]\big]\\
 & =\Tr\big[b_{\chi}\big[\chi_{1}(\diff\Gamma_{\!\varepsilon}(\eta))(\xi-\diff\Gamma_{\!\varepsilon}(\eta))^{2},u_{s}\big]\big]\displaybreak[0]\\
 & =\Tr\big[\big[b_{\chi}\,,\chi_{1}(\diff\Gamma_{\!\varepsilon}(\eta))(\xi-\diff\Gamma_{\!\varepsilon}(\eta))^{2}\big]u_{s}\big]\\
 & =\Tr\big[\chi(\diff\Gamma_{\!\varepsilon}(\eta))\,\tfrac{h}{i}\{b(x,\xi),\xi^{2}\}(\negmedspace-hD_{\xi},\xi-\diff\Gamma_{\!\varepsilon}(\eta))\, u_{s}\big]\\
 & =\Tr\big[\tfrac{h}{i}\chi(\diff\Gamma_{\!\varepsilon}(\eta))\:(2\xi.b)(\negmedspace-hD_{\xi},\xi-\diff\Gamma_{\!\varepsilon}(\eta))\, u_{s}\big]\,.\end{align*}
The bound~$\big\|\chi(\diff\Gamma_{\!\varepsilon}(\eta))\,(2\xi.b)\big(\negmedspace-hD_{\xi},\xi-\diff\Gamma_{\!\varepsilon}(\eta)\big)\big\|_{\mathcal{L}L_{\xi}^{2}\otimes\Gamma L_{\eta}^{2}}\leq C$
and a time integration bring\[
\Big|\frac{1}{\varepsilon}\int_{0}^{t}\Tr\big[b_{\chi}\big[(\xi-\diff\Gamma_{\!\varepsilon}(\eta))^{2},u_{s}\big]\big]\diff s\Big|\leq C\frac{h}{\varepsilon}\int_{0}^{t}\|u_{s}\|_{\mathcal{L}_{1}L_{\xi}^{2}\otimes\Gamma L_{\eta}^{2}}\diff s\,.\]
Then we use that both~$\hat{\boldsymbol{\rho}}_{t}$ and~$\hat{\boldsymbol{\rho}}_{t}^{app}$
have the same initial value~$\rho_{0}\otimes\proj\Omega$ with~$\rho_{0}=\sum_{j}\lambda_{j}|\psi_{0,j}\rangle\langle\psi_{0,j}|$,
$\sum_{j}\lambda_{j}=\Tr\rho$, $\lambda_{j}\geq0$, $\|\psi_{0,j}\|=1$
to write \[
\boldsymbol{\rho}_{t}=\sum_{j}\lambda_{j}|\varphi_{t,j}\rangle\langle\varphi_{t,j}|\,,\qquad\boldsymbol{\rho}_{t}^{app}=\sum_{j}\lambda_{j}|\varphi_{t,j}^{app}\rangle\langle\varphi_{t,j}^{app}|\,,\]
and then~$u_{t}=\sum_{j}\lambda_{j}\bigl(|\Psi_{t,j}-\Psi_{t,j}^{app}\rangle\langle\Psi_{t,j}|-|\Psi_{t,j}^{app}\rangle\langle\Psi_{t,j}^{app}-\Psi_{t,j}|\bigr)$
and\[
\left\Vert u_{t}\right\Vert _{\mathcal{L}_{1}L_{\xi}^{2}}\leq2\sum_{j}\lambda_{j}\left\Vert \Psi_{t,j}-\Psi_{t,j}^{app}\right\Vert \leq C(\tfrac{ht}{\varepsilon}/\sqrt{h})^{2}\,.\]
This and the integral above yield the result.

For Point~2, let~$\chi_{1}\succ\chi$,\begin{align*}
\Tr & \left[b_{\chi}\left[\diff\Gamma_{\!\varepsilon}(\eta)^{2}-\varepsilon\diff\Gamma_{\!\varepsilon}(\eta^{2})\,,u_{s}\right]\right]\\
 & =\Tr\left[b_{\chi}\left[\chi_{1}(\diff\Gamma_{\!\varepsilon}(\eta))\left(\diff\Gamma_{\varepsilon}(\eta)^{2}-\varepsilon\diff\Gamma_{\!\varepsilon}(\eta^{2})\right),u_{s}\right]\right]\\
 & =\Tr\left[\left[\chi_{1}(\diff\Gamma_{\!\varepsilon}(\eta))\left(\diff\Gamma_{\!\varepsilon}(\eta)^{2}-\varepsilon\diff\Gamma_{\!\varepsilon}(\eta^{2})\right),b_{\chi}\right]u_{s}\right]\end{align*}
which vanishes since~$\left[\chi_{1}(\diff\Gamma_{\!\varepsilon}(\eta))\left(\diff\Gamma_{\!\varepsilon}(\eta)^{2}-\varepsilon\diff\Gamma_{\!\varepsilon}(\eta^{2})\right),b_{\chi}\right]=0\,.$

For Point~3, we have, with~$\Delta\hat{\Psi}_{s}=\hat{\Psi}_{s}-\hat{\Psi}_{s}^{app}$,\[
\Tr\big[b_{\chi}[\Phi_{\varepsilon}(f_{h,\varepsilon})\,,u_{s}]\big]=\langle\Delta\hat{\Psi}_{s}|\,[b_{\chi}\,,\Phi_{\varepsilon}(f_{h,\varepsilon})]\,|\hat{\Psi}_{s}\rangle+\langle\hat{\Psi}_{s}^{app}|\,[b_{\chi}\,,\Phi_{\varepsilon}(f_{h,\varepsilon})]\,|\Delta\hat{\Psi}_{s}\rangle\,.\]
Taking the modulus we obtain\begin{multline*}
\left|\Tr\left[b_{\chi}\left[\Phi_{\varepsilon}(f_{h,\varepsilon})\,,u_{s}\right]\right]\right|\leq C\|\Delta\hat{\Psi}_{s}\|\Big(\big\|\Phi_{\varepsilon}(f_{h,\varepsilon})\,\hat{\Psi}_{s}\big\|+\big\|\Phi_{\varepsilon}(f_{h,\varepsilon})\, b_{\chi}\hat{\Psi}_{s}\big\|\\
+\big\|\Phi_{\varepsilon}(f_{h,\varepsilon})\, b_{\chi}^{*}\hat{\Psi}_{s}^{app}\big\|+\big\|\Phi_{\varepsilon}(f_{h,\varepsilon})\,\hat{\Psi}_{s}^{app}\big\|\Big)\end{multline*}
and we observe that\[
\max\big\{\big\|\Phi_{\varepsilon}(f_{h,\varepsilon})\,\Psi_{s}^{\sharp}\big\|,\big\|\Phi_{\varepsilon}(f_{h,\varepsilon})\, b_{\chi}\,\hat{\Psi}_{s}^{\sharp}\big\|\big\}\leq C\big\| f_{h,\varepsilon}\big\|\big\|(\varepsilon+N_{\varepsilon})^{1/2}\hat{\Psi}_{s}^{\sharp}\big\|\]
and thus, by the number estimate~(\ref{enu:Number_estimate_Psi})
in Proposition~\ref{pro:Approximated_solution},\[
\big|\Tr\big[b_{\chi}[\Phi_{\varepsilon}(f_{h,\varepsilon})\,,u_{s}]\big]\big|\leq C\|\Delta\hat{\Psi}_{s}\|\sqrt{\frac{h}{\varepsilon}\|\hat{G}\|_{L^{1}}}\Big(\sqrt{\varepsilon}+\tfrac{s}{\sqrt{2}}\|f_{h,\varepsilon}\|_{L_{\xi}^{2}}\Big)\,.\]
A time integration gives the result.
\end{proof}

\subsection{\label{sec:Step-3:-Release}Step 3: Release of the truncation on
the symbol}
\begin{prop}
Let~$b\in\mathcal{C}_{0}^{\infty}(\mathbb{R}_{x,\xi}^{2d})$ non-negative,
such that~$\Supp_{\xi}b\subset B_{R}\setminus B_{1/R}$ for some~$R>1$,
$\rho\in\mathcal{L}_{1}^{+}L_{x}^{2}$, $\Tr\rho\leq1$, with the
support of~$\hat{\rho}$ in~$B_{R+1}^{2}$ and~$\chi\in\mathcal{C}_{0}^{\infty}(\mathbb{R}_{\lambda}^{d})$
with values in~$\left[0,1\right]$, $\chi(B_{3R})=\left\{ 1\right\} $.
There is a constant~$C_{R,b,\chi}$ such that\[
m^{app}(b,\rho,t,\chi)-m(b,\rho_{t}^{app})\geq\mathcal{E}_{\ref{sec:Step-3:-Release}}\]
with~$\mathcal{E}_{\ref{sec:Step-3:-Release}}=\mathcal{E}_{\ref{par:Calculus-approximated}}+C_{R,b,\chi}h$,
i.e.\[
\mathcal{E}_{\ref{sec:Step-3:-Release}}=C\tfrac{ht}{\varepsilon}\big(\tfrac{ht}{\varepsilon}+h+\big[h(\tfrac{ht}{\varepsilon})^{-1}\big]^{d/2-1}+h^{\nu(d,\alpha)}+h^{\gamma\beta(d,\alpha)}\big)+C_{r,b,\chi}h\,.\]
\end{prop}
\begin{proof}
We restrict the proof to the case of~$\rho=\left|\psi\right\rangle \left\langle \psi\right|$
with~$\psi\in L_{x}^{2}$ since~$\rho$ is trace class, then~$\hat{\boldsymbol{\rho}}_{t}=|\hat{\Psi}_{t}^{app}\rangle\langle\hat{\Psi}_{t}^{app}|$.
We also define a positive symbol~$b_{1}\in\mathcal{C}_{0}^{\infty}(\mathbb{R}_{\xi}^{d})$
such that~$\Supp b_{1}\subset[R^{-2},R^{2}]$ and~$b_{1}(\xi^{2})\geq b(x,\xi)$.
Then \begin{align*}
 & m(b,\rho_{t}^{app})-m^{app}(b,\rho,t,\chi)\\
 & =\Tr\big[\left(1-\chi(\diff\Gamma_{\!\varepsilon}(\eta))\right)^{1/2}\, b^{W}\!\left(-hD_{\xi},\xi-\diff\Gamma_{\!\varepsilon}(\eta)\right)\left(1-\chi\left(\diff\Gamma_{\!\varepsilon}\left(\eta\right)\right)\right)^{1/2}\hat{\boldsymbol{\rho}}_{t}\big]\displaybreak[0]\\
 & \leq\Tr\big[b_{1}^{W}\!\big((\xi-\diff\Gamma_{\!\varepsilon}(\eta))^{.2}\big)\,\big(1-\chi(\diff\Gamma_{\!\varepsilon}(\eta))\big)\, b_{1}^{W}\!\big((\xi-\diff\Gamma_{\!\varepsilon}(\eta))^{.2}\big)\hat{\boldsymbol{\rho}}_{t}\big]+\mathcal{O}(h)\end{align*}
with~$\hat{\Psi}_{t}^{app}(\xi)=1_{\left[0,M\right]}(|\xi|)\,\hat{\Psi}_{t}^{app}(\xi)$
and~$\Supp b_{1}\subset[R^{-2},R^{2}]$. Then we decompose\[
\hat{\Psi}_{t}^{app}=1_{\left[1/2R,2R\right]}(|\xi|)\,\hat{\Psi}_{t}^{app}+1_{\left[0,M\right]\backslash\left[1/2R,2R\right]}(|\xi|)\,\hat{\Psi}_{t}^{app}=\hat{\Psi}_{t,1}^{app}+\hat{\Psi}_{t,2}^{app}\,.\]
With~$A=b_{1}^{W}\!\big(\left(\xi-\diff\Gamma_{\!\varepsilon}(\eta)\right)^{.2}\big)\,\left(1-\chi(\diff\Gamma_{\!\varepsilon}(\eta))\right)\, b_{1}^{W}\!\big(\left(\xi-\diff\Gamma_{\!\varepsilon}(\eta)\right)^{.2}\big)\geq0$
we have the estimate\[
\Tr\big[A\,|\hat{\Psi}_{t}^{app}\rangle\langle\hat{\Psi}_{t}^{app}|\big]\leq2\Tr\big[A\,|\hat{\Psi}_{t,1}^{app}\rangle\langle\hat{\Psi}_{t,1}^{app}|\big]+2\Tr\big[A\,|\hat{\Psi}_{t,2}^{app}\rangle\langle\hat{\Psi}_{t,2}^{app}|\big]\,.\]
The first term vanishes since\begin{multline*}
\Tr\Big[b_{1}^{W}\negmedspace\big((\xi-\diff\Gamma_{\!\varepsilon}(\eta))^{.2}\big)\big(1-\chi(\diff\Gamma_{\!\varepsilon}(\eta))\big)b_{1}^{W}\negmedspace\big((\xi-\diff\Gamma_{\!\varepsilon}(\eta))^{.2}\big)\,|\hat{\Psi}_{t,1}^{app}\rangle\langle\hat{\Psi}_{t,1}^{app}|\Big]\\
=\Tr\Big[1_{\left[1/2R,2R\right]}(|\xi|)\, b_{1}^{W}\negmedspace\big((\xi-\diff\Gamma_{\!\varepsilon}(\eta))^{.2}\big)\left(1-\chi(\diff\Gamma_{\!\varepsilon}(\eta))\right)\\
b_{1}^{W}\negmedspace\big((\xi-\diff\Gamma_{\!\varepsilon}(\eta))^{.2}\big)1_{\left[1/2R,2R\right]}(|\xi|)\,|\hat{\Psi}_{t,1}^{app}\rangle\langle\hat{\Psi}_{t,1}^{app}|\Big]\end{multline*}
and~$\left|\xi\right|\in\left[1/2R,2R\right]$, $\left|\xi-\diff\Gamma_{\!\varepsilon}(\eta)\right|\leq R$
implies $\left|\diff\Gamma_{\!\varepsilon}(\eta)\right|\leq3R$ and~$\chi(B_{3R})=\left\{ 1\right\} $.
For the second term,\begin{multline*}
\Tr\left[b_{1}^{W}\negmedspace\big((\xi-\diff\Gamma_{\!\varepsilon}(\eta))^{.2}\big)\left(1-\chi(\diff\Gamma_{\!\varepsilon}(\eta))\right)b_{1}^{W}\negmedspace\big((\xi-\diff\Gamma_{\!\varepsilon}(\eta))^{.2}\big)|\hat{\Psi}_{t,2}^{app}\rangle\langle\hat{\Psi}_{t,2}^{app}|\right]\\
\leq\Tr\left[b_{1}^{W}\big((\xi-\diff\Gamma_{\!\varepsilon}(\eta))^{.2}\big)|\hat{\Psi}_{t,2}^{app}\rangle\langle\hat{\Psi}_{t,2}^{app}|\right]\end{multline*}
since~$1-\chi(d\Gamma_{\varepsilon}(\eta))\leq\Id$. Then we use
the computation of the evolution of a symbol of~$|\xi|^{2}$ in the
case of the approximated equation as in Remark~\ref{rem:result-calculus-approximation-indep-x}
to get that, since~$b_{1}=b_{1}(|\xi|^{2})$ it is unchanged under
the evolution, and\begin{multline*}
\Tr\left[b_{1}^{W}\negmedspace\big((\xi-\diff\Gamma_{\!\varepsilon}(\eta))^{.2}\big)^{2}|\hat{\Psi}_{t,2}^{app}\rangle\langle\hat{\Psi}_{t,2}^{app}|\right]\\
\leq\Tr\left[b_{1}^{W}\negmedspace\big((\xi-\diff\Gamma_{\!\varepsilon}(\eta))^{.2}\big)^{2}|\hat{\psi}_{0,2}\otimes\Omega\rangle\langle\hat{\psi}_{0,2}\otimes\Omega|\right]+\mathcal{E}_{\ref{par:Calculus-approximated}}\end{multline*}
which brings the result observing that\[
\Tr\left[b_{1}^{W}\negmedspace\big((\xi-\diff\Gamma_{\!\varepsilon}(\eta))^{.2}\big)^{2}|\hat{\psi}_{0,2}\otimes\Omega\rangle\langle\hat{\psi}_{0,2}\otimes\Omega|\right]=\Tr\left[b_{1}^{W}\!(\xi^{.2})^{2}|\hat{\psi}_{0,2}\otimes\Omega\rangle\langle\hat{\psi}_{0,2}\otimes\Omega|\right]\]
vanishes since~$\Supp b_{1}\cap\Supp\hat{\psi}_{0,2}=\emptyset$.
\end{proof}

\section{\label{par:Gluing-estimates}The derivation of the Boltzmann equation
for the model}
\begin{prop}
\label{pro:estimate-below-m-b-rho}Let~$b\in\mathcal{C}_{0}^{\infty}(\mathbb{R}_{x,\xi}^{2d})$
with~$\Supp_{\xi}b\subset B_{R}\setminus B_{1/R}$. Let~$\rho$
a state and~$T>0$ then\[
\liminf_{h\to0}\left(m\left(\mathcal{B}^{T}(T)\, b,\rho\right)-m(b,\rho_{N,\Delta t}^{h})\right)\leq0\]
for a fixed~$\alpha\in(\frac{3}{4},1)$, $\Delta t=\Delta t(h)=h^{\alpha}$
and~$N(h)\,\Delta t(h)=T$.\end{prop}
\begin{lem}
With~$b_{t}=e^{tQ}e^{2t\xi.\partial_{x}}b$, and the hypotheses of
Proposition~\ref{pro:estimate-below-m-b-rho},\begin{multline*}
m(b_{\Delta t},\rho)-m(b,\rho_{\Delta t}^{h})\\
\leq C\big(h+(\Delta t/\sqrt{h})^{3}+(\Delta t/\sqrt{h})^{4}+\Delta t\big(\Delta t+h+(h/\Delta t)^{\tfrac{d}{2}-1}+h^{\mu}\big)\big)\,.\end{multline*}
\end{lem}
\begin{proof}
We recall that~$\rho_{\Delta t}^{h}=\rho_{\varepsilon\Delta t/h}^{\varepsilon}$
so that with~$\frac{ht}{\varepsilon}=\Delta t$, from Section~\ref{par:Comparison-original-approximated},\begin{align*}
m & \big(b,(\rho_{\chi_{2}})_{\Delta t}^{h,app}\big)-m(b,\rho_{\Delta t}^{h})\\
 & =m\big(b,(\rho_{\chi_{2}})_{t}^{\varepsilon,app}\big)-m(b,\rho_{t}^{\varepsilon})\\
 & \leq C\,\Big(h+\big(\tfrac{ht}{\varepsilon}/\sqrt{h}\big)^{3}+\big(\tfrac{ht}{\varepsilon}/\sqrt{h}\big)^{4}+\tfrac{ht}{\varepsilon}\big(\tfrac{ht}{\varepsilon}+h+(\varepsilon/t)^{d/2-1}+h^{\mu}\big)\Big)\\
 & \leq C\Big(h+(\Delta t/\sqrt{h})^{3}+(\Delta t/\sqrt{h})^{4}+\Delta t\big(\Delta t+h+(h/\Delta t)^{d/2-1}+h^{\mu}\big)\Big)\end{align*}
and from Section~\ref{par:Calculus-approximated} also used with~$\frac{ht}{\varepsilon}=\Delta t$
we get \[
m(b_{t},\rho_{\chi_{2}})-m\big(b,\left(\rho_{\chi_{2}}\right)_{t}^{\varepsilon,app}\big)\leq\mathcal{E}_{\ref{par:Calculus-approximated}}\leq\mathcal{E}_{\ref{par:Comparison-original-approximated}}\]
and this term will be in particular controlled if we control the previous
one. Finally from the conservation of the support in~$\xi$ of the
symbol by the approximated Boltzmann equation we get\[
m(b_{t},\rho)-m(b_{t},\rho_{\chi_{2}})\leq\mathcal{O}(h^{\infty})\]
for~$\chi_{2}$ a cutoff function chosen so that~$\chi_{2}(B_{R})=\left\{ 1\right\} $.

Thus we fix, for~$j=1,2$, two cutoff functions~$\chi_{j}\in\mathcal{C}_{0}^{\infty}(\mathbb{R}_{\lambda}^{d})$
with values in~$\left[0,1\right]$, $\chi_{j}(B_{M_{j}})=\left\{ 1\right\} $
for~$M_{1}=3R$ and~$M_{2}=1$ and with~$\chi_{2}(\mathbb{R}^{d}\setminus B_{R+1})=\left\{ 0\right\} $.
\end{proof}

\begin{proof}[Proof of Propostition~\ref{pro:estimate-below-m-b-rho}]
Let, for~$k\in\mathbb{N}$, $\Delta t>0$, $b_{k,\Delta t}=\left(e^{\Delta tQ}e^{2\Delta t\xi.\partial_{x}}\right)^{k}b$.
Iterating the estimation of the Lemma~$N(h)$ times brings\begin{multline*}
m\!\big(b_{N,\Delta t},\rho\big)-m\!\big(b,\rho_{N(h),\varepsilon\Delta t/h}^{\varepsilon}\big)\\
\leq CN\Big(h+(\Delta t/\sqrt{h})^{3}+(\Delta t/\sqrt{h})^{4}+\Delta t\big(\sqrt{\Delta t}+h+(h/\Delta t)^{\tfrac{d}{2}-1}\!+h^{\mu}\big)\Big)\end{multline*}
with~$N\Delta t=T$ and~$h^{\alpha}\leq\frac{ht}{\varepsilon}=\Delta t\leq1$
for some $\alpha\in(1/2,1)$. Thus we can choose $\Delta t=\frac{th}{\varepsilon}=h^{\alpha}$
and thus $N=Th^{-\alpha}$. Then we get the estimate\begin{align*}
m & \!\big(b_{N,\Delta t},\rho\big)-m\!\big(b,\rho_{N,\varepsilon\Delta t/h}\big)\\
 & \leq CTh^{-\alpha}\big(h+h^{3\alpha-3/2}+h^{4\alpha-2}+h^{\alpha}(h^{\alpha/2}+h+h^{(1-\alpha)(d/2-1)}+h^{\mu})\big)\\
 & \leq CT\: o_{h\to0}(1)\,,\end{align*}
for~$\alpha\in(\frac{3}{4},1)$. Finally it suffices to prove that\[
\lim_{h\to0}m\!\left(b_{N(h),\Delta t(h)},\rho\right)=m\!\left(b_{T},\rho\right)\]
which is true since the estimates of Proposition~\ref{pro:trotter-boltzmann}
prove that, for some constant~$C>0$, $\left\Vert b_{N,\Delta t}-b_{T}\right\Vert _{\mathcal{L}L_{x}^{2}}\leq\frac{C}{N}$.
\end{proof}
\appendix

\section*{}

\begin{acknowledgement*}
We would like to thank Francis Nier and Zied Ammari for very helpful
discussions, and many remarks and comments on this paper.
\end{acknowledgement*}

\nomenclature[0commutator]{$[A,B]$}{the commutator~$AB-BA$ of two operators}\nomenclature[0poisson bracket]{$\{f,g\}$}{the poisson bracket~$\partial_\xi f.\partial_x g -\partial_x f.\partial_\xi g$ of two functions on~$\mathbb R^{2d}_{x,\xi}$}\nomenclature[0symmetric tensor product]{$\vee,\,\bigvee$}{the symmetric tensor product}\nomenclature[0moyal]{$b_1\sharp \,b_2$}{the Moyal product of two symbols}

\nomenclature[boule de rayon R]{$B_R$}{the closed ball of radius $R$}

\nomenclature[complex numbers]{$\mathbb C$}{the field of complex numbers}\nomenclature[Cinfini_b]{$\mathcal C_b^\infty$}{the functions of class $\mathcal C^\infty$ bounded, with bounded derivatives}\nomenclature[C0infini]{$\mathcal C_\infty^0(X;\mathbb R)$}{the continuous function vanishing at infinity (on a locally compact, Hausdorff space~$X$)}\nomenclature[D(R)]{$\mathcal C^\infty_0(X)$}{with~$X$ an open subset of~$\mathbb R^{2d}_{x,\xi}$: the real valued functions on~$X$ of class~$\mathcal C^\infty$ with compact support \\\\with~$X$ an open subset of~$\mathbb R^d_x$: the complex valued functions on~$X$ of class~$\mathcal C^\infty$ with compact support }

\nomenclature[Deltax]{$\Delta_x$}{the Laplacian operator on~$L^2_x$}\nomenclature[Dx]{$D_x$}{$=-i\partial_x$}\nomenclature[D(A)]{$D(A)$}{the domain of an operator $A$}

\nomenclature[Fourier]{$\mathcal Fu,\,\hat u$}{the Fourier transform, $\mathcal Fu\left(\xi\right)=\int_{\mathbb R^d_x}e^{-ix.\xi}u\left(x\right)\diff x$}

\nomenclature[L]{$\mathcal L(\mathcal H_1,\mathcal H_2)$}{the continuous linear applications between the Hilbert spaces $\mathcal H_1$ and $\mathcal H_2$}\nomenclature[L1]{$\mathcal L_1(\mathcal H)$}{the trace class operators on a Hilbert space $\mathcal H$}\nomenclature[L1+]{$\mathcal L_1^+(\mathcal H)$}{the positive trace class operators on a Hilbert space $\mathcal H$}\nomenclature[L2]{$\mathcal L_2(\mathcal H)$}{the Hilbert-Schmidt operators on a Hilbert space $\mathcal H$}\nomenclature[L2x]{$L^2_x$}{$=L^2(\mathbb R^d_x;\mathbb C)$}\nomenclature[L2xi]{$L^2_\xi$}{$=L^2(\mathbb R^d_\xi;\mathbb C)$}\nomenclature[L2x-xi]{$L^2_{x,\xi}$}{$=L^2(\mathbb R^d_x\times\mathbb R^d_\xi;\mathbb R)$}

\nomenclature[Mb]{$\mathcal{M}_{b}\left(X;\mathbb{R}\right)$}{the set of Radon measures on locally compact, Hausdorff space~$X$ }

\nomenclature[natural numbers]{$\mathbb N$}{the non-negative integers}

\nomenclature[real numbers]{$\mathbb R$}{the field of real numbers}

\nomenclature[Supp]{$\Supp f$}{the support of a function~$f$}\nomenclature[Sd]{$\mathbb S^{d-1}$}{the unit sphere for the euclidean norm in~$\mathbb R^d$}

\nomenclature[taux]{$\tau_xf$}{the translation~$f(\cdot -x)$ of~$x\in\mathbb R^d$ of a function~$f$ with variable in~$\mathbb R^d$}

\bibliographystyle{plain}
\bibliography{biblio_breteaux_linear_boltzmann,biblio_breteaux_linear_boltzmann_2}

\begin{thebibliography}{10}

\bibitem{MR2465733}
Zied Ammari and Francis Nier.
\newblock Mean field limit for bosons and infinite dimensional phase-space
  analysis.
\newblock {\em Ann. Henri Poincar\'e}, 9(8):1503--1574, 2008.

\bibitem{MR2513969}
Zied Ammari and Francis Nier.
\newblock Mean field limit for bosons and propagation of {W}igner measures.
\newblock {\em J. Math. Phys.}, 50(4):042107, 16, 2009.

\bibitem{MR2312948}
St{{\'e}}phane Attal and Alain Joye.
\newblock Weak coupling and continuous limits for repeated quantum
  interactions.
\newblock {\em J. Stat. Phys.}, 126(6):1241--1283, 2007.

\bibitem{MR2205464}
St{{\'e}}phane Attal and Yan Pautrat.
\newblock From repeated to continuous quantum interactions.
\newblock {\em Ann. Henri Poincar{\'e}}, 7(1):59--104, 2006.

\bibitem{MR1888863}
Guillaume Bal, George Papanicolaou, and Leonid Ryzhik.
\newblock Radiative transport limit for the random {S}chr{\"o}dinger equation.
\newblock {\em Nonlinearity}, 15(2):513--529, 2002.

\bibitem{MR2205910}
Philippe Bechouche, Fr{\'e}d{\'e}ric Poupaud, and Juan Soler.
\newblock Quantum transport and {B}oltzmann operators.
\newblock {\em J. Stat. Phys.}, 122(3):417--436, 2006.

\bibitem{MR0208930}
Feliks~A. Berezin.
\newblock {\em The method of second quantization}.
\newblock Academic Press, New York, 1966.

\bibitem{MR725107}
Carlo Boldrighini, Leonid~A. Bunimovich, and Yakov~G. Sina{\u\i}.
\newblock On the {B}oltzmann equation for the {L}orentz gas.
\newblock {\em J. Statist. Phys.}, 32(3):477--501, 1983.

\bibitem{MR1441540}
Ola Bratteli and Derek~W. Robinson.
\newblock {\em Operator algebras and quantum statistical mechanics. 2}.
\newblock Texts and Monographs in Physics. Springer-Verlag, Berlin, second
  edition, 1997.
\newblock Equilibrium states. Models in quantum statistical mechanics.

\bibitem{MR1627111}
Nicolas Burq.
\newblock Mesures semi-classiques et mesures de d{\'e}faut.
\newblock {\em Ast{\'e}risque}, (245):Exp.\ No.\ 826, 4, 167--195, 1997.
\newblock S{{\'e}}minaire Bourbaki, Vol. 1996/97.

\bibitem{MR2165532}
Thomas Chen.
\newblock Localization lengths and {B}oltzmann limit for the {A}nderson model
  at small disorders in dimension 3.
\newblock {\em J. Stat. Phys.}, 120(1-2):279--337, 2005.

\bibitem{MR792484}
Robert Dautray and Jacques-Louis Lions.
\newblock {\em Analyse math\'ematique et calcul num\'erique pour les sciences
  et les techniques. {T}ome 1}.
\newblock Collection du Commissariat \`a l'\'Energie Atomique: S\'erie
  Scientifique. [Collection of the Atomic Energy Commission: Science Series].
  Masson, Paris, 1984.

\bibitem{MR902802}
Robert Dautray and Jacques-Louis Lions.
\newblock {\em Analyse math\'ematique et calcul num\'erique pour les sciences
  et les techniques. {T}ome 3}.
\newblock Collection du Commissariat \`a l'\'Energie Atomique: S\'erie
  Scientifique. [Collection of the Atomic Energy Commission: Science Series].
  Masson, Paris, 1985.

\bibitem{MR2283953}
L{\'a}szl{\'o} Erd\H{o}s, Manfred Salmhofer, and Horng-Tzer Yau.
\newblock Quantum diffusion of the random {S}chr{\"o}dinger evolution in the
  scaling limit. {II}. {T}he recollision diagrams.
\newblock {\em Comm. Math. Phys.}, 271(1):1--53, 2007.

\bibitem{MR2413135}
L{\'a}szl{\'o} Erd\H{o}s, Manfred Salmhofer, and Horng-Tzer Yau.
\newblock Quantum diffusion of the random {S}chr{\"o}dinger evolution in the
  scaling limit.
\newblock {\em Acta Math.}, 200(2):211--277, 2008.

\bibitem{MR1744001}
L\'aszl\'o Erd{\"o}s and Horng-Tzer Yau.
\newblock Linear {B}oltzmann equation as the weak coupling limit of a random
  {S}chr\"odinger equation.
\newblock {\em Comm. Pure Appl. Math.}, 53(6):667--735, 2000.

\bibitem{MR2436991}
Gerald~B. Folland.
\newblock {\em Quantum field theory}, volume 149 of {\em Mathematical Surveys
  and Monographs}.
\newblock American Mathematical Society, Providence, RI, 2008.
\newblock A tourist guide for mathematicians.

\bibitem{PhysRev.185.308}
Giovanni Gallavotti.
\newblock Divergences and the approach to equilibrium in the lorentz and the
  wind-tree models.
\newblock {\em Phys. Rev.}, 185(1):308--322, Sep 1969.

\bibitem{MR1131589}
Patrick G{{\'e}}rard.
\newblock Mesures semi-classiques et ondes de {B}loch.
\newblock In {\em S{\'e}minaire sur les \'{E}quations aux {D}{\'e}riv{\'e}es
  {P}artielles, 1990--1991}, pages Exp.\ No.\ XVI, 19. {\'E}cole Polytech.,
  Palaiseau, 1991.

\bibitem{MR1135919}
Patrick G{{\'e}}rard.
\newblock Microlocal defect measures.
\newblock {\em Comm. Partial Differential Equations}, 16(11):1761--1794, 1991.

\bibitem{MR1438151}
Patrick G{\'e}rard, Peter~A. Markowich, Norbert~J. Mauser, and Fr{\'e}d{\'e}ric
  Poupaud.
\newblock Homogenization limits and {W}igner transforms.
\newblock {\em Comm. Pure Appl. Math.}, 50(4):323--379, 1997.

\bibitem{MR1721376}
Patrick G{\'e}rard, Peter~A. Markowich, Norbert~J. Mauser, and Fr{\'e}d{\'e}ric
  Poupaud.
\newblock Erratum: ``{H}omogenization limits and {W}igner transforms'' [{C}omm.
  {P}ure {A}ppl. {M}ath. {\bf 50} (1997), no. 4, 323--379; {MR}1438151
  (98d:35020)].
\newblock {\em Comm. Pure Appl. Math.}, 53(2):280--281, 2000.

\bibitem{MR530915}
Jean Ginibre and Giorgio Velo.
\newblock The classical field limit of scattering theory for nonrelativistic
  many-boson systems. {I}.
\newblock {\em Comm. Math. Phys.}, 66(1):37--76, 1979.

\bibitem{MR539736}
Jean Ginibre and Giorgio Velo.
\newblock The classical field limit of scattering theory for nonrelativistic
  many-boson systems. {II}.
\newblock {\em Comm. Math. Phys.}, 68(1):45--68, 1979.

\bibitem{MR602197}
Jean Ginibre and Giorgio Velo.
\newblock The classical field limit of nonrelativistic bosons. {I}. {B}orel
  summability for bounded potentials.
\newblock {\em Ann. Physics}, 128(2):243--285, 1980.

\bibitem{MR605198}
Jean Ginibre and Giorgio Velo.
\newblock The classical field limit of nonrelativistic bosons. {II}.
  {A}symptotic expansions for general potentials.
\newblock {\em Ann. Inst. H. Poincar\'e Sect. A (N.S.)}, 33(4):363--394, 1980.

\bibitem{MR887102}
James Glimm and Arthur Jaffe.
\newblock {\em Quantum physics}.
\newblock Springer-Verlag, New York, second edition, 1987.
\newblock A functional integral point of view.

\bibitem{MR0332046}
Klaus Hepp.
\newblock The classical limit for quantum mechanical correlation functions.
\newblock {\em Comm. Math. Phys.}, 35:265--277, 1974.

\bibitem{MR1223523}
Ting-Guo Ho, Lawrence~J. Landau, and A.~J. Wilkins.
\newblock On the weak coupling limit for a {F}ermi gas in a random potential.
\newblock {\em Rev. Math. Phys.}, 5(2):209--298, 1993.

\bibitem{MR1251718}
Pierre-Louis Lions and Thierry Paul.
\newblock Sur les mesures de {W}igner.
\newblock {\em Rev. Mat. Iberoamericana}, 9(3):553--618, 1993.

\bibitem{MR1872698}
Andr{\'e} Martinez.
\newblock {\em An introduction to semiclassical and microlocal analysis}.
\newblock Universitext. Springer-Verlag, New York, 2002.

\bibitem{MR1996779}
Fr{{\'e}}d{{\'e}}ric Poupaud and Alexis Vasseur.
\newblock Classical and quantum transport in random media.
\newblock {\em J. Math. Pures Appl. (9)}, 82(6):711--748, 2003.

\bibitem{MR0493420}
Michael Reed and Barry Simon.
\newblock {\em Methods of modern mathematical physics. {II}. {F}ourier
  analysis, self-adjointness}.
\newblock Academic Press [Harcourt Brace Jovanovich Publishers], New York,
  1975.

\bibitem{MR529429}
Michael Reed and Barry Simon.
\newblock {\em Methods of modern mathematical physics. {III}}.
\newblock Academic Press [Harcourt Brace Jovanovich Publishers], New York,
  1979.
\newblock Scattering theory.

\bibitem{MR0489552}
Barry Simon.
\newblock {\em The {$P(\phi )_{2}$} {E}uclidean (quantum) field theory}.
\newblock Princeton University Press, Princeton, N.J., 1974.
\newblock Princeton Series in Physics.

\bibitem{MR0471824}
Herbert Spohn.
\newblock Derivation of the transport equation for electrons moving through
  random impurities.
\newblock {\em J. Statist. Phys.}, 17(6):385--412, 1977.

\bibitem{RevModPhys.52.569}
Herbert Spohn.
\newblock Kinetic equations from hamiltonian dynamics: Markovian limits.
\newblock {\em Rev. Mod. Phys.}, 52(3):569--615, Jul 1980.

\end{thebibliography}

\end{document}